\documentclass[11pt, reqno]{amsart}
\usepackage{amsmath,amsfonts,amssymb,amsthm, 
a4wide, 
url, mathscinet,mathdots}
\usepackage{graphicx}
\usepackage{float}

\usepackage
[colorlinks = true,
            linkcolor = blue,
            urlcolor  = blue,
            citecolor = blue,
            anchorcolor = blue]{hyperref}
\usepackage[usenames]{color}
\usepackage[nobysame,initials]{amsrefs}

\usepackage{xcolor}
\newtheorem{theorem}{Theorem}[section]
\newtheorem{lemma}[theorem]{Lemma}

\newtheorem{proposition}[theorem]{Proposition}
\newtheorem{corollary}[theorem]{Corollary}

\newtheorem{remark}[theorem]{Remark}
\newtheorem{obs}[theorem]{Observation}

\newtheorem{question}[theorem]{Question}
\newtheorem*{notat*}{Notation}

\newcommand{\N}{\mathbb{N}}
\newcommand{\Z}{\mathbb{Z}}
\newcommand{\R}{\mathbb{R}}

\newcommand{\E}{\mathbb{E}}
\newcommand{\eps}{\varepsilon}
\DeclareMathOperator{\sign}{sign}
\DeclareMathOperator{\supp}{supp}
\DeclareMathOperator{\lip}{Lip}
\DeclareMathOperator{\Dob}{Dob}

\DeclareMathOperator{\GE}{GE}
\DeclareMathOperator{\TV}{TV}

\DeclareMathOperator{\MG}{MG}
\DeclareMathOperator{\AS}{\mathcal{AS}}
\DeclareMathOperator{\comp}{comp}

\renewcommand{\P}{\mathbb{P}}

\newcommand{\osc}{\text{OSC}}
\newcommand{\esc}{\text{ESC}}
\newcommand{\nesc}{\text{NESC}}

\newcommand\mysubstack[2]{\genfrac{}{}{0pt}{}{#1}{#2}}

\DeclareMathOperator{\diam}{diam}
\DeclareMathOperator{\wid}{wid}
\DeclareMathOperator{\med}{med}

\DeclareMathOperator{\cov}{Cov}

\numberwithin{equation}{section}

\begin{document}

\title{Non-constant ground configurations in the disordered ferromagnet}

\author[M. Bassan]{Michal Bassan}
\address[M. Bassan]{School of Mathematical Sciences, Tel Aviv University, Tel Aviv, Israel}
\email{michalbassan@mail.tau.ac.il}

\author[S. Gilboa]{Shoni Gilboa}
\address[S. Gilboa]{Department of Mathematics, The Open University of Israel, Raanana, Israel}
\email{shoni@openu.ac.il}

\author[R. Peled]{Ron Peled}
\address[R. Peled]{School of Mathematical Sciences, Tel Aviv University, Tel Aviv, Israel.\hfill\break
School of Mathematics, Institute for Advanced Study and Department of Mathematics, Princeton University, New Jersey, United States}
\email{peledron@tauex.tau.ac.il}

\maketitle

\begin{abstract}
The disordered ferromagnet is a disordered version of the ferromagnetic Ising model in which the coupling constants are non-negative quenched random. A ground configuration is an infinite-volume configuration whose energy cannot be reduced by finite modifications. It is a long-standing challenge to ascertain whether the disordered ferromagnet on the $\Z^D$ lattice admits non-constant ground configurations. We answer this affirmatively in dimensions $D\ge 4$, when the coupling constants are sampled independently from a sufficiently concentrated distribution. The obtained ground configurations are further shown to be translation-covariant with respect to $\Z^{D-1}$ translations of the disorder.

Our result is proved by showing that the finite-volume interface formed by Dobrushin boundary conditions is localized, and converges to an infinite-volume interface. This may be expressed in purely combinatorial terms, as a result on the fluctuations of certain minimal cutsets in the lattice $\Z^D$ endowed with independent edge capacities.
\end{abstract}

\section{Introduction}

\noindent{\bf Disordered ferromagnet.}
The Ising model is among the most basic models of statistical physics. On the hypercubic lattice $\Z^D$, it is described by the formal Hamiltonian
\begin{equation}
H^{\text{Ising}}(\sigma):=-\sum_{\{x,y\}\in E(\Z^D)}\sigma_x\sigma_y
\end{equation}
on spin \emph{configurations} $\sigma\colon \Z^D\to\{-1,1\}$, where we write $E(\Z^D)$ for the edge set of $\Z^D$.

In this paper we study \emph{the disordered ferromagnet} (or \emph{ferromagnetic random-bond Ising model}), a version of the Ising model described by the formal Hamiltonian 
\begin{equation}\label{eq:Disordered ferromagnet formal Hamiltonian}
H^{\eta}(\sigma):=-\sum_{\{x,y\}\in E(\Z^D)}\eta_{\{x,y\}}\sigma_x\sigma_y,
\end{equation}
in which the coupling field $\eta = (\eta_e)_{e\in E(\Z^D)}$ is a \emph{non-negative} quenched random field. We refer to $\eta$ as the (quenched) \emph{disorder} and restrict throughout to the case that it is an \emph{independent} field. We first consider the \emph{isotropic} case, in which each $\eta_e$ is sampled independently from the same probability distribution $\nu$, termed \emph{the disorder distribution}.

\smallskip
\noindent{\bf Ground configurations.}
Our primary interest is in the set of \emph{ground configrations}, or zero-temperature configurations, of the disordered ferromagnet\footnote{The terminology \emph{ground states} is commonly used in the literature, but with two different meanings: while it may refer to the ground configurations defined here, it may also refer to probability measures over ground configurations. We use the terminology of ground configurations to avoid this ambiguity.}.
Precisely, a configuration $\sigma$ is said to be a ground configuration for the coupling field $\eta$ if it holds that $H^\eta(\sigma)\le H^{\eta}(\sigma')$ for every configuration $\sigma'$ which differs from $\sigma$ in \emph{finitely} many places. To make sense of the definition note that although $H^\eta(\sigma)$ is ill-defined, as a non-convergent infinite sum, the difference $H^\eta(\sigma)-H^\eta(\sigma')$ is well defined whenever $\sigma$ and $\sigma'$ differ in finitely many places.

As the coupling constants are non-negative, it is clear that the constant configurations, $\sigma\equiv +1$ and $\sigma\equiv -1$, are ground configurations of the disordered ferromagnet. We study the following basic challenge:
\begin{question}\label{q:non-constant ground configurations}
Does the disordered ferromagnet admit non-constant ground configurations?
\end{question}
Ergodicity implies that existence of non-constant ground configurations is an event of probability $0$ or $1$ for each dimension and disorder distribution $\nu$. 
The case of one dimension ($D=1$) is simple: 
non-constant ground configurations exist if and only if $\nu$ has an atom at the bottom of its support. 
However, the question has remained open in all higher dimensions. 
Much attention has been given to the two-dimensional case ($D=2$), 
where it is shown that existence of non-constant ground configurations is equivalent to the existence of infinite bigeodesics in a dual first-passage percolation model (see, e.g., \cite{N97}*{Page 8}). 
Such bigeodesics are believed not to exist under mild assumptions on the disorder distribution,  whence the answer to Question~\ref{q:non-constant ground configurations} is expected to be negative in $D=2$. However, so far this is only proved under assumptions on the model which are still unverified~\cite{A20}, or for other, related, models possessing a special integrable structure~\cites{BHS22, BBS20, GJRA21, BS22}. 
The higher-dimensional case, 
and the closely-related question of localization of domain walls of the disordered ferromagnet (see Question~\ref{q:localization of Dobrushin interface} below), were studied in the physics and mathematics literature by several authors, 
including Huse--Henley~\cite{HH85}, 
Bovier--Fr\"ohlich--Glaus~\cite{BFG86}, Fisher~\cite{F86}, Bovier--Picco~\cite{BP91}, Bovier--K\"ulske~\cites{BK, BK96}, Wehr~\cite{W97} and Wehr--Wasielak~\cite{WW16}. 
These studies predict, from non-rigorous methods~\cites{HH85, BFG86, F86} and from rigorous studies of simplified interface models~\cites{BP91, BK} (see also Section~\ref{sec:background}), an affirmative answer to Question~\ref{q:non-constant ground configurations} in dimensions $D\ge 4$ for sufficiently concentrated disorder distributions $\nu$ (see Section~\ref{sec:discussion and open problems} for a discussion of other settings). 
In this work we confirm this prediction.

\smallskip
\noindent{\bf Existence result.} To state our results, we need to give a precise meaning to the idea that the disorder distribution $\nu$ is sufficiently concentrated. To this end, we define a notion of ``width'' of the  distribution, applicable to either of the following two classes of disorder distributions:
\begin{enumerate}
\item 
Compact support: If the distribution $\nu$ has compact support, we set 
\begin{equation}
\diam(\nu):=\min\{b-a\,\colon\, b\ge a\text{ and } \nu([a,b])=1\}
\end{equation}
and otherwise we set $\diam(\nu):=\infty$.
\item 
Lipschitz image of a Gaussian: If there exists a Lipschitz continuous $f:\R\to[0,\infty)$ such that $\nu=f(N(0,1))$ (i.e., $\nu$ is the push-forward through $f$ of the standard normal distribution) then set
\begin{equation}\label{eq:lip_const_of_distribution}
\lip(\nu):=\inf\{\lip(f)\,\colon\, f:\R\to[0,\infty) \text{ is Lipschitz and }\nu=f(N(0,1))\},
\end{equation}
with $\lip(f):=\sup_{t\neq s}\frac{|f(t)-f(s)|}{|t-s|}$ being the Lipschitz constant of $f$. Otherwise, set $\lip(\nu):=\infty$.
\end{enumerate}
We then define the ``width'' of the disorder distribution $\nu$ by
\begin{equation}\label{eq:width dist def}
\wid(\nu):=\min\{\diam(\nu),\lip(\nu)\}.
\end{equation}

We further restrict attention to disorder distributions whose support is bounded away from zero and to this end we denote by
\begin{equation}
\min(\supp(\nu)):=\max\{\alpha\colon \nu([\alpha,\infty))=1\}
\end{equation}
the smallest point of the support. Lastly, to avoid issues with uniqueness of ground configurations in \emph{finite volume}, we assume that $\nu$ has no atoms.

The following is our main result.
\begin{theorem}\label{thm:non-constant ground configuration}
There exists $c>0$ such that the following holds in dimensions $D\ge 4$. Consider the (isotropic) disordered ferromagnet with disorder distribution $\nu$. If $\min(\supp(\nu))>0$, $\nu$ has no atoms
and
\begin{equation}
\label{eq:width assumption isotropic}
\frac{\wid(\nu)}{\min(\supp(\nu))}\le c\frac{\sqrt{D}}{\log D}
\end{equation}
then the disordered ferromagnet admits non-constant ground configurations.
\end{theorem}
We make two remarks regarding the theorem:
\begin{enumerate}
\item 
Condition~\eqref{eq:width assumption isotropic} is our notion of $\nu$ being sufficiently concentrated. It is invariant to dilations of $\nu$, as it should be since multiplying all coupling constants $(\eta_e)$ in the Hamiltonian~\eqref{eq:Disordered ferromagnet formal Hamiltonian} by a positive constant does not change the set of ground configurations of the disordered ferromagnet (or its Gibbs states).

The condition becomes easier to satisfy as the dimension $D$ increases. In particular, the theorem shows that for any non-atomic disorder distribution $\nu$ satisfying $\min(\supp(\nu))>0$ and $\wid(\nu)<\infty$ there exists $D_0(\nu)\ge 4$ such that the disordered ferromagnet with disorder distribution $\nu$ admits non-constant ground configurations in all dimensions $D\ge D_0(\nu)$.

Condition~\eqref{eq:width assumption isotropic} also becomes easier to satisfy if a positive constant is added to~$\nu$. More precisely, we see that for any non-atomic distribution $\mu$ supported in $[0,\infty)$ and satisfying $\wid(\mu)<\infty$ and any $D_0\ge 4$  there exists $C_0(\mu,D_0)\ge 0$ such that the following holds: for all $C\ge C_0(\mu,D_0)$, the disordered ferromagnet with disorder distribution $\nu = \mu + C$ admits non-constant ground configurations in all dimensions $D\ge D_0$ (where $\mu+C$ is the push-forward of $\mu$ by the translation $x\mapsto x+C$).
    
As an example, there exists $C>0$ such that the disordered ferromagnet whose disorder distribution is  uniform on the interval $[C,C+1]$ admits non-constant ground configurations in all dimensions $D\ge 4$.
    
\item For $\wid(\nu)$ to be finite, the disorder distribution needs to be either of compact support or a Lipschitz image of a Gaussian. The latter possibility allows some distributions of unbounded support such as the positive part of a Gaussian (plus a constant, so that $\min(\supp(\nu))>0$), and can also lead to a smaller value of $\wid(\nu)$ for some distributions of compact support.
\end{enumerate}

\smallskip
\noindent{\bf Covariant ground configurations.}
Once the \emph{existence} of non-constant ground configurations has been established, it is natural to turn to more refined properties. In the non-disordered setup, a key role is played by the notion of \emph{invariant} Gibbs states (e.g., translation-invariant Gibbs states). The corresponding notion in the disordered setup is that of \emph{covariant} states (going back at least to the pioneering work of Aizenman--Wehr~\cite{AW90}, who further introduced covariant \emph{metastates}). We apply this notion to ground configurations as follows: 
Let $G$ be a group of automorphisms of the lattice $\Z^D$ (each $g\in G$ is composed of translations, rotations and reflections). We naturally define
\begin{equation}\label{eq:automorphism actions}
g(\eta)_{\{x,y\}} := \eta_{g\{x,y\}}=\eta_{\{gx, gy\}}\quad\text{and}\quad g(\sigma)_x:=\sigma_{gx}    
\end{equation}
for coupling fields $\eta$ and configurations $\sigma$. 
Let $\mathcal{C}\subset[0,\infty)^{E(\Z^D)}$ be a measurable set of coupling fields which is $G$-invariant in the sense that $g(\eta)\in\mathcal{C}$ for each automorphism $g\in G$ and coupling field $\eta\in\mathcal{C}$. 
A \emph{$G$-covariant ground configuration} defined on $\mathcal{C}$ is a \emph{measurable} function $T:\mathcal{C}\to\{-1,1\}^{\Z^D}$ which satisfies the following properties for all $\eta\in\mathcal{C}$:
\begin{enumerate}
\item 
$T(\eta)$ is a ground configuration for the disordered ferromagnet with coupling field $\eta$.
\item 
$T(g(\eta))=g(T(\eta))$ for each automorphism $g\in G$. 
\end{enumerate}
If, moreover, $T(\eta)$ is non-constant for all $\eta\in\mathcal{C}$ we say that $T$ is a \emph{non-constant} $G$-covariant ground configuration defined on $\mathcal{C}$.

When a disorder distribution $\nu$ has been specified (in the isotropic setup), we may refer to a $G$-covariant ground configuration without reference to its domain $\mathcal{C}$. It is then understood that the $(\eta_e)$ are sampled independently from $\nu$ and that the $G$-covariant ground configuration is defined on some $G$-invariant $\mathcal{C}$ satisfying that $\P(\eta\in\mathcal{C})=1$. The analogous comment applies to the anistropic setup discussed below, when the two disorder distributions $\nu^\parallel$ and $\nu^\perp$ have been specified.

Wehr--Wasielak~\cite{WW16} prove, in all dimensions $D\ge 1$, that when the disorder distribution $\nu$ is non-atomic and has finite mean (or, more generally, sufficiently light tail) then there are no non-constant $\Z^D$-translation-covariant ground configurations ($\Z^D$-translation-covariant means that the group $G$ is the translation group of $\Z^D$). More generally, they prove that $\Z^D$-translation-covariant ground metastates (i.e., $G$-covariant \emph{probability distributions} on ground configurations) must be supported on the constant configurations.

Our next result shows that non-constant $G$-covariant ground configurations may exist already when $G$ is one rank lower than the full translation group of $\Z^D$.
\begin{theorem}\label{thm:covariant ground configurations isotropic}
Let $G^{D-1}$ be the group of automorphisms of $\Z^D$ which preserve the last coordinate, i.e.,
\begin{equation}\label{eq:G D-1 automorphism group}
G^{D-1}:=\{g\text{ automorphism of $\Z^D$}\colon (gx)_D=x_D\text{ for all }x=(x_1,\ldots, x_D)\}.
\end{equation}
Under the assumptions of Theorem~\ref{thm:non-constant ground configuration} (for a sufficiently small $c>0$), there exists a non-constant $G^{D-1}$-covariant ground configuration.
\end{theorem}

\smallskip
\noindent{\bf The anisotropic disordered ferromagnet.}
Our proof naturally extends to an \emph{anisotropic} setup, in which there is a distinguished lattice axis and the coupling constants of edges in the direction of that axis are sampled from a different disorder distribution. We next describe this setup (distinguishing the $D$th axis).

It is sometimes convenient to identify edges of $\Z^D$ with their dual plaquettes, i.e., to identify $\{x,y\}\in E(\Z^D)$ with the $(D-1)$-dimensional plaquette separating the unit cubes centered at $x$ and $y$. With this identification in mind we partition the plaquettes into those which are parallel to the hyperplane spanned by the first $D-1$ coordinate axes and those which are perpendicular to it. Thus we define
\begin{equation}
\begin{split}
E^{\parallel}(\Z^D):=\{\{x,y\}\in E(\Z^D)\colon x-y\in\{-e_D, e_D\}\},\\
E^{\perp}(\Z^D):=\{\{x,y\}\in E(\Z^D)\colon x-y\notin\{-e_D, e_D\}\},
\end{split}\label{eq:parallel and perpendicular edges}
\end{equation}
where $e_1,e_2,\ldots,e_D$ are the standard unit vectors in $\Z^D$.

By the \emph{anisotropic disordered ferromagnet} we mean that the disorder $\eta$ is sampled independently from two disorder distributions, $\nu^{\parallel}$ and $\nu^{\perp}$. Precisely, $(\eta_e)_{e\in E(\Z^D)}$ are independent, with $\eta_e$ sampled from $\nu^{\parallel}$ when $e\in E^{\parallel}(\Z^D)$, and sampled from $\nu^{\perp}$ when $e\in E^{\perp}(\Z^D)$. The isotropic setup is recovered when $\nu^\parallel=\nu^\perp$.

Our standard assumptions on the disorder distributions are
\begin{equation}\label{eq:disorder distributions assumptions}
\begin{alignedat}{2}
&\min(\supp(\nu^\parallel))>0,\quad&&\wid(\nu^\parallel)<\infty\quad\text{and $\nu^\parallel$ has no atoms},\\
&\min(\supp(\nu^\perp))>0,&&\wid(\nu^\perp)<\infty.
\end{alignedat}
\end{equation}
We do not assume that $\nu^\perp$ has no atoms (and, in fact, the case that $\nu^\perp$ is supported on a single point is of interest as it leads to the disordered Solid-On-Solid model; see Remark~\ref{rem:disordered SOS model special case}).

In the anisotropic setup, condition~\eqref{eq:width assumption isotropic} of Theorem~\ref{thm:non-constant ground configuration} is replaced by condition~\eqref{eq:anisotropic condition} below, which is based on the following quantity:
\begin{equation}\label{eq:kappa_definition}
\kappa(\nu^\parallel, \nu^\perp,d) := \left( \frac{1}{\underline{\alpha}^{\parallel}\underline{\alpha}^{\perp}} +\frac{1}{d (\underline{\alpha}^{\perp})^2}\right)\wid (\nu^{\parallel})^{2}+\frac{1}{(\underline{\alpha}^{\perp})^2}\wid(\nu^{\perp})^{2}
\end{equation}
where, for brevity, we denote
\begin{equation}
\label{eq:alphas}    \underline{\alpha}^{\parallel}:=\min(\supp(\nu^\parallel))\quad\text{and}\quad\underline{\alpha}^{\perp}:=\min(\supp(\nu^\perp)).
\end{equation}

\begin{theorem}\label{thm:anisotropic non-constant ground configuration}
There exists $c_0>0$ such that the following holds in dimensions $D\ge 4$. In the anisotropic disordered ferromagnet, suppose the disorder distributions $\nu^\parallel$ and $\nu^\perp$ satisfy~\eqref{eq:disorder distributions assumptions} and
\begin{equation}\label{eq:anisotropic condition}
\kappa(\nu^\parallel, \nu^\perp,D-1) \left(1+\frac{\underline{\alpha}^{\perp}}{\underline{\alpha}^{\parallel}} \right)\le c_0\frac{D}{(\log D)^2}.
\end{equation}
Then the anisotropic disordered ferromagnet admits non-constant ground configurations.

\noindent Moreover, there exists a non-constant $G^{D-1}$-covariant ground configuration, where $G^{D-1}$ is given by~\eqref{eq:G D-1 automorphism group}.
\end{theorem}

Theorem~\ref{thm:non-constant ground configuration} and Theorem~\ref{thm:covariant ground configurations isotropic} arise as the special case of Theorem~\ref{thm:anisotropic non-constant ground configuration} in which $\nu^\parallel = \nu^\perp$. We thus focus in the sequel on the anisotropic setup.

Similarly to condition~\eqref{eq:width assumption isotropic}, we make the following remark regarding condition~\eqref{eq:anisotropic condition}. For any pair of disorder distributions $\nu^\parallel$ and $\nu^\perp$ satisfying $\min(\supp(\nu^\parallel))>0$, $\min(\supp(\nu^\perp))>0$, $\wid(\nu^\parallel)<\infty$ and $\wid(\nu^\perp)<\infty$, condition~\eqref{eq:anisotropic condition} will be satisfied in either sufficiently high dimensions $D$, or in any fixed dimension $D\ge 4$ provided $\nu^\perp$ and $\nu^\parallel$ are replaced by $\nu^\perp+C$ and $\nu^\parallel+C$, respectively, for a sufficiently large $C$. 

\smallskip

\noindent{\bf Dobrushin boundary conditions.}
Our proof of Theorem ~\ref{thm:anisotropic non-constant ground configuration}
proceeds through an explicit construction of a non-constant ground configuration. Specifically, we will show that the infinite-volume limit of the model with \emph{Dobrushin boundary conditions} leads to such a ground configuration. In this section we explain the relevant result, following required notation.

We first generalize the notion of ground configuration to allow boundary conditions. Let $\Delta\subset\Z^D $ and $ \rho\colon \Z^D \to \{-1,1\} $. The configuration space in the domain $\Delta$ with boundary conditions $\rho$ is
\begin{equation}\label{eq:configuration space in Z^D}
\Omega^{\Delta,\rho}:=\{\sigma\colon \Z^D\to\{-1,1\}\colon \text{$\sigma_x = \rho_x$ for $x\notin\Delta$}\}.
\end{equation}
A \emph{ground configuration in $\Omega^{\Delta,\rho}$} (for the coupling field $\eta$) is any $\sigma\in\Omega^{\Delta,\rho}$ with the property that $H^\eta(\sigma)\le H^{\eta}(\sigma')$ for every configuration $\sigma'\in\Omega^{\Delta,\rho}$ which differs from $\sigma$ in finitely many places.
(noting, again, that the difference $H^{\eta}(\sigma)-H^{\eta}(\sigma')$ is then well defined). It is possible for multiple ground configurations to exist even when $\Delta$ is finite, though this will not be the case when $\eta$ is generic in a suitable sense (see Section \ref{sec:disorders energies and ground configurations}). 

We proceed to discuss ground configurations with Dobrushin boundary conditions. Assume henceforth that the dimension $D\ge 2$ and introduce the convenient notation
\begin{equation}
d := D-1.
\end{equation}
We often tacitly use the convention to denote vertices of $\Z^D$ by a pair $(v,k)$ with $v\in\Z^d$ and $k\in\Z$.
By Dobrushin boundary conditions we mean $\rho^{\Dob}\colon \Z^D\to\{-1,1\}$ given by
\begin{equation}\label{eq:rho Dob def}
\rho^{\Dob}_{(v,k)} := \sign(k - 1/2)
\end{equation}
where $\sign$ denotes the sign function. 
   
The following simple lemma shows that there is a unique ground configuration with Dobrushin boundary conditions on infinite cylinders of the form $\Lambda\times\Z$ for $\Lambda\subset\Z^d$ finite, under suitable conditions on the disorder distributions.
\begin{lemma}\label{lem:semi infinite volume ground configuration}
In the anisotropic disordered ferromagnet, suppose the disorder distributions $\nu^\parallel$ and $\nu^\perp$ satisfy~\eqref{eq:disorder distributions assumptions}. 
For each finite $\Lambda\subset\Z^d$ there exists almost surely a \emph{unique} ground configuration in $\Omega^{\Lambda\times{\mathbb Z},\rho^{\Dob}}$,
that we denote $\sigma^{\eta,\Lambda,\Dob}$. 
Moreover, $\sigma^{\eta,\Lambda,\Dob}_x =\rho^{\Dob}_x$ for all but finitely many $x\in\Z^D$.
\end{lemma}
The ground configuration $\sigma^{\eta,\Lambda,\Dob}$ necessarily contains an interface separating the $+1$ boundary values imposed at the ``top" of the cylinder $\Lambda\times\Z$ from the $-1$ boundary values imposed at the ``bottom" of the cylinder. To derive the existence of non-constant ground configurations in the whole of $\Z^D$ from these semi-infinite-volume ground configurations, the fundamental issue to be understood is whether this interface remains localized in the sense of the following question.
\begin{question}\label{q:localization of Dobrushin interface}
Is the interface under Dobrushin boundary conditions localized, uniformly in the finite volume $\Lambda$? More precisely, is it the case that for each $\eps>0$ there exists a finite $\Delta\subset\Z^D$ such that
\begin{equation}
\P(\sigma^{\eta,\Lambda,\Dob}\text{ is constant on }\Delta)<\eps\quad\text{for all finite $\Lambda\subset\Z^d$}.
\end{equation}
\end{question}
For the set $\Delta$ in the question, one may have in mind, e.g., the set $\{(0,\ldots,0,j)\colon -k\le j\le k\}$ with $k$ large.

A positive answer to Question~\ref{q:localization of Dobrushin interface} implies a positive answer to Question~\ref{q:non-constant ground configurations}. Indeed, by compactness, the distribution of the pair $(\eta, \sigma^{\eta,\Lambda,\Dob})$ admits sub-sequential limits along any sequence of finite volumes $(\Lambda_n)$ increasing to $\Z^d$. Any such limiting distribution is supported on pairs $(\eta', \sigma)$ with $\eta'$ having the distribution of $\eta$ and $\sigma$ an infinite-volume ground configuration for $\eta'$. A positive answer to the question ensures that $\sigma$ is almost surely non-constant.

The answer to Question 1.6 is known to be negative for $D=2$ in the isotropic setup under mild assumptions on the disorder distribution $\nu$; this is essentially the Benjamini--Kalai--Schramm midpoint problem~\cite{BKS03}, resolved conditionally by Damron--Hanson~\cite{DH17}, unconditionally by Ahlberg--Hoffman~\cite{AH16} and quantitatively by Dembin, Elboim and the third author~\cite{DEP22}. The following theorem, our main result on Dobrushin interfaces, proves that the answer is positive in dimensions $D\ge 4$ for sufficiently concentrated disorder distributions. 
Further discussion on Question~\ref{q:localization of Dobrushin interface}, including the possibility of a roughening transition in the disorder concentration, is in Section~\ref{sec:discussion and open problems}.

\begin{theorem}[Localization of Dobrushin interface]\label{thm:localization}
There exist $c_0,c>0$ such that the following holds in dimensions $D\ge 4$. In the anisotropic disordered ferromagnet, suppose the disorder distributions $\nu^\parallel$ and $\nu^\perp$ satisfy~\eqref{eq:disorder distributions assumptions} and condition~\eqref{eq:anisotropic condition} holds (with the constant $c_0$). Then for all finite $\Lambda\subset\Z^d$ and all $(v,k)\in \Z^{D}$,
\begin{equation}\label{eq:localization at v k}
\P\left(\sigma^{\eta,\Lambda,\Dob}_{(v,k)} \neq \rho^{\Dob}_{(v,k)}\right) \leq           
\exp\left(-\frac{c}{d^2 \kappa}  |k|^{\frac{d-2}{d-1}}\right)
\end{equation}
with $\kappa=\kappa(\nu^\parallel, \nu^\perp,d)$ defined in~\eqref{eq:kappa_definition}.
Moreover, for small $k$ we have an improved dependence on dimension: 
if $ |k| < 2^d$ then 
\begin{equation}
\P\left(\sigma^{\eta,\Lambda,\Dob}_{(v,k)} \neq \rho^{\Dob}_{(v,k)}\right) \leq \exp\left(-\frac{c}{\kappa} |k| \right).
\end{equation}
\end{theorem}
We add that the theorem resolves a version of~\cite{DG23}*{Open question 1} (see Section~\ref{sec:ground energy fluctuations}). We also remark that the power of $|k|$ in~\eqref{eq:localization at v k} is probably not optimal and may be increased with further optimization of the proof.

\begin{remark}[Combinatorial interpretation] Endow the edges of $\Z^D$ with independent, non-negative weights from a distribution $\nu$. Let $\Lambda\subset\Z^{D-1}$ finite. We study the minimal (edge) cutset in $\Lambda\times\Z$ separating the parts of the boundary of $\Lambda\times\Z$ above and below the plane $\Lambda\times\{0\}$. More precisely, writing
\begin{equation}\begin{split}
B^+:=\{(v,k)\in\Lambda^c\times\Z\colon k\ge 0\},\\
B^-:=\{(v,k)\in\Lambda^c\times\Z\colon k< 0\},
\end{split}
\end{equation}
we study the minimal cutset separating $B^+$ and $B^-$ (the cutset may only differ from the flat plane above $\Lambda$). Our result is proved in dimensions $D\ge 4$ when $\nu$ satisfies~\eqref{eq:width assumption isotropic} with a small $c>0$. It shows that the minimal cutset is localized close to the plane $\Lambda\times\{0\}$, in the sense that for any $v\in\Lambda$ the probability that the cutset contains an edge incident to $(v,k)$ decays as a stretched exponential in $|k|$. This holds uniformly in the finite set $\Lambda$. 

More generally, the edge weights may be sampled independently from distributions $\nu^{\parallel}$ and $\nu^\perp$ as explained after~\eqref{eq:parallel and perpendicular edges}, and the result is proved under~\eqref{eq:disorder distributions assumptions} and under condition~\eqref{eq:anisotropic condition} with a small $c_0>0$.
\end{remark}

As explained above, Theorem~\ref{thm:localization} already suffices to deduce the existence of non-constant ground configurations. To go further and prove the existence of a non-constant \emph{covariant} ground configuration we employ the next result, which proves the almost sure convergence of the semi-infinite-volume ground configurations with Dobrushin boundary conditions to an infinite-volume limit.
\begin{theorem}[Convergence]\label{thm:convergence}
Under the assumptions of Theorem~\ref{thm:localization} (for a sufficiently small $c_0>0$) there exists a configuration $\sigma^{\eta,\Dob}$ such that 
for every fixed sequence $(\Lambda_n)$ of finite subsets of $\Z^d$ satisfying that $\Lambda_n\supset\{-n,-n+1,\ldots,n\}^d$ for each $n$, almost surely,
\begin{equation} \label{eq:convergence}
\text{for each $v\in\Z^{d}$ there exists $n_0$ such that $\sigma^{\eta,\Lambda_n,\Dob}|_{\{v\}\times\Z}\equiv\sigma^{\eta,\Dob}|_{\{v\}\times\Z}$ for all $n\ge n_0$}.
\end{equation}
\end{theorem}

The following is deduced from Theorem~\ref{thm:localization} and Theorem~\ref{thm:convergence}.
\begin{corollary}\label{cor:covariant Dobrushin configuration}
There exists $c>0$ such that the following holds under the assumptions of Theorem~\ref{thm:localization} (for a sufficiently small $c_0>0$). 
The configuration $\sigma^{\eta,\Dob}$ of Theorem~\ref{thm:convergence} (possibly modified on a set of zero probability) is a non-constant $G^{D-1}$-covariant ground configuration, where $G^{D-1}$ is given by~\eqref{eq:G D-1 automorphism group}. In addition, for all $(v,k)\in \Z^{D}$,
\begin{equation}\label{eq:cor_localization at v k}
\P\left(\sigma^{\eta,\Dob}_{(v,k)} \neq \rho^{\Dob}_{(v,k)}\right) \leq           
\exp\left(-\frac{c}{d^2 \kappa}  |k|^{\frac{d-2}{d-1}}\right)
\end{equation}
with $\kappa=\kappa(\nu^\parallel, \nu^\perp,d)$ defined in~\eqref{eq:kappa_definition}.
Moreover, for small $k$ we have an improved dependence on dimension: 
if $ |k| < 2^d$ then 
\begin{equation}\label{eq:cor_localization at v k moreover}
\P\left(\sigma^{\eta,\Dob}_{(v,k)} \neq \rho^{\Dob}_{(v,k)}\right) \leq \exp\left(-\frac{c}{\kappa} |k| \right).
\end{equation}
\end{corollary}

Theorem~\ref{thm:anisotropic non-constant ground configuration} is an immediate consequence of the last corollary.

Our techniques further allow to quantify the rate of convergence in Theorem~\ref{thm:convergence} and to bound the rate of correlation decay in the infinite-volume configuration $\sigma^{\eta,\Dob}$. We record these in our final result (this result will not be needed elsewhere in the paper).
\begin{theorem}\label{thm:decay of correlations and rate of convergence}
There exist $C,c>0$ such that the following holds under the assumptions of Theorem~\ref{thm:localization} (for a sufficiently small $c_0>0$). Let
\begin{equation}
c(\nu^{\parallel},\nu^{\perp},d):=\frac{c}{\kappa  d^2} \left( \min \left\{\frac{\alpha^{\parallel}}{\alpha^{\perp}},1\right\}\right)^{\frac{d-2}{d-1}}
\end{equation}
using the notation~\eqref{eq:kappa_definition} (with $\kappa = \kappa(\nu^\parallel, \nu^\perp,d))$ and~\eqref{eq:alphas}.
Let also $\Lambda(k):= \{-k,\dots, k\}^{d}$ for integer $k\ge 0$.
\begin{enumerate}
\item \emph{Rate of convergence to infinite-volume limit}:
Let $L_1>L_0\ge 0$ integer. Let $\Lambda\subset\Z^d$ be a finite subset containing $\Lambda(L_1)$. Then 
\begin{equation}\label{eq:rate of convergence}
\P\left(\sigma^{\eta,\Dob}|_{\Lambda(L_0)\times\Z}\not\equiv \sigma^{\eta,\Lambda,\Dob}|_{\Lambda(L_0)\times\Z}\right)\le C\exp \left(-c(\nu^{\parallel},\nu^{\perp},d)\left(L_1 - L_0\right)^{\frac{d-2}{d}} \right).
\end{equation}
\item \emph{Correlation decay in infinite-volume limit}:
Let $u,v\in\Z^d$ and $L\ge 0$ integer, and suppose $\|u-v\|_\infty>2L$. Let $f,g:\{-1,1\}^{\Lambda(L)\times\Z}\to[-1,1]$ be measurable. Then
\begin{multline}\label{eq:correlation decay estimate}
\cov(f(\sigma^{\eta,\Dob}|_{(u+\Lambda(L))\times\Z}),g(\sigma^{\eta,\Dob}|_{(v+\Lambda(L))\times\Z}))\\
\le C\exp \left(-c(\nu^{\parallel},\nu^{\perp},d)\left(\|u-v\|_\infty-2L\right)^{\frac{d-2}{d}} \right)
\end{multline}
where $\cov(X,Y)$ denotes the covariance of the random variables $X,Y$.
\item \emph{Tail triviality in infinite-volume limit}:
The process $(\eta,\sigma^{\eta,\Dob})$ is $G^d$-invariant. Moreover, define the $\Z^d$-tail sigma algebra of $(\eta,\sigma^{\eta,\Dob})$ as the intersection of the sigma algebras $(\mathcal{T}_n)$, where $\mathcal{T}_n$ is generated by $\sigma^{\eta,\Dob}|_{(\Z^d\setminus\Lambda(n))\times\Z}$ and by $(\eta_e)$ for the edges $e = \{(u,k),(v,\ell)\}$ with $\{u,v\}\cap(\Z^d\setminus\Lambda(n))\neq\emptyset$. Then the $\Z^d$-tail sigma algebra of $(\eta,\sigma^{\eta,\Dob})$ is trivial. In particular, $(\eta,\sigma^{\eta,\Dob})$ is ergodic with respect to the group of translations in the first $d$ coordinates.
\end{enumerate}
\end{theorem}

Theorem~\ref{thm:localization} is proved in Section \ref{sec:deduction of localization theorem}, Theorem~\ref{thm:convergence}, Corollary~\ref{cor:covariant Dobrushin configuration} and Theorem~\ref{thm:decay of correlations and rate of convergence} are proved in Section \ref{sec:proof of convergence theorem} and Lemma~\ref{lem:semi infinite volume ground configuration} is proved in Appendix \ref{sec:small lemmas}.

The next section discusses related works. An overview of our proof is provided in Section~\ref{subsec:overview}. We provide a detailed reader's guide in Section~\ref{sec:readers guide}.

\subsection{Background}\label{sec:background}

\subsubsection{Localization predictions} The domain walls of the disordered Ising ferromagnet were studied by Huse--Henley~\cite{HH85}, Bovier--Fr\"ohlich--Glaus~\cite{BFG86} and Fisher~\cite{F86} using methods which are not mathematically rigorous. They predicted that the interface with Dobrushin boundary conditions is rough in dimensions $2\le D \le 3$, is localized in dimensions $D\ge 5$, and, for sufficiently concentrated disorder, is also localized in dimension $D=4$.

\subsubsection{Disordered Solid-On-Solid model}\label{sec:disordered SOS model}
The following simplified model for the interface under Dobrushin boundary conditions is considered by~\cites{HH85, BFG86}: The interface is described by a (height) function $\varphi:\Z^d\to\Z$, whose energy is given by the formal ``disordered Solid-On-Solid (SOS)'' Hamiltonian
\begin{equation}\label{eq:disordered SOS Hamiltonian}
H^{\text{SOS}, \zeta}(\varphi):=\sum_{\{u,v\}\in E(\Z^d)} |\varphi_u - \varphi_v| + \sum_{v\in\Z^d} \zeta_{v,\varphi_v}
\end{equation}
where $\zeta:\Z^d\times\Z\to\R$ is an environment describing the quenched disorder. This model is obtained from the disordered  ferromagnet with two approximations: (i) It is assumed that the interface in $D=d+1$ dimensions has \emph{no overhangs}, i.e., it may be described by a height function above a $d$-dimensional base, (ii) all the coupling constants corresponding to perpendicular plaquettes (i.e., all the $\eta_{\{u,v\}}$ for $\{u,v\}\in E^\perp(\Z^D)$) are set equal (with the normalization of~\eqref{eq:disordered SOS Hamiltonian} they are set equal to $1/2$ while $\zeta_{v,k}:=2\eta_{\{v,k\},\{v,k+1\}}$). 
\begin{remark}\label{rem:disordered SOS model special case} In fact, as part of the analysis of this paper, we prove the (possibly surprising) fact that at zero temperature the no overhangs approximation (i) is actually a consequence of the equal perpendicular couplings approximation (ii) (see Lemma~\ref{lem:ground configuration with no overhangs}). Thus, our main results for the anistropic disordered ferromagnet cover also the disordered SOS model~\eqref{eq:disordered SOS Hamiltonian}, as the special case in which the disorder distribution $\nu^\perp$ is a delta measure at $1/2$ (however, in this special case our proof may be greatly simplified due to the no overhangs property).
\end{remark}

A mathematically-rigorous study of the disordered SOS model~\eqref{eq:disordered SOS Hamiltonian} was carried out by Bovier--K\"ulske~\cites{BK,BK96}, following an earlier analysis by Bovier--Picco~\cite{BP91} of a hierarchical version of the model (see also~\cites{BK92, BK93}). It was shown in~\cite{BK} that in each dimension $d\ge 3$, at low temperature (including zero temperature), when the $(\zeta_{v,\cdot})_{v\in\Z^d}$ are independent and identically distributed, the sequence $k\mapsto\zeta_{v,k}$ is stationary for each $v$ (this is more general than being independent!) and the $\zeta_{v,k}$ are sufficiently concentrated, the finite-volume Gibbs measures of the Hamiltonian~\eqref{eq:disordered SOS Hamiltonian} converge, on a non-random sequence of volumes, to a limiting infinite-volume Gibbs measure, $\zeta$-almost-surely. Some control of the fluctuations of the infinite-volume Gibbs measure, at least at zero temperature, is also provided~\cite{BK}*{Proposition 3.6}. These results for the disordered SOS model~\eqref{eq:disordered SOS Hamiltonian} thus have the flavor of our Theorem~\ref{thm:localization} and Theorem~\ref{thm:convergence}, though they, on the one hand, apply also at low positive temperature and allow for more general disorder distributions and, on the other hand, do not quantify the dependence on the dimension $d$ (i.e., their sufficient concentration requirement may become more stringent as $d$ increases). The work~\cite{BK} further discusses alternative assumptions on $\zeta$ relevant to the interface in the \emph{random-field} Ising model (see also Section~\ref{sec:positive temperature and RFIM}).

The behavior of the disordered SOS model~\eqref{eq:disordered SOS Hamiltonian} in the low dimensions $d=1,2$ was studied in~\cite{BK96} (using a method of Aizenman--Wehr~\cite{AW90}), who proved a result showing a form of delocalization in these dimensions. Specifically, they prove that, at all finite non-zero temperatures, when the $(\zeta_{v,k})$ are independently sampled from a distribution with positive variance which either has no  isolated atoms or has compact support, the model does not admit translation-covariant and coupling-covariant metastates. Here, a metastate is a measurable mapping from $\zeta$ to probability distributions over (infinite-volume) Gibbs measures of the model, and the coupling covariance requirement is that, for each finite $\Lambda\subset\Z^d$, the metastate changes in a natural way under modification of $(\zeta_{v,k})_{v\in\Lambda, k\in\Z}$.

\subsubsection{Long-range order in the random-field Ising model} The localization proof of~\cite{BK} in dimensions $d\ge 3$ is closely tied to earlier developments on the problem of long-range order in the random-field Ising model (see~\eqref{eq:RFIM Hamiltonian} below). Imry--Ma~\cite{IM75} predicted that at low temperatures and weak disorder in dimensions $d\ge3$, the random-field Ising model retains the ferromagnetic ordered phase of the pure Ising model (and that this does not occur when $d=2$). The prediction for $d=3$ was initially challenged in the physics literature (e.g., ~\cite{PS79}), but received support in works of Chalker~\cite{C83} and Fisher--Fr\"ohlich--Spencer~\cite{FFS} and was finally confirmed in the breakthrough works of Imbrie~\cites{I84,I85} and Bricmont--Kupiainen~\cites{BK87, BK88}. The proof of~\cite{BK} adapts the proof technique of~\cite{BK88}. Recently, a short proof of the existence of an ordered phase in the random-field Ising model was found by Ding--Zhuang~\cite{DZ}. In this paper, we use an adaptation of the Ding--Zhuang argument as one of the ingredients in our proof of localization of the Dobrushin interface in the disordered Ising ferromagnet (see Section~\ref{subsec:overview} below). 

\subsubsection{Law of large numbers and large deviations of the ground energy} In dimensions $D>2$, following initial work by Kesten~\cite{K87} (and~\cites{ACCFR83, CC86} in the case of $\{0,1\}$ coupling constants), there has been significant advances in the understanding of the law of large numbers and large deviations of the ground energy of the disordered ferromagnet (or the maximal flow in a dual network) in various settings~\cites{Z00,Z18,T07,T08,W09,RT10,CT11a,CT11b,CT11c,T14,CT14,D20,DT20,DT21}.

\subsubsection{Ground energy fluctuations}\label{sec:ground energy fluctuations} Dembin--Garban~\cite{DG23} studied the fluctuations of the ground energy of the disordered ferromagnet in the cylinder $\{0,1,\ldots,L\}^d\times\{-H,{-H+1},\ldots, H\}$, with boundary conditions $-$ ($+$) imposed on the bottom face $\{0,1,\ldots, L\}^d\times\{-H\}$ (on the top face $\{0,1,\ldots, L\}^d\times\{H\}$) and free boundary conditions elsewhere (equivalently, they studied minimal weight cutsets separating the top and bottom faces of the cylinder). They obtain the first ``sub-surface-order'' bounds (superconcentration) on the variance of the ground energy, proving that when $H\ge CL$ then the variance is at most $C\frac{L^d}{\log L}$. The result is proved for coupling fields sampled independently from a disorder distribution supported on exactly two positive atoms. They mention that while their result applies in any dimension $d\ge 1$, this may be due to the free boundary conditions on the sides of the cylinder which allow a translation freedom for the cutset. Indeed, in~\cite{DG23}*{Open question 1} they ask to prove that the variance is of order $L^d$ when imposing Dobrushin boundary conditions on the cylinder and $d$ is sufficiently high. Their Proposition 5.1 implies that this is the case if the corresponding Dobrushin interface is localized. Our localization result, Theorem~\ref{thm:localization}, thus answers a version of~\cite{DG23}*{Open question 1} in the following sense: We prove that for any non-atomic disorder distribution $\nu$ satisfying $\min(\supp(\nu))>0$ and $\wid(\nu)<\infty$ there exists $d_0(\nu)\ge 3$ such that the Dobrushin interface in the infinite cylinder ($H=\infty$) is localized in dimensions $d\ge d_0(\nu)$ (with $d_0(\nu)$ being $D-1$ for the minimal $D\ge 4$ for which~\eqref{eq:width assumption isotropic}, with $c>0$ a small universal constant, is satisfied). Moreover, there exists $C(\nu, d)>0$ (depending on the disorder distribution $\nu$ and the dimension $d$) such that the same conclusion holds with any $H\ge C(\nu,d)(\log (L))^{\frac{d-1}{d-2}}$, as the Dobrushin interface of the infinite cylinder coincides with the interface in such a finite cylinder with high probability, by the probability estimate~\eqref{eq:localization at v k}.

\subsubsection{Concentrated disorder distributions} Our results apply for sufficiently concentrated disorder distributions (see conditions~\eqref{eq:width assumption isotropic} and~\eqref{eq:anisotropic condition}). Related assumptions were also useful in the study of the limit shape in first-passage percolation, by Basdevant--Gou\'er\'e--Th\'eret~\cite{BGT22} (for $\{0,1\}$ passage times) and by Dembin--Elboim--Peled~\cite{DEP22}*{Theorem 1.5}.

\subsubsection{Number of ground configurations} Wehr~\cite{W97} proved that the number of ground configurations of the disordered ferromagnet is two or infinity. The result applies for coupling fields sampled independently from a non-atomic disorder distribution with finite mean. It thus follows from our main result, Theorem~\ref{thm:non-constant ground configuration}, that there are infinitely many ground configurations under the assumptions there (see also Section~\ref{sec:set of non-constant covariatn ground configurations}).

\subsubsection{Translation-covariant ground metastates} As previously mentioned, Wehr--Wasielak~\cite{WW16} proved that $\Z^D$-translation-covariant ground metastates must be supported on the constant configurations when the disorder distribution $\nu$ is non-atomic and has finite mean (or, more generally, has sufficiently light tail). This result is applied in the discussion in Section~\ref{sec:set of non-constant covariatn ground configurations}.

\subsubsection{The Dobrushin interface in other settings} Our localization result, Theorem~\ref{thm:localization}, extends the seminal work of Dobrushin~\cite{Dob} to the setting of the zero-temperature disordered ferromagnet. Dobrushin's result has previously been extended to various (non-disordered) settings, of which we mention the Widom--Rowlinson model~\cites{BLPO79, BLP79}, lattice Gauge theories~\cite{B84} (see also Section~\ref{sec:higher codimension surfaces}), the Falicov–Kimball model~\cite{DMN00}, percolation and the random-cluster model~\cites{GG02,CK03} and in studying fine properties of the Dobrushin interface of the Ising model~\cite{GL}. Alternative approaches for showing the existence of non-translation-invariant Gibbs states include the correlation inequality approach of van Beijeren~\cite{vB} and the restricted-reflection-positivity approach of Shlosman--Vignaud~\cite{SV07}. These alternative approaches do not seem to be applicable in our disordered setting.

\subsection{Overview of the proof}
\label{subsec:overview}
In this section we overview the proof of the localization of the Dobrushin interface stated in Theorem~\ref{thm:localization}.

The basic idea is to synthesize Dobrushin's approach~\cite{Dob} for proving the localization of the Dobrushin interface in the \emph{pure} (i.e., non-disordered) Ising model with the simple method for proving long-range order in the random-field Ising model presented by Ding--Zhuang~\cite{DZ}.  
As it turns out, difficulties arise in this synthesis which necessitate the development of additional tools.

\subsubsection{The random-field Ising model} The random-field Ising model (RFIM) is the model on $\sigma:\Z^d\to\{-1,1\}$ given by the formal Hamiltonian
\begin{equation}\label{eq:RFIM Hamiltonian}
H^{\text{RFIM}, \zeta}(\sigma):=-\sum_{\{u,v\}\in E(\Z^d)}\sigma_u \sigma_v - \lambda\sum_v \zeta_v\sigma_v
\end{equation}
where $(\zeta_v)$ are independently sampled from the standard Gaussian distribution (more general distributions may be allowed) and $\lambda>0$ denotes the random-field strength. Imbrie~\cites{I84,I85} established long-range order in the RFIM in dimensions $d\ge 3$ at zero temperature and small $\lambda$ (and Bricmont--Kupiainen~\cites{BK87, BK88} proved the analogous fact at low, positive temperatures). It is instructive to begin our overview by describing Ding--Zhuang's~\cite{DZ} short approach to this (while~\cite{DZ} present their argument at low, positive temperatures, below we describe its zero-temperature version). Let $\sigma^{\zeta,L}$ be the ground configuration of the RFIM in $\{-L,\ldots, L\}^d$ with $+1$ boundary conditions. Let us show that it is unlikely that there exists some $A\subset\Z^d$, connected with connected complement and containing the origin, such that $\sigma^{\zeta,L}\equiv -1$ ($\sigma^{\zeta,L}\equiv +1$) on the interior (exterior) vertex boundary of $A$. Suppose $A$ is such a set. Define a modified configuration and random field by
\begin{equation}\label{eq:RFIM discrete symmetry}
\begin{split}
&\sigma^{\zeta,L,A}_v:=\begin{cases}
-\sigma^{\zeta,L}_v&v\in A\\
\sigma^{\zeta,L}_v&v\notin A
\end{cases},\\
&\zeta^A_v:=\begin{cases}-\zeta_v&v\in A\\
\zeta_v&v\notin A
\end{cases}.
\end{split}
\end{equation}
The discrete $\pm1$ symmetry of the RFIM then leads to the energy gap
\begin{equation}
H^{\text{RFIM}, \zeta}(\sigma^{\zeta,L})-H^{\text{RFIM}, \zeta^A}(\sigma^{\zeta,L,A}) \ge 2|\partial A|
\end{equation}
where $\partial A$ is the edge boundary of $A$. This implies that also
\begin{equation}
\GE^{\text{RFIM}, \zeta,L}-\GE^{\text{RFIM}, \zeta^A,L} \ge 2|\partial A|
\end{equation}
where $\GE^{\text{RFIM}, \zeta,L}:=H^{\text{RFIM}, \zeta}(\sigma^{\zeta,L})$ denotes the energy of the ground configuration in the random field $\zeta$. The argument will be (eventually) concluded by proving that, for each $\ell$,
\begin{equation}\label{eq:RFIM bound for energetic gain on A}
\P\left(\substack{\exists A\subset\Z^d\text{ connected with connected complement, $0\in A$ and $|\partial A|=\ell$},\\
|\GE^{\text{RFIM}, \zeta,L}-\GE^{\text{RFIM}, \zeta^A,L}| \ge 2|\partial A|}\right)\le C_d\exp\left(-c_d\frac{\ell^{\frac{d-2}{d-1}}}{\lambda^2}\right) 
\end{equation}
(with $C_d,c_d>0$ depending only on $d$).
To understand~\eqref{eq:RFIM bound for energetic gain on A} better, let us first explain a version of it (see~\eqref{eq:RFIM direct application of concentration} below) for a fixed \emph{deterministic} set $A\subset\Z^d$. First, observe that $\GE^{\text{RFIM}, \zeta,L}$ satisfies the conditional concentration inequality (see Theorem~\ref{thm:Maurey-Pisier} below)
\begin{equation}\label{eq:RFIM concentration}
\P\left(\big|\GE^{\text{RFIM}, \zeta,L} - \E(\GE^{\text{RFIM}, \zeta,L}\,|\,\zeta|_{A^c})\big|\ge t\,|\, \zeta|_{A^c}\right)\le C\exp\left(-c\frac{t^2}{\lambda^2|A|}\right)
\end{equation}
(with $C,c>0$ absolute constants). Next, note that $\zeta^{A}$ and $\zeta$ have the same distribution, even conditioned on $\zeta|_{A^c}$, whence the same is true for $\GE^{\text{RFIM}, \zeta^{A},L}$ and $\GE^{\text{RFIM}, \zeta,L}$. 
Consequently, the \emph{difference of ground energies} satisfies the same concentration inequality (with different constants),
\begin{equation}\label{eq:RFIM GE concentration}
\P\left(\big|\GE^{\text{RFIM}, \zeta,L}-\GE^{\text{RFIM}, \zeta^{A},L}\big|\ge t\right)\le C\exp\left(-c\frac{t^2}{\lambda^2|A|}\right).
\end{equation}
Thus, using the isoperimetric inequality $|A|\le C_d |\partial A|^{d/(d-1)}$,
\begin{equation}\label{eq:RFIM direct application of concentration}
\P(\GE^{\text{RFIM}, \zeta,L}-\GE^{\text{RFIM}, \zeta^{A},L} \ge 2|\partial A|)\le C\exp\left(-c\frac{|\partial A|^2}{\lambda^2|A|}\right)\le  C\exp\left(-c_d\frac{|\partial A|^{\frac{d-2}{d-1}}}{\lambda^2}\right).
\end{equation}
Such an estimate, however, does not suffice to establish~\eqref{eq:RFIM bound for energetic gain on A} via a union bound, since the number of subsets $0\in A\subset\Z^d$, connected with connected complement, which have $|\partial A|=\ell$ is at least $c_d\exp(C_d \ell)$ (see \cite{BB}*{Theorem 6 and Theorem 7} and Appendix~\ref{app:BB}). Instead, the estimate~\eqref{eq:RFIM bound for energetic gain on A} is derived from the concentration bound~\eqref{eq:RFIM GE concentration} using a coarse-graining technique (or chaining argument) introduced by Fisher--Fr\"ohlich--Spencer~\cite{FFS} in a closely-related context. To this end one defines $A_N$, the $N$-coarse-grained version of $A\subset\Z^d$, as the union of all cubes $B\subset\Z^d$, of the form $v+\{0,1,\ldots,N-1\}^d$ with $v\in N\Z^d$, which satisfy $|A\cap B|\ge \frac{1}{2}|B|$. Then, one writes the chaining expansion
\begin{equation}\label{eq:chaining RFIM}
\GE^{\text{RFIM}, \zeta,L}-\GE^{\text{RFIM}, \zeta^A,L} = \sum_{k=0}^{K-1} \left(\GE^{\text{RFIM}, \zeta^{A_{2^{k+1}}},L}-\GE^{\text{RFIM}, \zeta^{A_{2^k}},L}\right)
\end{equation}
where $K$ is chosen sufficiently large that $A_{2^K}=\emptyset$ (so that $\zeta^{A_{2^K}}=\zeta$), and noting that $A_{2^0}=A_1=A$. A version of the concentration inequality~\eqref{eq:RFIM GE concentration} is available (with the same proof) for any two finite $A', A''\subset\Z^d$,
\begin{equation}\label{eq:RFIM GE concentration2}
\P\left(\big|\GE^{\text{RFIM}, \zeta^{A'},L}-\GE^{\text{RFIM}, \zeta^{A''},L}\big|\ge t\right)\le C\exp\left(-c\frac{t^2}{\lambda^2|A'\Delta A''|}\right).
\end{equation}
The idea of the coarse-graining technique is to apply the concentration bound~\eqref{eq:RFIM GE concentration2} to each of the terms on the right-hand side of~\eqref{eq:chaining RFIM} (with suitable $t_k$ summing to $2|\partial A|$), using a union bound over the possible $A_{2^k}$ and bounds for $|A_{2^k}\Delta A_{2^{k+1}}|$, for $0\le k\le K-1$. The gain over the direct application~\eqref{eq:RFIM direct application of concentration} of~\eqref{eq:RFIM GE concentration} lies in the smaller denominator in the right-hand side of the concentration inequality~\eqref{eq:RFIM GE concentration2} compared to~\eqref{eq:RFIM GE concentration}, and the fact that the number of possibilities for $A_{2^k}$ is greatly reduced as $k$ increases (roughly, $|\partial A_N|\approx |\partial A|$ so that $A_N$ may be regarded as a set with surface volume $|\partial A|/N^{d-1}$ after shrinking the lattice $\Z^d$ by a factor $N$. This is complicated, however, by the fact that $A_N$ need not be connected or have connected complement).

\subsubsection{The disordered Solid-On-Solid model} It is instructive to first try and adapt the above approach to the disordered SOS model~\eqref{eq:disordered SOS Hamiltonian}, before discussing the disordered ferromagnet. The goal there is to recover a version of the result of~\cite{BK}, showing that in dimensions $d\ge 3$ when, say, the disorder $(\zeta_{v,k})$ is given by independent Gaussians with \emph{small variance}, then there is localization of the ground configuration  $\varphi^{\zeta,L}$ in $\{-L,\ldots, L\}^d$ with zero boundary values. To this end, it suffices to show that it is unlikely that there exists an integer $m\ge 0$ and some $A\subset\Z^d$, connected with connected complement and containing the origin, such that $\varphi\ge m+1$ ($\varphi\le m$) on the interior (exterior) vertex boundary of $A$. We have checked (but do not provide full details here) that a proof may be carried out very similarly to the RFIM case with the main difference being that the discrete $\pm1$ symmetry of the RFIM is now replaced by the discrete translation symmetry of adding an integer constant to $\varphi$. Thus, instead of~\eqref{eq:RFIM discrete symmetry}, a new configuration and disorder are defined by
\begin{equation}\label{eq:SOS discrete symmetry}
\begin{split}
&\varphi^{\zeta,L,A}:=\varphi^{\zeta,L}-1_A,\\
&\zeta^A_{(v,k)}:=\begin{cases}\zeta_{(v,k+1)}&v\in A\\
\zeta_{(v,k)}&v\notin A
\end{cases}
\end{split}
\end{equation}
(where $1_A$ is the indicator function of $A$), leading to the energy gap
\begin{equation}\label{eq:SOS energy gap}
H^{\text{SOS}, \zeta}(\varphi^{\zeta,L})-H^{\text{SOS}, \zeta^A}(\varphi^{\zeta,L,A}) \ge |\partial A|.
\end{equation}
While we do not enter into further detail, we remind that the disordered SOS model may be seen as a special case of the anisotropic disordered ferromagnet; see Remark~\ref{rem:disordered SOS model special case}.

The above sketch for the disordered SOS model may also be adapted to low, positive temperatures, similarly to the argument of~\cite{DZ}. However, such an extension for the disordered ferromagnet requires additional ideas (see Section~\ref{sec:positive temperature and RFIM} for further discussion).

\subsubsection{The disordered ferromagnet} 
We proceed to overview 
our approach to proving Theorem~\ref{thm:localization} - localization of the Dobrushin interface in the disordered ferromagnet.
While the approach adapts several of the ideas appearing above, it is significantly more complicated, essentialy due to the fact that the Dobrushin interface may have overhangs (i.e., have several parallel interface plaquettes in the same ``column''). Below we describe the obstacles that arise and our methods for overcoming them.

We work in dimension $D=d+1\ge 4$ under the assumptions of Theorem~\ref{thm:localization}. For finite $\Lambda\subset\Z^d$, we write $\sigma^{\eta,\Lambda,\Dob}$ for the ground configuration in $\Omega^{\Lambda\times{\mathbb Z},\rho^{\Dob}}$ as given by Lemma~\ref{lem:semi infinite volume ground configuration}, and we let $\GE^\Lambda(\eta)$ be its energy (i.e., the ground energy) in the coupling field $\eta$ (see~\eqref{eq:ground energy def} below for a precise definition). Our goal is to show that, for $(v_0,k_0)\in\Z^D$, the event \begin{equation}\label{eq:excitation at v_0}
\sigma^{\eta,\Lambda,\Dob}_{(v_0,k_0)} \neq \rho^{\Dob}_{(v_0,k_0)}
\end{equation}
is unlikely (with $\rho^{\Dob}$ defined in~\eqref{eq:rho Dob def}).

\smallskip
\paragraph{\emph{Shifts and energy gap}}  
We aim to obtain an energy gap after a transformation of the configuration and the disorder based on a discrete symmetry, 
in a similar way to~\eqref{eq:SOS discrete symmetry} and~\eqref{eq:SOS energy gap}. 
The symmetry used is the translation of the $\Z^D$ lattice along its last coordinate, but its use is more complicated than in the disordered SOS model. The amount to translate by is encoded by a function~$\tau:\Z^d\to\Z$ having finite $\supp(\tau):=\{v\in\Z^d\colon \tau(v)\neq 0\}$; we call any such function a \emph{shift}.

The shifted disorder $\eta^\tau$ is defined as follows: We fix, once and for all, an arbitrary function $\iota:E(\Z^d)\to\Z^d$ that chooses an endpoint for each edge (i.e., $\iota(e)\in e$). Then
\begin{equation}
\eta^\tau_e:=\begin{cases}\eta_{e+(0,\tau(u))}&e = \{(u,k),(u,k+1)\}\in E^\parallel(\Z^D),\\
\eta_{e+(0,\tau(\iota(\{u,v\})))}&e = \{(u,k),(v,k)\}\in E^\perp(\Z^D),
\end{cases}
\end{equation}
where $\{x,y\}+z=\{x+z,y+z\}$ for $x,y,z\in\Z^D$
(i.e., the ``column of disorders'' above a base vertex $u$ is shifted by $\tau(u)$, and the ``column'' above a base edge $\{u,v\}\in E(\Z^d)$ is shifted by $\tau(\iota(\{u,v\}))$; see also~\eqref{eq:parallel disorder shift}) and~\eqref{eq:transversal disorder shift}). Two useful features of this definition are that $\iota(\{u,v\})$ is unimportant when $\tau(u)=\tau(v)$ and that $(\eta^{\tau_1})^{\tau_2}=\eta^{\tau_1+\tau_2}$ for $\tau_1,\tau_2$ shifts.

The action on the configuration $\sigma^{\eta,\Lambda,\Dob}$ is more complicated, since a simple shift would not suffice to eliminate overhangs. Instead, our definition involves an additional subset $\tilde{A}\subset\Z^d$ (we take $\tilde{A}$ to be the projection to $\Z^d$ of the overhangs and ``interface walls'' 
that we would like to remove; see Section~\ref{sec:tau and tilde A def overview} below for our precise definition) and we define, for $(u,k)\in\Z^D$,
\begin{equation}
\sigma^{\eta,\Lambda,\Dob,\tau,\tilde{A}}_{(u,k)}:=\begin{cases}
\sigma^{\eta,\Lambda,\Dob}_{(u,k+\tau(u))}&u\notin\tilde{A},\\
\rho^{\Dob}_{(u,k)}&u\in\tilde{A}.
\end{cases}
\end{equation}
The energy gap obtained from this definition is the difference
\begin{equation}\label{eq:energy gap}
\GE^\Lambda(\eta) - \GE^\Lambda(\eta^\tau) \ge H^{\eta}(\sigma^{\eta,\Lambda,\Dob}) - H^{\eta^\tau}(\sigma^{\eta,\Lambda,\Dob,\tau,\tilde{A}}).
\end{equation}
We choose $\tau$ and $\tilde{A}$ so that the right-hand side consists exactly of (twice) the coupling constants corresponding to the overhangs and walls of $\sigma^{\eta,\Lambda,\Dob}$ above $\tilde{A}$ (more precisely, regarding the overhangs, for each $u\in\tilde{A}$ such that $\{u\}\times\Z$ has multiple parallel interface plaquettes we gain the coupling constants of all these plaquettes except the one between $(u,\tau(u))$ and $(u,\tau(u)+1)$). This is implied by the following compatibility relations:
\begin{alignat}{1}
&\text{If $\{u,v\}\in E(\Z^d)$ and $\{u,v\}\not\subset\tilde{A}$ then $\tau(u)=\tau(v)$.}\\
&\text{If $u\in\tilde{A}$, $v\notin\tilde{A}$ and $\{u,v\}\in E(\Z^d)$ then $\sigma^{\eta,\Lambda,\Dob}_{(u,k+\tau(u))}=\rho^{\Dob}_{u,k}$ for $k\in\Z$.}\\
&\text{If $u\in\tilde{A}$ then $\sigma^{\eta,\Lambda,\Dob}_{(u,\tau(u))}=-\sigma^{\eta,\Lambda,\Dob}_{(u,\tau(u)+1)}$ (our construction also gives $\sigma^{\eta,\Lambda,\Dob}_{(u,\tau(u))}=-1$).}\label{eq:third compatibility relation}
\end{alignat}
A key role in our proof of Theorem~\ref{thm:localization} is thus played by defining $\tau$ and $\tilde{A}$ as above so that: (i) a sufficient energy gap is generated when~\eqref{eq:excitation at v_0} holds, and (ii) the shift $\tau$ is taken from a small enough class (that we call \emph{admissible shifts}; see Section~\ref{sec:admissible shifts overview} below) for which we may develop suitable enumeration theorems for the required union bounds (there is no need to also enumerate over $\tilde{A}$ as it does not appear on the left-hand side of~\eqref{eq:energy gap}).

\smallskip
\paragraph{\emph{Definition of $\tau$ and $\tilde{A}$}}\label{sec:tau and tilde A def overview} 
Let $E\subset\Lambda$ (initially $E=\{v_0\}$ for the $v_0$ of~\eqref{eq:excitation at v_0}. However, later parts in our argument necessitate consideration of more general $E$). We aim to define $\tilde{A}$ as the ``projection to $\Z^d$ of the places with overhangs and interface walls which surround $E$'' and to define $\tau$ in a compatible and admissible manner. Our definitions are motivated by the ideas of Dobrushin~\cite{Dob}, to which we add a new result (Lemma~\ref{lem:ground configuration with no overhangs}) in order to define $\tau$ as an admissible shift. In fact, the absence of bubbles (\emph{finite} connected components of spins of one sign) in our zero-temperature setup allows us to simplify the approach of~\cite{Dob} and we present a self-contained treatment in Section~\ref{sec:obtaining admissible shifts} (with some inspiration from~\cites{PS23,PS20}). This also yields an improved dependence on the dimension~$d$. A brief description of our construction follows.

First, we define a function $I:\Z^d\to\Z\cup\{{\rm ``layered"}\}$, which partitions $\Z^d$ into different regions according to the height of the Dobrushin interface, as follows:
\begin{enumerate}
\item $I(v)=k$ if $\sigma^{\eta,\Lambda,\Dob}$ has a \emph{unique} sign change in $\{v\}\times\Z$, with $\sigma^{\eta,\Lambda,\Dob}_{(v,k)}=-1$ and $\sigma^{\eta,\Lambda,\Dob}_{(v,k+1)}=1$,
\item $I(v)={\rm ``layered"}$ if $\sigma^{\eta,\Lambda,\Dob}$ has \emph{multiple} sign changes in $\{v\}\times\Z$.
\end{enumerate} Define the set $V_{\sigma^{\eta,\Lambda,\Dob}} \subset Z^d$ (the ``projected interface vertices'') as those $v$ satisfying that there exists an edge $\{u,v\}\in E(\Z^d)$ with either $I(u) \neq I(v)$ or $I(u)=I(v)=``\text{layered}''$ (i.e., all layered vertices and their neighbors and all non-layered vertices having a neighbor with a different value of $I$). We then define $\tilde{A}$ to be the union of those connected components of $V_{\sigma^{\eta,\Lambda,\Dob}}$ which surround $E$, i.e., those connected components $C$ for which some vertex of $E$ lies in a finite connected component of $\Z^d\setminus C$ (see Figure \ref{fig:pre_shift} in Section \ref{sec:s_0}). 

Second, we define a ``pre-shift'' $\tau_0:\Z^d\to\Z\cup\{\text{``layered''}\}$ as follows: For $v\in\tilde{A}$ we set $\tau_0(v)=I(v)$. For $v\notin\tilde{A}$, we let $B_v$
be the connected component of $v$ in $\Z^d\setminus\tilde{A}$ and observe that $I$ is necessarily some constant integer on the external vertex boundary of $B_v$; then we set $\tau_0(v)$ equal to this constant (necessarily $\tau_0(v)=0$ if $B_v$ is infinite). 

Third, the requisite shift $\tau$ is formed from $\tau_0$ by setting $\tau(v)=\tau_0(v)$ whenever $\tau_0(v)\in\Z$ and choosing values $\tau(v)\in\Z$ at those $v$ where $\tau_0(v)=\text{``layered''}$ (such $v$ are necessarily in~$\tilde{A}$). While our choice is limited by the compatibility relation~\eqref{eq:third compatibility relation}, this still leaves it significant freedom; the main limiting factor is our requirement that $\tau$ be an admissible shift. To choose the values we use our Lemma~\ref{lem:ground configuration with no overhangs}, which gives a mechanism for modifying the configuration $\sigma^{\eta,\Lambda,\Dob}$ on each connected component of layered vertices into a configuration $\sigma'$ with the properties: (i) $\sigma'$ has no overhangs,  (ii) if $\sigma'_{(v,k)}=-\sigma'_{(v,k+1)}$ at some $(v,k)$ then the same holds for $\sigma^{\eta,\Lambda,\Dob}$ at $(v,k)$, and (iii) $\sigma'$ has a fewer or equal number of perpendicular interface plaquettes than $\sigma^{\eta,\Lambda,\Dob}$.
Choosing $\tau$ on layered vertices to be the height of the unique sign change in $\sigma'$ is shown to yield the requisite admissible shift.

\smallskip
\paragraph{\emph{Chaining and concentration}} 
The above discussion implies that on the event~\eqref{eq:excitation at v_0} there exists an admissible shift $\tau$ inducing an energy  in~\eqref{eq:energy gap} (and this gap is large if $k_0$ of~\eqref{eq:excitation at v_0} is large). Consequently, it remains to prove Theorem~\ref{thm:Existence of admissible ensemble is unlikely} below, which states that it is unlikely that there exists any admissible shift producing an energy gap which is large in absolute value. To this end, motivated by the chaining expansion~\eqref{eq:chaining RFIM} of the RFIM, our first step is to write
\begin{equation}\label{eq:chaining disordered ferromagnet}
\GE^\Lambda(\eta) - \GE^\Lambda(\eta^\tau) = \sum_{k=0}^{K-1} \left(\GE^\Lambda(\eta^{\tau_{2^{k+1}}})-\GE^\Lambda(\eta^{\tau_{2^k}})\right)
\end{equation}
where $\tau_N$ represents some notion of $N$-coarse-graining of $\tau$, with $\tau_{2^0}=\tau_1=\tau$ and with $K$ large enough that $\tau_{2^K}\equiv 0$ (so that $\eta^{\tau_{2^K}}=\eta$). We choose to define $\tau_N:\Z^d\to\Z$ as a function which is constant on cubes of the form $v+\{0,1,\ldots,N-1\}^d$ with $v\in N\Z^d$, and equal on each such cube $B$ to the average of $\tau$ on $B$ rounded to the closest integer (arbitrarily rounding $k+1/2$ to $k$ for integer $k$). Significant effort is then devoted in Section~\ref{sec:basic properties of grainings of shifts} and Section~\ref{sec:Enumeration} to develop an enumeration theory (reminiscent of~\cite{FFS}) for the number of possibilities for $\tau_N$ according to the complexity of $\tau$ (complexity is discussed in Section~\ref{sec:admissible shifts overview} below). The proof also introduces an extra ``fine grained'' shift $\tau_I$, for $I\subset[d]=\{1,\ldots, d\}$, which ``lies between'' $\tau$ and $\tau_2$ and is obtained by averaging and rounding $\tau$ on boxes of the form $v+\{0,1\}^I\times\{0\}^{[d]\setminus I}$. This extra ingredient allows our assumptions on the disorder distributions (\eqref{eq:width assumption isotropic} and~\eqref{eq:anisotropic condition}) to become less restrictive as the dimension $d$ increases.

The next step following~\eqref{eq:chaining disordered ferromagnet} is to obtain a concentration inequality for the ground energy differences appearing in its right-hand side, similar to the concentration inequality~\eqref{eq:RFIM GE concentration2} of the RFIM. Here, however, lies a major hurdle in our analysis, as the available inequality is significantly weaker than the one available for the RFIM or the disordered SOS model. Let us describe the inequality that we have. Let $\tau_1,\tau_2$ be shift functions. We introduce a version of the ground energy in which we minimize over a restricted set of configurations: For $A\subset\Z^d$ and $b^\parallel,b^\perp\ge 0$, let $\GE^{\Lambda,A,(b^\parallel,b^\perp)}(\eta)$ be the minimal energy in the coupling field $\eta$ among configurations in $\Omega^{\Lambda\times\Z,\rho^{\Dob}}$ which have at most $b^\parallel$ parallel plaquettes and at most $b^\perp$ perpendicular plaquettes above $A$ in the Dobrushin interface (see~\eqref{eq:restricted ground energy} for a precise definition). Then (see Lemma~\ref{lem:concentration for energetic gain conditioned on layering bound})
\begin{multline}\label{eq:two point estimate overview}
\P\left(\big|\GE^{\Lambda,\supp(\tau_1-\tau_2),(b^\parallel,b^\perp)}(\eta^{\tau_1}) - \GE^{\Lambda,\supp(\tau_1-\tau_2),(b^\parallel,b^\perp)}(\eta^{\tau_2})\big|\ge t\right)\\
\le C\exp \left(-c \frac{t^2}{\wid (\nu^{\parallel})^{2}  b^{\parallel}+\wid (\nu^{\perp})^{2} b^{\perp}} \right),
\end{multline} 
so that the concentration estimate deteriorates as $b^\parallel$ and $b^\perp$ grow.
Thus, in order to apply~\eqref{eq:two point estimate overview} to the $k$th term in~\eqref{eq:chaining disordered ferromagnet} (and successfully use a union bound over the possible $\eta^{\tau_{2^k}}$ and $\eta^{\tau_{2^{k+1}}}$) we need that for sufficiently small $b_k^\parallel(s)$ and $b_k^\perp(s)$ (depending on $k$ and the energy gap $s$; see Lemma~\ref{lem:layering bounds for given maximal gain}), \begin{equation}\label{eq:layering is met}
\begin{split}
&\GE^\Lambda(\eta^{\tau_{2^k}})=\GE^{\Lambda,\supp(\tau_{2^{k+1}}-\tau_{2^k}),(b_k^\parallel(s),b_k^\perp(s))}(\eta^{\tau_{2^k}}),\\
&\GE^\Lambda(\eta^{\tau_{2^{k+1}}})=\GE^{\Lambda,\supp(\tau_{2^{k+1}}-\tau_{2^k}),(b_k^\parallel(s),b_k^\perp(s))}(\eta^{\tau_{2^{k+1}}}).
\end{split}
\end{equation}
However, we are not guaranteed that these equalities hold!

\smallskip
\paragraph{\emph{The maximal energy gap.}}\label{sec:maximal energy gap overview} It remains to deal with the case that one of the inequalities~\eqref{eq:layering is met} is violated. Our strategy is to show that in this case there is a new admissible shift $\tau'$ inducing a significantly larger absolute energy gap $|\GE^\Lambda(\eta) - \GE^\Lambda(\eta^{\tau'})|$ than the shift $\tau$. The argument then proceeds by focusing on the admissible shift with the \emph{maximal} energy gap and deducing that for that shift all the equalities~\eqref{eq:layering is met} hold.

To this end, suppose, e.g., that the first equality in~\eqref{eq:layering is met} is violated. Set $E := \supp(\tau_{2^{k+1}}-\tau_{2^k})$. By definition, this means that $\sigma^{\eta^{\tau_{2^k}},\Lambda,\Dob}$ either has more than $b_k^\parallel(s)$ parallel interface plaquettes above $E$ or has more than $b_k^\perp(s)$ perpendicular interface plaquettes above $E$. We may thus use the construction of Section~\ref{sec:tau and tilde A def overview}, with $\sigma^{\eta^{\tau_{2^k}},\Lambda,\Dob}$ and $E$, to define $\tau'$ and $\tilde{A}'$ inducing a large energy gap $\GE^\Lambda(\eta^{\tau_{2^k}}) - \GE^\Lambda((\eta^{\tau_{2_k}})^{\tau'})$. If $b_k^\parallel(s)$ and $b_k^\perp(s)$ are not too small (see Lemma~\ref{lem:layering bounds for given maximal gain} for their value) then the new gap will indeed be much greater than the old one, as we require. One difficulty, however, is that the new gap is induced for the shifted disorder $\eta^{\tau_{2^k}}$ rather than for the original disorder~$\eta$. This is simply resolved though, since
\begin{equation}\label{eq:new shift from old}
\begin{split}
&\GE^\Lambda(\eta^{\tau_{2^k}}) - \GE^\Lambda((\eta^{\tau_{2^k}})^{\tau'}) = \GE^\Lambda(\eta^{\tau_{2^k}}) - \GE^\Lambda(\eta^{\tau_{2^k}+\tau'})\\
&= \left(\GE^\Lambda(\eta^{\tau_{2^k}}) - \GE^\Lambda(\eta)\right) - \left(\GE^\Lambda(\eta^{\tau_{2^k}+\tau'}) - \GE^\Lambda(\eta)\right)
\end{split}
\end{equation}
so that a large energy gap induced for the shifted disorder $\eta^{\tau_{2^k}}$ implies a large energy gap in absolute value for the original disorder (induced either by the shift $\tau_{2^k}$ or by the shift $\tau_{2^k}+\tau'$).

\smallskip
\paragraph{\emph{Admissible shifts.}}\label{sec:admissible shifts overview} The above argument may give rise to shift functions with a complicated structure. Initially, given the input~\eqref{eq:excitation at v_0}, we construct a relatively simple shift $\tau$ (and set $\tilde{A}$) in order to remove the interface walls surrounding the vertex $v_0$. However, as explained in Section~\ref{sec:maximal energy gap overview} above, we may need to replace the shift $\tau$ by the shifts $\tau_{2^k}$ or $\tau_{2^k}+\tau'$ appearing in~\eqref{eq:new shift from old}, and upon iterating this procedure the shifts may become more and more complicated.
We thus need to define a class of shifts which, on the one hand, is broad enough to be closed under such operations and, on the other hand, is narrow enough to enable efficient enumeration (of the number of possibilities for the shift and its coarse grainings), allowing the union bounds in the chaining argument to go through. This is our motivation for defining in Section~\ref{sec:admissible shifts} the class of \emph{admissible} shifts, which depends on the coupling field $\eta$. We measure the complexity of a shift $\tau$ by its total variation $\TV(\tau)$ (i.e., the $\ell_1$-norm of its gradient) and by a quantity $R(\tau)$ that we call \emph{trip entropy}, which is the minimal length of a path visiting all level components of $\tau$ (i.e., visiting all connected components of level sets of $\tau$).
Admissible shifts are then defined as those that induce an energy gap for the coupling field $\eta$ that is sufficiently large compared to the complexity of the shift. This definition turns out to strike the requisite balance between broadness and narrowness.

\subsection{Reader's guide}\label{sec:readers guide}

Section~\ref{sec:notation} presents basic conventions and graph notation used throughout the paper. We also include there the statements of two concentration inequalities for Lipschitz functions of independent random variables and their adaptation to our setup. 

The short Section~\ref{sec:Disorders which are constant on transversal plaquettes} is devoted to showing that when the coupling field $\eta$ is constant on perpendicular plaquettes then there exists a ground configuration with no overhangs. 

Section~\ref{sec:4} introduces the concept of shifts, their properties (including total variation, level components and trip entropy), their action on coupling fields, and the notion of an admissible shift. It is here that we define the important concept of the energy gap of a shift: the change in ground energy (under Dobrushin boundary conditions) when the shift acts on the disorder. We then state our main technical result, Theorem~\ref{thm:Existence of admissible ensemble is unlikely}, which bounds the probability that there exists an admissible shift with a large energy gap. Complementing this, we state there (Lemma~\ref{lem:spin different than sign shift existence}) that such a shift necessarily exists if the interface of the ground configuration is not flat around the origin. Together, these statements imply our main result on localization, Theorem~\ref{thm:localization}.

Section~\ref{sec:coarse and fine graining} defines the coarse, and fine, grainings of shifts, which provide a chaining approach to proving Theorem~\ref{thm:Existence of admissible ensemble is unlikely} (inspired by~\cite{FFS}, who introduced coarse graining of \emph{sets} in their chaining argument for the random-field Ising model). This chaining argument is written for the admissible shift attaining the maximal energy gap. We postpone the bounds on individual ``links of the chain'' to the next section.

Section~\ref{sec:Concentration of ground-energy differences between consecutive grainings} develops tools for deducing the requisite probability bounds of Section~\ref{sec:coarse and fine graining}: concentration for ground energy differences given a priori bounds on interface layering (``number of excessive faces in the interface'') and bounds for the number of shifts and their grainings of a given complexity (in the sense of total variation and trip entropy). The a priori bounds on interface layering are reduced to showing that excessive layering contradicts the fact that the shift attains the maximal energy gap.

Section~\ref{sec:obtaining admissible shifts} is then devoted to showing that admissible shifts with a high energy gap exist in two scenarios, proving Lemma~\ref{lem:spin different than sign shift existence} and Proposition \ref{prop:existance_Dob_shift_ground_config}, and thereby finishing the proofs of Theorem~\ref{thm:Existence of admissible ensemble is unlikely} and Theorem~\ref{thm:localization}.

Section~\ref{sec:proof of convergence theorem} proves our results on infinite-volume limits, Theorem~\ref{thm:convergence}, Corollary~\ref{cor:covariant Dobrushin configuration} and Theorem~\ref{thm:decay of correlations and rate of convergence}. Their proof relies on Theorem~\ref{thm:Existence of admissible ensemble is unlikely} and the construction of a suitable admissible shift to deduce that, with high probability, the ground configurations in a sequence of growing domains stabilize in the neighborhood of the origin.

Section~\ref{sec:discussion and open problems} provides further discussion and a selection of open problems and conjectures.

Appendix~\ref{app:BB} shows that primitive contours in the sense of Bollob\'as--Balister ~\cite{BB} (following Lebowitz--Mazel~\cite{LM}) are the same as the edge boundaries of connected sets in $\Z^d$ whose complement is also connected.

Appendix~\ref{sec:small lemmas} proves various statements related to the passage from ground configurations in finite volume to ground configurations in semi-infinite volume, $\Lambda\times\Z$ for $\Lambda\subset\Z^d$ finite.
\section{Notation, conventions and concentration results}
\label{sec:notation}

We use the convention $\N:=\{1,2,\ldots\}$. 
For $k\in\N$, we let $[k]:=\{1,2,\ldots, k\}$ and for any set $A$, let
\begin{equation*}
\binom{A}{k} := \{I\subseteq A \colon |I|=k\}
\end{equation*}
be the family of subsets of size $k$ of $A$.

For $x \in \R^m$ and $p\ge 1$ we let
$\lVert x \rVert_{p}= (\sum_{i=1}^{m}|x_i|^p)^{1/p}$ be the standard $p$-norm.

Unless explicitly stated otherwise, 
all ``geometric" notions in ${\mathbb Z}^d$ are with respect to the $\ell_1$ metric. In particular, the (closed) ball of radius $r\geq0$ around $a\in{\mathbb Z}^d$ is 
$$
{\mathcal B}_r(a):=\{v\in{\mathbb Z}^d\colon  \lVert v-a\rVert_1\leq r\},
$$
the diameter of a bounded set $A\subset {\mathbb Z}^d$ is ${\rm diam}(A)=\max_{u_1,u_2\in A}\lVert u_1-u_2\rVert_1$, the distance from $\omega\in{\mathbb Z}^d$ to a non-empty set $A\subset{\mathbb Z}^d$ is ${\rm dist}(\omega,A)=\min_{u\in A}\lVert \omega-u\rVert_1$ and the distance between two non-empty sets $A,B\subset {\mathbb Z}^d$ is ${\rm dist}(A,B)=\min_{u\in A,\,v\in B}\lVert u-v\rVert_1$;
we say that $u,v\in{\mathbb Z}^d$ are adjacent, and denote it by $u\sim v$, if $\lVert u-v\rVert_1=1$;
let $E({\mathbb Z}^d):=\{\{u,v\}\in\binom{{\mathbb Z}^d}{2}\colon u\sim v\}$;
the \emph{edge boundary} of a 
set $A\subset{\mathbb Z}^d$ is 
$$
\partial A:=\{(u,v)\colon u\in A,\, v\in {\mathbb Z}^d\setminus A,\, u\sim v\}
$$
its \emph{inner vertex boundary} is
$$
\partial^{\rm in} A:=\{u\in A \colon \exists v\in {\mathbb Z}^d\setminus A \text{ such that }  u\sim v\},
$$
and its \emph{outer vertex boundary} is
$$
\partial^{\rm out} A:=\{v\in \Z^d\setminus A \colon \exists u\in A \text{ such that }  u\sim v\}.
$$

Denote by $\pi$ the projection from $\Z^{d+1}$  to $\Z^{d}$ defined by $\pi(x_{1},\dots,x_{d},x_{d+1})=(x_1,\dots,x_d)$.

\medskip
The proofs of our main results require a concentration inequality for the minimal energy of configurations of the disordered ferromagnet. According to whether the disorder distributions have compact support or are Lipschitz functions of a Gaussian, one of the following two inequalities will be used.

A function $f:D\to\R$, defined on a subset $D\subset\R^n$, is said to be Lipschitz with constant $L>0$ if
\begin{equation}\label{eq:Lipschitz function def}
|f(x)-f(y)|\le L\|x-y\|_2,\qquad x,y\in D.
\end{equation}
\begin{theorem}
[Gaussian concentration inequality; see, e.g.~\cite{BLM13}*{Theorem 5.6}]\label{thm:Maurey-Pisier}
\mbox{}

\noindent 
Let $g_1,\ldots, g_n$ be independent standard Gaussian random variables. Suppose $f\colon \R^n\to\R$ is Lipschitz with constant $L$.
Set $X := f(g_1,\ldots, g_n)$. Then $\E(|X|)<\infty$ and for each $t>0$,
\begin{equation*}
\P(|X - \E(X)|\ge t)\le 2e^{-\frac{t^2}{2L^2}}.
\end{equation*}
\end{theorem}
The theorem is part of the Gaussian concentration phenomenon as initiated by Paul L\'{e}vy, Christer Borell, Tsirelson--Ibragimov--Sudakov and Maurey--Pisier.

A function $f:\R^n\to\R$ is called \emph{quasi-convex} if $\{x\in\R^n\colon f(x)\le s\}$ is a convex set for every $s\in\R$.

\begin{theorem}[\cite{BLM13}*{Theorem 7.12}, going back to Johnson--Schechtman \cite{JS91}, following Talagrand \cite{T88}]\label{thm:Lipschitz fcn of bdd} Let $z_1, . . . , z_n$ be independent random variables taking values in the interval $[0, 1]$ and let $f:[0,1]^n\to\R$ be a quasi-convex function which is also Lipschitz with constant $1$.
Set $X:=f(z_1,\ldots, z_n)$. Then, for each $t > 0$,
\begin{equation}\label{eq:med conc Lipschitz fcn of bdd}
\P(|X - \med(X)|\ge t)\le 4e^{-\frac{t^2}{4}}
\end{equation}
where $\med(X)$ is any median of $X$.
\end{theorem}

A function $f:\R^n\to\R$ is called \emph{quasi-concave} if $-f$ is quasi-convex. Theorem \ref{thm:Lipschitz fcn of bdd} clearly applies to quasi-concave functions as well.

We remark that it is standard (and simple; see \cite{MS86}*{p. 142}) that~\eqref{eq:med conc Lipschitz fcn of bdd} implies the same conclusion with the median replaced by the (necessarily finite) expectation, in the form
\begin{equation}\label{eq:exp conc Lipschitz of bdd}
\P(|X - \E(X)|\ge t)\le Ce^{-ct^2}
\end{equation}
for some universal constants $C,c>0$.

For our later use, it is convenient to deduce a unified result from the previous two theorems, applicable to distributions of finite width. For a random variable $W$ we set
\begin{equation}
\wid(W) := \wid(\mathcal{L}(W))
\end{equation}
where $\mathcal{L}(W)$ is the distribution of $W$ (and $\wid(\mathcal{L}(W))$ is defined by~\eqref{eq:width dist def}).
\begin{corollary}\label{cor:concentration with bounded width}
There exist $C,c>0$ such that the following holds. Let $W_1,\ldots, W_n$ be independent random variables with $0<\wid(W_i)<\infty$ for all $i$. Suppose $f\colon \R^n\to\R$ is a quasi-convex or a quasi-concave function which is Lipschitz with constant $L>0$ in the sense of~\eqref{eq:Lipschitz function def}. 
Set
\begin{equation}
X := f\left(\frac{W_1}{\wid(W_1)},\ldots, \frac{W_n}{\wid(W_n)}\right). 
\end{equation}
Then $\E(|X|)<\infty$ and for each $t>0$,
\begin{equation}
\P(|X - \E(X)|\ge t)\le Ce^{-c\frac{t^2}{L^2}}.
\end{equation}
\end{corollary}

We remark regarding the restriction $\wid(W_i)>0$ that a distribution $\nu$ with $\wid(\nu)=0$ is supported on a single point. Indeed, this is clear if $\wid(\nu)=\diam(\nu)$, while if $\wid(\nu)=\lip(\nu)$ then one may either argue directly or deduce the fact from Theorem~\ref{thm:Maurey-Pisier}.
\begin{proof}[Proof of Corollary~\ref{cor:concentration with bounded width}]
Let us assume, without loss of generality, that $0\le k\le n$ is such that $\wid(W_i) = \lip(\mathcal{L}(W_i))$ for $1\le i\le k$ while $\wid(W_i)=\diam(\mathcal{L}(W_i))$ for $k+1\le i\le n$. By subtracting suitable constants from the $W_i$ with $k+1\le i\le n$ we may further assume, without loss of generality, that each such $W_i$ is supported on an interval of the form $[0,a_i]$ with $\diam(W_i) = a_i$. This implies that $W_i / \wid(W_i)\in[0,1]$ for $k+1\le i\le n$, as will be required for using Theorem~\ref{thm:Lipschitz fcn of bdd}.
    
It suffices to prove that for any $t>0$ we have
\begin{equation}
\label{eq:conc wrt cond expectation}
\P(|X - \E(X\,|\, W_1,\ldots, W_k)|\ge t\,|\, W_1, \ldots, W_k)\le Ce^{-c\frac{t^2}{L^2}}
\end{equation}
almost surely, and
\begin{equation}\label{eq:conc of cond expectation}
\P(|\E(X | W_1,\ldots, W_k) - \E(X)|\ge t)\le Ce^{-c\frac{t^2}{L^2}}.
\end{equation}    
Inequality~\eqref{eq:conc wrt cond expectation} follows from Theorem~\ref{thm:Lipschitz fcn of bdd}, in the form~\eqref{eq:exp conc Lipschitz of bdd}. To see this, first note that $f/L$ is a quasi-convex or a quasi-concave function which is Lipschitz with constant $1$. Conclude that, for each fixed values of $x_1,\ldots, x_k\in\R$, the restricted function $x_{k+1},\ldots,x_n\mapsto f(x_1,\ldots, x_k, x_{k+1},\ldots, x_n)/L$ satisfies the same properties, and finally recall that $W_i / \wid(W_i)\in[0,1]$ for $k+1\le i\le n$.

We proceed to deduce inequality~\eqref{eq:conc of cond expectation}. Observe first that the average of a Lipschitz function with respect to some of its variables is still a Lipschitz function, with the same constant, of the remaining variables. In particular, the function
\begin{equation}
\tilde{f}(x_1,\ldots, x_k):= \E\left(f\left(x_1,\ldots, x_k,\frac{W_{k+1}}{\wid(W_{k+1})},\ldots, \frac{W_n}{\wid(W_n)}\right)\right)
\end{equation}
is Lipschitz with constant $L$. 
Fix $\eps>0$. Let $g_1,\ldots, g_k$ be a independent standard Gaussian random variables. Write, for $1\le i\le k$, $W_i = h_i(g_i)$ where $h_i:\R\to\R$ satisfies $\lip(h_i)\le \lip(W_i)(1+\eps)$. It follows that $(y_1,\ldots, y_k)\mapsto\tilde{f}\left(\frac{h_1(y_1)}{\wid(W_1)},\ldots, \frac{h_k(y_k)}{\wid(W_k)}\right)$ is a Lipschitz function with constant $L(1+\eps)$.  
Inequality~\eqref{eq:conc of cond expectation} then follows from Theorem~\ref{thm:Maurey-Pisier}, taking into account that $\eps$ is arbitrary.
\end{proof}

\section{Disorders which are constant on perpendicular plaquettes}
\label{sec:Disorders which are constant on transversal plaquettes}

Say that an Ising configuration $\sigma\colon \Z^{d+1}\to\{-1,1\}$ is \emph{interfacial} if for every $v\in\Z^d$
\begin{equation}
\text{$\lim_{k\to-\infty}\sigma_{(v,k)}=-1$ and $\lim_{k\to\infty}\sigma_{(v,k)}=1$}.
\end{equation}
A configuration $\sigma$ is said to have \emph{no overhangs} if it is interfacial and for every $v\in\Z^d$, there is a \emph{unique} $k$ for which $\sigma_{(v,k)}=-\sigma_{(v,k+1)}$.

Recall the definition of $\Omega^{\Delta,\rho}$ and the definition of a ground configuration in $\Omega^{\Delta, \rho}$ from~\eqref{eq:configuration space in Z^D}. We use these here with $\Delta = \Lambda\times\Z$ for a finite $\Lambda\subset\Z^d$ and a $\rho$ with no overhangs. Note that a ground configuration in $\Omega^{\Lambda\times\Z,\rho}$ may not exist for a general coupling field $\eta$. However, such a ground configuration, which is moreover interfacial, will exist if $\inf_{e\in E(\Z^D)}\eta_e>0$ (see a related discussion after Observation~\ref{obs:different Hamiltonians equivalent}).

\begin{lemma}\label{lem:ground configuration with no overhangs}
Let $\Lambda\subset\Z^d$ be finite and let $\rho\colon \Z^{d+1}\to\{-1,1\}$ have no overhangs. Suppose the coupling field  $\eta\colon E(\Z^{d+1})\to[0,\infty)$ satisfies that $\eta$ is \emph{constant} on $E^{\perp}(\Z^{d+1})$. Then for each interfacial configuration $\sigma\in\Omega^{\Lambda\times\Z,\rho}$ there exists $\sigma'\in\Omega^{\Lambda\times\Z,\rho}$ \emph{with no overhangs} such that $H^\eta(\sigma')\le H^\eta(\sigma)$ and whenever $\{x,x+e_{d+1}\}\in E(\Z^{d+1})$ is such that $\sigma'_{x}=-1$ and $\sigma'_{x+e_{d+1}}=1$ then also $\sigma_{x}=-1$ and $\sigma_{x+e_{d+1}}=1$.
 
Consequently, if $\eta$ is such that there exists an interfacial ground configuration in $\Omega^{\Lambda\times\Z,\rho}$, then there also exists a ground configuration in $\Omega^{\Lambda\times\Z,\rho}$ which has no overhangs.
 
\end{lemma}
We note that in the terminology of~\eqref{eq:odd and even sign changes} below, the lemma asserts that the sign changes of~$\sigma'$ (having no overhangs) are contained in the odd sign changes of $\sigma$.

The proof of the lemma uses the following preliminary definitions and proposition.

Fix a finite $\Lambda\subset\Z^d$ and a configuration $\rho$ with no overhangs. We make the following definitions for an interfacial configuration $\sigma\in\Omega^{\Lambda\times\Z,\rho}$:
\begin{enumerate}
\item 
The next definitions capture a notion of ``odd'' and ``even'' sign changes in $\sigma$,
\begin{equation}\label{eq:odd and even sign changes}
\begin{split}
&\osc(\sigma):=\{\{x,x+e_{d+1}\}\in E(\Z^{d+1})\colon \sigma_x = -1, \sigma_{x+e_{d+1}}=1\},\\
&\esc(\sigma):=\{\{x,x+e_{d+1}\}\in E(\Z^{d+1})\colon \sigma_x = 1, \sigma_{x+e_{d+1}}=-1\}.
\end{split}
\end{equation}
Note that as $\rho$ has no overhangs and $\sigma$ is interfacial, then
\begin{itemize}
\item 
for each $v\in\Z^d$ there are finitely many $k$ for which $\{(v,k),(v,k+1)\}\in\osc(\sigma)$, with a unique such $k$ when $v\in\Z^d\setminus\Lambda$.
\item 
for each $v\in\Z^d$, the number of $\{(v,k),(v,k+1)\}\in\esc(\sigma)$ equals the number of $\{(v,k),(v,k+1)\}\in\osc(\sigma)$ minus $1$. In particular,
if $\{(v,k),(v,k+1)\}\in\esc(\sigma)$ then $v\in\Lambda$.
\end{itemize}
\item 
Let $\nesc(\sigma)$ be the number of ``adjacent even sign changes'' in $\sigma$, defined as the number of pairs $\{\{(u,k),(u,k+1)\}, \{(v,\ell),(v,\ell+1)\}\}\subset\esc(\sigma)$ satisfying that $\{u,v\}\in E(\Z^d)$ and $k=\ell$.
\item 
Define the number of perpendicular domain wall plaquettes above $\Lambda$ to be
\begin{equation*}
D^\Lambda(\sigma):=|\{\{x,y\}\in E^{\perp}(\Z^{d+1}) \colon \{x,y\}\cap(\Lambda\times \Z)\neq\emptyset,\, \sigma_x\neq\sigma_y\}|.
\end{equation*}
\end{enumerate}
Finally, we define a partial order on interfacial configurations in $\Omega^{\Lambda\times\Z,\rho}$ as follows: Say that $\sigma'<\sigma$ if \begin{equation*}
D^\Lambda(\sigma')\le D^\Lambda(\sigma)
\end{equation*}
and, either
\begin{equation*}
\osc(\sigma')\subsetneq\osc(\sigma)
\end{equation*}
or 
\begin{equation*}
\osc(\sigma')=\osc(\sigma)\text{ and } \nesc(\sigma')>\nesc(\sigma).
\end{equation*}
The following proposition is the key step in proving Lemma~\ref{lem:ground configuration with no overhangs}.
\begin{proposition}\label{prop:partial order proposition}
Let $\sigma\in\Omega^{\Lambda\times\Z,\rho}$ be interfacial. 
If $\esc(\sigma)\neq\emptyset$ then there exists $\sigma'\in\Omega^{\Lambda\times\Z,\rho}$ such that $\sigma'<\sigma$ (in particular, $\sigma'$ is interfacial).
\end{proposition}
\begin{proof}
Fix a $\sigma\in\Omega^{\Lambda\times\Z,\rho}$ with $\esc(\sigma)\neq\emptyset$. Fix some $\{(v_0,k_0),(v_0,k_0+1)\}\in\esc(\sigma)$. Consider the set of all positions which are directly below even sign changes at height $k_0$,
\begin{equation*}
\Delta:=\{v\in\Lambda\colon \{(v,k_0),(v,k_0+1)\}\in\esc(\sigma)\}.
\end{equation*}
For a given height $k$, define the sum of the configuration surrounding $\Delta$ at height $k$,
\begin{equation*}
S(k):=\sum_{v\in \partial^{\rm out}\Delta}\sigma_{(v,k)}.
\end{equation*}
The definition of $\Delta$ implies that $S(k_0)\le S(k_0+1)$. Thus, either $S(k_0)\le 0$ or $S(k_0+1)\ge 0$ (or both). Let us assume without loss of generality that $S(k_0)\le 0$
as the other case can be treated analogously.

Define $k_1\le k_0$ to be the smallest integer with the following property: For all $k_1\le k\le k_0$ it holds that
\begin{equation*}
\begin{split}
&\sigma_{(v,k)} = \sigma_{(v,k_0)} = 1\quad\text{for $v\in\Delta$},\\
&\sigma_{(v,k)}\le \sigma_{(v,k_0)}\quad\text{for $v\in\partial^{\rm out}\Delta$}.
\end{split}
\end{equation*}
The definition implies, in particular, that
\begin{equation}\label{eq:non-positive S of k}
S(k)\le S(k_0)\le 0 
\end{equation}
for all $k_1\le k\le k_0$. Finally, define a configuration $\sigma'$ as follows
\begin{equation*}
\sigma'_{(v,k)} = \begin{cases}
-1&v\in\Delta, k_1\le k\le k_0,\\
\sigma_{(v,k)}&\text{otherwise}.
\end{cases}
\end{equation*}
The inequality~\eqref{eq:non-positive S of k} implies that $D^\Lambda(\sigma')\le D^\Lambda(\sigma)$. Moreover, the definition of $k_1$ implies that either $\osc(\sigma')\subsetneq\osc(\sigma)$ or $\osc(\sigma')=\osc(\sigma)$ and $\nesc(\sigma')>\nesc(\sigma)$. Thus, $\sigma'<\sigma$, as we wanted to prove.
\end{proof}

A repeated use of Proposition \ref{prop:partial order proposition} yields the following corollary.

\begin{corollary}\label{cor:no_overhangs}
For every interfacial $\sigma\in\Omega^{\Lambda\times\Z,\rho}$, there is an interfacial configuration $\sigma'$ that has a unique sign change above every vertex (i.e., $\sigma'$ has no overhangs), with $\sigma$ having the same sign change at the same height, and $D^\Lambda(\sigma')\le D^\Lambda(\sigma)$, i.e.,
\begin{multline}\label{eq:no_overhangs}
\lvert\{\{x,y\}\in E^{\perp}(\Z^{d+1}) \colon \{x,y\}\cap(\Lambda\times \Z)\neq\emptyset,\, \sigma'_x\neq\sigma'_y\}\rvert\\
\leq \lvert\{\{x,y\}\in E^{\perp}(\Z^{d+1}) \colon \{x,y\}\cap(\Lambda\times \Z)\neq\emptyset,\, \sigma_x\neq\sigma_y\}\rvert.
\end{multline}  
\end{corollary}

\begin{proof}
If $\esc(\sigma)=\emptyset$ then $\sigma$ has no overhangs and we are done. Otherwise, apply Proposition~\ref{prop:partial order proposition} iteratively to produce a sequence $\sigma_m<\sigma_{m-1}<\cdots<\sigma_0 = \sigma$ of configurations in $\Omega^{\Lambda\times\Z,\rho}$, with $\esc(\sigma_m)=\emptyset$ (the iterations necessarily terminate at some finite $m\ge 1$ since the number of odd sign changes above each $v\in\Z^d$ cannot increase and the number of even sign changes above each $v\in\Lambda$ is no larger than the number of odd sign changes above $v$). 
Then, $\sigma_m$ has no overhangs, and by the definition of the partial order, $\osc(\sigma_m)\subset\osc(\sigma)$ and $D^\Lambda(\sigma_m)\le D^\Lambda(\sigma)$. 
\end{proof}

Lemma \ref{lem:ground configuration with no overhangs} immediately follows from Corollary \ref{cor:no_overhangs}.

\begin{proof}[Proof of Lemma~\ref{lem:ground configuration with no overhangs}]
Let $\eta\colon E(\Z^{d+1})\to[0,\infty)$ satisfy the properties in the lemma. Let $\sigma$ be an interfacial configuration in $\Omega^{\Lambda\times\Z,\rho}$. 
Let $\sigma'$ be the configuration guaranteed by  Corollary \ref{cor:no_overhangs}.
Since $\eta$ is constant on $E^{\perp}(\Z^{d+1})$, 
it follows from \eqref{eq:no_overhangs} 
that $H^\eta(\sigma')\le H^\eta(\sigma)$. 
\end{proof}

\section{Stability of the ground energy under shifts of the disorder and a deduction of Theorem~\ref{thm:localization}}
\label{sec:4}

In this section we present our main technical result, Theorem~\ref{thm:Existence of admissible ensemble is unlikely} below, which bounds the probability that certain ``admissible shifts of the disorder'' lead to a significant change in the energy of the ground configuration under Dobrushin boundary conditions. Our main localization theorem, Theorem~\ref{thm:localization}, follows by combining Theorem~\ref{thm:Existence of admissible ensemble is unlikely} with the fact, stated in Lemma~\ref{lem:spin different than sign shift existence} below,
that admissible shifts inducing large energy changes necessarily exist whenever the interface in the ground configuration deviates (in prescribed locations) from the flat interface. Theorem~\ref{thm:Existence of admissible ensemble is unlikely} will also be instrumental in the proof of Theorem~\ref{thm:convergence} 
(presented in Section~\ref{sec:proof of convergence theorem})
on the convergence of the semi-infinite-volume ground configurations in the infinite-volume limit.

We begin in Section~\ref{sec:shift preliminaries} with required definitions, continue in Section~\ref{sec:stability of the ground energy} with the statement of our main technical result and finally deduce Theorem~\ref{thm:localization} in Section~\ref{sec:deduction of localization theorem}.

\subsection{Preliminaries} 
\label{sec:shift preliminaries}

This section contains the required definitions of ground energies, shifts and their actions on the disorder, and admissibility of shifts. 

\subsubsection{Coupling fields, energies and ground configurations}\label{sec:disorders energies and ground configurations}

\mbox{}

\smallskip

\noindent{\bf Generic coupling fields.}
We often work with coupling fields $\eta$ whose values on all edges are uniformly bounded away from $0$.
In addition, in order to ensure uniqueness of finite-volume ground configurations we ask that the coupling field $\eta$ satisfies the assumption
\begin{equation}\label{eq:no finite binary zero combination}
\sum_{i=1}^k s_i \eta_{f_i}\neq 0,\quad k\in\N, \{s_i\}_{i=1}^k\subseteq\{-1,1\}, \{f_i\}_{i=1}^k\subset E(\Z^{d+1}) \text{ and } \{f_i\}_{i=1}^k \nsubseteq E^\perp(\Z^{d+1}).
\end{equation}
This is captured with the following notation: Given $\alpha^{\parallel},\alpha^{\perp}\in(0,\infty)$ let
\begin{equation}\label{eq:D alpha def}
\mathcal{D}(\alpha^{\parallel},\alpha^{\perp}):=\left\{\eta\colon E(\Z^{d+1})\to(0,\infty)\colon \substack{\eta_e\in(\alpha^{\parallel},
\infty)\text{ for } e\in E^\parallel(\Z^{d+1}),\\
\eta_e\in(\alpha^{\perp},\infty)\text{ for } e\in E^\perp(\Z^{d+1}),\\
\eta \text{ satisfies \eqref{eq:no finite binary zero combination}}} \right\}.
\end{equation}
In addition, we set
\begin{equation}
\mathcal{D}:=\bigcup_{\alpha^{\parallel}, \alpha^{\perp}\in(0,\infty)} \mathcal{D}(\alpha^{\parallel}, \alpha^{\perp}).
\end{equation}
Such disorders arise, almost surely, under our assumption that $\nu^{\parallel}$ and $\nu^{\perp}$ satisfy \eqref{eq:disorder distributions assumptions}.

\smallskip

\noindent{\bf Semi-infinite volumes and an equivalent Hamiltonian.} 
Most of our analysis will be done on the semi-infinite volume $\Lambda\times\Z$ for a finite $\Lambda\subset\Z^d$. Recalling the definition of the Dobrushin boundary conditions $\rho^{\Dob}$ from~\eqref{eq:rho Dob def} and of $\Omega^{\Delta,\rho}$ from \eqref{eq:configuration space in Z^D}, 
we work on the following configuration space
\begin{equation}\label{eq:Dobrushin space infinite volume}
\Omega^{\Lambda,\Dob}:=\left\{\sigma\in\Omega^{\Lambda\times{\mathbb Z},\rho^{\Dob}}\colon\sigma=\rho^{\Dob}\ \text{at all but finitely many points of $\Lambda\times\Z$} \right\}.
\end{equation}
Now, given also a coupling field $\eta\colon E(\Z^{d+1})\to[0,\infty]$ and a finite $\Lambda\subset\Z^d$ define the Hamiltonian
\begin{equation}\label{eq:Hamiltonian}
\mathcal{H}^{\eta,\Lambda}(\sigma):=\sum_{\substack{\{x,y\}\in E(\Z^{d+1})\\\{x,y\}\cap(\Lambda\times \Z)\neq\emptyset}} \eta_{\{x,y\}}(1-\sigma_x\sigma_y)=2\sum_{\substack{\{x,y\}\in E(\Z^{d+1})\\\{x,y\}\cap(\Lambda\times \Z)\neq\emptyset}} \eta_{\{x,y\}}1_{\sigma_x\neq\sigma_y}
\end{equation}
and note that it is well defined on $\Omega^{\Lambda,\Dob}$. 

From the following observation it follows that the minimizers of the Hamiltonian $\mathcal{H}^{\eta,\Lambda}$ in $\Omega^{\Lambda,\Dob}$ coincide with the ground configurations discussed in Lemma~\ref{lem:semi infinite volume ground configuration}. It is proved in appendix \ref{sec:small lemmas} for completeness.
 
\begin{obs}\label{obs:different Hamiltonians equivalent}
Let $\sigma, \sigma' \in \Omega^{\Lambda,\Dob}$, and $\eta:E(\Z^{D}) \rightarrow \left[0,\infty\right) $. The following holds \begin{equation*}
H^{\eta}(\sigma)- 
H^{\eta}(\sigma') =
\mathcal{H}^{\eta,\Lambda}(\sigma)- 
\mathcal{H}^{\eta, \Lambda}(\sigma').
\end{equation*}
\end{obs}

We note that when $\eta\in\mathcal{D}$ then there is a unique minimizer of $\mathcal{H}^{\eta,\Lambda}$ in $\Omega^{\Lambda,\Dob}$. Indeed, there are only finitely many configurations in $\Omega^{\Lambda,\Dob}$ whose energy is lower than $\mathcal{H}^{\eta,\Lambda}(\rho^{\Dob})$, and no two of them have equal energy by~\eqref{eq:no finite binary zero combination}. With a slight abuse of notation, we will denote this unique minimizer by $\sigma^{\eta,\Lambda,\Dob}$, noting that it coincides with the minimizer of Lemma~~\ref{lem:semi infinite volume ground configuration} under the assumptions there. We will use the terminology \emph{ground energy} to refer to the energy of the minimizing configuration. We thus define, for each $\eta\in\mathcal{D}$ and finite $\Lambda\subset\Z^d$,
\begin{equation}\label{eq:ground energy def}
\GE^{\Lambda}(\eta):=\mathcal{H}^{\eta,\Lambda}(\sigma^{\eta,\Lambda,\Dob}).
\end{equation}

\subsubsection{Shifts of the coupling field}
\label{sec:shifts of the disorder}
\mbox{}

\smallskip

\noindent{\bf Shifts and shifted coupling fields.}
We use the term \emph{shift} to denote any function~$\tau\colon \Z^d\to\Z$ which equals zero except at finitely many vertices.
We denote the (finite) support of $\tau$ by
\begin{equation*}
\supp(\tau):=\{v\in\Z^d\colon \tau(v)\neq 0\}.
\end{equation*}
We occasionally refer to the $\ell_1$ norm of a shift, $\|\tau\|_1:=\sum_{v\in\Z^d}|\tau_v|$.
The set of all shifts will be denoted by $\mathcal S$.

We define an operation of shifts on coupling fields $\eta$: 
first fix an arbitrary choice function $\iota:E({\mathbb Z}^d)\to{\mathbb Z}^d$ that chooses for each edge one of its endpoints, i.e., $\iota(e)\in e$ for every $e\in E({\mathbb Z}^d)$;
the shifted coupling field $\eta^\tau$ is defined by shifting the ``column of disorders'' above a base vertex $u$ by $\tau(u)$, and a similar shift up for ``columns'' above any base edge $\{u,v\}$ such that $\iota(\{u,v\})=u$. 
Precisely, given a shift $\tau$  and a disorder $\eta\colon E({\mathbb Z^{d+1}})\to [0,\infty)$, define $\eta^\tau\colon E({\mathbb Z^{d+1}})\to [0,\infty)$ as follows: for every $u\in\Z^d$ and $k\in{\mathbb Z}$,
\begin{equation}\label{eq:parallel disorder shift}
\eta^{\tau}_{\{(u,k),(u,k+1)\}}:=\eta_{\{(u,k+\tau(u)),(u,k+1+\tau(u))\}},
\end{equation}
and for every $\{u,v\}\in E(\Z^d)$ and $k\in{\mathbb Z}$, 
\begin{equation}\label{eq:transversal disorder shift}
\eta^{\tau}_{\{(u,k),(v,k)\}}:=\eta_{\{(u,k+\tau(\iota(\{u,v\}))),(v,k+\tau(\iota(\{u,v\})))\}}.
\end{equation}
Note that if $\tau(u)=\tau(v)$ for adjacent $u,v\in{\mathbb Z}^d$, then for every $k\in{\mathbb Z}$,
\begin{equation}
\label{eq:shifted_noise_perp}
\eta^{\tau}_{\{(u,k),(v,k)\}}=\eta_{\{(u,k+\tau(u)),(v,k+\tau(u))\}}=\eta_{\{(u,k+\tau(v)),(v,k+\tau(v))\}}.
\end{equation}

\noindent{\bf Changes to the ground energy.} Of central importance in our arguments will be the change in ground energy induced by shifts of the coupling field. 
This is captured by the following definition. For each $\eta\in\mathcal{D}$, finite $\Lambda\subset\Z^d$ and shifts $\tau,\tau'$ we set
\begin{equation}
G^{\eta,\Lambda}(\tau, \tau'):= \GE^{\Lambda}(\eta^{\tau'}) - \GE^{\Lambda}(\eta^{\tau}).
\end{equation}
We also abbreviate
\begin{equation}\label{eq:shift energetic gain}
G^{\eta,\Lambda}(\tau) := G^{\eta,\Lambda}(\tau,0)= \GE^{\Lambda}(\eta) - \GE^{\Lambda}(\eta^{\tau}).
\end{equation}
With these definitions, for any shifts $\tau_1,\ldots, \tau_k$ we have the telescopic sum
\begin{equation}\label{eq:telescopic sum}
G^{\eta,\Lambda}(\tau_1) = \sum_{i=1}^{k-1} G^{\eta,\Lambda}(\tau_i,\tau_{i+1}) + G^{\eta,\Lambda}(\tau_k).
\end{equation}

\subsubsection{Enumeration of shifts, admissible shifts and the maximal energetic change}\label{sec:admissible shifts}

The counting of various classes of shifts plays an important role in our arguments (the shifts play a role somewhat analogous to that of contours in the classical Peierls argument). To facilitate it, we need a way to succinctly describe shifts. To this end, the following notations regarding a shift $\tau$ are handy:
\begin{itemize}
\item 
The \emph{total variation} of $\tau$ is defined as 
\begin{equation*}
\TV{(\tau}):= 
\sum_{\{u,v\}\in E(\Z^d)} \lvert \tau(u)-\tau(v) \rvert.
\end{equation*}
\item   
A \emph{level component} of a shift $\tau$ is a connected set on which $\tau$ is constant and which is not strictly contained in another set with this property (i.e., a connected component of $\tau^{-1}(k)$ for some $k$).
Denote the collection of all level components of $\tau$ by $\mathcal{LC}(\tau)$.
\item 
A finite sequence $(v_i)_{i\ge 0}$ of points in $\Z^d$ with $v_0 = 0$ is a \emph{root sequence} for a collection $\mathcal F$ of sets in ${\mathbb Z}^d$ if there is a point of $\{v_i\}_{i\geq 0}$ in every set in $\mathcal F$. We further define the \emph{trip entropy} $R(\tau)$ of the shift $\tau$ as
\begin{equation*}
R(\tau):=\min\left\{\sum_{i\ge 1} \|v_i - v_{i-1}\|_1\colon (v_i)_{i\geq 0}\text{ is a root sequence for } \mathcal{LC}(\tau)\right\}.
\end{equation*}
Similarly, define the trip entropy $R(E)$ of a set $E\subseteq{\mathbb Z}^d$ as
\begin{equation*}
R(E):=\min\left\{\sum_{i\ge 1} \|v_i - v_{i-1}\|_1\colon \substack{(v_i)_{i\geq 0}\text{ is a root sequence for the collection} \\  \text{of connected components of } E}\right\}.
\end{equation*}   
\end{itemize}
These definitions are put to use in estimating the number of shifts in Proposition~\ref{prop:enumerate_shift_functions}.

We next define a restricted class of shifts, depending on $\eta$, that we term \emph{admissible shifts} (while restricted, the set of admissible shifts is still defined in a broad enough fashion to contain all the shifts arising in our proof). Very roughly, the class is defined as those shifts whose action on the coupling field induces a sufficiently large energetic change to the ground energy (as defined in~\eqref{eq:shift energetic gain}) to compensate for the number of shifts in the class. Here, the first notion one has for the number of shifts with given parameters is that coming from our later Proposition~\ref{prop:enumerate_shift_functions}. However, we will see later that this notion will need to be further refined in our argument, where we will also need to take care of the number of coarse grainings (and fine grainings) of our shifts. The need to account also for these more refined counting problems lies at the heart of our choice for the definition of root sequence above and the precise definition of admissible shifts below.

Given a coupling field $\eta\colon E(\Z^{d+1})\to[0,\infty]$, finite $\Lambda\subset\Z^d$ and positive $\alpha^{\parallel},\alpha^{\perp}$, the class of $(\alpha^{\parallel},\alpha^{\perp})$-admissible shifts is defined by 
\begin{equation*}
\AS^{\eta,\Lambda}(\alpha^{\parallel},\alpha^{\perp}) := \left\{\tau\in{\mathcal S}\colon |G^{\eta,\Lambda}(\tau)| \geq \max\left\{\frac{\alpha^{\perp}}{2}\TV(\tau), \min \{ \alpha^{\parallel},\alpha^{\perp}\} \frac{d}{200}R(\tau)\right\}\right\}.
\end{equation*}
Lastly, we give a notation to the maximal change in the ground energy that is induced by an $(\alpha^{\parallel},\alpha^{\perp})$-admissible shift,
\begin{equation}\label{eq:MG def}
\MG^{\eta,\Lambda}(\alpha^{\parallel},\alpha^{\perp}):=\sup_{\tau \in \AS^{\eta,\Lambda} (\alpha^{\parallel},\alpha^{\perp})} |G^{\eta,\Lambda}(\tau)|.
\end{equation}

Our proof will make use of the fact that $\MG^{\eta,\Lambda}(\alpha^{\parallel},\alpha^{\perp})<\infty$, almost surely, under suitable assumptions on the disorder distributions. This is implied by the following lemma, proved in appendix \ref{sec:small lemmas}.

\begin{lemma}\label{lem:finite number of admissible shifts}
In the anisotropic disordered ferromagnet, suppose the disorder distributions
$\nu^{\parallel}$ and $\nu^{\perp}$ satisfy \eqref{eq:disorder distributions assumptions}.
Then, for any finite $\Lambda\subset\Z^d$ and positive $\alpha^{\parallel},\alpha^{\perp}$ we have that $|\AS^{\eta,\Lambda}(\alpha^{\parallel},\alpha^{\perp})|<\infty$ almost surely.
\end{lemma}

\smallskip

\subsection{Stability of the ground energy}
\label{sec:stability of the ground energy}~

We proceed to state our main technical result, Theorem~\ref{thm:Existence of admissible ensemble is unlikely} below. It gives a quantitative bound on the probability that there exists an admissible shift whose action on the disorder yields a large change in the ground energy.

\begin{theorem}\label{thm:Existence of admissible ensemble is unlikely}
There exist constants $c_0,c,C>0$ such that the following holds. 
In the anisotropic disordered ferromagnet, suppose the disorder distributions
$\nu^{\parallel}$ and $\nu^{\perp}$ satisfy \eqref{eq:disorder distributions assumptions}.
Let $\kappa=\kappa(\nu^\parallel, \nu^\perp,d)$ be as in definition~\eqref{eq:kappa_definition}. 
Let $\underline{\alpha}^{\parallel}$ and $\underline{\alpha}^{\perp}$ be the minimums of the supports of $\nu^\parallel$ and $\nu^\perp$, as in~\eqref{eq:alphas}.
Let $D=d+1\ge 4$ and suppose that condition~\eqref{eq:anisotropic condition} holds (with the constant $c_0$).
Then the following holds for all finite $\Lambda\subset\Z^d$ and $t>0$,
\begin{equation}\label{eq:tau localization bound}
\P\left(\MG^{\eta,\Lambda}(\underline{\alpha}^{\parallel},\underline{\alpha}^{\perp}) \geq t \right)\leq C\exp \left(-\frac{c}{\kappa  d^2} \left( \frac{t}{\underline{\alpha}^{\perp}}\right)^{\frac{d-2}{d-1}} \right).
\end{equation}
Moreover, for small $t$ we have an improved dependence on dimension: 
if $t < \underline{\alpha}^{\perp} 2^d$ then
\begin{equation}
\P\left(\MG^{\eta,\Lambda}(\underline{\alpha}^{\parallel},\underline{\alpha}^{\perp}) \geq t \right)\leq C\exp \left(-\frac{c t}{\kappa \underline{\alpha}^{\perp}} \right). \end{equation}
\end{theorem}

The theorem will be proven at the end of Section \ref{sec:The chaining argument}. 

\subsection{Deduction of Theorem~\ref{thm:localization}}\label{sec:deduction of localization theorem}

The following deterministic lemma shows that if the interface of the ground configuration is not flat around the origin then there necessarily exists an admissible shift whose action on the coupling field induces a large change in the ground energy.

\begin{lemma} \label{lem:spin different than sign shift existence}
Let $\eta\in\mathcal{D}(\alpha^{\parallel}, \alpha^{\perp})$ for some  $\alpha^{\parallel}, \alpha^{\perp}>0$
and let $\Lambda \subset \Z^d$ be a finite subset. 
If $\sigma^{\eta,\Lambda,\Dob}_{(0,k)} \neq \rho^{\Dob}_{(0,k)}$ then there exists a shift $\tau\in \AS^{\eta,\Lambda}({\alpha}^{\parallel},{\alpha}^{\perp})$ for which
\begin{equation*}
G^{\eta,\Lambda}(\tau) \geq 2|k| {\alpha}^{\perp}.
\end{equation*}
\end{lemma}

The lemma is proved in Section~\ref{sec:4_4}.

\smallskip
Theorem~\ref{thm:localization} follows as a direct consequence of Lemma~\ref{lem:spin different than sign shift existence} and Theorem~\ref{thm:Existence of admissible ensemble is unlikely}. First, it suffices to establish the inequality~\eqref{eq:localization at v k} at the vertex $v=0$, since the choice of $\Lambda$ in Theorem~\ref{thm:localization} is arbitrary. Then, inequality~\eqref{eq:localization at v k} with $v=0$ follows directly by combining Lemma~\ref{lem:spin different than sign shift existence} with~\eqref{eq:tau localization bound}.

\section{Coarse and fine grainings of shifts and their use in proving the stability of the ground energy}\label{sec:coarse and fine graining}

In this section we take the first step towards proving Theorem~\ref{thm:Existence of admissible ensemble is unlikely}, describing a form of ``chaining'' argument on the set of admissible shifts which is used to control their effect on the ground energy. The notion of coarse grainings of shifts which lies at the heart of our chaining argument is modelled after a similar graining method for sets which was introduced by Fisher--Fr\"ohlich--Spencer~\cite{FFS} in their discussion of the domain walls in the random-field Ising model.

\subsection{Coarse and fine grainings of shifts}\label{sec:coarse and fine grainings}

The chaining argument is based on the notions of coarse and fine grainings of shifts that we now describe. 

Given a partition $\mathcal{P}$ of $\Z^d$ into finite sets and a shift $\tau$, we write $\tau_\mathcal{P}$ for the shift obtained by averaging the value of $\tau$ on each partition element of $\mathcal{P}$ and rounding to the closest integer. Precisely, we set
\begin{equation}\label{eq:average of tau}
\tau_\mathcal{P}(v):=\left[\frac{1}{|P(v)|}\sum_{u\in P(v)}\tau(u)\right]
\end{equation}
where we write $P(v)$ for the unique partition element of $\mathcal{P}$ containing $v$, and where $\left[a\right]$ is the rounding of $a$ to the nearest integer, with the convention $\left[k+\frac{1}{2}\right]=k$ for $k\in\Z$.

We make use of two special cases of the above definition:
\begin{itemize}
\item Coarse graining: Given an integer $N\ge 1$, we use the notation $\tau_N := \tau_{\mathcal{P}_N}$ (as in~\eqref{eq:average of tau}), with $\mathcal{P}_N$ is the following partition into discrete cubes of side length $N$,
\begin{equation*}
\mathcal{P}_N=\{Q_N(v)\}_{v\in N\Z^d}\quad\text{where}\quad Q_N(v):=v+\{0,1,\ldots,N-1\}^d. 
\end{equation*}
\item Fine graining: Given a subset of the coordinates $I\subset[d]$, we use the notation $\tau_I := \tau_{\mathcal{P}_I}$ (as in~\eqref{eq:average of tau}), with $\mathcal{P}_I$ is the following partition into discrete boxes with side length $2$ in the directions in $I$ and side length $1$ in the directions in $[d]\setminus I$,
\begin{equation*}
\mathcal{P}_I=\{Q_I(v)\}_{v\in (2\Z)^I\times \Z^{[d]\setminus I}}\quad\text{where}\quad Q_I(v):=v+\{0,1\}^I\times\{0\}^{[d]\setminus I}. 
\end{equation*}
\end{itemize}

\subsection{The chaining argument}\label{sec:The chaining argument} 
We work under the assumptions of Theorem~\ref{thm:Existence of admissible ensemble is unlikely}. Precisely, let the disorder $\eta$ be sampled as in the anisotropic disordered ferromagnet, with the disorder distributions
$\nu^{\parallel}$ and $\nu^{\perp}$ satisfying \eqref{eq:disorder distributions assumptions}.
Let $D\ge 4$ and suppose that condition~\eqref{eq:anisotropic condition} holds with a constant $c>0$ chosen sufficiently small for the arguments below.

Fix a finite $\Lambda\subset\Z^d$. For brevity, we will write 
$G$ for $G^{\eta,\Lambda}$, $\AS$ for $\AS^{\eta,\Lambda}(\underline{\alpha}^{\parallel},\underline{\alpha}^{\perp})$ (recall~\eqref{eq:alphas}), \emph{admissible} for $(\underline{\alpha}^{\parallel},\underline{\alpha}^{\perp})$-\emph{admissible}
and $\MG$ for $\MG^{\eta,\Lambda}(\underline{\alpha}^{\parallel},\underline{\alpha}^{\perp})$.
First, since $\MG<\infty$ almost surely by Lemma~\ref{lem:finite number of admissible shifts}, we have
\begin{equation*}
\P(\MG\ge t) = \sum_{k=0}^\infty\P\left(\MG\in[t2^k,t2^{k+1})\right).
\end{equation*}
Next, for any $s>0$, integer $K\ge 1$, integer $1\leq r \leq d$, any positive $(\gamma_j)_{j\in [K] \cup \{(0,r),(r,1) \}}$ with $\gamma_{(0,r)}+\gamma_{(r,1)}+\sum_{1\le k\le K}\gamma_k\le 1$ and any function $I_\tau$ which assigns a subset of $[d]$ of size $r$ to each shift $\tau$, we have the chaining argument (noting that the supremum is realized in~\eqref{eq:MG def} due to Lemma~\ref{lem:finite number of admissible shifts}, and also recalling~\eqref{eq:telescopic sum})
\begin{subequations}
\label{eq:term_tau_tau}
\begin{align}
\P(&\MG\in[s,2s))=\P\left(\{\MG\le 2s\}\cap\left\{\max_{\tau\in\AS} |G(\tau)|\ge s\right\}\right)\nonumber\\
=&\P\left(\{\MG\le 2s\}\cap\left\{\max_{\tau\in\AS} |G(\tau,\tau_{I_\tau})+G(\tau_{I_\tau},\tau_2)+\sum_{k=1}^{K-1}G(\tau_{2^k},\tau_{2^{k+1}})+G(\tau_{2^K})|\ge s\right\}\right)\nonumber\\
\le&\P\left(\{\MG\le 2s\}\cap\left\{\max_{\tau\in\AS} |G(\tau,\tau_{I_\tau})|\ge \gamma_{(0,|I_{\tau}|)} s\right\}\right)\label{eq:term_tau_tauI}\\
&+\P\left(\{\MG\le 2s\}\cap\left\{\max_{\tau\in\AS} |G(\tau_{I_\tau},\tau_2)|\ge \gamma_{(|I_{\tau}|,1)} s\right\}\right)\label{eq:term_tauI_tau2}\\
&+\sum_{k=1}^{K-1}\P\left(\{\MG\le 2s\}\cap\left\{\max_{\tau\in\AS} |G(\tau_{2^k},\tau_{2^{k+1}})|\ge \gamma_{k} s\right\}\right)\label{eq:term_tauM_tau2M}\\
&+\P\left(\{\MG\le 2s\}\cap\left\{\max_{\tau\in\AS} |G(\tau_{2^K})|\ge \gamma_{K} s\right\}\right).\label{eq:term_tau_last}
\end{align}
\end{subequations}

The following notion will become useful for bounding the terms \eqref{eq:term_tau_tauI} and \eqref{eq:term_tauI_tau2} from above.
A set of indices $I\subseteq[d]$ will be called \emph{compatible} with a shift $\tau$ if it holds that 
\begin{equation*}  \TV(\tau_{I}) \leq 20(2\lvert I\rvert+1) \TV(\tau)\qquad\text{and}\qquad\|\tau_{I}-\tau\|_1 \leq 
\frac{4\lvert I\rvert}{d} \TV(\tau).
\end{equation*}
Denote 
$$
\comp(\tau):=\left\{I\subset [d]\colon  I\text{ is compatible with }\tau\right\}.
$$

The following proposition is proved in Section \ref{subsubsec:bounds_fine}. 
\begin{proposition}\label{prop:Itau}
Let $\tau$ be a shift. For each $0\le r\le d$ there exists $I\in\comp(\tau)$ with $|I|=r$.
\end{proposition}
It is clear that sufficiently coarsed grainings of a shift will yield the identity (all zero) shift. The following proposition, proved in Section~\ref{subsubsec:bounds_coarse}, quantifies this statement.
\begin{proposition}\label{prop:coarse graining to the identity}
Let $\tau$ be a shift. 
For each integer $N>\sqrt[d]{2}\left(\frac{\TV\left(\tau\right)}{2d}\right)^{\frac{1}{d-1}}$ 
it holds that 
$\tau_N\equiv 0$.
\end{proposition}
The next lemma, whose proof will be the focus of Section \ref{sec:Concentration of ground-energy differences between consecutive grainings} allows to estimate the expressions \eqref{eq:term_tau_tauI}, \eqref{eq:term_tauI_tau2} and \eqref{eq:term_tauM_tau2M}.
\begin{lemma}\label{lem:Chaining bounds coarse graining}   
Define $\kappa=\kappa(\nu^{\parallel},\nu^{\perp},d)$ as in \eqref{eq:kappa_definition}:
\begin{equation*}
\kappa = \kappa(\nu^\parallel, \nu^\perp,d) := \left( \frac{1}{\underline{\alpha}^{\parallel}\underline{\alpha}^{\perp}} +\frac{1}{d (\underline{\alpha}^{\perp})^{2}}\right)\wid (\nu^{\parallel})^{2}+\frac{1}{(\underline{\alpha}^{\perp})^{2}}\wid(\nu^{\perp})^{2}.
\end{equation*}        
There exist universal constants $C,c>0$ such that the following hold for every $s>0$.
\begin{enumerate}
\item 
For any $1 \leq r \leq d$, any map $\tau \mapsto I_\tau$  assigning to each shift $\tau$ a compatible set $I_{\tau}\in \comp(\tau)$ with $|I|=r$, and 
$ C r \kappa \frac{\log d}{ d} \left(1+\frac{\underline{\alpha}^{\perp}}{\underline{\alpha}^{\parallel}} \right) \leq \Gamma \leq 1$,
\begin{equation*}
\P\left( \left\{\MG \leq 2s \right\} \cap \left\{ \exists \tau \in \AS\colon|G(\tau,\tau_{I_{\tau}})|>\sqrt{\Gamma} s\right\} \right) \leq C\exp \left(-c\frac{\Gamma}{\kappa \underline{\alpha}^{\perp} r} s \right).
\end{equation*}
\item
For any $1 \leq r \leq d$, any map $\tau \mapsto I_\tau$ map assigning to each shift $\tau$ a compatible set $I_{\tau}\in \comp(\tau)$ with $|I|=r$, and 
$C \kappa \frac{dr}{2^r}\left(r+\log\left(\frac{\underline{\alpha}^{\perp}}{ \underline{\alpha}^{\parallel}dr}+1 \right) \right)\leq \Gamma \leq 1 $,
\begin{equation*}
\P\left( \left\{\MG \leq 2s \right\} \cap \left\{ \exists \tau \in \AS\colon|G(\tau_{I_{\tau}},\tau_{2})|>\sqrt{\Gamma} s \right\} \right) \leq C\exp \left( -c\frac{\Gamma}{\kappa \underline{\alpha}^{\perp} d} s \right).
\end{equation*}
\item
For any $k\ge 1$ and  
$C \kappa \frac{d^3}{2^{k(d-2)}}\left(dk +\log\left(\frac{\underline{\alpha}^{\perp}}{\underline{\alpha}^{\parallel}d^2}+1\right)\right) \leq \Gamma \leq 1$,
\begin{equation*}
\P\left(\left\{\MG \leq 2s \right\} \cap \left\{\exists \tau \in \AS\colon |G(\tau_{2^k},\tau_{2^{k+1}})|\ge \sqrt{\Gamma} s \right\} \right) \leq C\exp\left(-c \frac{\Gamma}{\kappa \underline{\alpha}^{\perp} d^2 2^k}s  \right).
\end{equation*}
\end{enumerate}
\end{lemma}

For small $s$, we may obtain an improved dependence on the dimension $d$,  using the following lemma.

\begin{lemma}\label{lem:bounds without graining for small s} 
There exist universal constants $C,c>0$ such that the following holds.
Assume that $0<s < \underline{\alpha}^{\perp}4^d$, and let $\kappa=\kappa(\nu^{\parallel},\nu^{\perp},d)$ be as in \eqref{eq:kappa_definition}, then 
\begin{equation}
\P \left(\{ \MG \leq 2s \} \cap \{\exists \tau \in \AS\colon |G(\tau)|>s \} \right) \leq C\exp \left(- \frac{c}{\kappa \underline{\alpha}^{\perp}} s \right).
\end{equation}
\end{lemma}

\begin{proof}[Proof of Theorem \ref{thm:Existence of admissible ensemble is unlikely}]
Throughout the proof we will use $C$ and $c$ to denote positive absolute constants;
the values of these constants will be allowed to change from line to line, 
even within the same calculation, with the value of $C$ increasing and the value of $c$ decreasing.

Set $K:= \lceil \frac{1}{d-1} \log_2\left(\frac{4s}{\underline{\alpha}^{\perp} d }  \right) \rceil+1$. By Proposition \ref{prop:coarse graining to the identity} and the definition of admissibility,
the term \eqref{eq:term_tau_last} vanishes for any choice of $\gamma_K$.

Set $\gamma_{(0,|I|)}=\gamma_{(|I|,1)}:=\frac{1}{4}$ and $\gamma_k:= \gamma\, 2^{-\frac{1}{4}\min\{k, K-k\}}$ for any $1\leq k\leq K-1$, 
where $\gamma=\left(2\sum_{k=1}^{K-1} 2^{-\frac{1}{4}\min\{k, K-k\}}\right)^{-1}$. 
Set $r:=\lceil \min \{ 10\log_{2}d,\frac{d}{2}\}\rceil$ 
and a map $\tau\mapsto I_\tau$ assigning to each shift $\tau$ a compatible set $I_\tau \in \comp(\tau)$ of  size $r$, existing by Proposition \ref{prop:Itau}. 
Notice that $\gamma_{(0,r)}+\gamma_{(r,1)}+\sum_{k=1}^{K-1} \gamma_k =1$ and that $1 \leq r\leq d $ so one may use the chaining argument \eqref{eq:term_tau_tau}.
    
Recall that $\kappa\left(1+\frac{\underline{\alpha}^{\perp}}{\underline{\alpha}^{\parallel}} \right) \leq c_0 \frac{d+1}{ \log^{2} (d+1)} $, by assumption \eqref{eq:anisotropic condition} on the noise distributions. 
It is easy to verify that for $c_0$ sufficiently small, this enables us to use the first part of Lemma \ref{lem:Chaining bounds coarse graining} to bound the term \eqref{eq:term_tau_tauI}, 
the second part of Lemma \ref{lem:Chaining bounds coarse graining} to bound the term \eqref{eq:term_tauI_tau2}, 
and the third part of Lemma \ref{lem:Chaining bounds coarse graining} to bound the term \eqref{eq:term_tauM_tau2M} 
(recall that the term \eqref{eq:term_tau_last} vanishes).

This yields that for every positive $s$,
\begin{align*}
\P(\MG\in[s,2s)) \leq & C\exp\left(-c\frac{s}{\kappa \underline{\alpha}^{\perp} r}\right)+
C\exp \left(-c \frac{s}{\kappa \underline{\alpha}^{\perp} d}\right)+ 
C\sum_{k=1}^{ \lceil \frac{K}{2} \rceil } \exp \left(-c\frac{s}{\kappa \underline{\alpha}^{\perp} d^{2} 2^{\frac{3}{2} k} }\right)\\
&+C\sum_{k= \lceil \frac{K}{2} \rceil+1 }^{K-1}  \exp \left(-c\frac{s}{\kappa \underline{\alpha}^{\perp} d^{2} 2^{K} }\right)
\end{align*}
Now, noticing that the $K-\lceil \frac{K}{2} \rceil-1 $ last summands are asymptotically dominant and that $2^{K} \leq 4 \left(\frac{4s}{\underline{\alpha}^{\perp} d} \right)^{\frac{1}{d-1}}$, 
one gets the bound
\begin{equation*}
\P(\MG\in[s,2s) ) \leq C\exp \left(-\frac{c}{\kappa d^2} \left( \frac{s}{\underline{\alpha}^{\perp}} \right)^{\frac{d-2}{d-1}} \right)
\end{equation*}
for every positive integer $s$. 
Hence, 
\begin{align*}
\P(\MG\geq t )&=\sum_{i=0}^{\infty}\P(\MG\in[2^it,2^{i+1}t) ) \leq \sum_{i=0}^{\infty}C\exp \left(-\frac{c}{\kappa d^2} \left(\frac{2^i t}{\underline{\alpha}^{\perp}} \right)^{\frac{d-2}{d-1}} \right)\\
&\leq C\exp \left(-\frac{c}{\kappa  d^2} \left( \frac{t}{\underline{\alpha}^{\perp}}\right)^{\frac{d-2}{d-1}} \right).
\end{align*}
For 
$t < \underline{\alpha}^{\perp} 2^d$, 
by Lemma \ref{lem:bounds without graining for small s} and \eqref{eq:tau localization bound}, 
\begin{align*}
\P(\MG \geq t)&= \sum_{i=0}^{d-1} \P\left(\MG \in\left[2^{i}t,2^{i+1}t \right) \right)+ \P\left(\MG \geq 2^d t\right) \\
&\leq \sum_{i=0}^{d-1} C\exp\left(-\frac{c}{\kappa \underline{\alpha}^{\perp}} 2^{i} t \right)+ C\exp \left( -\frac{c}{\kappa d^2} \left( \frac{2^d t}{\underline{\alpha}^{\perp}} \right)^{\frac{d-2}{d-1}} \right)
\leq C\exp\left(-\frac{c t}{\kappa\underline{\alpha}^{\perp}}\right).
\qedhere\end{align*}
\end{proof}

\section{Concentration of ground-energy differences between consecutive grainings}
\label{sec:Concentration of ground-energy differences between consecutive grainings}
The goal of this section is to prove Proposition~\ref{prop:Itau}, Proposition~\ref{prop:coarse graining to the identity}, Lemma~\ref{lem:Chaining bounds coarse graining} and Lemma~\ref{lem:bounds without graining for small s}. 
The proofs of Lemma~\ref{lem:Chaining bounds coarse graining} and Lemma~\ref{lem:bounds without graining for small s} are achieved via the following pivotal statements.

\smallskip

\noindent{\bf Interface layering.} 
In the following lemmas, the concept of interface layering plays a significant role. 
By such layering, we mean the number of interface plaquettes (in a ground configuration with Dobrushin boundary conditions) lying above a given position in the base plane (i.e., having the same projection). 
We use the following definitions:
The \emph{parallel layering} of a configuration $\sigma\in \Omega^{\Lambda,\Dob}$ over a set $A\subset\Lambda\subset\Z^d$ is defined as
\begin{equation}\label{eq:parallel_layering}
{\mathcal L}^{\parallel}_{A}(\sigma):=\left\lvert \left\{ \left\{x,x+e_{d+1} \right\} \in E^{\parallel}(\Z^{d+1}) \colon \pi\left(x\right)\in A,\,
\sigma_x \neq \sigma_{x+e_{d
+1}}\right\} \right\rvert .
\end{equation}
The \emph{perpendicular layering} of  $\sigma$ over $A$ is defined as
\begin{equation}\label{eq:perpendicular_layering}
{\mathcal L}^{\perp}_{A}(\sigma):=\left\lvert \left\{ \left\{x,y\right\} \in E^{\perp}(\Z^{d+1}) \colon  \pi\left(x\right)\in A , \, \pi\left(y\right)\in A , \, 
\sigma_x \neq \sigma_{y}\right\}\right\rvert .
\end{equation}
With these definitions in mind one may think of ${\mathcal L}^{\perp}_{A}(\tau)+{\mathcal L}^{\parallel}_{A}(\tau)-|A|$ as the number of excessive plaquettes in the interface created by the minimal energy configuration above $A$, compared to the interface of the configuration $\rho^{\Dob}$.

For $A\subset \Lambda \subset\Z^{d}$ and integers
$b^{\parallel},b^{\perp}\ge 0$, define:
\begin{equation*}
\Omega^{\Lambda,A,(b^{\parallel},b^{\perp})}:=\left\{ \sigma \in \Omega^{\Lambda,\Dob}\colon\sum\limits_{\mysubstack{\{x,y\}\in E^{\theta}(\Z^{d+1})}{ \{\pi(x),\pi(y)\} \cap A\neq \emptyset}} \mathrm{1}_{\sigma_x\neq \sigma_y}\leq b^{\theta}\text{ for }\theta\in \{\parallel, \perp \}\right\},    
\end{equation*}
as well as
\begin{equation}\label{eq:restricted ground energy}
\GE ^{\Lambda,A,(b^{\parallel},b^{\perp})}(\eta):=
\min\left\{\mathcal{H}^{\eta,\Lambda}(\sigma) \colon\sigma \in \Omega^{\Lambda,A,(b^{\parallel},b^{\perp})}\right\}.
\end{equation}
When $\Lambda$ is fixed, we will occasionally abbreviate by omitting it.
Also define
\begin{equation*}
G^{\eta, \Lambda,(b^{\parallel},b^{\perp})}(\tau, \tau'):=\GE^{\Lambda,\supp(\tau-\tau'),(b^{\parallel},b^{\perp})}(\eta^{\tau'}) - \GE^{\Lambda, \supp(\tau-\tau'),(b^{\parallel},b^{\perp})}(\eta^{\tau}),
\end{equation*}
and abbreviate
\begin{equation*}
G^{\eta, \Lambda,(b^{\parallel},b^{\perp})}(\tau) := G^{\eta, \Lambda,(b^{\parallel},b^{\perp})}(\tau,0)= \GE^{\Lambda,\supp(\tau),(b^{\parallel},b^{\perp})}(\eta) - \GE^{\Lambda,\supp(\tau),(b^{\parallel},b^{\perp})}(\eta^{\tau}).
\end{equation*}

\noindent{\bf Concentration of ground energy differences.} 
First, we provide a bound on the probability of a given shift producing a large energetic gap, given some a priori bound on the ``number of excessive faces in the interface'' above the support of the shift.
\begin{lemma}\label{lem:concentration for energetic gain conditioned on layering bound}
There exist universal $C,c>0$ such that the following holds. Suppose the disorder distributions $\nu^{\parallel},\nu^{\perp}$ satisfy \eqref{eq:disorder distributions assumptions}.
Then for any two shifts $\tau,\tau'$ and any non-negative $b^{\parallel},b^{\perp}$ for which $\rho^{\Dob}\in\Omega^{\Lambda,\supp(\tau-\tau'),(b^{\parallel},b^{\perp})}$,
\begin{equation}\label{eq:two point estimate}
\P\left(\lvert G^{\eta, \Lambda,(b^{\parallel},b^{\perp})} (\tau,\tau') \rvert \ge t\right)\le
C\exp \left(-c \frac{t^2}{\wid (\nu^{\parallel})^{2}  b^{\parallel}+\wid (\nu^{\perp})^{2} b^{\perp}} \right).
\end{equation} 
\end{lemma}

\noindent{\bf Layering bounds.} 
Lemma \ref{lem:concentration for energetic gain conditioned on layering bound} provides a concentration estimate for the ground energy of a restricted set of configurations. In the following lemma, we show that at each step of the graining, the non-restricted ground energy coincides with an appropriate restricted ground energy.   
Recall the definition of $\mathcal{D}(\alpha^{\parallel},\alpha^{\perp})$ from \eqref{eq:D alpha def}.

\begin{lemma}\label{lem:layering bounds for given maximal gain}
There exists a universal $C>0$ such that the following holds. 
Let $\eta\in\mathcal{D}(\alpha^{\parallel},\alpha^{\perp})$ for some  $\alpha^{\parallel},\alpha^{\perp}>0$.
Let $\Lambda \subset \Z^d$ be a finite subset.
Let $s>0$ such that $\MG(\alpha^{\parallel},\alpha^{\perp})\leq 2s$, and let $\tau$ be an $(\alpha^{\parallel},\alpha^{\perp})$-admissible shift. 
\begin{enumerate}
\item 
For any $\emptyset\neq I\in\comp(\tau)$,
\begin{align*}
&\GE^{\Lambda}(\#)=
\GE^{\Lambda,\supp(\tau-\tau_{I}),(b_{(0,|I|)}^{\parallel}(s),b_{(0,|I|)}^{\perp}(s))}(\#)       & \text{\emph{for } }\#\in \{\eta^{\tau},\eta^{\tau_{I}} \}, \\
&\GE^{\Lambda}(\#)=
\GE^{\Lambda,\supp(\tau_{2}-\tau_I),(b_{(|I|,1)}^{\parallel}(s),b_{(|I|,1)}^{\perp}(s))}(\#) & \text{\emph{for } }\#\in \{\eta^{\tau_2}, \eta^{\tau_{I}} \},
\end{align*}
where 
\begin{align*}
b_{(0,|I|)}^{\parallel}(s):=C \left( \frac{1}{{\alpha}^{\parallel}}+\frac{1}{{\alpha}^{\perp}d}  \right) |I|s, &\quad
b_{(0,|I|)}^{\perp}(s):=\frac{C}{{\alpha}^{\perp}}|I| s,  \\
b_{(|I|,1)}^{\parallel}(s):=
C\left( \frac{1}{{\alpha}^{\parallel}}+ \frac{1}{{\alpha}^{\perp}d}  \right) d s,
&\quad
b_{(|I|,1)}^{\perp}(s):= \frac{C}{{\alpha}^{\perp}} d s .
\end{align*}
\item  
For any $k\geq 1$,
\begin{align*}
\GE^{\Lambda}(\eta^{\tau_{2^{k}}})&=
\GE^{\Lambda,\supp(\tau_{2^{k}}-\tau_{2^{k+1}}),(b_k^{\parallel}(s),b_k^{\perp}(s))}(\eta^{\tau_{2^{k}}})\\&\stackrel{(k\geq 2)}{=}\GE^{\Lambda,\supp(\tau_{2^{k-1}}-\tau_{2^{k}}),(b_{k-1}^{\parallel},b_{k-1}^{\perp})}(\eta^{\tau_{2^{k}}}),
\end{align*}
where 
\begin{equation*}
b_{k}^{\parallel}(s):= C \left( \frac{1}{{\alpha}^{\parallel}} +\frac{1 }{
 {\alpha}^{\perp}d}\right) d^2 2^{k} s,\quad
b_{k}^{\perp}(s):=\frac{C}{{\alpha}^{\perp}} d^{2} 2^{k} s.
\end{equation*}
\item 
For $s<{\alpha}^{\perp} 4^d$ it holds that
\begin{align*}
& \GE^{\Lambda}(\#) =
\GE^{\Lambda,\supp(\tau),(b^{\parallel} ,b^{\perp}) }(\#)  & \text{\emph{for } }\#\in \{\eta^{\tau}, \eta \},     
\end{align*}
where
\begin{equation*}
 b^{\parallel}(s):=C \left(\frac{1} {{\alpha}^{\parallel}} +\frac{1}{{\alpha}^{\perp}d} \right)  s, \quad
b^{\perp}(s):= \frac{C}{{\alpha}^{\perp}}  s.
\end{equation*}       
 \end{enumerate}
\end{lemma}
            
\noindent{\bf Enumeration of shifts.}
Lastly, we provide a bound on the number of shifts and their coarse/fine grainings, satisfying that their total variation and trip entropy are bounded by given constants. This is done by the following proposition and the corollary that follows it.

\begin{proposition}[Counting shift functions]\label{prop:enumerate_shift_functions} There exists $C>0$ such that for each $\lambda,\rho>0$,
\begin{equation*}
\lvert\{\tau\in{\mathcal S}\colon \TV(\tau) \leq \lambda,\, R(\tau) \leq \rho\}\rvert\leq
\exp\left(C\min\left\{\lambda+\lambda\log\left(\frac{\rho}{\lambda}+1\right), \lambda\frac{\log d}{d}+\rho\log d\right\}\right),
\end{equation*}
\end{proposition}

\begin{corollary}\label{cor:enumerate_grainings_shift_functions}
There exists a universal $C>0$ such that the following holds.
\begin{enumerate}
\item 
For each integer $N\geq 2$ and $\lambda,\rho>0$,
\begin{align}
\label{eq:coarse_enumerate_shift_functions}
\lvert \{\tau_{N}\colon\tau\in{\mathcal S},\,\TV(\tau) \leq \lambda,\, &R(\tau)\leq   \rho \} \rvert \nonumber\\
&\leq \exp\left(C\frac{d\lambda}{N^{d-1}}  \left(d\log N +\log\left(\frac{\rho}{d\lambda}+1 \right)\right)\right).
\end{align}
\item 
For each integer $1\le r\le d$, a mapping $\tau \mapsto I_{\tau}$ such that $I_{\tau}\in \comp(\tau)$ and $|I_{\tau}|=r$, and $\lambda,\rho>0$,
\begin{equation}
\lvert\{\tau_{I_\tau}\colon \tau\in{\mathcal S},\,\TV(\tau) \leq \lambda,\, R(\tau)\leq  \rho\}\rvert\leq \exp\left( C\frac{r \lambda}{2^{r}}\left(r+\log\left(\frac{\rho}{r\lambda}+1\right)\right)\right). \label{eq:fine_enumerate_shift_functions}
\end{equation}
\end{enumerate}
\end{corollary}

Lemma \ref{lem:concentration for energetic gain conditioned on layering bound} will be proved in Section \ref{subsec:concentration}, using a concentration estimate.
Lemma \ref{lem:layering bounds for given maximal gain} will be proved in Section \ref{sec:layering}, requiring both basic properties of grainings to be established as well as using Lemmas inspired by Dobrushin's work. Proposition \ref{prop:enumerate_shift_functions} and  Corollary \ref{cor:enumerate_grainings_shift_functions} will be proved in Section \ref{sec:Enumeration} using the work of Bollob\'as--Balister \cite{BB} (continuing on  Lebowitz--Mazel~\cite{LM}).

\subsection{Proof of Lemmas \ref{lem:Chaining bounds coarse graining} and \ref{lem:bounds without graining for small s}} 

In this section we show how Lemma \ref{lem:concentration for energetic gain conditioned on layering bound},
Lemma \ref{lem:layering bounds for given maximal gain}, Proposition \ref{prop:enumerate_shift_functions} and  Corollary \ref{cor:enumerate_grainings_shift_functions} imply Lemmas \ref{lem:Chaining bounds coarse graining} and \ref{lem:bounds without graining for small s}.
We will continue to use the abbreviations of Section \ref{sec:coarse and fine graining}, specifically $G$ for $G^{\eta,\Lambda}$, $\AS$ for $\AS^{\eta,\Lambda}(\underline{\alpha}^{\parallel},\underline{\alpha}^{\perp})$
and $\MG$ for $\MG^{\eta,\Lambda}(\underline{\alpha}^{\parallel},\underline{\alpha}^{\perp})$.
Throughout this section we will use $C$ and $c$ to denote positive absolute constants;
the values of these constants will be allowed to change from line to line, 
even within the same calculation, with the value of $C$ increasing and the value of $c$ decreasing.

In the proofs of Lemmas \ref{lem:Chaining bounds coarse graining} and \ref{lem:bounds without graining for small s} we will use the following corollary of Proposition \ref{prop:enumerate_shift_functions} and  Corollary \ref{cor:enumerate_grainings_shift_functions}. 
\begin{corollary}\label{cor:enumerate_adms}
The following bounds hold.
\begin{enumerate}
\item For every $t>0$,
\begin{equation}
\label{eq:enumerate_adms}
\lvert\{\tau\in\mathcal{AS}\colon G(\tau) \leq t\}\rvert\leq \exp\left(C\frac{(\log d)t}{\min\{\underline{\alpha}^{\parallel},\underline{\alpha}^{\perp}\} d}\right).
\end{equation}
\item For every integer $N\geq 2$ and $t>0$,
\begin{equation}
\label{eq:coarse_enumerate_adms}
\lvert \{\tau_{N}\colon\tau\in\mathcal{AS},\, G(\tau) \leq t \} \rvert \leq 
\exp \left( C \frac{2d\,t}{\underline{\alpha}^{\perp} N^{d-1}}  \left(d\log N +\log\left(\frac{\underline{\alpha}^{\perp}}{\underline{\alpha}^{\parallel}d^2}+1 \right)\right)\right).
\end{equation}
\item For every integer $1\le r\le d$, a mapping $\tau \mapsto I_{\tau}$ such that $I_{\tau}\in \comp(\tau)$ and $|I_{\tau}|=r$, and $t>0$,
\begin{equation}\label{eq:fine_enumerate_adms}
\lvert\{\tau_{I_\tau}\colon \tau\in\mathcal{AS},\,G(\tau) \leq t\}\rvert\leq 
\exp\left( C \frac{2r t}{ \underline{\alpha}^{\perp} 2^{r}} \left(r+\log\left(\frac{\underline{\alpha}^{\perp}}{ \underline{\alpha}^{\parallel}dr}+1 \right) \right) \right).
\end{equation}
\end{enumerate}
\end{corollary}

\begin{proof}
For every $t>0$, 
$$
\{\tau\in\mathcal{AS}\colon G(\tau) \leq t\}\subseteq\left\{\tau\in{\mathcal S}\colon \TV(\tau) \leq \frac{2t}{\underline{\alpha}^{\perp}},\, R(\tau) \leq \frac{200t}{\min \{ \underline{\alpha}^{\parallel}, \underline{\alpha}^{\perp}\} d}\right\}.
$$
Hence, by Proposition \ref{prop:enumerate_shift_functions},
\begin{align*}
\lvert\{&\tau\in\mathcal{AS}\colon G(\tau)\leq t\}\rvert \leq\left\lvert\left\{\tau\in{\mathcal S}\colon \TV(\tau) \leq \frac{2t}{\underline{\alpha}^{\perp}},\, R(\tau) \leq \frac{200t}{\min \{\underline{\alpha}^{\parallel},\underline{\alpha}^{\perp}\} d}\right\}\right\rvert\\
&\leq\exp\left(C\min\left\{\frac{2t}{\underline{\alpha}^{\perp}}+\frac{2t}{\underline{\alpha}^{\perp}}\log\left(\frac{100 \underline{\alpha}^{\perp}}{\min \{\underline{\alpha}^{\parallel},\underline{\alpha}^{\perp}\}d}+1\right),
\frac{2t}{\underline{\alpha}^{\perp}d}\log{d}+ \frac{200t}{\min \{\underline{\alpha}^{\parallel},\underline{\alpha}^{\perp}\}d} \log d,
\right\}\right)\\
&\leq\exp\left(C\frac{(\log d)t}{\min \{\underline{\alpha}^{\parallel},\underline{\alpha}^{\perp}\} d}\right),
\end{align*}
for every integer $N\geq 2$, by \eqref{eq:coarse_enumerate_shift_functions},
\begin{align*}
\lvert\{\tau_N\colon\tau\in\mathcal{AS},\, G(\tau)\leq t\}\rvert &\leq\left\lvert\left\{\tau_N\colon \tau\in{\mathcal S},\, \TV(\tau) \leq \frac{2t}{\underline{\alpha}^{\perp}},\, R(\tau) \leq \frac{200t}{\min \{ \underline{\alpha}^{\parallel}, \underline{\alpha}^{\perp}\}} d\right\}\right\rvert\\
&\leq\exp\left(C\frac{2d\,t}{\underline{\alpha}^{\perp} N^{d-1}}  \left(d\log N +\log\left(\frac{100}{d^2}\left(1+\frac{\underline{\alpha}^{\perp}}{\underline{\alpha}^{\parallel}} \right)+1 \right)\right)\right) \\
&\leq\exp\left(C \frac{d\,t}{\underline{\alpha}^{\perp} N^{d-1}}  \left(d\log N +\log\left(\frac{\underline{\alpha}^{\perp}}{\underline{\alpha}^{\parallel}d^2}+1 \right)\right)\right),
\end{align*}
and for every integer $1\le r\le d$ and mapping $\tau \mapsto I_{\tau}$ such that $I_{\tau}\in \comp(\tau)$ and $|I_{\tau}|=r$, by \eqref{eq:fine_enumerate_shift_functions},
\begin{align*}
\lvert\{\tau_{I_{\tau}}\colon\tau\in\mathcal{AS},\, G(\tau)\leq t\}\rvert &\leq\left\lvert\left\{\tau_{I_{\tau}}\colon \tau\in{\mathcal S},\, \TV(\tau) \leq \frac{2t}{\underline{\alpha}^{\perp}},\, R(\tau) \leq \frac{200t}{\min \{ \underline{\alpha}^{\parallel}, \underline{\alpha}^{\perp}\}d}
\right\}\right\rvert\\
&\leq\exp\left( C \frac{2r t}{ \underline{\alpha}^{\perp} 2^{r}} \left(r+\log\left(\frac{100 \underline{\alpha}^{\perp}}{\min \{\underline{\alpha}^{\parallel},\underline{\alpha}^{\perp} \} dr}+1 \right) \right) \right)\\
&\leq\exp\left(C \frac{r t}{ \underline{\alpha}^{\perp} 2^{r}} \left(r+\log\left(\frac{\underline{\alpha}^{\perp}}{ \underline{\alpha}^{\parallel}dr}+1 \right) \right) \right).
\qedhere\end{align*}
\end{proof}

\begin{proof}[Proof of Lemma \ref{lem:Chaining bounds coarse graining}] 
For the first part of the Lemma, let $1\leq r \leq d$,  $\tau \mapsto I_{\tau}$ be a map assigning to each shift $\tau$ a compatible set $I_\tau \in \comp(\tau) $ with $|I_{\tau}|=r$, and $ C r \kappa  \frac{\log d}{ d} \left(1+\frac{\underline{\alpha}^{\perp}}{\underline{\alpha}^{\parallel}} \right) \leq \Gamma \leq 1$.
It holds that 
\begin{align*}
\P&\left( \{\MG\leq 2s \} \cap \{\exists \tau \in \AS, |G(\tau,\tau_{I_{\tau}})|>\sqrt{\Gamma} s \} \right)\\
&=\P \left( \{\MG\leq 2s \} \cap \{\exists \tau \in \AS, |G^{\eta, \Lambda, (b^{\parallel}_{(0,r)}(s),b^{\perp}_{(0,r)}(s))}(\tau,\tau_{I_{\tau}})|>\sqrt{\Gamma} s \} \right) \\
&\leq C \exp\left(C\frac{(\log d) s}{\min\{\underline{\alpha}^{\parallel},\underline{\alpha}^{\perp}\} d}\right) \exp \left(-c\frac{\Gamma s^2}{ \lip(\nu^{\parallel})^{2} b^{\parallel}_{(0,r)}+\lip(\nu^{\perp})^{2} b^{\perp}_{(0,r)} } \right) \\
&= C\exp \left(\left(\frac{C\log d}{\min\{\underline{\alpha}^{\parallel},\underline{\alpha}^{\perp}\}d}-c\frac{\Gamma }{\kappa \underline{\alpha}^{\perp} r}\right) s \right),
\end{align*}
with the first equality by Lemma \ref{lem:layering bounds for given maximal gain}, the first inequality by union bound, \eqref{eq:enumerate_adms}, and Lemma \ref{lem:concentration for energetic gain conditioned on layering bound}, 
and the second equality by the definition of $b^{\parallel}_{(0,r)},b^{\perp}_{(0,r)}$ and of $\kappa$. 
Now, for $C$ sufficiently large, if 
\begin{equation*}
\Gamma\geq C r \kappa  \frac{\log d}{ d} \left(1+\frac{\underline{\alpha}^{\perp}}{\underline{\alpha}^{\parallel}} \right)
\end{equation*}
then the second term in the exponent is the asymptotically dominant one and one gets
\begin{equation*}
\P\left( \{\MG\leq 2s \} \cap \{\exists \tau \in \AS, |G(\tau,\tau_{I_{\tau}})|>\sqrt{\Gamma} s \} \right)\leq C\exp \left(-c\frac{\Gamma s}{ \kappa \underline{\alpha}^{\perp} r} \right).
\end{equation*}
In an identical manner, it holds that 
\begin{align*}
\P&\left( \{\MG\leq 2s \} \cap \{\exists \tau \in \AS, |G(\tau_{I_{\tau}}, \tau_{2})|>\sqrt{\Gamma} s \} \right)\\
&=\P \left( \{\MG\leq 2s \} \cap \{\exists \tau \in \AS, |G^{\eta, \Lambda, (b^{\parallel}_{(r,1)}(s),b^{\perp}_{(r,1)}(s))}(\tau_{I_{\tau}},\tau_{2})|>\sqrt{\Gamma} s \} \right) \\
&\leq C\exp \left(C \frac{r s}{ \underline{\alpha}^{\perp} 2^{r}} \left(r+\log\left(\frac{\underline{\alpha}^{\perp}}{ \underline{\alpha}^{\parallel}dr}+1 \right) \right)\right) 
\exp \left(-c\frac{\Gamma s^2}{\lip(\nu^{\parallel})^{2} b^{\parallel}_{(r,1)}+\lip(\nu^{\perp})^{2} b^{\perp}_{(r,1)} } \right) \\
&= C\exp \left( \left( C \frac{r }{ \underline{\alpha}^{\perp} 2^{r}} \left(r+\log\left(\frac{\underline{\alpha}^{\perp}}{ \underline{\alpha}^{\parallel}dr}+1 \right)\right) - c\frac{\Gamma }{ \kappa \underline{\alpha}^{\perp}  d}\right)s \right),
\end{align*}
with the first equality by Lemma \ref{lem:layering bounds for given maximal gain}, the first inequality by union bound, \eqref{eq:fine_enumerate_adms}, and Lemma \ref{lem:concentration for energetic gain conditioned on layering bound},
and the second equality by the definition of $b^{\parallel}_{(r,1)},b^{\perp}_{(r,1)}$ and of $\kappa$.
Now, for $C$ sufficiently large, if 
\begin{equation*}
\Gamma  \geq  C \kappa \frac{d r}{ 2^{r}} \left(r+\log\left(\frac{\underline{\alpha}^{\perp}}{ \underline{\alpha}^{\parallel}dr}+1 \right)\right) 
\end{equation*}
then the second term in the exponent is the asymptotically dominant one and one gets
\begin{equation*}
\P\left( \{\MG\leq 2s \} \cap \{\exists \tau \in \AS, |G(\tau,\tau_{I_{\tau}})|>\sqrt{\Gamma} s \} \right)\leq \exp C\left(-c\frac{\Gamma s }{\kappa \underline{\alpha}^{\perp} d } \right).
\end{equation*}
For the third bound, again in an identical manner 
\begin{align*}
\P&\left( \{\MG\leq 2s \} \cap \{\exists \tau \in \AS, |G(\tau_{2^{k}}, \tau_{2^{k+1}})|>\sqrt{\Gamma} s \} \right)\\
&=\P \left( \{\MG\leq 2s \} \cap \{\exists \tau \in \AS, |G^{\eta, \Lambda, (b^{\parallel}_{k}(s),b^{\perp}_{k}(s))}(\tau_{2^{k}},\tau_{2^{k+1}})|>\sqrt{\Gamma} s \} \right) \\
&\leq C\exp \left(C \frac{d s}{\underline{\alpha}^{\perp} 2^{k(d-1)}}  \left(dk+\log\left(\frac{\underline{\alpha}^{\perp}}{\underline{\alpha}^{\parallel}d^2}+1 \right)\right)\right) \exp \left(-c\frac{\Gamma s^2}{ \lip(\nu^{\parallel})^{2} b^{\parallel}_{k}+\lip(\nu^{\perp})^{2} b^{\perp}_{k} } \right) \\
&= C\exp \left( \left( C \frac{d }{\underline{\alpha}^{\perp} 2^{k(d-1)}}  \left(dk +\log\left(\frac{\underline{\alpha}^{\perp}}{\underline{\alpha}^{\parallel}d^2}+1 \right)\right)-c\frac{\Gamma }{\kappa \underline{\alpha}^{\perp} d^2 2^k}\right)s \right),
\end{align*}
with the first equality by Lemma \ref{lem:layering bounds for given maximal gain}, the first inequality by union bound, 
\eqref{eq:coarse_enumerate_adms}, and Lemma \ref{lem:concentration for energetic gain conditioned on layering bound}, 
and the second equality by the definition of $b^{\parallel}_{k},b^{\perp}_{k}$ and of $\kappa$.
For $C>0$ sufficiently large, if 
\begin{equation*}
\Gamma \geq C \kappa \frac{d^3}{2^{k(d-2)}}\left(dk +\log\left(\frac{\underline{\alpha}^{\perp}}{\underline{\alpha}^{\parallel}d^2}+1 \right)\right)
\end{equation*}
then the second term in the exponent is the asymptotically dominant one and one gets
\begin{equation*}
\P\left( \{\MG\leq 2s \} \cap \{\exists \tau \in \AS, |G(\tau,\tau_{I_{\tau}})|>\sqrt{\Gamma} s \} \right) \leq C\exp \left(-c\frac{\Gamma s} {\kappa \underline{\alpha}^{\perp} d^{2} 2^{k}}  \right).
\end{equation*}
This concludes the proof.
\end{proof}

\begin{proof}[Proof of Lemma \ref{lem:bounds without graining for small s}]
It holds that
\begin{align*}
\P&\left( \{\MG\leq 2s \} \cap \{\exists \tau \in \AS, |G(\tau)|> s \} \right)\\
&=\P \left( \{\MG\leq 2s \} \cap\{\exists \tau \in \AS, |G^{\eta,\Lambda,(b^{\parallel},b^{\perp})}(\tau)|> s \} \right)\\           
&\leq \exp \left(C 
\frac{(\log d) s}{\min \{\underline{\alpha}^{\parallel}, \underline{\alpha}^{\perp} \} d} \right) C\exp \left(-\frac{s^2}{ \lip(\nu^{\parallel})^{2} b^{\parallel} +
\lip(\nu^{\perp})^{2} b^{\perp} } \right) \\
&= C\exp \left( \left(\frac{C\log d}{\min \{\underline{\alpha}^{\parallel}, \underline{\alpha}^{\perp}\}d } - \frac{c }{\kappa \underline{\alpha}^{\perp}} 
\right)s \right),
\end{align*}
with the first equality by Lemma \ref{lem:layering bounds for given maximal gain}, the first inequality by union bound, \eqref{eq:enumerate_adms}, and Lemma \ref{lem:concentration for energetic gain conditioned on layering bound},
and the second equality by the definition of $b^{\parallel},b^{\perp}$ and of $\kappa$.
Then, for $c>0$ sufficiently small, if 
$\kappa \left(1+\frac{\underline{\alpha}^{\perp}} {\underline{\alpha}^{\parallel}}\right) \leq c \frac{d}{\log d }$
which is always the case by condition \eqref{eq:anisotropic condition},
then the second term in the exponent is the asymptotically dominant one and one gets 
\begin{equation*}
\P\left( \{\MG\leq 2s \} \cap \{\exists \tau \in \AS, |G(\tau)|> s \} \right)\leq \exp \left(-\frac{c}{\kappa \underline{\alpha}^{\perp}} s \right).   \qedhere\end{equation*}
\end{proof}
   
\subsection{Basic properties of grainings of shifts}\label{sec:basic properties of grainings of shifts}
This section provides estimates on basic parameters (total variation, support size, trip-entropy, etc.) for coarse and fine grainings of shifts. These estimates will be used in the proof of several of our preliminary statements in the subsequent sections. 

\subsubsection{Isoperimetric inequalities} The following pair of lemmas present basic isoperimetric inequalities for the lattice $\Z^d$. We also deduce Proposition~\ref{prop:coarse graining to the identity} as an immediate consequence.

\begin{lemma}[Isoperimetric inequality on ${\mathbb Z}^d$]\label{lem:isoperimetric inequality}
Let $A$ be a finite set of points in ${\mathbb Z}^d$.
Then,
\begin{equation}\label{eq:bdry_small}
\lvert  \partial A\rvert \geq 2d\lvert  A \rvert ^{1-\frac{1}{d}}.    
\end{equation}
Moreover, for 
every $N\times N\times\cdots\times N$ $d$-dimensional cube $B$ in $\mathbb Z^d$, 
\begin{equation}\label{eq:bdry_rough}
\lvert  \partial A\cap(B\times B)\rvert 
\geq \frac{2}{3N}\min\left\{\lvert A\cap B\rvert ,\lvert B\setminus A\rvert \right\}.
\end{equation}
\end{lemma}
\begin{lemma}[Functional isoperimetric inequality on $\Z^d$]\label{lem:functional_isoperimetry}
For every shift function $\tau$,
\begin{equation*}
\TV(\tau) \geq 2d  \left( \sum_{u\in \Z^d}|\tau(u)| \right) ^{1-\frac{1}{d}}.
\end{equation*}
\end{lemma}

\begin{proof}[Proof of Proposition \ref{prop:coarse graining to the identity}]
Let $\tau$ be a shift. 
Observe that if $N$ is a positive integer such that $\sum_{u\in \Z^{d}}|\tau(u)| < N^d/2$ then necessarily $\tau_N\equiv 0$.
Hence, by Lemma
~\ref{lem:functional_isoperimetry}, $\tau_N\equiv 0$ if 
\begin{equation*}
\left( \frac{\TV(\tau)}{2d} \right)^{\frac{d}{d-1}}<\frac{1}{2}N^d,
\end{equation*}
i.e.,  
\begin{equation*}
N>\sqrt[d]{2}\left(\frac{\TV\left(\tau\right)}{2d}\right)^{\frac{1}{d-1}}.\qedhere\end{equation*}
\end{proof}

\begin{proof}[Proof of Lemma~\ref{lem:isoperimetric inequality}]
For any $S\subset\mathbb Z^d$ and $1\leq i\leq d$, let $\pi_i(S)$ be the projection of $S$ on the hyperplane spanned by $\{e_1,e_2,\ldots,e_d\}\setminus\{e_i\}$. 

Recall that the Loomis-Whitney inequality \cite{LW} states that for every finite set $S$ of points in ${\mathbb Z}^d$ it holds that
$\prod_{i=1}^d\lvert  \pi_i(S) \rvert \geq \lvert  S \rvert ^{d-1}$
and hence, by the inequality of arithmetic and geometric means,
\begin{equation}\label{eq:addLW}
\sum_{i=1}^d\lvert  \pi_i(S) \rvert \geq d\lvert  S \rvert ^{1-\frac{1}{d}}.
\end{equation}    

For every $1\leq i\leq d$, let $\partial_i A:=\{(u,v)\in\partial A\colon u-v\in\{-e_i,e_i\}\}$. Obviously, $\lvert \partial_i A\rvert\geq 2\lvert\pi_i(A)\rvert$ for every $1\leq i\leq d$, and hence,
\begin{equation*}\lvert  \partial A\rvert =\sum_{i=1}^d\lvert  \partial_i A\rvert \geq 2\sum_{i=1}^d \lvert \pi_i(A)\rvert \geq 2d\lvert  A \rvert ^{1-\frac{1}{d}}.
\end{equation*}

We proceed to prove \eqref{eq:bdry_rough}. Since $\partial A\cap(B\times B)\rvert =  \partial(B\setminus A)\cap(B\times B)$, we may assume with no loss of generality that $\lvert A\cap B\rvert\leq \frac{1}{2}\lvert B\rvert$.
For every $1\leq i\leq d$, let 
$$
F_i:=\left\{x\in \pi_i(A\cap B)\colon  \lvert A\cap B\cap \pi_i^{-1}(x)\rvert=\frac{\lvert B\rvert }{\lvert \pi_i(B)\rvert }\right\}.
$$
For every $x\in\pi_i(A\cap B)\setminus F_i$ it holds that $\emptyset\neq A\cap B\cap \pi_i^{-1}(x)\neq B\cap \pi_i^{-1}(x)$ and therefore $\lvert\partial_i A\cap(B\times B)\cap(\pi_i^{-1}(x)\times\pi_i^{-1}(x))\rvert\geq 1$.
Hence, 
$$
\lvert  \partial_i A\cap(B\times B)\rvert =\sum_{x\in\pi_i(A\cap B)}\lvert\partial_i A\cap(B\times B)\cap(\pi_i^{-1}(x)\times\pi_i^{-1}(x))\rvert\geq \lvert \pi_i(A\cap B)\rvert -\lvert  F_i\rvert,
$$
and since
\begin{equation*}
\lvert A\cap B\rvert\geq\sum_{x\in F_i}\lvert A\cap B\cap \pi_i^{-1}(x)\rvert=\lvert F_i\rvert \frac{\lvert B\rvert }{\lvert \pi_i(B)\rvert },
\end{equation*}
it follows that
\begin{equation*}
\lvert  \partial_i A\cap(B\times B)\rvert \geq\lvert \pi_i(A\cap B)\rvert -\frac{\lvert \pi_i(B)\rvert }{\lvert B\rvert }\lvert A\cap B\rvert .  
\end{equation*}
Hence, by \eqref{eq:addLW},
\begin{align*}
\lvert  \partial A\cap(B\times B)\rvert &=\sum_{i=1}^d\lvert  \partial_i A\cap(B\times B)\rvert \geq\sum_{i=1}^d\left(\lvert \pi_i(A\cap B)\rvert -\frac{\lvert \pi_i(B)\rvert }{\lvert B\rvert }\lvert A\cap B\rvert \right)\\
&=\left(\sum_{i=1}^d\lvert\pi_i(A\cap B)\rvert\right)-\frac{d}{N}\lvert A\cap B\rvert
\geq d\lvert  A\cap B \rvert ^{1-\frac{1}{d}}-\frac{d}{N}\lvert A\cap B\rvert\\ 
&=\frac{d}{N}\left(\frac{\sqrt[d]{\lvert B\rvert}}{\sqrt[d]{\lvert A\cap B\rvert }}-1\right)\lvert A\cap B\rvert \geq \frac{d(\sqrt[d]{2}-1)}{N}\lvert A\cap B\rvert,
\end{align*}
and \eqref{eq:bdry_rough} follows since $d(\sqrt[d]{2}-1)>\ln2>2/3$.
\end{proof}

\begin{proof}[Proof of Lemma \ref{lem:functional_isoperimetry}]
WLOG we may assume that $\tau$ is non-negative, since $\TV(|\tau|) \leq \TV(\tau)$ by the triangle inequality.
Define the family of sets $A_{k}:= \left\{u \in \Z^{d}\colon  \tau(u)>k \right\}$. 
By definition, we have that 
$\TV(\tau):= \sum_{k\geq 0} |\partial A_k|$ and $\sum_{u\in \Z^{d}} |\tau(u)| = \sum_{k\geq 0}|A_k |$ (note that in both sums all but finitely many of the terms are non-zero) and so showing 
\begin{equation*}
\sum_{k \geq 0}|\partial A_k| \geq 
2d \left( \sum_{k\geq 0}|A_k | \right) ^{1-\frac{1}{d}} 
\end{equation*}
would be sufficient.
Using \eqref{eq:bdry_small} for each of the $A_k$'s one gets $\sum_{k\geq 0}|\partial A_k| \geq 2d  \sum_{k\geq 0} 
|A_k |^{1-\frac{1}{d}}$. 
The result follows by setting $\lambda_k:=\lvert A_k\rvert/\sum_{k\geq 0}\lvert A_k\rvert$ for every $k\geq 0$ and noting that $0\leq\lambda_k\leq 1$ for every $k\geq 0$ and hence 
\begin{equation*}
\sum_{k\geq 0}\lvert A_k\rvert^{1-\frac{1}{d}}\Bigg/\left(\sum_{k\geq 0}\lvert A_k\rvert\right)^{1-\frac{1}{d}}=\sum_{k\geq 0}\lambda_k^{1-\frac{1}{d}}\geq\sum_{k\geq 0}\lambda_k=1.      \qedhere\end{equation*}
\end{proof}

\subsubsection{Bounding the total variation and weighted difference of coarse grainings}
\label{subsubsec:bounds_coarse}

This section is devoted to the proof of the following two propositions, which control parameters of the coarse graining of a shift in terms of the total variation of that shift.

\begin{proposition}\label{prop:functional_bdry}
For every shift $\tau$ and every positive integer $N$, 
$$
\TV(\tau_N) \leq 10d \, \TV(\tau). 
$$
\end{proposition}

\begin{proposition}\label{prop:functional_diff}
For every shift $\tau$,
$$
\lvert\supp(\tau-\tau_2)\rvert\leq\lVert \tau-\tau_2\rVert_1
\leq 2\, \TV(\tau) 
$$
and moreover, for every positive integer $N$,
$$
\lvert\supp(\tau_N-\tau_{2N})\rvert\leq\lVert \tau_N-\tau_{2N}\rVert_1
\leq (4d+9)N\, \TV(\tau). 
$$
\end{proposition}

We introduce the notation
$$\tau^{\rm rough}_N (u):=
\frac{1}{N^d}\sum_{v\in Q_N(N\,w)}
\tau(v),
$$
where $w$ is the unique point in ${\mathbb Z}^d$ such that $u\in Q_N(N\,w)$. Recalling the definition of $\tau_N$ from Section~\ref{sec:coarse and fine grainings}, we then have, for every $u\in{\mathbb Z}^d$,
$$
\tau_N (u)= 
\left[ \tau^{\rm rough}_N (u)\right],
$$
where by $\left[ a \right]$ we denote the nearest integer to $a$.

For a shift $\tau\colon \Z^d\to\Z$ and a set $A\subset\Z^d$, define
$$
\TV(\tau;A):= 
\sum_{\{u,v\}\in E(\Z^d)\cap\binom{A}{2}} \lvert \tau(u)-\tau(v) \rvert.
$$

The proofs of Propositions \ref{prop:functional_bdry} and \ref{prop:functional_diff} use the following lemmas.

\begin{lemma}\label{lem:Ron}
For every $0<\alpha<\frac{1}{2}$ and $u\in\Z^d$ such that $\lvert\tau^{\rm rough}_N(N\,u)-\tau_N(N\,u)\rvert>\alpha$, it holds that 
\begin{equation*}
\TV\left(\tau; Q_N(N\,u)\right)\geq \frac{2\alpha}{3} N^{d-1}.
\end{equation*}
\end{lemma}

\begin{lemma}
\label{lem:rough}
For every $\{u,v\}\in E(\Z^d)$, 
\begin{equation*}
\left\lvert\tau^{\rm rough}_N(N\,u)-\tau^{\rm rough}_N(N\,v)\right\rvert
\leq\frac{1}{N^{d-1}}\TV\left(\tau; Q_N(N\,u)\cup Q_N(N\,v)\right).
\end{equation*}
\end{lemma}

\begin{lemma}
\label{lem:diff_rough}
Let $w\in{\mathbb Z}^d$, let $g$ be a function from $B:=Q_2(2\,w)$ to $\R$, and let $\mu:=\frac{1}{2^d}\sum_{u\in B}g(u)$. Then,
$$
\sum_{u\in B} \lvert g(u)-\mu \rvert
\leq\sum_{\{u,v\}\in E(\Z^d)\cap\binom{B}{2}} \lvert g(u)-g(v) \rvert.
$$
\end{lemma}

Let us now deduce the propositions from the lemmas.

\begin{proof}[Proof of Proposition \ref{prop:functional_bdry}]
Note that $\TV(\tau_N)=\sum_{\{u,v\}\in E(\Z^d)}N^{d-1}\left\lvert\tau_N(N\,u)-\tau_N(N\,v)\right\rvert$ and 
$\sum_{\{u,v\}\in E(\Z^d)}\TV\left(\tau;Q_N(N\,u)\cup Q_N(N\,v)\right)\leq 2d\TV(\tau)$.
Hence,
it is enough to prove that for every $\{u,v\}\in E(\Z^d)$, 
\begin{equation*}
\left\lvert\tau_N(N\,u)-\tau_N(N\,v)\right\rvert
\leq\frac{5}{N^{d-1}}\TV\left(\tau;Q_N(N\,u)\cup Q_N(N\,v)\right).
\end{equation*}
If $\tau_N(N\,u)=\tau_N(N\,v)$ there is nothing to prove.
If $\left\lvert\tau^{\rm rough}_N(N\,u)-\tau^{\rm rough}_N(N\,v)\right\rvert\geq\frac{1}{3}$, then 
$$
\left\lvert\tau_N(N\,u)-\tau_N(N\,v)\right\rvert\leq \left\lvert\tau^{\rm rough}_N(N\,u)-\tau^{\rm rough}_N(N\,v)\right\rvert+1\leq 4\left\lvert\tau^{\rm rough}_N(N\,u)-\tau^{\rm rough}_N(N\,v)\right\rvert
$$
and the result follows from Lemma \ref{lem:rough}.

Therefore, suppose that 
$\left\lvert\tau^{\rm rough}_N(N\,u)-\tau^{\rm rough}_N(N\,v)\right\rvert<\frac{1}{3}$
but $\tau_N(N\,u)\neq\tau_N(N\,v)$.
Then, necessarily, 
$$
\max\left\{\left\lvert\tau^{\rm rough}_N(N\,u)-\tau_N(N\,u)\right\rvert,\left\lvert\tau^{\rm rough}_N(N\,v)-\tau_N(N\,v)\right\rvert\right\}>\frac{1}{3}.
$$
If $\lvert\tau^{\rm rough}_N(N\,u)-\tau_N(N\,u)\rvert>\frac{1}{3}$ then by Lemma \ref{lem:Ron},
\begin{equation*}
\left\lvert\tau_N(N\,u)-\tau_N(N\,v)\right\rvert=1
\leq\frac{9}{2N^{d-1}}\TV\left(\tau; Q_N(N\,u)\right)
\end{equation*}
and similarly, if $\lvert\tau^{\rm rough}_N(N\,v)-\tau_N(N\,v)\rvert>\frac{1}{3}$ then by Lemma \ref{lem:Ron},
\begin{equation*}
\left\lvert\tau_N(N\,u)-\tau_N(N\,v)\right\rvert=1
\leq\frac{9}{2N^{d-1}}\TV\left(\tau; Q_N(N\,v)\right).
\qedhere\end{equation*}
\end{proof}

\begin{proof}[Proof of Proposition \ref{prop:functional_diff}]

It is clear that $|\supp(\tilde{\tau})|\le \|\tilde{\tau}\|_1$ for every shift $\tilde{\tau}$, so that we only need to bound the $\ell_1$ norms appearing in the statement.

We first prove that $\lVert \tau-\tau_2\rVert_1
\leq 2\, \TV(\tau)$. 
For every $u\in\Z^d$ such that $\tau(u)\neq\tau_2(u)$, necessarily $\lvert \tau(u)-\tau^{\rm rough}_2(u)\rvert\geq\frac{1}{2}$ and hence, 
$$
\left\lvert\tau(u)-\tau_2(u)\right\rvert\leq \left\lvert\tau(u)-\tau^{\rm rough}_2(u)\right\rvert+\frac{1}{2}\leq 2\left\lvert\tau(u)-\tau^{\rm rough}_2(u)\right\rvert.$$
It follows that 
$$ 
\lVert \tau-\tau_2\rVert_1=\sum_{u\in\Z^d}\lvert\tau(u)-\tau_2(u)\rvert\leq 2\sum_{u\in\Z^d}\left\lvert\tau(u)-\tau^{\rm rough}_2(u)\right\rvert.
$$
For every $w\in\Z^d$, by Lemma \ref{lem:diff_rough}, 
$$
\sum_{u\in B_w}\lvert\tau(u)-\tau^{\rm rough}_2(u)\rvert\leq \sum_{\{u,v\}\in E(\Z^d)\cap\binom{B_w}{2}}\lvert\tau(u)-\tau(v)\rvert=\TV\left(\tau;B_w\right),
$$
where $B_w:=Q_2(2\,w)$.
Hence,
\begin{align*} 
\lVert \tau-\tau_2\rVert_1&\leq 2\sum_{u\in\Z^d}\left\lvert\tau(u)-\tau^{\rm rough}_2(u)\right\rvert\\&=2\sum_{w\in\Z^d}\sum_{u\in B_w}\lvert\tau(u)-\tau^{\rm rough}_2(u)\rvert
\leq 2\sum_{w\in\Z^d}\TV\left(\tau;B_w\right)\leq 2\TV(\tau).
\end{align*}

We proceed to show that for any positive integer $N$, it holds that $\lVert \tau_N-\tau_{2N}\rVert_1
\leq (4d+9)N\, \TV(\tau)$.
It is clearly enough to show that for every $w\in \mathbb Z^d$, 
$$
\sum_{u\in Q_{2N}(2N\,w)} \lvert \tau_N(u) -\tau_{2N} (u) \rvert\leq (4d+9)N \TV\left(\tau;Q_{2N}(2N\,w)\right).
$$
Denote $B_w:=Q_2(2\,w)$ once more, and let 
$$A_w:=\left\{v\in B_w\colon \lvert\tau^{\rm rough}_N(N\,v)-\tau_N(N\,v)\rvert>\frac{1}{6}\right\}.$$

For $v\in B_w$ such that $\left\lvert\tau^{\rm rough}_N(N\,v)-\tau^{\rm rough}_{2N}(N\,v)\right\rvert<\frac{1}{3}$, necessarily $\left\lvert\tau_N(N\,v)-\tau_{2N}(N\,v)\right\rvert\leq 1$ and if 
$\left\lvert\tau_N(N\,v)-\tau_{2N}(N\,v)\right\rvert=1$, then $\left\lvert\tau_N(N\,v)-\tau^{\rm rough}_{2N}(N\,v)\right\rvert\geq\frac{1}{2}$ and hence,
$$
\lvert\tau^{\rm rough}_N(N\,v)-\tau_N(N\,v)\rvert\geq \left\lvert\tau_N(N\,v)-\tau^{\rm rough}_{2N}(N\,v)\right\rvert-\left\lvert\tau^{\rm rough}_N(N\,v)-\tau^{\rm rough}_{2N}(N\,v)\right\rvert>\frac{1}{2}-\frac{1}{3}=\frac{1}{6},
$$
i.e., $v\in A_w$;
if $v\in B_w$ such that
$\left\lvert\tau^{\rm rough}_N(N\,v)-\tau^{\rm rough}_{2N}(N\,v)\right\rvert\geq\frac{1}{3}$, then 
$$
\left\lvert\tau_N(N\,v)-\tau_{2N}(N\,v)\right\rvert\leq \left\lvert\tau^{\rm rough}_N(N\,v)-\tau^{\rm rough}_{2N}(N\,v)\right\rvert+1\leq 4\left\lvert\tau^{\rm rough}_N(N\,v)-\tau^{\rm rough}_{2N}(N\,v)\right\rvert.$$
It follows that
\begin{align*}
\frac{1}{N^d}\sum_{u\in Q_{2N}(2N\,w)}\lvert\tau_N(u)-\tau_{2N}(u)\rvert&=
\sum_{v\in B_w}\lvert\tau_N(N\,v)-\tau_{2N}(N\,v)\rvert\\
&\leq \lvert A_w\rvert+
4\sum_{v\in B_w}\lvert\tau^{\rm rough}_N(N\,v)-\tau^{\rm rough}_{2N}(N\,v)\rvert,
\end{align*}
and we are done since by Lemma \ref{lem:Ron}, $$\lvert A_w\rvert\leq \frac{9}{N^{d-1}}\sum_{v\in A_w}\TV\left(\tau; Q_N(N\,v)\right)\leq \frac{9}{N^{d-1}}\TV\left(\tau; Q_{2N}(2N\,w)\right)$$
and by Lemma \ref{lem:diff_rough} and Lemma \ref{lem:rough}, 
\begin{align*}
\sum_{v\in B_w}\lvert\tau^{\rm rough}_N(N\,v)-\tau^{\rm rough}_{2N}(N\,v)\rvert&\leq\sum_{\{u,v\}\in E(\Z^d)\cap\binom{B_w}{2}}\lvert\tau^{\rm rough}_N(N\,u)-\tau^{\rm rough}_N(N\,v)\rvert\\
&\leq \frac{1}{N^{d-1}}\sum_{\{u,v\}\in E(\Z^d)\cap\binom{B_w}{2}}\TV\left(\tau; Q_N(N\,u)\cup Q_N(N\,v)\right)\\
&\leq \frac{d}{N^{d-1}}\TV\left(\tau; Q_{2N}(2N\,w)\right).\qedhere\end{align*}
\end{proof}

Finally, we will now prove Lemmas \ref{lem:Ron}, \ref{lem:rough} and \ref{lem:diff_rough}.

\begin{proof}[Proof of Lemma \ref{lem:Ron}]
For simplicity, denote $B:=Q_N(N\,u)$ and let $m:=\min_{v\in B}\tau(v),\, M:=\max_{v\in B}\tau(v)$. 
For every integer $k$, let
$A_k:=\{v\in B\colon  \tau(v)> k\}$.
Note that $A_{m-1}=B$ and $A_M=\emptyset$. Let $\ell:=\min\{m\leq k\leq M\colon  |A_k|<\frac{1}{2}N^d\}$. 

Note that $\TV\left(\tau; B\right)=\sum_{k=m}^{M-1}\lvert \partial A_k\cap(B\times B)\rvert$ and hence, by \eqref{eq:bdry_rough},
\begin{equation*}
\TV\left(\tau; B\right)\geq\frac{ 2}{3N}\sum_{k=m}^{M-1}\min\left\{\lvert A_k\rvert,\lvert B\setminus A_k\rvert\right\}=\frac{ 2}{3N}\sum_{k=m}^{\ell-1}\lvert B\setminus A_k\rvert+\frac{2}{3N}\sum_{k=\ell}^{M-1}\lvert A_k\rvert.
\end{equation*}
Hence, the result follows if $\sum_{k=m}^{\ell-1}\lvert B\setminus A_k\rvert\geq \alpha N^d$ or $\sum_{k=\ell}^{M-1}\lvert A_k\rvert\geq \alpha N^d$.
Assume by way of contradiction that
$\sum_{k=m}^{\ell-1}\lvert B\setminus A_k\rvert<\alpha N^d$ and 
$\sum_{k=\ell}^{M-1}\lvert A_k\rvert<\alpha N^d$. 
Now, since
\begin{equation*}
\tau^{\rm rough}_N(N\,u)=\frac{1}{N^d}\sum_{v\in B}\tau(v)
=m+\sum_{k=m}^{M-1}\frac{\lvert A_k\rvert}{N^d}
=\ell-\sum_{k=m}^{\ell-1}\frac{\lvert B\setminus A_k\rvert}{N^d}+\sum_{k=\ell}^{M-1}\frac{\lvert A_k\rvert}{N^d}
\end{equation*}
it follows that $\lvert\tau^{\rm rough}_N(N\,u)-\ell\rvert<\alpha$. In particular, $\tau_N(N\, u)=\ell$ and we get a contradiction.
\end{proof}

\begin{proof}[Proof of Lemma \ref{lem:rough}]
With no loss of generality, assume that $v=u+e_d$. 
Then,
$$
\tau^{\rm rough}_N(N\,u)-\tau^{\rm rough}_N(N\,v)=\frac{1}{N^d}\sum_{w\in B}\left(\sum_{i=0}^{N-1}\tau(w+ie_d)-\sum_{i=N}^{2N-1}\tau(w+ie_d)\right),
$$
where $B:=N\,u+\{0,1,2,\ldots,N-1\}^{d-1}\times\{0\}$.
For every $w\in B$, using summation by parts, it holds that
\begin{align*}
\sum_{i=0}^{N-1} \tau(w+ie_d)&=N\tau(w+(N-1)e_d)+\sum_{i=1}^{N-1}i\left(\tau(w+(i-1)e_d)-\tau(w+ie_d)\right),\\
\sum_{i=N}^{2N-1}\tau(w+ie_d)&=N\tau(w+Ne_d)-\sum_{i=N+1}^{2N-1}(2N-i)\left(\tau(w+(i-1)e_d)-\tau(w+ie_d)\right).
\end{align*}
Therefore, for every $w\in B$,
\begin{equation*}
\sum_{i=0}^{N-1} \tau(w+ie_d)-\sum_{i=N}^{2N-1}\tau(w+ie_d)=\sum_{i=1}^{2N-1}\min\{i,2N-i\}\left(\tau(w+(i-1)e_d)-\tau(w+ie_d)\right)
\end{equation*}
and hence
\begin{equation*}
\left\lvert\sum_{i=0}^{N-1} \tau(w+ie_d)-\sum_{i=N}^{2N-1}\tau(w+ie_d)\right\rvert\leq N\sum_{i=1}^{2N-1}\lvert\tau(w+(i-1)e_d)-\tau(w+ie_d)\rvert.
\end{equation*}
Therefore,
\begin{align*}
\lvert\tau^{\rm rough}_N(N\,u)-\tau^{\rm rough}_N(N\,v)\rvert&=\frac{1}{N^d}\sum_{w\in B}\left\lvert\sum_{i=0}^{N-1}\tau(w+ie_d)-\sum_{i=N}^{2N-1}\tau(w+ie_d)\right\rvert\\
&\leq\frac{1}{N^{d-1}}\sum_{w\in B}\sum_{i=1}^{2N-1}\lvert\tau(w+(i-1)e_d)-\tau(w+ie_d)\rvert\\
&\leq\frac{1}{N^{d-1}}\TV\left(\tau; Q_N(N\,u)\cup Q_N(N\,v)\right).
\qedhere
\end{align*}
\end{proof}

\begin{proof}[Proof of Lemma \ref{lem:diff_rough}]
We first prove by induction on $k$ that for every $1\leq k\leq d$ that
\begin{equation}\label{eq:induction}
\sum_{\substack{\{u,v\}\in\binom{B}{2}\\ \lVert u-v\rVert_1=k}} \lvert g(u)-g(v) \rvert
\leq\binom{d-1}{k-1}\sum_{\{u,v\}\in E(\Z^d)\cap\binom{B}{2}} \lvert g(u)-g(v) \rvert.
\end{equation}
The base case $k=1$ obviously holds as equality, and if \eqref{eq:induction} holds for some $k$, then it follows that it holds for $k+1$ as well, since
\begin{align*}
\sum_{\substack{\{u,v\}\in\binom{B}{2}\\ \lVert u-v\rVert_1=k+1}} \lvert g(u)&-g(v) \rvert=\frac{1}{2}\sum_{u\in B}\sum_{\substack{v\in B\\\lVert u-v\rVert_1=k+1}}\lvert g(u)-g(v)\rvert\\
\leq&\frac{1}{2}\sum_{u\in B}\sum_{\substack{v\in B\\\lVert u-v\rVert_1=k+1}} \frac{1}{k+1}\sum_{\substack{w\in B\\\lVert u-w\rVert_1=k,\, w\sim v}}\left(\lvert g(u)-g(w) \rvert+\lvert g(w)-g(v) \rvert \right)\\
=&\frac{1}{2}\sum_{u\in B} \frac{d-k}{k+1}\sum_{\substack{w\in B\\\lVert u-w\rVert_1=k}}\lvert g(u)-g(w) \rvert+\frac{1}{2}\sum_{v\in B} \frac{1}{k+1}\binom{d-1}{k}\sum_{\substack{w\in B\\w\sim v}}\lvert g(w)-g(v) \rvert \\
=&\frac{d-k}{k+1}\sum_{\substack{\{u,w\}\in\binom{B}{2}\\ \lVert u-w\rVert_1=k}} \lvert g(u)-g(w) \rvert+\frac{1}{k+1}\binom{d-1}{k}\sum_{\{v,w\}\in E(\Z^d)\cap\binom{B}{2}} \lvert g(v)-g(w) \rvert.
\end{align*}
Now we can conclude the proof of the lemma.
\begin{align*}
\sum_{u\in B} \lvert g(u)-\mu \rvert&\leq\sum_{u\in B} \frac{1}{2^d}\sum_{v\in B}\lvert g(u)-g(v) \rvert\\
&=\frac{1}{2^{d-1}}\sum_{\{u,v\}\in\binom{B}{2}}\lvert g(u)-g(v) \rvert=\frac{1}{2^{d-1}}\sum_{k=1}^d\sum_{\substack{\{u,v\}\in\binom{B}{2}\\ \lVert u-v\rVert_1=k}} \lvert g(u)-g(v) \rvert\\
&\leq \frac{1}{2^{d-1}}\sum_{k=1}^d\binom{d-1}{k-1} \sum_{\{u,v\}\in E(\Z^d)\cap\binom{B}{2}} \lvert g(u)-g(v) \rvert\\
&=\sum_{\{u,v\}\in E(\Z^d)\cap\binom{B}{2}} \lvert g(u)-g(v) \rvert.
\qedhere\end{align*}
\end{proof}

\subsubsection{Bounding the total variation and weighted difference for fine grainings}
\label{subsubsec:bounds_fine}

In this section we prove analogous results to Propositions \ref{prop:functional_bdry} and \ref{prop:functional_diff} for fine grainings, and deduce Proposition \ref{prop:Itau}.
For $I\in\binom{[d]}{r}$, let $T_I\colon {\mathbb Z}^d\to{\mathbb Z}^d$ be defined as follows: for any $x=(x_1,x_2,\ldots,x_d)\in{\mathbb Z}^d$ and every $1\leq i\leq d$, the $i$th coordinate of $T_I(x)$ is $2x_i$ if $i\in I$ and $x_i$ otherwise.
We introduce the notation
$$\tau^{\rm rough}_I (u):=
\frac{1}{2^r}\sum_{v\in Q_I(T_I(w))}
\tau(v),
$$
where $w$ is the unique point in ${\mathbb Z}^d$ such that $u\in Q_I(T_I(w))$. 
Recalling the definition of $\tau_I$ from Section~\ref{sec:coarse and fine grainings}, we then have, for every $u\in{\mathbb Z}^d$,
$$
\tau_I (u)= 
\left[ \tau^{\rm rough}_I (u)\right],
$$
where by $\left[ a \right]$ we denote the nearest integer to $a$.

\begin{lemma}\label{lem:rough_fine}
For every $I\in\binom{[d]}{r}$ and $\{u,v\}\in E(\Z^d)$, 
\begin{equation}\label{eq:rough_fine}
\left\lvert\tau^{\rm rough}_I(T_I(u))-\tau^{\rm rough}_I(T_I(v))\right\rvert
\leq\frac{1}{2^{r-1}}\TV\left(\tau; Q_I(T_I(u))\cup Q_I(T_I(v))\right).
\end{equation}
\end{lemma}
\begin{proof}
If $u-v\in\{-e_i,e_i\}$ for $i\notin I$, then \eqref{eq:rough_fine} easily follows by a straightforward use of the triangle inequality; otherwise, \eqref{eq:rough_fine} is simply the claim of Lemma \ref{lem:rough} for $N=2$ in the $r$-dimensional affine subspace $u+{\rm span}(\{e_i\}_{i\in I})=v+{\rm span}(\{e_i\}_{i\in I})$ of ${\mathbb Z}^d$.
\end{proof}

\begin{proposition}\label{prop:functional_bdry_fine}
Let $I\in\binom{[d]}{r}$ be chosen uniformly at random. Then, for every $\tau\colon \Lambda \to \Z$, the following holds:
$$
\mathbb{E} \TV(\tau_I) \leq 10(2r+1)\TV(\tau) 
$$
\end{proposition}

\begin{proof}
For every $\{x,y\}\in E(\Z^d)$, let $X_{\{x,y\}}$ be the random variable defined as follows: 
$X_{\{x,y\}}=2d$ if $x-y\in\{e_i,-e_i\}$ for $i\in I$ and $X_{\{x,y\}}=1$ otherwise. 
Note that for every $\{x,y\}\in E(\Z^d)$, it holds that ${\mathbb E}X_{\{x,y\}}=\frac{r}{d}\cdot 2d+\left(1-\frac{r}{d}\right)\cdot 1<2r+1$.
Note that for every $\{x,y\}\in E(\Z^d)$, 
$$\lvert \{u,v\}\in E(\Z^d)\colon  \{x,y\}\subset Q_I(T_I(u)\cup Q_I(T_I(v))\}\rvert\leq X_{\{x,y\}}.$$
The same argument as in the proof of Proposition \ref{prop:functional_bdry}, where Lemma \ref{lem:rough} is replaced by Lemma \ref{lem:rough_fine}, and Lemma \ref{lem:Ron} is applied in an appropriate $r$-dimensional affine subspace of ${\mathbb Z}^d$ (either $u+{\rm span}(\{e_i\}_{i\in I})$ or $v+{\rm span}(\{e_i\}_{i\in I})$), yields that for every $\{u,v\}\in E(\Z^d)$, 
\begin{equation*}
\left\lvert\tau_I(T_I(u))-\tau_I(T_I(v))\right\rvert
\leq\frac{5}{2^{r-1}}\TV\left(\tau; Q_I(T_I(u))\cup Q_I(T_I(v))\right).
\end{equation*}
Therefore, 
\begin{align*}
\TV(\tau_I)&\leq
2^r\sum_{\{u,v\}\in E(\Z^d)}\left\lvert\tau_I(T_I(u))-\tau_I(T_I(v))\right\rvert\\
&\leq 10\sum_{\{u,v\}\in E(\Z^d)}
\TV(\tau;Q_I(T_I(u)\cup Q_I(T_I(v)))\\
&=10\sum_{\{x,y\}\in E(\Z^d)}\lvert \tau(x)-\tau(y)\rvert\cdot\lvert \{u,v\}\in E(\Z^d)\colon  \{x,y\}\subset Q_I(T_I(u)\cup Q_I(T_I(v))\}\rvert\\
&\leq10\sum_{\{x,y\}\in E(\Z^d)}\lvert \tau(x)-\tau(y)\rvert X_{\{x,y\}}.
\end{align*}
Hence,
\begin{equation*}
\mathbb{E}\TV(\tau_I)\leq 10\sum_{\{x,y\}\in E(\Z^d)}\lvert \tau(x)-\tau(y)\rvert\, \mathbb{E}X_{\{x,y\}}\leq 10(2r+1)\TV(\tau).\qedhere
\end{equation*}
\end{proof}

\begin{proposition}\label{prop:functional_diff_fine}
Let $I\in\binom{[d]}{r}$ be chosen uniformly at random. Then, for every $\tau\colon \Lambda \to \Z$, the following holds:
$$
{\mathbb E}\lVert \tau_I-\tau\rVert_1\leq \frac{2r}{d}\TV(\tau). 
$$
\end{proposition}

\begin{proof}
For every $1\leq i\leq d$, let $X_i$ be the random variable defined as follows: $X_i=1$ if $i\in I$ and $X_i=0$ otherwise.
For every $u\in{\mathbb Z}^d$ it holds that
$\left\lvert\tau(u)-\tau_I(u)\right\rvert\leq 2\left\lvert\tau(u)-\tau^{\rm rough}_I(u)\right\rvert$.
Hence, for every $w\in{\mathbb Z}^d$, by Lemma \ref{lem:diff_rough},
\begin{align*}
\sum_{v\in Q_I(T_I(w))}\lvert\tau(u)-\tau_I(u)\rvert&\leq 2\sum_{v\in Q_I(T_I(w))}\lvert\tau(u)-\tau^{\rm rough}_I(u)\rvert\\
&\leq 2\sum_{\{u,v\}\in E(\Z^d)\cap\binom{Q_I(T_I(w))}{2}}\lvert\tau(u)-\tau(v)\rvert.
\end{align*}

Therefore,
$$
\lVert \tau_I-\tau\rVert_1
\leq 2\sum_{i\in I}\sum_{\substack{\{u,v\}\in E(\Z^d)\\ u-v\in\{-e_i,e_i\}}}\lvert \tau(u)-\tau(v)\rvert=2\sum_{i=1}^d \sum_{\substack{\{u,v\}\in E(\Z^d)\\ u-v\in\{-e_i,e_i\}}}\lvert \tau(u)-\tau(v)\rvert X_i
$$
and hence,
\begin{align*}
{\mathbb E}\lVert \tau_I-\tau\rVert_1
&\leq 2\sum_{i=1}^d \sum_{\substack{\{u,v\}\in E(\Z^d)\\ u-v\in\{-e_i,e_i\}}}\lvert \tau(u)-\tau(v)\rvert {\mathbb E}X_i\\
&=\frac{2r}{d}\sum_{i=1}^d\sum_{\substack{\{u,v\}\in E(\Z^d)\\ u-v\in\{-e_i,e_i\}}}\lvert \tau(u)-\tau(v)\rvert=\frac{2r}{d}\TV(\tau).
\qedhere\end{align*}
\end{proof}

\begin{proof}[Proof of Proposition \ref{prop:Itau}]
Let $I\in\binom{[d]}{r}$ be chosen uniformly at random, By Markov inequality and Propositions \ref{prop:functional_bdry_fine} and \ref{prop:functional_diff_fine},
$$
\P\left(\TV(\tau_I)\geq 20(2r+1)\TV(\tau)+1\right)\leq\frac{\mathbb{E} \TV(\tau_I)}{20(2r+1)\TV(\tau)+1}<\frac{1}{2} 
$$
and
$$
\P\left(\lVert \tau_I-\tau\rVert_1\geq\frac{4r}{d}\TV(\tau)\right)\leq\frac{{\mathbb E}\lVert \tau_I-\tau\rVert_1}{\frac{4r}{d}\TV(\tau)}\leq\frac{1}{2}. 
$$
Hence,
$$
\P\left(\TV(\tau_I)\leq 20(2r+1)\TV(\tau) \text{ and } \lVert \tau_I-\tau\rVert_1<\frac{4r}{d}\TV(\tau)\right)>0
$$
and the result follows.
\end{proof}

\subsubsection{Entropy bounds} 

\begin{lemma}\label{lem:net}
For every $A\subset {\mathbb Z}^d$ and finite $B\subseteq\partial^{\rm out}A$, there is a 
set $S\subseteq B$ such that $\lvert S\rvert<\frac{1}{d}\lvert\partial A\rvert$ and $B\subseteq\bigcup_{a\in S}{\mathcal B}_4(a)$.
\end{lemma}

\begin{proof}
We first show that for every $a\in\partial^{\rm out}A$,
\begin{equation}\label{eq:tictactoe}
\left\lvert \partial 
A\cap\left({\mathbb Z}^d\times{\mathcal B}_2(a)\right)\right\rvert>d.    
\end{equation}
With no loss of generality assume that $a+e_d\in A$, and let $E:=\{-e_i\}_{i=1}^{d-1}\cup\{e_i\}_{i=1}^{d-1}$.
For every $u\in E$, denote
$$
{\mathcal T}_u:=\left\{(a+u,a),(a+e_d,a+u+e_d), (a+u+e_d,a+u)\right\}.
$$
If $a+u\in A$ then $(a+u,a)\in\partial A$; if $a+u+e_d\notin A$ then $(a+e_d,a+u+e_d)\in\partial 
A$; finally, if $a+u\notin A$ and $a+u+e_d\in A$ then $(a+u+e_d,a+u)\in\partial A$.
Hence, $\partial A\cap{\mathcal T}_u\neq\emptyset$ for every $u\in E$, and \eqref{eq:tictactoe} follows since the $2d-2>d$ sets $\{{\mathcal T}_u\}_{u\in E}$
are mutually disjoint.

Now, let $S$ be a set of maximal cardinality in $B$ such that the sets $\{{\mathcal B}_2(a)\}_{a\in S}$ are mutually disjoint.
The maximality of $S$ implies that $B\subseteq \bigcup_{a\in S}{\mathcal B}_4(a)$, and by \eqref{eq:tictactoe}, 
\begin{equation*}
\lvert S\rvert<\frac{1}{d}\sum_{a\in S}\left\lvert\partial A\cap\left({\mathbb Z}^d\times{\mathcal B}_2(a)\right)\right\rvert\leq\frac{1}{d}\lvert\partial A\rvert.    
\qedhere\end{equation*}
\end{proof}

We will say that a set $A\subseteq{\mathbb Z}^d$ is \emph{$\ell_1^+$-connected} if for any two points $a,b\in A$ there is a sequence  $a=s_0,s_1,\ldots,s_n=b$ of points in $A$ such that $\lVert s_{i-1}-s_i\rVert_1\leq 2$ for every $1\leq i\leq n$.

\begin{lemma}\label{lem:ham}
Let $A\subset {\mathbb Z}^d$ be an $\ell_1^+$-connected finite set, and assume that there is a set $S\subseteq A$ such that $A\subseteq\bigcup_{a\in S}{\mathcal B}_4(a)$.
Then, 
\begin{equation}\label{eq:diam}
{\rm diam}(A)<10\lvert S\rvert
\end{equation}
Moreover, there is an ordering $a_1,a_2,\ldots, a_{\lvert S\rvert}$ of $S$ such that, denoting $a_{\lvert S\rvert+1}:=a_1$,
\begin{equation}\label{eq:ham}
\sum_{i=1}^{\lvert S\rvert}\lVert a_i-a_{i+1}\rVert_1<20\lvert S\rvert.
\end{equation}
Consequently, for every finite $\Omega\subset{\mathbb Z}^d$, there is an ordering $\omega_1,\omega_2,\ldots,\omega_{\lvert \Omega\rvert}$ of $\Omega$ such that, denoting $\omega_{\lvert \Omega\rvert+1}:=\omega_1$,
\begin{equation}\label{eq:ham+}
\sum_{i=1}^{\lvert \Omega\rvert}\lVert \omega_i-\omega_{i+1}\rVert_1<20\lvert S\rvert+8\lvert\Omega\rvert+2\sum_{\omega\in\Omega}{\rm dist}(\omega,A).
\end{equation}
\end{lemma}

\begin{proof}
Consider the complete graph $K$ on the vertex set $S$, and its spanning subgraph $G$ in which $a,b\in S$ are adjacent if there are $u \in {\mathcal B}_4(a)$, $v\in{\mathcal B}_4(b)$ such that $\lVert u-v\rVert_1\leq 2$.
For any edge $e=\{a,b\}$ of $K$, denote $\lVert e\rVert:=\lVert a-b\rVert_1$.
Note that $\lVert e\rVert\leq 10$ for every edge $e$ of $G$.
Since $A$ is $\ell_1^+$-connected, it follows that the graph $G$ is connected.
Let $\mathcal T$ be a spanning tree of $G$.

To prove \eqref{eq:diam}, we need to show that $\lVert a-\tilde{a}\rVert_1<10\lvert S\rvert$ for every $a,\tilde{a}\in A$. There are $s,\tilde{s}\in S$ such that $a\in{\mathcal B}_4(s)$ and $\tilde{a}\in{\mathcal B}_4(\tilde{s})$. Let $s=s_0,s_1,\ldots,s_k=\tilde{s}$ be the unique path from $s$ to $\tilde{s}$ in $\mathcal T$. Then,
$$\lVert a-\tilde{a}\rVert_1\leq \lVert a-s\rVert_1+\sum_{i=1}^k\lVert s_{i-1}-s_i\rVert_1+\lVert \tilde{s}-\tilde{a}\rVert_1\leq 4+10k+4<10(k+1)\leq 10\lvert S\rvert.$$

Using the structure of the tree, we may arrange the edges of $\mathcal T$, each taken in both directions to create a directed cycle ${\mathcal C}_0$ that goes through all the vertices. Let ${\mathcal C}_1$ be the simple cycle in $K$ obtained from ${\mathcal C}_0$ by omitting multiple occurrences of vertices.
Then, by using the triangle inequality,
$$
\sum_{e\in E({\mathcal C}_1)}\lVert e \Vert\leq
\sum_{e\in E({\mathcal C}_0)}\lVert e \Vert=2\sum_{e\in E({\mathcal T})}\lVert e \Vert\leq 20\lvert E(\mathcal T)\rvert=20(\lvert S\rvert-1)$$
which proves \eqref{eq:ham}.

Finally, let $\Omega\subset{\mathbb Z}^d$ be a finite set.
Let $a_1,a_2,\ldots,a_{\lvert S\rvert}$ be an ordering of $S$ such that (denoting $a_{\lvert S\rvert+1}:=a_1$)
$\sum_{i=1}^{\lvert S\rvert}\lVert a_i-a_{i+1}\rVert_1<20\lvert S\rvert$.
For every $\omega\in\Omega$, there is $1\leq n(\omega)\leq |S|$ such that $\lVert \omega-a_{n(\omega)}\rVert_1\leq 4+{\rm dist}(\omega,A)$.
Let $\omega_1,\omega_2,\ldots,\omega_{|\Omega|}$ be an ordering of $\Omega$ such that $n(\omega_j)\leq n(\omega_{j+1})$ for every $1\leq j<|\Omega|$ (it is easy to see that such orderings exist), and denote $\omega_{\lvert \Omega\rvert+1}:=\omega_1$.
Then, for every $1\leq j\leq|\Omega|$,
\begin{align*}
\lVert \omega_j-\omega_{j+1}\rVert_1&\leq \lVert \omega_j-a_{n(\omega_j)}\rVert_1+\sum_{i=n(\omega_j)}^{n(\omega_{j+1})-1}\lVert a_i-a_{i+1}\rVert_1+\lVert a_{n(\omega_{j+1})}-\omega_{j+1}\rVert_1\\
&\leq \sum_{i=n(\omega_j)}^{n(\omega_{j+1})-1}\lVert a_i-a_{i+1}\rVert_1+8+{\rm dist}(\omega_j,A)+{\rm dist}(\omega_{j+1},A),
\end{align*}
where, for $j=\lvert\Omega\rvert$, the sum $\sum_{i=n(\omega_{\lvert\Omega\rvert})}^{n(\omega_1)-1}\lVert a_i-a_{i+1}\rVert_1$  should be interpreted as $\sum_{i=n(\omega_{\lvert\Omega\rvert})}^{\lvert S\rvert}\lVert a_i-a_{i+1}\rVert_1+\sum_{i=1}^{n(\omega_1)-1}\lVert a_i-a_{i+1}\rVert_1$.
Hence,
\begin{align*}
\sum_{j=1}^{\lvert\Omega\rvert}\lVert \omega_j-\omega_{j+1}\rVert_1&\leq \sum_{i=1}^{\lvert S\rvert}\lVert a_i-a_{i+1}\rVert_1+8\lvert\Omega\rvert+2\sum_{\omega\in\Omega}{\rm dist}(\omega,A)\\
&<20\lvert S\rvert+8\lvert\Omega\rvert+2\sum_{\omega\in\Omega}{\rm dist}(\omega,A).
\qedhere
\end{align*}
\end{proof}

Following Tim\'{a}r \cite{T}, we define, for a set $A\subseteq{\mathbb Z}^d$ and $v\in{\mathbb Z}^d\cup\{\infty\}$, the outer vertex boundary of $A$ visible from $v$:
\begin{equation}\label{eq:viz}
\partial_{{\rm vis}(v)}A:=
\left\{u\in\partial^{\rm out}A\colon \text{there exists a path from } u \text{ to } v \text{ not intersecting } A \right\}.
\end{equation}

\begin{obs}\label{obs:diam_viz}
For every bounded $A\subseteq{\mathbb Z}^d$ and every $u\in A$ it holds that  \begin{equation*}
{\rm dist}(u,\partial_{{\rm vis}(\infty)}A)<{\rm diam}(\partial_{{\rm vis}(\infty)}A).
\end{equation*}
\end{obs}

\begin{proof}
There is $w\in\partial_{{\rm vis}(\infty)}A$ such that
${\rm dist}(u,\partial_{{\rm vis}(\infty)}A)=\lVert u-w\rVert_1$. With no loss of generality we may assume that the first coordinate of $u-w$ is non-negative.
Let $n_0:=\max\{n\in{\mathbb Z}\colon u+ne_1\in A\}+1$. 
Obviously, $u+n_0e_1\in\partial_{{\rm vis}(\infty)}A$.
Therefore,
\begin{equation*}
{\rm dist}(u,\partial_{{\rm vis}(\infty)}A)=\lVert u-w\rVert_1<\lVert (u+n_0e_1)-w\rVert_1\leq{\rm diam}(\partial_{{\rm vis}(\infty)}A).
\qedhere
\end{equation*}
\end{proof}

The following lemma is a special case of \cite{T}*{Theorem 3}. 

\begin{lemma}\label{lem:Timar3}
For every connected $A\subseteq{\mathbb Z}^d$ and every $v\in{\mathbb Z}^d\cup\{\infty\}$, the set $\partial_{vis(v)}A$ is $\ell_1^+$-connected.  
\end{lemma}

\begin{obs}\label{obs:no_LC}
The number of level components of $\tau_N$ satisfies the following bound
\begin{equation*}
\lvert\mathcal{LC}(\tau_N)\rvert \leq\frac{\TV(\tau_N)}{d N^{d-1}}.
\end{equation*}
In particular, 
\begin{equation}\label{eq:no_LC}
\lvert\mathcal{LC}(\tau)\rvert \leq\frac{\TV(\tau)}{d}.
\end{equation}
Similarly, for the number of level components of $\tau_I$,
\begin{equation*}
|\mathcal{LC}(\tau_I)| \leq 
\frac{\TV(\tau_I)}{d \,2^{|I|-1}}.
\end{equation*}
\end{obs}

\begin{proof}
Every level component $A$ of $\tau_N$ is a (disjoint) union of discrete cubes of side length $N$; in particular, $\lvert A\rvert\geq N^d$ and hence $\lvert\partial A\rvert\geq 2d\, N^{d-1}$, by \eqref{eq:bdry_small}. Therefore,
$$
\lvert\mathcal{LC}(\tau_N)\rvert\leq\frac{1}{2d\,N^{d-1}}\sum_{A\in\mathcal{LC}(\tau_N)}\lvert\partial A\rvert\leq\frac{1}{2d\,N^{d-1}}2\,\TV(\tau_N).
$$
Similarly, for every level component $A$ of $\tau_I$ it holds, by \eqref{eq:bdry_small}, that $\lvert\partial A\rvert\geq 2d\lvert A\rvert^{1-\frac{1}{d}}\geq 2d\, 2^{\lvert I\rvert-\frac{\lvert I\rvert}{d}}\geq 2d\, 2^{\lvert I\rvert-1}$. Hence,
\begin{equation*}
\lvert\mathcal{LC}(\tau_N)\rvert\leq\frac{1}{2d\,2^{\lvert I\rvert-1}}\sum_{A\in\mathcal{LC}(\tau_I)}\lvert\partial A\rvert\leq\frac{1}{2d\,2^{\lvert I\rvert-1}}2\,\TV(\tau_I).
\qedhere\end{equation*}
\end{proof}

\begin{proposition}\label{prop: trip_entropy_bound_general}
Let $\tau, \tilde{\tau}$ be two shifts and let $r$ be a positive integer such that for every level component $\tilde{A}\in \mathcal{LC}(\tilde{\tau})$ of $\tilde{\tau}$, there exists a level component $A\in \mathcal{LC}(\tau)$ of $\tau$ such that ${\rm dist}(\tilde{A}, \partial_{{\rm vis}(\infty)}A) \leq r $. Then,
\begin{equation*}
R(\tilde{\tau}) \leq R(\tau) + \frac{88}{d} \TV(\tau)+(2r+8)\lvert\mathcal{LC}(\tilde{\tau})\rvert.
\end{equation*}
\end{proposition}

\begin{proof}
For simplicity, denote $N:=\lvert\mathcal{LC}(\tau)\rvert$.
Let $(u_i)_{i=0}^{N-1}$ be a sequence of points in ${\mathbb Z}^d$ such that $u_0=0$, $\sum_{i=1}^{N-1}\lVert u_{i-1}-u_i\rVert_1=R(\tau)$ and each $u_i$ is in a different level component of $\tau$, which we denote $A_i$.
For every $\tilde{A}\in \mathcal{LC}(\tilde{\tau})$ there are $0\leq i(\tilde{A})\leq N-1$ and $\omega(\tilde{A})\in \tilde{A}$ such that ${\rm dist}(\omega(\tilde{A}), \partial_{{\rm vis}(\infty)} A_{i(\tilde{A})})\leq r$.
For every $0\leq i\leq N-1$, let $$\Omega_i:=\{u_i\}\cup\{\omega(\tilde{A})\colon \tilde{A}\in\mathcal{LC}(\tilde{\tau}),\,i(\tilde{A})=i\}.
$$
By Lemma \ref{lem:net}, there is a 
set $S_i\subseteq \partial_{{\rm vis}(\infty)} A_i$ such that $\lvert S_i\rvert<\frac{1}{d}\lvert\partial A_i\rvert$ and $\partial_{{\rm vis}(\infty)} A_i\subseteq\bigcup_{a\in S_i}{\mathcal B}_4(a)$. 
By Observation \ref{obs:diam_viz} and \eqref{eq:diam}, 
$${\rm dist}(x_i,\partial_{{\rm vis}(\infty)} A_i)<{\rm diam}(\partial_{{\rm vis}(\infty)} A_i)<10\lvert S_i\rvert<\frac{10}{d}\lvert\partial A_i\rvert.$$
The set $\partial_{{\rm vis}(\infty)} A_i$ is $\ell_1^+$-connected, by Lemma \ref{lem:Timar3}. Hence, by \eqref{eq:ham+}, there is an ordering $\omega^{(i)}_1,\omega^{(i)}_2,\ldots,\omega^{(i)}_{\lvert \Omega_i\rvert}$ of $\Omega_i$ such that, denoting $\omega^{(i)}_{\lvert \Omega_i\rvert+1}:=\omega^{(i)}_1$,
\begin{align*}
\sum_{j=1}^{\lvert\Omega_i\rvert}\lVert\omega^{(i)}_j-\omega^{(i)}_{j+1}\rVert_1<& 20\lvert S_i\rvert+8\lvert\Omega_i\rvert+2\sum_{\omega\in\Omega_i}{\rm dist}(\omega,\partial_{{\rm vis}(\infty)} A_i)\\
<&\frac{20}{d}\lvert \partial A_i\rvert+8\lvert\Omega_i\rvert+\frac{20}{d}\lvert \partial A_i\rvert+2(\lvert\Omega_i\rvert-1)r\\
=&\frac{40}{d}\lvert \partial A_i\rvert+(2r+8)(\lvert\Omega_i\rvert-1)+8.
\end{align*}
With no loss of generality we may assume that $\omega^{(i)}_1=x_i$.
Hence, for every $1\leq i\leq N-1$,
\begin{align*}
\lVert \omega^{(i-1)}_{\lvert \Omega_{i-1}\rvert}-\omega^{(i)}_2\rVert_1&\leq\lVert \omega^{(i-1)}_{\lvert \Omega_{i-1}\rvert}-\omega^{(i-1)}_1\rVert_1+\lVert \omega^{(i-1)}_1-\omega^{(i)}_1\rVert_1+\lVert \omega^{(i)}_1-\omega^{(i)}_2\rVert_1\\
&=\lVert \omega^{(i-1)}_{\lvert \Omega_{i-1}\rvert}-\omega^{(i-1)}_1\rVert_1+\lVert u_{i-1}-u_i\rVert_1+\lVert \omega^{(i)}_1-\omega^{(i)}_2\rVert_1.
\end{align*}
Therefore, considering the sequence
$$
0=\omega^{(0)}_1,\omega^{(0)}_2,\ldots,\omega^{(0)}_{\lvert\Omega_0\rvert},\omega^{(1)}_2,\omega^{(1)}_3,\ldots,\omega^{(1)}_{\lvert\Omega_1\rvert},\omega^{(2)}_2,\omega^{(2)}_3,\ldots,\omega^{(N-1)}_{\lvert\Omega_{N-1}\rvert},
$$
we conclude that
\begin{align*}
R(\tilde{\tau}&)\leq\sum_{j=1}^{\lvert\Omega_0\rvert-1}\lVert \omega^{(0)}_j-\omega^{(0)}_{j+1}\rVert_1+\sum_{i=1}^{N-1}\left(\lVert \omega^{(i-1)}_{\lvert\Omega_{i-1}\rvert}-\omega^{(i)}_2\rVert_1+\sum_{j=2}^{\lvert\Omega_i\rvert-1}\lVert \omega^{(i)}_j-\omega^{(i)}_{j+1}\rVert_1\right)\\
\leq&\sum_{j=1}^{\lvert\Omega_0\rvert-1}\lVert \omega^{(0)}_j-\omega^{(0)}_{j+1}\rVert_1+\sum_{i=1}^{N-1}\left(\lVert \omega^{(i-1)}_{\lvert\Omega_{i-1}\rvert}-\omega^{(i-1)}_1\rVert_1+\lVert u_{i-1}-u_i\rVert_1+\sum_{j=1}^{\lvert\Omega_i\rvert-1}\lVert \omega^{(i)}_j-\omega^{(i)}_{j+1}\rVert_1\right)\\
=&\sum_{i=1}^{N-1}\sum_{j=1}^{\lvert\Omega_i\rvert}\lVert \omega^{(i)}_j-\omega^{(i)}_{j+1}\rVert_1-\lVert \omega^{(N-1)}_{\lvert\Omega_{N-1}\rvert}-\omega^{(N-1)}_1\rVert_1+R(\tau)\\
<&\sum_{i=1}^{N-1}\left(\frac{40}{d}\lvert \partial A_i\rvert+(2r+8)(\lvert\Omega_i\rvert-1)+8\right)+R(\tau)\\
=&\frac{40}{d}\sum_{i=1}^{N-1}\lvert\partial A_i\rvert+(2r+8)\sum_{i=1}^{N-1}(\lvert\Omega_i\rvert-1)+8(N-1)+R(\tau),
\end{align*}
and the result follows since $\sum_{i=1}^{N-1}\lvert\partial A_i\rvert\leq 2\TV(\tau)$, $\sum_{i=1}^{N-1}(\lvert\Omega_i\rvert-1)=\lvert\mathcal{LC}(\tilde{\tau})\rvert$ and by \eqref{eq:no_LC}, $N-1<\TV(\tau)/d$.
\end{proof}

\begin{lemma}
There is a universal constant $C>0$ such that for every shift $\tau$ and for every integer $N\geq 2$,
\begin{equation}\label{eq:upper_bound_R_of_coarse}
R(\tau_{N}) \leq R(\tau)+ \frac{C}{d}    \TV(\tau), 
\end{equation}
and for every $I\in \comp(\tau)$,
\begin{equation}\label{eq:upper_bound_R_of_fine}
R(\tau_I) \leq R(\tau) +\frac{C}{d} \TV(\tau).
\end{equation}
\end{lemma}

\begin{proof}
Let $\tilde{A}$ be a level component of $\tau_{N}$. By the definition of $\tau_{N}$, there are necessarily $u_1\sim u_2$, both at distance at most $N$ from $\tilde{A}$, such that $u_1\in A_1$ and $u_2\in A_2$, where $A_1,A_2$ are distinct level components of $\tau$. 
It is easy to see that for every two disjoint connected sets $A_1, A_2\subseteq {\mathbb Z}^d$, it holds that
\begin{equation}\label{eq:1_of_2_in_viz}
E({\mathbb Z}^d)\cap(A_1\times A_2)\subseteq\left(A_1\times \partial_{{\rm vis}(\infty)}(A_1)\right)\cup\left(\partial_{{\rm vis}(\infty)}(A_2)\times A_2\right).
\end{equation}
It follows that for every level component $\tilde{A}$ of $\tau_{N}$, there is a level component $A$ of $\tau$ such that ${\rm dist}(\tilde{A},\partial_{{\rm vis}(\infty)}A)\leq N$.
Hence, by Proposition \ref{prop: trip_entropy_bound_general},
\begin{equation*}
R(\tau_N) \leq R(\tau) + \frac{88}{d} \TV(\tau)+(2N+8)\lvert\mathcal{LC}(\tau_N)\rvert,
\end{equation*}
and \eqref{eq:upper_bound_R_of_coarse} follows, since by Observation \ref{obs:no_LC} and Proposition \ref{prop:functional_bdry}, 
\begin{equation*}
\lvert\mathcal{LC}(\tau_{N})\rvert\leq \frac{1}{d N^{d-1} }{\rm TV}(\tau_{N})\leq\frac{10}{N^{d-1}}{\rm TV}(\tau).
\end{equation*}

Similarly, if $I\subseteq[d]$, then for every level component $\tilde{A}$ of $\tau_I$, there is a level component $A$ of $\tau$ 
such that ${\rm dist}(\tilde{A},\partial_{{\rm vis}(\infty)}A)\leq 2$.
Hence, by Proposition \ref{prop: trip_entropy_bound_general},
\begin{equation*}
R(\tau_I) \leq R(\tau) + \frac{88}{d} \TV(\tau)+12\lvert\mathcal{LC}(\tau_I)\rvert,
\end{equation*}
and \eqref{eq:upper_bound_R_of_fine} follows, for $I\in{\rm comp}(\tau)$, since then, by Observation \ref{obs:no_LC},
\begin{equation*}
\lvert\mathcal{LC}(\tau_{I})\rvert\leq \frac{1}{d 2^{\lvert I\rvert-1} }{\rm TV}(\tau_I)\leq\frac{20(2\lvert I\rvert+1)}{d\,2^{\lvert I\rvert-1}}{\rm TV}(\tau).
\qedhere
\end{equation*}
\end{proof}

The following observation is a simple one, deriving from the definition of total variation and triangle inequality:
\begin{obs}\label{obs:upper_bound_TV_of_sum}
For any two shifts $\tau,\tau'$ the following holds:
\begin{equation*}
\TV(\tau+\tau')\leq \TV(\tau)+\TV(\tau').
\end{equation*}
\end{obs}

\begin{lemma}\label{lem:upper_bound_R_of_sum}
There is a universal constant $c$ such that for any two shifts $\tau,\tau'$,
    \begin{align*}
R\left(\tau+\tau'\right)& \leq 2R\left(\tau\right)+R\left(\tau'\right)+\frac{88}{d}\left(\TV\left(\tau\right)+\TV\left(\tau'\right)\right)+\frac{10}{d} \TV(\tau+\tau')\\ & \leq 2R\left(\tau\right)+R\left(\tau'\right)+\frac{98}{d}\left(\TV\left(\tau\right)+\TV\left(\tau'\right)\right),
\end{align*}
where the second inequality follows 
by Observation \ref{obs:upper_bound_TV_of_sum}.
\end{lemma}

\begin{proof}
Suppose that $u_1\sim u_2$ belong to different level components of $\tau+\tau'$. Then, $u_1\in A_1$ and $u_2\in A_2$, where $A_1$ and $A_2$ are distinct level components of the same function in $\{\tau,\tau'\}$. By \eqref{eq:1_of_2_in_viz}, $u_1\in \partial_{{\rm vis}(\infty)}(A_2)$ or $u_2\in \partial_{{\rm vis}(\infty)}(A_1)$.
It follows that for every $\tilde{A}\in\mathcal{LC}(\tau+\tau')$, there is $A\in\mathcal{LC}(\tau)\cup\mathcal{LC}(\tau')$
such that ${\rm dist}(\tilde{A}, \partial_{{\rm vis}(\infty)}A) \leq 1 $. Then, a similar argument to that of the proof of Proposition \ref{prop: trip_entropy_bound_general} yields that 
\begin{equation*}
R(\tau+\tau') \leq 2R(\tau) + R(\tau')+\frac{88}{d} \left(\TV(\tau)+\TV(\tau')\right) + 10\lvert\mathcal{LC}(\tau+\tau')\rvert
\end{equation*}
and the result follows since $\lvert\mathcal{LC}(\tau+\tau')\rvert\leq \frac{1}{d} {\rm TV}(\tau+\tau')$, 
by \eqref{eq:no_LC}.
\end{proof}

\subsection{Enumeration of shifts}
\label{sec:Enumeration} 
The goal of this section is to prove Proposition \ref{prop:enumerate_shift_functions} and Corollary \ref{cor:enumerate_grainings_shift_functions}.
Before proving Proposition \ref{prop:enumerate_shift_functions}, we first show how it easily implies Corollary \ref{cor:enumerate_grainings_shift_functions}.
Intuitively it is clear that the number of possible grainings of a shift of bounded complexity will decrease significantly as the scale of the grainings grows. 
Corollary \ref{cor:enumerate_grainings_shift_functions} quantifies this simple statement and is a direct result of the previously obtained total variation and trip entropy bounds for coarse and fine grainings, a simple scaling argument, and Proposition~\ref{prop:enumerate_shift_functions} bounding the number of general shifts with limited total variation and trip entropy.

\begin{proof}[Proof of Corollary \ref{cor:enumerate_grainings_shift_functions}]
To show the first bound, let ${\mathcal S}_N$ be the set of shifts which are constant in each set of the partition $\mathcal{P}_N$. Then, by Proposition \ref{prop:functional_bdry} and  \eqref{eq:upper_bound_R_of_coarse} there is a universal constant $C>0$ such that
\begin{equation*}\{\tau_{N}\colon\tau\in{\mathcal S},\, \TV(\tau) \leq \lambda,\, R(\tau)\leq \rho\} \subseteq \left\{\tau\in{\mathcal S}_N\colon \TV(\tau)\leq 10d \lambda,\, R(\tau)\leq  \rho +\frac{C}{d}\lambda\right\}. 
\end{equation*}
Denote by $\mu_N:{\mathbb Z}^d\to{\mathbb Z}^d$ the multiplication by $N$.   
The mapping $\tau\mapsto\tau\circ\mu_N$ is obviously a bijection of ${\mathcal S}_N$ onto $\mathcal S$, and moreover, for every $\tau\in{\mathcal S}_N$, clearly $\TV(\tau\circ\mu_N)\leq \frac{1}{N^{d-1}}\TV(\tau)$ and $R(\tau\circ\mu_N)\leq R(\tau)$. Hence,
\begin{align*}
\lvert\{\tau_{N}\colon\tau\in{\mathcal S},\, \TV(\tau) \leq \lambda,\, R(\tau)\leq \rho\}\rvert
&\leq\left\lvert\left\{\tau\in{\mathcal S}_N\colon \TV(\tau)\leq 10d \lambda,\, R(\tau)\leq  \rho +\frac{C}{d}\lambda\right\}\right\rvert\\
&\leq \left\lvert\left\{\tau\in{\mathcal S}\colon \TV(\tau)\leq \frac{10d \lambda}{N^{d-1}} ,\, R(\tau)\leq  \rho+\frac{C}{d}\lambda \right\}\right\rvert. 
\end{align*} 
The bound~\eqref{eq:coarse_enumerate_shift_functions} now follows directly from Proposition~\ref{prop:enumerate_shift_functions}.
The proof of \eqref{eq:fine_enumerate_shift_functions} is similar.
\end{proof}

\begin{proof}[Proof of Proposition \ref{prop:enumerate_shift_functions}]
Fix a shift $\tau$.
Let $J(\tau)$ be the number of level components of $\tau$, let $(v_j(\tau))_{j=1}^{J(\tau)}$ be a sequence of points in ${\mathbb Z}^d$ such that $v_1(\tau)=0$, $\sum_{j=1}^{J(\tau)-1}\lVert v_j(\tau)-v_{j+1}(\tau)\rVert_1=R(\tau)$ and there is a unique element of $\{v_j(\tau)\}_{j=1}^{J(\tau)}$ in each level component of $\tau$.
For every $1\leq j\leq J(\tau)$, let $L_j(\tau)$ be the level component of $\tau$ containing $v_j(\tau)$. 
Let $j_0(\tau)$ be the index of the unique unbounded level component of $\tau$.

Define a partial order $\leq_{\tau}$ on the set $[J(\tau)]$ as follows: $i\leq_{\tau}j$ if every path from $L_i$ to $\infty$ necessarily intersects $L_j$.
For every $j\in[J(\tau)]$, let $U_j(\tau):=\bigcup_{i\leq_{\tau}j}L_i(\tau)$.
Clearly, $U_{j_0(\tau)}(\tau)={\mathbb Z}^d$. 

Let ${\mathcal G}(\tau)$ be the graph on the vertex set $[J(\tau)]$ in which $i\neq j$ are adjacent if there are neighbouring $u\in L_i$ and $v\in L_j$.
Define a rooted spanning tree ${\mathcal T}(\tau)$ of ${\mathcal G}(\tau)$ in the following inductive manner. 
Set $V_0:=\{j_0(\tau)\}$, $\tilde{V}_0:=\{j_0(\tau)\}$ and $E_0:=\emptyset$, and for every $1\leq r\leq J(\tau)$, let $i_r:=\min \tilde{V}_{r-1}$ and set $V_r:=V_{r-1}\cup\{j\in [J(\tau)]\setminus V_{r-1}: j \text{ is adjacent to } i_r \text{ in } {\mathcal G}(\tau)\}$, $\tilde{V}_r:=(\tilde{V}_{r-1}\setminus\{i_r\})\cup(V_r\setminus V_{r-1})$ and $E_r:=E_{r-1}\cup\{(i_r,j):j\in V_r\setminus V_{r-1}\}$.
Finally, let ${\mathcal T}(\tau)$ be the tree on the vertex set $V_{J(\tau)}=[J(\tau)]$ whose set of (directed) edges is $E_{J(\tau)}$. 
For every directed edge $e=(i,j)$ of ${\mathcal T}(\tau)$, let $s_e(\tau):=\tau(L_i)-\tau(L_j)$.

Clearly, $\sum_{j=1}^{J(\tau)} \lvert\partial L_j(\tau)\rvert\leq 2{\rm TV}(\tau)$ and by \eqref{eq:bdry_small}, $\lvert\partial L_j(\tau)\rvert\geq 2d$ for  every $1\leq j\leq J(\tau)$.
Consequently, $J(\tau)\leq\frac{1}{2d}{\rm TV}(\tau)$. Moreover, clearly $J(\tau)\leq 1+R(\tau)$.
For every positive integer $J\leq\min\{\frac{\lambda}{2d}.1+\rho\}$, let
$$\tilde{\mathcal S}_J:=\{\tau\in{\mathcal S}\colon {\rm TV}(\tau) \leq \lambda,\, R(\tau) \leq \rho,\, J(\tau)=J\}.$$
The map 
$$
\chi\colon \tau\mapsto
\left(J(\tau),
j_0(\tau),
(U_j(\tau))_{j_0(\tau)\neq j\in[J(\tau)]},
(s_e(\tau))_{e\in E({\mathcal T}(\tau))}\right)
$$
is clearly injective, hence
\begin{equation}\label{eq:by_injectivity}
\lvert\{\tau\in{\mathcal S}\colon {\rm TV}(\tau) \leq \lambda,\, R(\tau) \leq \rho\}\rvert=\sum_{J\leq\frac{\lambda}{2d}}\lvert\tilde{\mathcal S}_J\rvert\leq \sum_{J\leq\frac{\lambda}{2d}}\lvert\chi(\tilde{\mathcal S}_J)\rvert.     
\end{equation}

In the estimates below we will use the following estimates several times. First, for every positive integers $k$ and $n$,
\begin{equation}
\label{eq:binom_estimate}
\binom{n}{k}\leq 
\frac{n^k}{k!}<\left(\frac{en}{k}\right)^k<\left(\frac{3n}{k}\right)^k.
\end{equation}
For every positive integers $k$ and $m$ there are no more than $\min\{k,m\}$ non-zero terms in every sequence $(a_i)_{i=1}^k$ such that $\sum_{i=1}^k\lvert a_i\rvert\leq m$ and hence
\begin{align*}
\lvert\{(a_i)_{i=1}^k\in{\mathbb Z}^k\colon \sum_{i=1}^k\lvert a_i\rvert\leq m\}\rvert&\leq 2^{\min\{k,m\}}\lvert\{(p_i)_{i=1}^k\in({\mathbb Z}\cap[0,\infty))^k\colon \sum_{i=1}^k p_i\leq m\}\rvert\\
&=2^{\min\{k,m\}}\binom{m+k}{k}=2^{\min\{k,m\}}\binom{m+k}{m}.
\end{align*}
Therefore, for every positive integers $k$ and $m$ and real $\alpha\geq k$, by \eqref{eq:binom_estimate},
\begin{equation}\label{eq:l_1_cube_1}
\lvert\{(a_i)_{i=1}^k\in{\mathbb Z}^k\colon \sum_{i=1}^k\lvert a_i\rvert\leq m\}\rvert\leq 2^k\binom{m+k}{k}\leq \left(6\left(\frac{m}{k}+1\right)\right)^k\leq \left(6\left(\frac{m}{\alpha}+1\right)\right)^{\alpha},
\end{equation}
where the last inequality holds since the function $t\mapsto \left(\frac{m}{t}+1\right)^t$ is increasing in the interval $(0,\infty)$, and also
\begin{equation}\label{eq:l_1_cube_2}
\lvert\{(a_i)_{i=1}^k\in{\mathbb Z}^k\colon \sum_{i=1}^k\lvert a_i\rvert\leq m\}\rvert\leq 2^m\binom{m+k}{m}\leq \left(6\left(\frac{k}{m}+1\right)\right)^m\leq \left(6\left(\frac{\alpha}{m}+1\right)\right)^m.
\end{equation}

Let ${\mathbb B}_d$ denote the family of finite $A\subset{\mathbb Z}^d$ such that both $A$ and ${\mathbb Z}^d\setminus A$ are connected.
For every shift $\tau$ and $j_0(\tau)\neq j\in[J(\tau)]$, the set $U_j(\tau)$ is obviously in ${\mathbb B}_d$.
By \cite{BB}*{Theorem 6} (improving on \cite{LM}*{Corollary 1.2}; see more details in Appendix \ref{app:BB}), it holds that for every $v\in{\mathbb Z}^d$ and integer $b\geq 2d$, 
\begin{equation}\label{eq:BB}
\lvert \left\{A\in{\mathbb B}_d\colon v\in A,\,\lvert\partial A \rvert=b \right\} \rvert\leq (8d)^{2b/d}.
\end{equation}
Hence, for every $v_1,\ldots,v_J\in{\mathbb Z}^d$, $1\leq j_0\leq J$ and integers $(b_j)_{j_0\neq j\in[J]}\in([2d,\infty))^{J-1}$ such that $\sum_{j_0\neq j\in[J]}b_j\leq\lambda$,
\begin{multline*}
\lvert\{(A_j)_{j_0\neq j\in[J]}\colon \forall j_0\neq j\in[J] \text{ it holds that } A_j\in{\mathbb B}_d,\,v_j\in A_j,\,\lvert\partial A_j\rvert=b_j\rvert\leq \prod_{j_0\neq j\in[J]}(8d)^{2b_j/d}\\
=(8d)^{2\sum_{j_0\neq j\in[J]}b_j/d}\leq(8d)^{2\lambda/d}. 
\end{multline*}
For every shift $\tau$ it holds that $\sum_{j_0(\tau)\neq j\in[J(\tau)]} \lvert\partial U_j(\tau)\rvert\leq {\rm TV}(\tau)$ and by \eqref{eq:bdry_small}, $\lvert\partial U_j(\tau)\rvert\geq 2d$ for every $j_0(\tau)\neq j\in[J(\tau)]$. 
Therefore, since by \eqref{eq:l_1_cube_1} (noting that $d(J-1)<\lambda/2$) and  \eqref{eq:l_1_cube_2} (noting that $d(J-1)\leq d\rho$),
\begin{multline*}
\lvert\{(v_j)_{j=1}^J\in({\mathbb Z}^d)^J\colon  v_1=0,\,\sum_{j=1}^{J-1}\lVert v_j-v_{j+1}\rVert_1\leq\rho\}\rvert=\lvert\{(y_j)_{j=1}^{J-1}\in({\mathbb Z}^d)^{J-1}\colon  \sum_{j=1}^{J-1}\lVert y_j\rVert_1\leq\rho\}\rvert\\
=\lvert\{(a_i)_{i=1}^{d(J-1)}\in{\mathbb Z}^{d(J-1)}\colon \sum_{i=1}^{d(J-1)}\lvert a_i\rvert\leq\rho\}\rvert\leq \min\left\{\left(6\left(\frac{2\rho}{\lambda}+1\right)\right)^{\lambda/2},\left(6(d+1)\right)^{\rho}\right\}\\
\leq \min\left\{\left(6\left(\frac{2\rho}{\lambda}+1\right)\right)^{\lambda/2},\left(8d\right)^{\rho}\right\}
\end{multline*}
and for every $1\leq j_0\leq J$, by \eqref{eq:binom_estimate},
\begin{multline*}
\lvert\{(b_j)_{j_0\neq j\in[J]}\in({\mathbb Z}\cap[2d,\infty))^{J-1}\colon \sum_{j_0\neq j\in[J]}b_j\leq\lambda\}\rvert=\binom{\lambda-2d(J-1)+J-1}{J-1}\\
<\left(\frac{3(\lambda-2d(J-1)+J-1)}{J-1}\right)^{J-1}<\left(\frac{3\lambda}{J-1}\right)^{J-1}\leq (6d)^{\frac{\lambda}{2d}}
\end{multline*}
(where the last inequality holds since $J-1<\frac{\lambda}{2d}$ and the function $t\mapsto (e\lambda/t)^t$ is increasing in the interval $(0,\lambda]$), we conclude that
\begin{multline}\label{eq:first_chi}
\lvert\{(j_0(\tau),
(U_j(\tau))_{j_0(\tau)\neq j\in[J(\tau)]})\colon \tau\in\tilde{\mathcal S}_J\}\rvert\\
\leq \frac{\lambda}{2d}\min\left\{\left(6\left(\frac{2\rho}{\lambda}+1\right)\right)^{\lambda/2},\left(8d\right)^{\rho}\right\}(6d)^{\frac{\lambda}{2d}}(8d)^{\frac{2\lambda}{d}}
\end{multline}
Additionally, for every $\tau$, clearly $\sum_{e\in E({\mathcal T}(\tau))}\lvert s_e(\tau)\rvert\leq{\rm TV}(\tau)$. Hence, by using \eqref{eq:l_1_cube_1}, since $J-1<J\leq\frac{\lambda}{2d}$,
$$
\lvert\{(s_e(\tau))_{e\in E({\mathcal T}(\tau))}\colon \tau\in\tilde{\mathcal S}_J\}\rvert\leq \left(6(2d+1)\right)^{\frac{\lambda}{2d}}\leq\left(14d\right)^{\frac{\lambda}{2d}}.
$$
Combining this with \eqref{eq:first_chi} we get that for every $J\leq\min\{\frac{\lambda}{2d},\rho+1\}$,
\begin{equation*}
\lvert\chi(\tilde{\mathcal S}_J)\rvert\leq \min\left\{C_1^{\lambda}\left(\frac{2\rho}{\lambda}+1\right)^{\lambda/2}, \left(C_2 d^3\right)^{\lambda/d}\left(8d\right)^{\rho}\right\}
\end{equation*}
for some universal positive constants $C_1,C_2$, and the result follows by \eqref{eq:by_injectivity}.
\end{proof}

\subsection{Concentration of ground energy differences} 
\label{subsec:concentration}

In this section we prove Lemma~\ref{lem:concentration for energetic gain conditioned on layering bound} in the following equivalent formulation.

\smallskip
\noindent
{\it There exist universal $C,c>0$ such that the following holds. Suppose the disorder distributions $\nu^{\parallel},\nu^{\perp}$ satisfy \eqref{eq:disorder distributions assumptions}.
Then for any shift $\tau$ and any non-negative $b^{\parallel},b^{\perp}$ for which $\rho^{\Dob}\in\Omega^{\Lambda,\supp(\tau),(b^{\parallel},b^{\perp})}$,}
\begin{equation}
\label{eq:6.2_equiv}
\P\left(\lvert G^{\eta, \Lambda,(b^{\parallel},b^{\perp})} (\tau) \rvert \ge t\right)\le
C\exp \left(-c \frac{t^2}{\wid (\nu^{\parallel})^{2}  b^{\parallel}+\wid (\nu^{\perp})^{2} b^{\perp}} \right).
\end{equation} 
(This is a special case of Lemma \ref{lem:concentration for energetic gain conditioned on layering bound}, for $\tau'\equiv 0$, and it implies the lemma in full generality, since $G^{\eta, \Lambda,(b^{\parallel},b^{\perp})} (\tau,\tau')=G^{\eta^{\tau'}, \Lambda,(b^{\parallel},b^{\perp})} (\tau-\tau')$ for any two shifts $\tau,\tau'$.)

We aim to use Corollary~\ref{cor:concentration with bounded width} to show concentration of the ground energy difference under changes in the disorder $\eta$ induced by shifts. To be able to do so would require approximating the ground energy difference by a function of the disorder on finitely many edges. 

To be exact, for $\Lambda\subset\Z^d$ finite we write $\Delta_M:=\Lambda \times \{-M,\dots, M\}$, and define 
\begin{equation*}
\Omega^{\Delta_M, A, (b^{\parallel},b^{\perp})}:=\Omega^{\Delta_M,\rho^{\Dob}}\cap 
\Omega^{\Lambda, A, (b^{\parallel},b^{\perp})}    
\end{equation*}
to be the space of configurations on $\Delta_M$ satisfying the Dobrushin boundary conditions and layering bounds $(b^\parallel,b^{\perp})$ in $A\subseteq\Lambda$.
Let $\eta:E(\Z^{d+1}) \rightarrow \left[0,\infty \right)$. Denote by 
\begin{equation*}
\GE^{\Delta_M,A, (b^{\parallel},b^{\perp})}(\eta):=
\min \left\{ \mathcal{H}^{\eta,\Lambda}(\sigma): \sigma \in \Omega^{\Delta_M, A, (b^{\parallel},b^{\perp})} \right\}
\end{equation*}
and by $G^{\Delta_M, (b^{\parallel},b^{\perp})}(\tau)=
\GE^{\Delta_M,\supp(\tau), (b^{\parallel},b^{\perp})}(\eta)-\GE^{\Delta_M,\supp(\tau), (b^{\parallel},b^{\perp})}(\eta^{\tau})
$.
Then under the above definitions, the following holds: 
\begin{lemma}\label{lem: bounded_layering_finite_to_semi_infinite_convergence}
Let $\Lambda\subset \Z^d$ be finite, $\tau$ a shift and non-negative $b^{\parallel},b^{\perp}$ for which $\rho^{\Dob}\in\Omega^{\Lambda,\supp(\tau),(b^{\parallel},b^{\perp})}$.
Then
\begin{equation*}
\lim_{M\rightarrow \infty}\P\left(G^{\eta,\Delta_M, (b^{\parallel},b^{\perp})}(\tau) =G^{\eta,\Lambda,(b^{\parallel},b^{\perp})}(\tau) \right)=1.
\end{equation*}    
\end{lemma}
The proof of Lemma \ref{lem: bounded_layering_finite_to_semi_infinite_convergence} is routine, and appears in appendix \ref{sec:small lemmas} for completeness.

By the lemma above proving \eqref{eq:6.2_equiv} may be reduced to the following.

\begin{proposition}\label{prop:two point estimate finite volume}
There exist universal  $C,c>0$ such that the following holds. Suppose the disorder distributions $\nu^{\parallel},\nu^{\perp}$ satisfy \eqref{eq:disorder distributions assumptions}.
Then for any shift $\tau$ and non-negative $b^{\parallel},b^{\perp}$ for which $\rho^{\Dob}\in\Omega^{\Lambda,\supp(\tau),(b^{\parallel},b^{\perp})}$ and every positive integer $M$,
\begin{equation}
\label{eq:Talagrand_Delta_M}
\P \left(\left\lvert G^{\eta,\Delta_M,(b^{\parallel},b^{\perp})} (\tau)\right\rvert\geq t \right) \leq C\exp \left(-c\frac{t^{2}}{\wid(\nu^{\parallel})^2 b^{\parallel}+
\wid(\nu^{\perp})^2 b^{\perp}} \right).
\end{equation}
\end{proposition}

The rest of this section will be devoted to proving Proposition \ref{prop:two point estimate finite volume}.
Note that condition \eqref{eq:disorder distributions assumptions} implies that $\wid(\nu^{\parallel})\neq 0$. Also notice that if $\wid(\nu^{\perp})=0$ then $\nu^{\perp}$ is supported on one point, so the value of the coupling field on perpendicular plaquettes is fixed rather than random. This simplifies the argument and hence we will assume that $\wid(\nu^{\perp})\neq 0$.
Let
\begin{align*}
\mathfrak{H}&:=\{\eta_e\colon e\in E(\Z^{d+1}),\, e \nsubseteq  \supp(\tau) \times{\mathbb Z}\},\\
A_{M}&:=\{e\in E(\Z^{d+1}) \colon e \subseteq  \supp(\tau) \times \{-M,\dots,M \} \},
\end{align*}
and for every $e\in A_M$, let
\begin{equation*}
X_e:=\begin{cases}
\frac{\eta_e}{\wid(\nu^{\parallel})} & e\in E^{\parallel}(\Z^{d+1}), \\
\frac{\eta_e}{\wid(\nu^{\perp})} & e\in E^{\perp}(\Z^{d+1}).\end{cases}
\end{equation*}
Conditioned on $\mathfrak{H}$, the ground energy $\GE^{\Delta_M,\supp(\tau),(b^{\parallel},b^{\perp})}(\eta)$ may be viewed as a function of $\{ X_e \}_{e\in A_{M}}$.
Moreover, it is easy to verify that it is quasi-concave. 
Therefore, the following lemma will allow us to apply Corollary~\ref{cor:concentration with bounded width} to it.

\begin{lemma}\label{lem:GE_Lip}
The ground energy $\GE^{\Delta_M,\supp(\tau),(b^{\parallel},b^{\perp})}(\eta)$, conditioned on $\mathfrak{H}$, 
is Lipschitz, as a function of $\{ X_e \}_{e\in A_{M}}$, with Lipschitz constant $2\sqrt{\wid (\nu^{\parallel})^{2}  b^{\parallel}+\wid (\nu^{\perp})^{2} b^{\perp}}$.
\end{lemma}

Before proving Lemma \ref{lem:GE_Lip}, we show how it implies Proposition \ref{prop:two point estimate finite volume}. 
By Lemma \ref{lem:GE_Lip} and Corollary~\ref{cor:concentration with bounded width}, 
$\E\lvert(\GE^{\Delta_M,\supp(\tau),(b^{\parallel},b^{\perp})}(\eta) \mid \mathfrak{H})\rvert<\infty$ and for each $t>0$,
\begin{multline}
\P\left(\left\lvert\left(\GE^{\Delta_M,\supp(\tau),(b^{\parallel},b^{\perp})}(\eta) \mid \mathfrak{H}\right)-\E \left(\GE^{\Delta_M,\supp(\tau),(b^{\parallel},b^{\perp})}(\eta) \mid \mathfrak{H}\right)\right\rvert \geq t \right)\\
\le C\exp\left(-c\frac{t^2}{\wid (\nu^{\parallel})^{2}  b^{\parallel}+\wid (\nu^{\perp})^{2} b^{\perp}}\right).
\label{eq:Talagrand1}\end{multline}
Similarly,
$\E\lvert(\GE^{\Delta_M,\supp(\tau),(b^{\parallel},b^{\perp})}(\eta^{\tau}) \mid \mathfrak{H})\rvert<\infty$ and for each $t>0$,
\begin{multline}
\P\left(\left\lvert\left(\GE^{\Delta_M,\supp(\tau),(b^{\parallel},b^{\perp})}(\eta^{\tau}) \mid \mathfrak{H}\right)-\E \left(\GE^{\Delta_M,\supp(\tau),(b^{\parallel},b^{\perp})}(\eta^{\tau}) \mid \mathfrak{H}\right)\right\rvert \geq t \right)\\
\le C\exp\left(-c\frac{t^2}{\wid (\nu^{\parallel})^{2}  b^{\parallel}+\wid (\nu^{\perp})^{2} b^{\perp}}\right).
\label{eq:Talagrand2}\end{multline}

Observe that the following holds, by linearity of expectation and the facts that the disorder is independent and $\eta^\tau$ has the same distribution as $\eta$.

\begin{obs}\label{obs:connditional expectation on ground energy is zero}
It holds that
$$
\E\left(\GE^{\Delta_M,\supp(\tau),(b^{\parallel},b^{\perp})}(\eta) \mid \mathfrak{H}\right)=\E\left(\GE^{\Delta_M,\supp(\tau),(b^{\parallel},b^{\perp})}(\eta^{\tau}) \mid \mathfrak{H}\right).
$$
\end{obs}

Observation \ref{obs:connditional expectation on ground energy is zero} implies that for every $t>0$,
\begin{multline*}
\P\left(\left\lvert\left(G^{\eta,\Delta_M,(b^{\parallel},b^{\perp})} \mid \, \mathfrak{H}\right)\right\rvert\geq t\right)\\
\leq \P\left(\left\lvert\left(\GE^{\Delta_M,\supp(\tau),(b^{\parallel},b^{\perp})}(\eta)\mid\mathfrak{H}\right)-\E \left(\GE^{\Delta_M,\supp(\tau),(b^{\parallel},b^{\perp})}(\eta) \mid \mathfrak{H}\right)\right\rvert \geq t/2 \right) \\
+\P\left(\left\lvert\left(\GE^{\Delta_M,\supp(\tau),(b^{\parallel},b^{\perp})}(\eta^{\tau})\mid\mathfrak{H}\right)-\E \left(\GE^{\Delta_M,\supp(\tau),(b^{\parallel},b^{\perp})}(\eta^{\tau}) \mid \mathfrak{H}\right)\right\rvert \geq t/2 \right) 
\end{multline*}
and \eqref{eq:Talagrand_Delta_M} follows by \eqref{eq:Talagrand1} and \eqref{eq:Talagrand2}.

\begin{proof}[Proof of Lemma \ref{lem:GE_Lip}]
For $y\in [0,\infty)^{A_M}$ and
$\tilde{y} \in [0,\infty)^{E(\Z^{d+1})\setminus A_M}$ 
let $h(y,\tilde{y}): E(\Z^{d+1}) \rightarrow \left[0, \infty \right)$ be defined by
\begin{equation*}
\left(h(y,\tilde{y})\right)_{e}=\begin{cases}
\wid(\nu^{\parallel}) y_e & e\in A_M\cap E^{\parallel}(\Z^{d+1}), \\
\wid(\nu^{\perp}) y_e & e\in A_M\cap E^{\perp}(\Z^{d+1}),\\
\tilde{y} & e\in E({\mathbb Z}^{d+1})\setminus A_M.
\end{cases}
\end{equation*}        
We need to verify that for any $\tilde{y} \in [0,\infty)^{E(\Z^{d}) \setminus A_M}$ and $y,y'\in [0,\infty)^{A_M}$ it holds that 
\begin{equation*}
\lvert \GE^{\Delta_M,\supp(\tau),(b^{\parallel},b^{\perp})} (h)-\GE^{\Delta_M,\supp(\tau),(b^{\parallel},b^{\perp})} (h')\rvert    
\leq 2\sqrt{\wid (\nu^{\parallel})^{2}  b^{\parallel}+\wid (\nu^{\perp})^{2} b^{\perp}}\,\lVert y-y' \rVert_2,
\end{equation*}
where $h:=h(y,\tilde{y})$ and $h':=h(y',\tilde{y})$.
Let $\sigma'$ be (some) ground configuration in $\Omega^{\Lambda, \supp(\tau),(b^{\parallel},b^{\perp})}$
with respect to the coupling field $h'$. 
Then, 
\begin{align*}   
\mathcal{H}^{h,\Lambda} (\sigma')-\mathcal{H}^{h',\Lambda} (\sigma') \leq &
\sum_{\{x,y\}\in A_M} |h_{\{x,y\}}-h'_{\{x,y\}}| \left(1- \sigma'_x \sigma'_y  \right)\\
=&\wid(\nu^{\parallel}) \sum_{\{x,y\}\in A_M\cap E^{\parallel}(\Z^{d+1})}  |y_{\{x,y\}}-y'_{\{x,y\}}| \left(1- \sigma'_x \sigma'_y  \right)+\\
&\wid(\nu^{\perp}) \sum_{\{x,y\}\in A_M\cap E^{\perp}(\Z^{d+1})}  |y_{\{x,y\}}-y'_{\{x,y\}}| \left(1- \sigma'_x \sigma'_y  \right)\\
\leq &2\sqrt{\wid (\nu^{\parallel})^{2}  b^{\parallel}+\wid (\nu^{\perp})^{2} b^{\perp}}\,\lVert y- y' \rVert_{2},
\end{align*}
where the last inequality is by the Cauchy-Schwarz inequality (and the layering bound on $\sigma'$ deriving from the fact $\sigma' \in \Omega^{\Delta_M,\supp(\tau),(b^{\parallel},b^{\perp})}$).
Since $\GE^{\Delta_M,\supp(\tau),(b^{\parallel},b^{\perp})}(h)\leq\mathcal{H}^{h,\Lambda} (\sigma')$ and $\GE^{\Delta_M,\supp(\tau),(b^{\parallel},b^{\perp})}(h')=\mathcal{H}^{h',\Lambda} (\sigma')$, it follows that  
$$
\GE^{\Delta_M,\supp(\tau),(b^{\parallel},b^{\perp})}(h)-\GE^{\Delta_M,\supp(\tau),(b^{\parallel},b^{\perp})}(h') \leq 
2\sqrt{\wid (\nu^{\parallel})^{2}  b^{\parallel}+\wid (\nu^{\perp})^{2} b^{\perp}}\,\lVert y- y'\rVert_{2}.
$$
A symmetric inequality of the form
$$
\GE^{\Delta_M,\supp(\tau),(b^{\parallel},b^{\perp})}(h')-\GE^{\Delta_M,\supp(\tau),(b^{\parallel},b^{\perp})}(h) \leq 
2\sqrt{\wid (\nu^{\parallel})^{2}  b^{\parallel}+\wid (\nu^{\perp})^{2} b^{\perp}}\,\lVert y- y'\rVert_{2}.
$$
holds with identical reasoning, and so we are done.
\end{proof}

\subsection{Layering bounds}
\label{sec:layering}
In this section we prove Lemma~\ref{lem:layering bounds for given maximal gain}.
Fix positive $\alpha^{\parallel},\alpha^{\perp}$ and a finite set $\Lambda\subset \Z^{d}$.
Recall the definitions of parallel and perpendicular layering in \eqref{eq:parallel_layering} and \eqref{eq:perpendicular_layering}.
Introduce the following convenient notation: For $E\subset\Lambda$, $\eta \in \mathcal{D}(\alpha^{\parallel},\alpha^{\perp})$, $\theta \in \{ \parallel,\perp \}$ write $\mathcal{L}_{E}^{\theta}(\eta):=\mathcal{L}_{E}^{\theta}(\sigma^{\eta,\Lambda,\Dob})$ where $\sigma^{\eta,\Lambda,\Dob}$ is the unique ground configuration, existing by Lemma~\ref{lem:semi infinite volume ground configuration}.
We will use the following proposition, which will be proved in Section \ref{sec:obtaining admissible shifts}.

\begin{proposition} 
\label{prop:existance_Dob_shift_ground_config}
Let $E\subset \Lambda$ and let $\eta\in D(\alpha^{\parallel},\alpha^{\perp})$ be a coupling field.
Then, there exists a shift function ${\tau}$ satisfying the following:
\begin{align*}
G^{\eta,\Lambda}({\tau})
&\geq    
2\alpha^{\parallel} \left( \mathcal{L}^{\parallel}_E(\eta)-|E|\right)+2\alpha^{\perp}\max\{\mathcal{L}^{\perp}_{E}(\eta),\TV({\tau})\},
\\  
G^{\eta,\Lambda}({\tau})
&\geq \frac{1}{16}\min\{\alpha^{\parallel},\alpha^{\perp}\}d\left(R({\tau})-R(E)-2(|E|-1)\right).
\end{align*}
\end{proposition}

Fix any coupling field $\eta\in D(\alpha^{\parallel},\alpha^{\perp})$.
For brevity, we will once more write 
$G$ for $G^{\eta,\Lambda}$, $\AS$ for $\AS^{\eta,\Lambda}(\alpha^{\parallel},\alpha^{\perp})$, \emph{admissible} for $(\alpha^{\parallel},\alpha^{\perp})$-admissible and $\MG$ for $\MG^{\eta,\Lambda}(\alpha^{\parallel},\alpha^{\perp})$.

The following proposition establishes a useful general layering bound.

\begin{proposition}\label{prop:general_shift_functions_layering}
Let $E\subset \Lambda$, and $\tau$ be a shift. The following layering bound holds:
\begin{multline}
\label{eq:layering_bound_general}
\alpha^{\perp}\mathcal{L}^{\perp}_{E}(\eta^{\tau})+ \alpha^{\parallel} \left( \mathcal{L}^{\parallel}_{E}(\eta^{\tau}) -|E| \right) \\
\leq \max\left\{  \MG,\, 2\alpha^{\perp}\TV\left(\tau\right)+2\min \{ \alpha^{\parallel}, \alpha^{\perp} \}d\left(R(\tau)+R\left(E\right)+\left|E\right|\right)\right\}. 
\end{multline}
\end{proposition}
\begin{proof}
By Proposition \ref{prop:existance_Dob_shift_ground_config}, applied to the set $E$ and the shifted disorder $\eta^{\tau}$, there is a shift $\tau'$ such that
\begin{align*}
G^{\eta^{\tau},\Lambda}(\tau')
&\geq    
2\alpha^{\parallel} \left( \mathcal{L}^{\parallel}_E(\eta^{\tau})-|E|\right)+2\alpha^{\perp}\max\{\mathcal{L}^{\perp}_{E}(\eta^{\tau}),\TV(\tau')\},\\  
G^{\eta^{\tau},\Lambda}(\tau')
&\geq \frac{1}{16}\min\{\alpha^{\parallel},\alpha^{\perp}\}d\left(R(\tau')-R(E)-2(|E|-1)\right).
\end{align*}
Note that 
\begin{align*}
G^{\eta^{\tau},\Lambda}(\tau')&= \GE^{\Lambda}(\eta^{\tau}) - \GE^{\Lambda}\left((\eta^{\tau})^{\tau'}\right)=\GE^{\Lambda}(\eta^{\tau})- \GE^{\Lambda}(\eta^{\tau+\tau'})\\
&= \left(\GE^{\Lambda}(\eta)- \GE^{\Lambda}(\eta^{\tau+\tau'})\right)-\left(\GE^{\Lambda}(\eta) - \GE^{\Lambda}(\eta^{\tau})\right)=
G(\tau+\tau')-G(\tau)\\
&\leq 2\max\{G(\tau+\tau'),\,\lvert G(\tau)\rvert\}.
\end{align*}
Hence,
\begin{align}
\max\{G(\tau+\tau'),\,\lvert G(\tau)\rvert\}&
\geq \alpha^{\parallel} \left( \mathcal{L}^{\parallel}_E(\eta^{\tau})-|E|\right)+\alpha^{\perp}\max\{\mathcal{L}^{\perp}_{E}(\eta^{\tau}),\TV(\tau')\},\label{eq:Dob_shift_ground_config_TV}\\  
\max\{G(\tau+\tau'),\,\lvert G(\tau)\rvert\}&\geq \frac{1}{32}\min\{\alpha^{\parallel},\alpha^{\perp}\}d\left(R(\tau')-R(E)-2(|E|-1)\right).
\label{eq:Dob_shift_ground_config_R}\end{align}
By way of contradiction, assume that \eqref{eq:layering_bound_general} does not hold. Then, by \eqref{eq:Dob_shift_ground_config_TV},
\begin{multline}
\label{eq:Dob_shift_ground_config_TV_0}
\max\{G(\tau+\tau'),\,\lvert G(\tau)\rvert\} \\
>\max\left\{  \MG,\, 2\alpha^{\perp}\TV\left(\tau\right)+2\min \{ \alpha^{\parallel}, \alpha^{\perp} \}d\left(R(\tau)+R\left(E\right)+\left|E\right|\right)\right\}.
\end{multline}

If $\lvert G(\tau)\rvert\geq G(\tau+\tau')$ then by \eqref{eq:Dob_shift_ground_config_TV_0}, $|G(\tau)| > 2\alpha^{\perp}\TV\left(\tau\right)$, 
$|G(\tau)|>2\min\{\alpha^{\parallel},\alpha^{\perp}\} d R\left(\tau\right)$ and $|G(\tau)|>\MG$, 
hence $\tau$ is admissible and of larger energetic gap than $\MG$, a contradiction.

Assume that $G(\tau+\tau')>\lvert G(\tau)\rvert$.
By Lemma \ref{lem:upper_bound_R_of_sum}, \eqref{eq:Dob_shift_ground_config_R} and \eqref{eq:Dob_shift_ground_config_TV},
\begin{align*}
R(\tau+\tau') &\leq 2R(\tau) +R(\tau')+\frac{98}{d}\left(\TV(\tau)+\TV(\tau')\right) \\
&\leq 2R(\tau)+R(E)+2|E|+\frac{32\, G(\tau+\tau')}{\min \left\{\alpha^{\parallel}, \alpha^{\perp}\right\} d} +\frac{98}{d}\TV(\tau)+\frac{98\, G(\tau+\tau')}{\alpha^{\perp} d}\\
&\leq 100\left(\frac{\TV(\tau)}{d}+R(\tau)+R(E)+|E|\right)+\frac{130\, G(\tau+\tau')}{\min \left\{\alpha^{\parallel}, \alpha^{\perp}\right\} d}
\end{align*}
and hence, by \eqref{eq:Dob_shift_ground_config_TV_0}, $R(\tau+\tau')<\frac{180}{\min \left\{\alpha^{\parallel}, \alpha^{\perp}\right\} d} G(\tau+\tau')$, and by Observation \ref{obs:upper_bound_TV_of_sum}, 
\eqref{eq:Dob_shift_ground_config_TV_0}
and \eqref{eq:Dob_shift_ground_config_TV},
\begin{equation*}\label{eq:TV_G}
\TV(\tau+\tau')\leq\TV(\tau) +\TV(\tau') < \frac{G(\tau+\tau')}{2\alpha^{\perp}} 
+\frac{G(\tau+\tau')}{\alpha^{\perp}}< \frac{2}{\alpha^{\perp}} G(\tau+\tau').
\end{equation*}
Therefore, $\tau+\tau'$ is admissible and of larger energetic gap than $\MG$, by \eqref{eq:Dob_shift_ground_config_TV_0}, a contradiction.
\end{proof}

Now we use Proposition \ref{prop:general_shift_functions_layering} to prove Lemma \ref{lem:layering bounds for given maximal gain}.

\begin{proof} [Proof of Lemma \ref{lem:layering bounds for given maximal gain}]
First note that for every set $E\subseteq\Lambda$ and every coupling field $\tilde{\eta}\in \mathcal{D}(\alpha^{\parallel},\alpha^{\perp})$, 
it holds that $\GE^{\Lambda}(\tilde{\eta})=\GE^{\Lambda,E,(b^{\parallel},b^{\perp})}(\tilde{\eta})$ if and only if 
$\mathcal{L}^{\perp}_{E}(\tilde{\eta})  \leq b^\perp$ and $\mathcal{L}^{\parallel}_{E}(\tilde{\eta})  \leq b^\parallel$.
Also note that $\eta^{\tau}\in\mathcal{D}(\alpha^{\parallel},\alpha^{\perp})$ for every shift $\tau$.
Throughout the proof we will use $C$ to denote a positive absolute constant;
the values of this constant will be allowed to change from line to line, 
even within the same calculation, with its value increasing.

To prove the third part of the lemma, we use Proposition \ref{prop:general_shift_functions_layering} twice, for the same subset $E:=\supp(\tau)$. 
Recall that, by the admissibility of $\tau$,
\begin{align}
\TV(\tau)&\leq \frac{2}{\alpha^{\perp}}\lvert G(\tau)\rvert\leq \frac{2}{\alpha^{\perp}}\MG\leq \frac{4s}{\alpha^{\perp}},\label{eq:TVs}\\
R(\tau)&\leq \frac{200}{\min\{\alpha^{\parallel},\alpha^{\perp}\}d}\lvert G(\tau)\rvert\leq \frac{200}{\min\{\alpha^{\parallel},\alpha^{\perp}\}d}\MG\leq \frac{400 s}{\min\{\alpha^{\parallel},\alpha^{\perp}\}d}.\label{eq:Rs}
\end{align}
Hence, $R(E) \leq R(\tau)\leq \frac{C s}{\min\{\alpha^{\parallel},\alpha^{\perp}\}d}$ and by Lemma \ref{lem:functional_isoperimetry} and the assumption that $s < \alpha^{\perp} 4^d$, 
\begin{equation}\label{eq:Es}
|E|\leq \sum_{u\in{\mathbb Z}^d}\lvert\tau(u)\rvert\leq\left( \frac{\TV(\tau)}{2d}
\right)^{\frac{d}{d-1}}
\leq \left( \frac{2s}{ \alpha^{\perp} d}
\right)^{\frac{d}{d-1}}
\leq \frac{Cs}{ \alpha^{\perp}d}.
\end{equation}
The first use of Proposition \ref{prop:general_shift_functions_layering} will be for the shift $\tau$, and it gives 
\begin{multline*}
\alpha^{\perp} \mathcal{L}^{\perp}_{E}(\eta^{\tau})+\alpha^{\parallel} \left( \mathcal{L}^{\parallel}_{E}(\eta^{\tau}) -|E| \right)\\
\leq \max\left\{ \MG, 2\alpha^{\perp}\TV\left(\tau\right)+2\min \{ \alpha^{\parallel}, \alpha^{\perp} \}d\left(R(\tau)+ R(E)+|E| \right)\right\}\leq C s
\end{multline*}
and the second use will be for the shift $\tau_0\equiv 0$, and it gives 
\begin{equation*}
\alpha^{\perp} \mathcal{L}^{\perp}_{E}(\eta)+\alpha^{\parallel} \left( \mathcal{L}^{\parallel}_{E}(\eta) -|E| \right)\leq \max\left\{ \MG, 2\min \{ \alpha^{\parallel}, \alpha^{\perp} \}d\left( R(E)+|E| \right) \right\}\leq C s.
\end{equation*}
Hence, for $\# \in \{\eta^{\tau},\eta\}$, by using \eqref{eq:Es}, 
\begin{equation*}
\mathcal{L}^{\perp}_{E}(\#) \leq  \frac{C s}{\alpha^{\perp}}  ,
\quad \mathcal{L}^{\parallel}_{E} (\#) \leq C\left(\frac{1}{\alpha^{\parallel}}+\frac{1}{\alpha^{\perp}d}\right)s.
\end{equation*}

We proceed to prove the second part of the lemma. 
For every positive integer $k$, let $E_k:=\supp \left(\tau_{2^k}-\tau_{2^{k+1}} \right)$.
By Proposition~\ref{prop:functional_bdry},
\begin{equation}\label{eq:TVks}
\TV(\tau_{2^k})\leq 10d\, \TV(\tau)\leq \frac{C d s}{\alpha^{\perp}},\end{equation}
and by \eqref{eq:upper_bound_R_of_coarse},
\begin{equation}\label{eq:Rks}
R(\tau_{2^k})\leq R(\tau)+\frac{C}{d}\TV(\tau)\leq \frac{Cs}{\min\{\alpha^{\parallel},\alpha^{\perp}\}d}.
\end{equation}
Hence, by Lemma \ref{lem:upper_bound_R_of_sum}, 
$$
R\left(E_k\right)\leq R\left(\tau_{2^k}-\tau_{2^{k+1}}\right)\leq R\left(\tau_{2^k}\right)+2R\left(\tau_{2^{k+1}}\right)+\frac{98}{d}\left(\TV\left(\tau_{2^k}\right)+\TV\left(\tau_{2^{k+1}}\right)\right)\leq \frac{C s}{\min\{\alpha^{\parallel},\alpha^{\perp}\}},
$$
and by Proposition \ref{prop:functional_diff},
\begin{equation}\label{eq:Eks}
\lvert E_k\rvert\leq\lVert \tau_{2^k}-\tau_{2^{k+1}}\rVert_1
\leq (4d+9)2^k\, \TV(\tau)\leq\frac{C 2^k d s}{\alpha^{\perp}}. 
\end{equation}
Therefore, by Proposition \ref{prop:general_shift_functions_layering},
\begin{multline*}
\alpha^{\perp}\mathcal{L}^{\perp}_{E_k}(\eta^{\tau_{2^k}})+ \alpha^{\parallel} \left( \mathcal{L}^{\parallel}_{E_k}(\eta^{\tau_{2^k}}) -|E_k| \right) \\
\leq \max\left\{  \MG,\, 2\alpha^{\perp}\TV\left(\tau_{2^k}\right)+2\min \{ \alpha^{\parallel}, \alpha^{\perp} \}d\left(R(\tau_{2^k})+R\left(E_k\right)+\lvert E_k\rvert\right)\right\}\leq C 2^k d^2 s
\end{multline*}
and
\begin{multline*}
\alpha^{\perp}\mathcal{L}^{\perp}_{E_{k-1}}(\eta^{\tau_{2^k}})+ \alpha^{\parallel} \left( \mathcal{L}^{\parallel}_{E_{k-1}}(\eta^{\tau_{2^k}}) -|E_{k-1}| \right) \\
\leq \max\left\{  \MG,\, 2\alpha^{\perp}\TV\left(\tau_{2^k}\right)+2\min \{ \alpha^{\parallel}, \alpha^{\perp} \}d\left(R(\tau_{2^k})+R\left(E_{k-1}\right)+\lvert E_{k-1}\rvert\right)\right\}\leq C 2^k d^2 s.
\end{multline*}
Hence, for $\# \in \{k-1,k\}$, by using \eqref{eq:Eks},
\begin{equation*}
\mathcal{L}^{\perp}_{E_{\#}}(\tau^{2^k}) \leq  \frac{C2^k d^2 s }{\alpha^{\perp}},\quad \mathcal{L}^{\parallel}_{E_{\#}}(\tau_{2^k}) \leq \frac{C 2^k d^2 s}{\alpha^{\parallel}}+\lvert E_{\#}\rvert\leq C \left( \frac{1}{\alpha^{\parallel}} +\frac{1 }{\alpha^{\perp}d}\right) d^2 2^{k} s.    \end{equation*}
       
To prove the first part of the Lemma, fix $\emptyset\neq I\in {\rm comp}(\tau) $, and let $E_{(0,I)}:=\supp \left(\tau- \tau_{I}\right)$. 
By the compatibility of $I$ and \eqref{eq:TVs},
\begin{equation}\label{eq:TVIs}
\TV(\tau_I)\leq 20(2\lvert I\rvert+1)\TV(\tau)\leq\frac{C\lvert I\rvert s}{\alpha^{\perp}},
\end{equation}
and by \eqref{eq:upper_bound_R_of_fine}, \eqref{eq:TVs}
and \eqref{eq:Rs},\begin{equation}\label{eq:RIs}
R(\tau_I)\leq R(\tau)+\frac{C}{d}\TV(\tau)\leq \frac{C s}{\min\{\alpha^{\parallel},\alpha^{\perp}\}d},
\end{equation}
and hence, by Lemma \ref{lem:upper_bound_R_of_sum},  
\begin{equation}\label{eq:RE0Is}
R\left(E_{(0,I)}\right)\leq R\left(\tau-\tau_I\right)\leq R\left(\tau\right)+2R\left(\tau_I\right)+\frac{100}{d}\left(\TV\left(\tau\right)+\TV\left(\tau_I\right)\right)\leq \frac{C\lvert I\rvert s}{\min\{\alpha^{\parallel},\alpha^{\perp}\}d},
\end{equation}
and by the compatibility of $I$, 
\begin{equation}\label{eq:E0Is}
\lvert E_{(0,I)}\rvert\leq\lVert \tau-\tau_I\rVert_1\leq\frac{4\lvert I\rvert}{d}\TV(\tau)\leq \frac{C \lvert I\rvert s}{\alpha^{\perp}d}.
\end{equation}
Therefore, by Proposition \ref{prop:general_shift_functions_layering},
\begin{multline*}
\alpha^{\perp}\mathcal{L}^{\perp}_{E_{(0,I)}}(\eta^{\tau})+ \alpha^{\parallel} \left( \mathcal{L}^{\parallel}_{E_{(0,I)}}(\eta^{\tau}) -|E_{(0,I)}| \right) \\
\leq \max\left\{  \MG,\, 2\alpha^{\perp}\TV\left(\tau\right)+2\min \{ \alpha^{\parallel}, \alpha^{\perp} \}d\left(R(\tau)+R\left(E_{(0,I)}\right)+\lvert E_{(0,I)}\rvert\right)\right\}\leq C\lvert I\rvert s
\end{multline*}
and
\begin{multline*}
\alpha^{\perp}\mathcal{L}^{\perp}_{E_{(0,I)}}(\eta^{\tau_{I}})+ \alpha^{\parallel} \left( \mathcal{L}^{\parallel}_{E_{(0,I)}}(\eta^{\tau_{I}}) -\lvert E_{(0,I)}\rvert \right) \\
\leq \max\left\{  \MG,\, 2\alpha^{\perp}\TV\left(\tau_I\right)+2\min \{ \alpha^{\parallel}, \alpha^{\perp} \}d\left(R(\tau_I)+R\left(E_{(0,I)}\right)+\lvert E_{(0,I)}\rvert\right)\right\}\leq C\lvert I\rvert s.
\end{multline*}
Hence, for $\# \in \{\eta^{\tau},\eta^{\tau_{I}}\}$, by using \eqref{eq:E0Is},
\begin{equation*}
\mathcal{L}^{\perp}_{E_{(0,I)}}(\#) \leq  \frac{C\lvert I\rvert s}{\alpha^{\perp}}  ,\quad \mathcal{L}^{\parallel}_{E_{(0,I)}} (\#) \leq \frac{C\lvert I\rvert s}{\alpha^{\parallel}}+\lvert E_{(0,I)}\rvert\leq C\left( \frac{1}{\alpha^{\parallel}}+\frac{1}{\alpha^{\perp}d}  \right) |I|s.      
\end{equation*}
Finally, let $E_{(I,1 )}:=\supp \left(\tau_I-\tau_2\right)$ and note that by Lemma \ref{lem:upper_bound_R_of_sum}, \eqref{eq:RIs}, \eqref{eq:Rks}, \eqref{eq:TVIs} and \eqref{eq:TVks},
$$
R\left(E_{(I,1)}\right)\leq R\left(\tau_I-\tau_2\right)\leq R\left(\tau_I\right)+2R\left(\tau_2\right)+\frac{98}{d}\left(\TV\left(\tau_I\right)+\TV\left(\tau_2\right)\right)\leq \frac{C s}{\min\{\alpha^{\parallel},\alpha^{\perp}\}}, 
$$
and by the compatibility of $I$, Proposition \ref{prop:functional_diff} and \eqref{eq:TVs}, 
\begin{equation}\label{eq:EI1s}
\lvert E_{(I,1)}\rvert\leq\lVert \tau_I-\tau_2\rVert_1\leq \lVert \tau_I-\tau\rVert_1+\lVert \tau-\tau_2\rVert_1\leq \frac{4\lvert I\rvert}{d}\TV(\tau)+2\TV(\tau)\leq \frac{C s}{\alpha^{\perp}}.
\end{equation}
Therefore, by Proposition \ref{prop:general_shift_functions_layering},
\begin{multline*}
\alpha^{\perp}\mathcal{L}^{\perp}_{E_{(I,1)}}(\eta^{\tau_I})+ \alpha^{\parallel} \left( \mathcal{L}^{\parallel}_{E_{(I,1)}}(\eta^{\tau_I}) -|E_{(I,1)}| \right) \\
\leq \max\left\{  \MG,\, 2\alpha^{\perp}\TV\left(\tau_I\right)+2\min \{ \alpha^{\parallel}, \alpha^{\perp} \}d\left(R(\tau_I)+R\left(E_{(I,1)}\right)+\lvert E_{(I,1)}\rvert\right)\right\}\leq C d s
\end{multline*}
and
\begin{multline*}
\alpha^{\perp}\mathcal{L}^{\perp}_{E_{(I,1)}}(\eta^{\tau_2})+ \alpha^{\parallel} \left( \mathcal{L}^{\parallel}_{E_{(I,1)}}(\eta^{\tau_2}) -\lvert E_{(I,1)}\rvert \right) \\
\leq \max\left\{  \MG,\, 2\alpha^{\perp}\TV\left(\tau_2\right)+2\min \{ \alpha^{\parallel}, \alpha^{\perp} \}d\left(R(\tau_2)+R\left(E_{(I,1)}\right)+\lvert E_{(I,1)}\rvert\right)\right\}\leq C d s.
\end{multline*}
Hence, for $\# \in \{\eta^{\tau_I},\eta^{\tau_2}\}$, by using \eqref{eq:EI1s},
\begin{equation*}
\mathcal{L}^{\perp}_{E_{(I,1)}}(\#) \leq  \frac{C ds}{\alpha^{\perp}} ,\quad \mathcal{L}^{\parallel}_{E_{(I,1)}} (\#) \leq \frac{C d s}{\alpha^{\parallel}}+\lvert E_{(I,1)}\rvert\leq C\left( \frac{1}{\alpha^{\parallel}}+\frac{1}{\alpha^{\perp}d}  \right) d s.
\qedhere
\end{equation*}
\end{proof}

\section{Obtaining admissible shifts from interfaces}
\label{sec:obtaining admissible shifts}

The goal of this section will be to prove Proposition \ref{prop:existance_Dob_shift_ground_config}
and Lemma~\ref{lem:spin different than sign shift existence}. 

\medskip

Fix $\eta\in\mathcal{D}(\alpha^{\parallel},
\alpha^{\perp})$ for some  $\alpha^{\parallel},\alpha^{\perp}>0$
and let $\Lambda \subset \Z^d$ be finite.
Recall the definition of the configuration space for semi-infinite-volume under Dobrushin boundary conditions $\Omega^{\Lambda,\Dob}$ in \eqref{eq:Dobrushin space infinite volume}, the
Hamiltonian $\mathcal{H}^{\eta,\Lambda}$ in \eqref{eq:Hamiltonian}, and of the ground energy with respect to it $\GE^{\Lambda,\eta}$ in \eqref{eq:ground energy def}. Recall as well the definition of parallel and perpendicular layering in \eqref{eq:parallel_layering} and \eqref{eq:perpendicular_layering}.

\subsection{Defining ${\tau}_0$}
\label{sec:s_0}

Let $E\subseteq\Lambda$ and let $\sigma_0\colon {\mathbb Z}^{d+1} \to \{-1,1\}$ be a configuration in $\Omega^{\Lambda, \Dob}$.

Define a function $I_{\sigma_0}\colon \Z^d \to\Z \cup \{{\rm ``layered"}\}$ as follows. For each $v\in \Z^{d}$, if $\sigma_0$ has a unique sign change above $v$, then set $I_{\sigma_0}(v)$ to be the location of this sign change (precisely, if $\sigma_0(v,k)=-1$ and $\sigma_0(v,k+1)=1$ then set $I_{\sigma_0}(v) = k$). If $\sigma_0$ has more than one sign change above $v$ then set $I_{\sigma_0}(v) = {\rm ``layered"}$.

Define a graph $G_{\sigma_0}$ to be the induced subgraph of $\Z^d$ on the vertex set $V_{\sigma_0}$, where $V_{\sigma_0}\subset \Z^d$  is defined to be the set of vertices $v$ satisfying that there exists a neighbor $u\sim v$ (in the usual connectivity of $\Z^d$) such that either $I_{\sigma_0}(u) \neq I_{\sigma_0}(v)$ or $I_{\sigma_0}(u)=I_{\sigma_0}(v)={\rm ``layered"}$ (see Figure \ref{fig:pre_shift}).

Recall from \eqref{eq:viz} the definition of $\partial_{{\rm vis}(w)}(A)$, the outer vertex boundary of a set $A\subseteq{\mathbb Z}^d$ visible from a point $w\in{\mathbb Z}^d$. 

\begin{obs}\label{obs:svC}
For every connected component $A$ of $G_{\sigma_0}$ and every $v\in({\mathbb Z}^d\setminus A)\cup\{\infty\}$, there is an integer, which we denote $\tilde{I}_{\sigma_0}(v;A)$, such that $I_{\sigma_0}(u)=\tilde{I}_{\sigma_0}(v;A)$ for every $u\in\partial_{{\rm vis}(v)}(A)$.
\end{obs}

\begin{proof}
For any $u\in{\mathbb Z}^d\setminus V_{\sigma_0}$, it holds that $I_{\sigma_0}(w)=I_{\sigma_0}(u)$, for every $w\in{\mathcal B}_1(u)$.
Hence, for every $u_1,u_2 \in{\mathbb Z}^d\setminus V_{\sigma_0}$ such that $\lVert u_1-u_2\rVert_1\leq 2$, i.e., such that ${\mathcal B}_1(u_1)\cap {\mathcal B}_1(u_2)\neq\emptyset$, it holds that $I_{\sigma_0}(u_1)=I_{\sigma_0}(u_2)$.

The claim follows since $\partial_{{\rm vis}(v)}(A)\subseteq\partial^{\rm out}A\subseteq{\mathbb Z}^d\setminus V_{\sigma_0}$ and the set $\partial_{{\rm vis}(v)}(A)$ is $\ell_1^+$-connected, by Lemma \ref{lem:Timar3}. 
\end{proof}

For every $A\subseteq{\mathbb Z}^d$, let
$$
{\rm in}(A):=\left\{ u\in {\mathbb Z}^d\setminus A\colon  \text{every path from } u \text{ to }\infty \text{ intersects }A\right\}.
$$

\begin{lemma}\label{lem:unique smallest contour} 
Let $A_1, A_2\subseteq {\mathbb Z}^d$ be nonempty connected sets such that ${\rm dist}(A_1,A_2)>1$.
\begin{enumerate}
\item 
If ${\rm dist}(A_1\cup{\rm in}(A_1), A_2\cup{\rm in}(A_2))\leq 1$, then $(A_1\cup{\rm in}(A_1))\cap(A_2\cup{\rm in}(A_2))\neq\emptyset$.
\item
If $(A_1\cup{\rm in}(A_1))\cap(A_2\cup{\rm in}(A_2))\neq\emptyset$ then $A_1\subseteq{\rm in}(A_2)$ or $A_2\subseteq{\rm in}(A_1)$. 
\item
If $A_1\subseteq{\rm in}(A_2)$ then ${\rm in}(A_1)\subsetneq{\rm in}(A_2)$. (Similarly, if $A_2\subseteq{\rm in}(A_1)$ then ${\rm in}(A_2)\subsetneq{\rm in}(A_1)$.)
\end{enumerate}
\end{lemma}

\begin{proof}
We first show that the first statement holds. 
There are $u_1\in A_1\cup{\rm in}(A_1)$ and $u_2\in A_2\cup{\rm in}(A_2)$ such that $\lVert u_1-u_2\rVert_1\leq 1$.
If $u_1=u_2$ there is nothing to prove, hence we assume that $\lVert u_1-u_2\rVert_1=1$, i.e., $u_1\sim u_2$.
Since ${\rm dist}(A_1,A_2)>1$, necessarily $u_1\notin A_1$ or $u_2\notin A_2$.
With no loss of generality assume that $u_1\notin A_1$.
Hence, $u_1\in{\rm in}(A_1)$.
If $P$ is a path from $u_2$ to $\infty$, then starting at $u_1$ and continuing along $P$ is a path from $u_1\in{\rm in}(A_1)$ to $\infty$, therefore it must intersect $A_1$; hence, since $u_1\notin A_1$, the path $P$ necessarily intersects $A_1$.
Therefore, every path from $u_2$ to $\infty$ intersects $A_1$, i.e., $u_2\in A_1\cup{\rm in}(A_1)$.

To prove the second statement, assume by contradiction that there are $a_1\in A_1\setminus{\rm in}(A_2)$ and $a_2\in A_2\setminus{\rm in}(A_1)$.
Consider an arbitrary path $P_0$ from an arbitrary vertex $u_0\in(A_1\cup{\rm in}(A_1))\cap(A_2\cup{\rm in}(A_2))$ to $\infty$. 
The path $P_0$ necessarily intersects both $A_1$ and $A_2$. Let $a$ be the first intersection point of $P_0$ with $A_1
\cup A_2$. With no loss of generality we may assume that $a\in A_1$. 
Since $A_1$ is connected, there is a path $P_1$ in $A_1$ from $a$ to $a_1$. Since $a_1\notin{\rm in}(A_2)$, there is a path $P_2$ from $a_1$ to $\infty$ that does not intersect $A_2$.
Then, the path that is obtained by taking $P_0$ up to the point $a$, then $P_1$ and then $P_2$ is a path from $u_0\in A_2\cup{\rm in}(A_2)$ to $\infty$ which does not intersect $A_2$, and we get a contradiction. 

Finally, we show that the last statement holds. 
If $P$ is a path from a point of ${\rm in}(A_1)$ to $\infty$, then it must intersect $A_1$; let $a_1$ be such an intersecting point; the part of the path $P$ that starts at $a_1$ is a path from $a_1\in A_1\subseteq{\rm in}(A_2)$ to $\infty$, hence it intersects $A_2$.
Therefore, every path from any point of ${\rm in}(A_1)$ to $\infty$ intersects $A_2$, i.e., ${\rm in}(A_1)\subseteq A_2\cup{\rm in}(A_2)$.
By way of contradiction assume that ${\rm in}(A_1)\nsubseteq {\rm in}(A_2)$; then, there is $a_2\in A_2\cap{\rm in}(A_1)$; let $P$ be a path from $a_2$ to $\infty$ and let $a$ be the last intersection point of $P$ with $A_1\cup A_2$; 
if $a\in A_1$, then the part of $P$ that starts at $a$ is a path from $a$ to $\infty$ that does not intersect $A_2$, contradicting the assumption that $A_1\subseteq{\rm in}(A_2)$;
if $a\in A_2$, then since $A_2$ is connected there is a path $P_2$ in $A_2$ from $a_2$ to $a$ and then, the path that is obtained by taking $P_2$ and then the part of $P$ that starts at $a$ is a path from $a_2\in{\rm in}(A_1)$ to $\infty$ which does not intersect $A_1$, and we get a contradiction.
Hence, ${\rm in}(A_1)\subseteq {\rm in}(A_2)$, and the inclusion is obviously proper, since $\emptyset\neq A_1\subseteq {\rm in}(A_2)\setminus{\rm in}(A_1)$.
\end{proof}

\begin{corollary}\label{cor:sameI}
Let $\mathcal{C}_0$ be a collection of connected components of $G_{\sigma_0}$, and let $S: = \bigcup_{A \in \mathcal{C}_0} A \cup in(A)$.
Let $v_0,v_1,\ldots,v_N$ be a path in $\Z^d$ such that $N>1$, $v_i \in S$ for all $0<i<N$, but $v_0, v_N \notin S$ (in particular, $v_0, v_N \notin V_{\sigma_0}$). 
Then, $I_{\sigma_0}(v_0) = I_{\sigma_0}(v_N)$. \end{corollary}

\begin{figure}[t]\label{fig:pre_shift}
\centering
\includegraphics[width=0.8\textwidth]{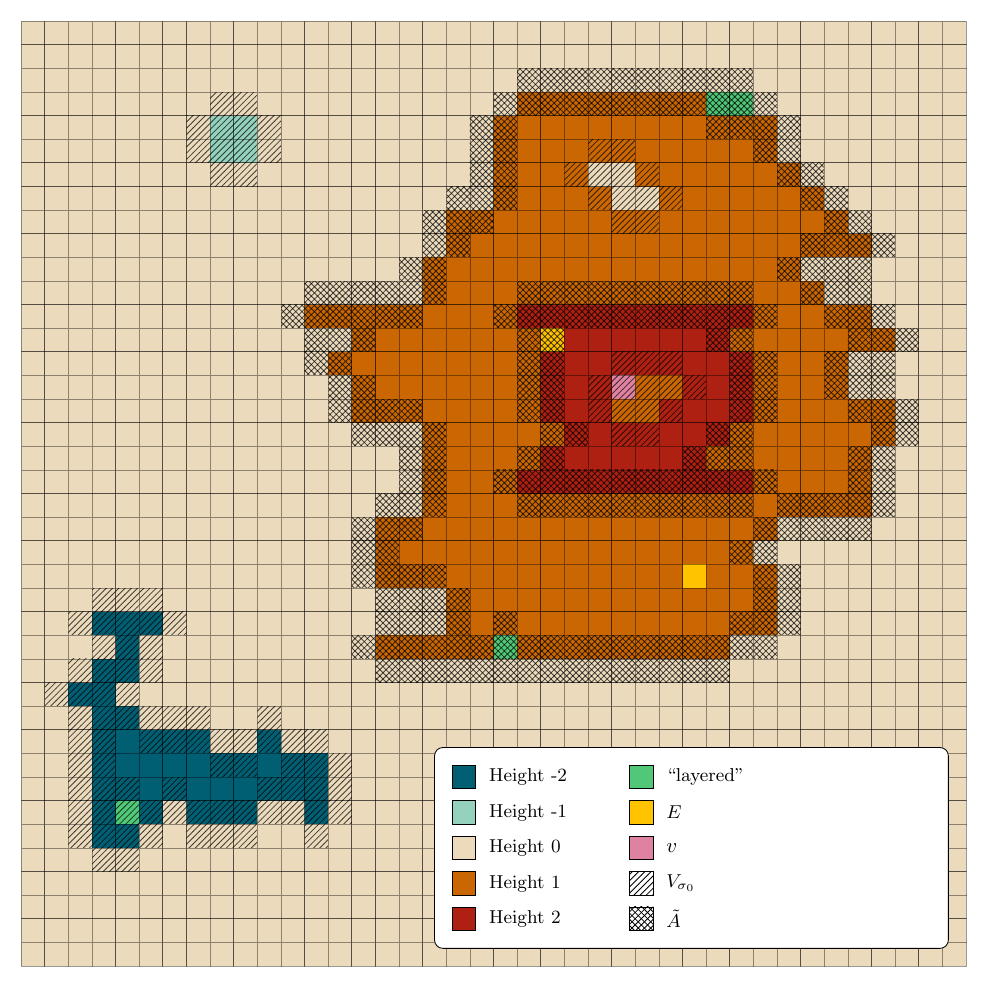}:
\caption{An illustration of the corresponding $V_{\sigma_{0}}$ and $\tilde{A}$, given a fixed function $I_{\sigma_0}$, vertex $v$ and set $E$ (of two vertices). Notice that the innermost cross-hatched component is $A_{v}$.}
\end{figure}

\begin{proof}
We will show that there is $A\in{\mathcal C}_0$ such that $v_1,\ldots,v_{N-1}\in A\cup{\rm in}(A)$. 
Then, obviously $v_0, v_N\in\partial^{\rm out}(A\cup{\rm in}(A)))=\partial_{{\rm vis}(\infty)}
A$ and hence $I_{\sigma_0}(v_0)=\tilde{I}_{\sigma_0}(\infty;A)=I_{\sigma_0}(v_N)$, by Observation \ref{obs:svC}.

Let $i^*$ be the maximal $2\leq i\leq N$ for which there is $A\in\mathcal{C}_0$ such that $v_1,\ldots,v_{i-1}\in A\cup{\rm in}(A)$. 
By way of contradiction, assume that $i^*<N$.
There is $A^*\in\mathcal{C}_0$ such that $v_{i^*}\in A^*\cup{\rm in}(A^*)$, and by the definition of $i^*$, there is $A\in\mathcal{C}_0$ such that $v_1,\ldots,v_{i^*-1}\in A\cup{\rm in}(A)$.
The maximality of $i^*$ implies that $A^*\neq A$ and hence ${\rm dist}(A^*,A)>1$, since both $A^*$ and $A$ are connected components of $G_{\sigma_0}$.
Obviously, ${\rm dist}(A^*\cup{\rm in}(A^*), A\cup{\rm in}(A))\leq\lVert v_{i^*}-v_{i^*-1}\rVert_1=1$.
Then, Lemma \ref{lem:unique smallest contour} implies that $A^*\cup{\rm in}(A^*)\subseteq A\cup{\rm in}(A)$ or $A\cup{\rm in}(A)\subseteq A^*\cup{\rm in}(A^*)$, contradicting the maximality of $i^*$.
\end{proof}

Let $\mathcal C$ be the collection of all connected components of $G_{\sigma_0}$, let
\begin{equation*}
{\mathcal A}=\{A\in{\mathcal C}\colon E\cap\left( A\cup{\rm in}(A)\right)\neq\emptyset\},\qquad \qquad \qquad
\tilde{A}:=\bigcup_{A\in\mathcal{A}}A,
\end{equation*}
(see Figure \ref{fig:pre_shift})
and let 
$$B_{\infty}:=\left\{ u\in {\mathbb Z}^d\colon  \text{there is a path from } u \text{ to } \infty \text{ that does not  intersect } \tilde{A}\right\}$$
be the unique infinite connected component of ${\mathbb Z}^d\setminus\tilde{A}$.
For  $v\in{\mathbb Z}^d\setminus\left(B_{\infty}\cup\tilde{A}\right)$, the second and third parts of Lemma \ref{lem:unique smallest contour} imply that there exists $A_v\in{\mathcal A}$ such that $v\in {\rm in}(A_v)$ and ${\rm in}(A_v)\subsetneq{\rm in}(A)$ for any other $A\in{\mathcal A}$ for which $v\in{\rm in}(A)$.

\begin{lemma}\label{lem:tau_v_Av_Iv}
For every $v\in\partial^{\rm out}\tilde{A}$, 
$$
I_{\sigma_0}(v)=\begin{cases}
0 & v\in B_{\infty},\\
\tilde{I}_{\sigma_0}(v;A_v) & v\notin B_{\infty}.
\end{cases}
$$
\end{lemma}

\begin{proof}
Let $\mathcal{A}_v:=\{A\in\mathcal{C}\colon v\in{\rm in}(A)\}$
and let $S_v:=\bigcup_{A\in\mathcal{C}\setminus\mathcal{A}_v}(A\cup{\rm in}(A))$.
Note that $v\notin S_v$.

If $v\in B_{\infty}$ then there is a path from $v$ to $\infty$ that does not intersect $\tilde{A}$;
this path eventually reaches a point in ${\mathbb Z}^d\setminus(\Lambda\cup\partial^{\rm out}\Lambda)\subseteq {\mathbb Z}^d\setminus S_v$, where $I_{\sigma_0}=0$ since $\sigma_0\in\Omega^{\Lambda\times\Z, \rho^{\Dob}}$. 

If $v\notin B_{\infty}$ then any path from $v$ to $\infty$ intersects $A_v$, and the last point in any such path before it first meets $A_v$ is necessarily in $\partial_{{\rm vis}(v)}A_v\subseteq {\mathbb Z}^d\setminus S_v$. 

In any case, there is a path $v_0,v_1,\ldots,v_N$ of points in ${\mathbb Z}^d\setminus\bigcup_{A\in\mathcal{A}_v}A$ such that $v_{i-1}\sim v_i$ for every $1\leq i\leq N$, $v_0=v$, $v_N\notin S_v$ and
\begin{equation*}
I_{\sigma_0}(v_N)=\begin{cases}
0 & v\in B_{\infty},\\
\tilde{I}_{\sigma_0}(v;A_v) & v\notin B_{\infty}.
\end{cases}
\end{equation*}

With this equality, the lemma would follow once we show that
\begin{equation}\label{eq:same value at endpoints}
    I_{\sigma_0}(v_0)=I_{\sigma_0}(v_N).
\end{equation} 
To this end,
recall first that $v_0, v_N\notin S_v$. 
Note that $v_{i}\notin V_{\sigma_{0}}$ for every $0\leq i\leq N$ such that $v_{i} \notin S_{v}$;
hence, $I_{\sigma_{0}}(v_{i-1})=I_{\sigma_{0}}(v_{i})$ for every $1\leq i\leq N$ such that at least one of $v_{i-1}, v_i$ is not in $S_{v}$. 
Finally, for the times that the path spends in $S_{v}$, the value of $I_{\sigma_{0}}$ at the point just preceding an entry point to $S_v$ is the same as at the following leaving point;
namely, if $0<j\leq k<N$ are such that $v_{j-1}, v_{k+1}\notin S_v$ and $v_i\in S_v$ for every $j\leq i\leq k$, then $I_{\sigma_{0}}(v_{j-1})=I_{\sigma_{0}}(v_{k+1})$, by Corollary \ref{cor:sameI}.
\end{proof}

Define a ``pre-shift" ${\tau}_0\colon {\mathbb Z}^d \to \Z \cup\{{\rm ``layered"}\}$ as follows.
$$
{\tau}_0(v):=\begin{cases}
0 & v\in B_{\infty},\\
I_{\sigma_0}(v) & v\in\tilde{A},\\
\tilde{I}_{\sigma_0}(v;A_v) & \text{otherwise.}
\end{cases}
$$

\begin{obs}\label{obs:s_0_constant}
For every connected component $B$ of ${\mathbb Z}^d\setminus\tilde{A}$, the function ${\tau}_0$ is constant on the set $B\cup \partial^{\rm out}B$.
\end{obs}

\begin{proof}
Let $v\in\partial^{\rm out}B$. Necessarily $v\in\tilde{A}$ and hence $\tau_0(v)=I_{\sigma_0}(v)$.
There is $u\in\partial^{\rm in}B$ such that $u\sim v$. 
Necessarily $u\in\partial^{\rm out}\tilde{A}\subseteq{\mathbb Z}^d\setminus\tilde{A}\subseteq{\mathbb Z}^d\setminus V_{\sigma_0}$, therefore $I_{\sigma_0}(v)=I_{\sigma_0}(u)$, and by  Lemma \ref{lem:tau_v_Av_Iv},
$I_{\sigma_0}(u)=\tau_0(u)$.
Hence, $\tau_0(v)=\tau_0(u)$.

To obtain the claim it is therefore enough to show that $\tau_0$ is constant on $B$.
By definition, $\tau_0=0$ on $B_{\infty}$.
For every $u,v\in B\neq B_{\infty}$, clearly $A_u=A_v$, and moreover  
$\partial_{{\rm vis}(u)}(A_u)=\partial_{{\rm vis}(v)}(A_v)$, and hence $\tau_0(u)=\tau_0(v)$.
\end{proof}

\subsection{Defining $\tau$}
\label{sec:s}
Now, we turn the "pre-shift" ${\tau}_0$ into a shift ${\tau}\colon {\mathbb Z}^d \to{\mathbb Z}$. 
Define a configuration $\sigma$ which has a unique sign change above every vertex $v$ where ${\tau}_0(v)$ is an integer, and the sign change occurs at height ${\tau}_0(v)$. 
The value of $\sigma$ above vertices $v$ where ${\tau}_0(v) = {\rm ``layered"}$ equals the value of $\sigma_0$ at these vertices.
Note that, by Observation \ref{obs:s_0_constant},
\begin{multline}\label{eq:s_0_vs_I}
\{\{x,y\}\in E^{\perp}(\Z^{d+1}) \colon \{x,y\}\cap(\Lambda\times \Z)\neq\emptyset,\, \sigma_x\neq \sigma_y\}\\
=\{\{x,y\}\in E^{\perp}(\Z^{d+1}) \colon x,y\in\tilde{A}\times{\mathbb Z},\, (\sigma_0)_x\neq(\sigma_0)_y\}.
\end{multline}
By Corollary \ref{cor:no_overhangs} there is an interfacial configuration $\sigma'$ that has a unique sign change above every vertex, with $\sigma$ having the same sign change at the same height, and 
\begin{multline}\label{eq:s_vs_s_0}
\lvert\{\{x,y\}\in E^{\perp}(\Z^{d+1}) \colon \{x,y\}\cap(\Lambda\times \Z)\neq\emptyset,\, \sigma'_x\neq\sigma'_y\}\rvert\\
\leq \lvert\{\{x,y\}\in E^{\perp}(\Z^{d+1}) \colon \{x,y\}\cap(\Lambda\times \Z)\neq\emptyset,\, \sigma_x\neq\sigma_y\}\rvert.
\end{multline}

For every $v\in{\mathbb Z}^d$, let ${\tau}(v)$ be the height where $\sigma'$ changes sign over $v$.

Note that ${\tau}(v)={\tau}_0(v)$ for every $v\in{\mathbb Z}^d$ such that ${\tau}_0(v)\neq{\rm ``layered"}$, in particular for every $v\in({\mathbb Z}^d\setminus\tilde{A})\cup\partial^{\rm out}({\mathbb Z}^d\setminus\tilde{A})$.
This observation, combined with Observation \ref{obs:s_0_constant} yields the following useful corollary.

\begin{corollary}\label{cor:s_constant}
For every connected component $B$ of ${\mathbb Z}^d\setminus\tilde{A}$, the function ${\tau}$ is constant on the set $B\cup \partial^{\rm out}B$.
\end{corollary}

\subsection{Bounding the total variation of ${\tau}$ via a layering bound}

This section is devoted to proving the following proposition.

\begin{proposition}\label{prop:total_variation_bound_s}
The shift ${\tau}$ as defined above satisfies
\begin{equation}\label{eq:G_C}
\mathcal{H}^{\eta,\Lambda}(\sigma_0) -\GE^{\Lambda}(\eta^{{\tau}})\geq 2\alpha^{\parallel} \left( \mathcal{L}^{\parallel}_{\tilde{A}}(\sigma_0)-|\tilde{A}|\right)+2\alpha^{\perp}\mathcal{L}^{\perp}_{\tilde{A}}(\sigma_0)
\end{equation}
and consequently,
\begin{equation}\label{eq:G_E}
\mathcal{H}^{\eta,\Lambda}(\sigma_0) -\GE^{\Lambda}(\eta^{{\tau}})\geq 2\alpha^{\parallel} \left( \mathcal{L}^{\parallel}_E(\sigma_0)-|E|\right)+2\alpha^{\perp}\max\{\mathcal{L}^{\perp}_{E}(\sigma_0),\TV({\tau})\}.
\end{equation}
\end{proposition}

For every coupling field $\tilde{\eta}\colon E({\mathbb Z}^{d+1})\to[0,\infty)$ and configuration
$\sigma\in\Omega^{\Lambda,\Dob}$, denote
\begin{align*}
\mathcal{H}^{\tilde{\eta},\Lambda,\parallel}(\sigma)&:=2\sum_{\substack{\{x,y\}\in E^{\parallel}({\mathbb Z}^{d+1})\\\{x,y\}\cap(\Lambda\times{\mathbb Z})\neq\emptyset}} \tilde{\eta}_{\{x,y\}}
1_{\sigma_x\neq\sigma_y}=2\sum_{u\in\Lambda}\sum_{k\in{\mathbb Z}}\tilde{\eta}_{\{(u,k),(u,k+1)\}}1_{\sigma_{(u,k)}\neq\sigma_{(u,k+1)}},\\ \mathcal{H}^{\tilde{\eta},\Lambda,\perp}(\sigma)&:=2\sum_{\substack{\{x,y\}\in E^{\perp}({\mathbb Z}^{d+1})\\\{x,y\}\cap(\Lambda\times{\mathbb Z})\neq\emptyset}} \tilde{\eta}_{\{x,y\}}1_{\sigma_x\neq \sigma_y}=2\sum_{\substack{\{u,v\}\in E({\mathbb Z}^d)\\\{u,v\}\cap\Lambda\neq\emptyset}}\sum_{k\in{\mathbb Z}}\tilde{\eta}_{\{(u,k),(v,k)\}}1_{\sigma_{(u,k)}\neq\sigma_{(v,k)}}.
\end{align*}

Define a configuration $\tilde{\sigma}\in\Omega^{\Lambda,\Dob}$ as follows:
$$
\tilde{\sigma}_{(u,k)}=\begin{cases}
(\sigma_0)_{\left(u,k+\tau(u)\right)} & u\in\Lambda\setminus\tilde{A},\\
1 & u\in \tilde{A},\, k>0,\\
-1 & u\in \tilde{A},\, k\leq 0.
\end{cases}
$$

\begin{obs}\label{obs:def_overlap}
If $(u,v)\in\partial\tilde{A}$, i.e., $u\in\tilde{A}$ and $v\notin\tilde{A}$, then  $\tilde{\sigma}_{(u,k)}=
(\sigma_0)_{\left(u,k+\tau(u)\right)}$ for all~$k\in\Z$.
\end{obs}

\begin{proof}
Necessarily $v\notin V_{\sigma_0}$ and hence ${\tau}_0(u)=I_{\sigma_0}(u)=I_{\sigma_0}(v)\neq{\rm ``layered"}$.
Therefore, ${\tau}(u)={\tau}_0(u)=I_{\sigma_0}(u)$. Hence, $(\sigma_0)_{(u,k)}=1$ for every $k>{\tau}(u)$ and $(\sigma_0)_{(u,k)}=-1$ for every $k\leq {\tau}(u)$, i.e., $\tilde{\sigma}_{(u,k)}=
(\sigma_0)_{\left(u,k+{\tau}(u)\right)}$.
\end{proof}

Proposition \ref{prop:total_variation_bound_s} will easily follow from the following lemmas.

\begin{lemma}\label{lem:horizontal}
It holds that
\begin{equation*}
\mathcal{H}^{\eta,\Lambda,\parallel}(\sigma_0)-\mathcal{H}^{\eta^{\tau},\Lambda,\parallel}(\tilde{\sigma})\geq 2\alpha^{\parallel} \left(\mathcal{L}_{\tilde{A}}^{\parallel}(\sigma_0)-\lvert\tilde{A}\rvert\right).
\end{equation*}
\end{lemma}

\begin{lemma}\label{lem:vertical}
It holds that
\begin{equation*}\mathcal{H}^{\eta,\Lambda,\perp}(\sigma_0)-
\mathcal{H}^{\eta^{\tau},\Lambda,\perp}(\tilde{\sigma}) \geq 2\alpha^{\perp}\mathcal{L}^{\perp}_{\tilde{A}}(\sigma_0).
\end{equation*}
\end{lemma}

Before proving Lemmas \ref{lem:horizontal} and \ref{lem:vertical}, let us show how they imply Proposition \ref{prop:total_variation_bound_s}.

\begin{proof}[Proof of Proposition \ref{prop:total_variation_bound_s}]
Combining Lemmas \ref{lem:horizontal} and \ref{lem:vertical} yields that
\begin{align*}
\mathcal{H}^{\eta,\Lambda}(\sigma_0)-
\mathcal{H}^{\eta^{\tau},\Lambda}(\tilde{\sigma})
&=\left(\mathcal{H}^{\eta,\Lambda,\parallel}(\sigma_0)+\mathcal{H}^{\eta,\Lambda,\perp}(\sigma_0)\right)-
\left(\mathcal{H}^{\eta^{\tau},\Lambda,\parallel}(\tilde{\sigma})+\mathcal{H}^{\eta^{\tau},\Lambda,\perp}(\tilde{\sigma})\right)\\
&=\left(\mathcal{H}^{\eta,\Lambda,\parallel}(\sigma_0)-\mathcal{H}^{\eta^{\tau},\Lambda,\parallel}(\tilde{\sigma})\right)+\left(\mathcal{H}^{\eta,\Lambda,\perp}(\sigma_0)-\mathcal{H}^{\eta^{\tau},\Lambda,\perp}(\tilde{\sigma})\right)\\
&\geq 2\alpha^{\parallel} \left( \mathcal{L}^{\parallel}_{\tilde{A}}(\sigma_0)-|\tilde{A}|\right)+2\alpha^{\perp}\mathcal{L}^{\perp}_{\tilde{A}}(\sigma_0),
\end{align*}
and \eqref{eq:G_C} follows since obviously $\GE^{\Lambda}(\eta^{\tau})\leq \mathcal{H}^{\eta^{\tau},\Lambda}(\tilde{\sigma})$.
We proceed to show how \eqref{eq:G_C} implies \eqref{eq:G_E}. 
Note first that
$E\cap V_{\sigma_0}\subseteq\tilde{A}$ and hence
$$
\mathcal{L}^{\parallel}_{\tilde{A}}(\sigma_0)-|\tilde{A}|\geq \mathcal{L}^{\parallel}_{E\cap V_{\sigma_0}}(\sigma_0)-|E\cap V_{\sigma_0}|,\quad \mathcal{L}^{\perp}_{\tilde{A}}(\sigma_0)\geq \mathcal{L}^{\perp}_{E\cap V_{\sigma_0}}(\sigma_0).
$$
Additionally, note that $\mathcal{L}^{\parallel}_{\{u\}}(\sigma_0)=1$ for every $u\notin V_{\sigma_0}$ and $\mathcal{L}^{\perp}_{\{u,v\}}(\sigma_0)=0$ for every $\{u,v\}\in E(\Z^d)$ such that $\{u,v\}\nsubseteq V_{\sigma_0}$ and hence,
\begin{equation*}
\mathcal{L}^{\parallel}_{E\cap V_{\sigma_0}}(\sigma_0)-|E\cap V_{\sigma_0}|=\mathcal{L}^{\parallel}_E(\sigma_0)-|E|,\quad \mathcal{L}^{\perp}_{E\cap V_{\sigma_0}}(\sigma_0)=\mathcal{L}^{\perp}_E(\sigma_0).
\end{equation*}
Finally, by \eqref{eq:s_vs_s_0} and \eqref{eq:s_0_vs_I},
\begin{align*}
\mathcal{L}^{\perp}_{\tilde{A}}(\sigma_0)&=\lvert\{\{x,y\}\in E^{\perp}({\mathbb Z}^{d+1})\colon x,y\in\tilde{A}\times{\mathbb Z},\,(\sigma_0)_x\neq(\sigma_0)_y\}\rvert\\
&=\lvert\{\{x,y\}\in E^{\perp}({\mathbb Z}^{d+1})\colon \{x,y\}\cap(\Lambda\times{\mathbb Z})\neq\emptyset,\,\sigma_x\neq\sigma_y\}\rvert\\
&\geq \lvert\{\{x,y\}\in E^{\perp}({\mathbb Z}^{d+1})\colon \{x,y\}\cap(\Lambda\times{\mathbb Z})\neq\emptyset,\,\sigma'_x\neq\sigma'_y\}\rvert=\TV({\tau}).
\qedhere\end{align*}
\end{proof} 

We now prove Lemmas \ref{lem:horizontal} and \ref{lem:vertical}.

\begin{proof}[Proof of Lemma \ref{lem:horizontal}]
If $u\in\Lambda\setminus\tilde{A}$, then for every $k\in{\mathbb Z}$,
\begin{align*}
&\eta^{\tau}_{\{(u,k),(u,k+1)\}}=\eta_{\{(u,k+{\tau}(u)),(u,k+1+{\tau}(u))\}},\\ &\tilde{\sigma}_{(u,k)}=(\sigma_0)_{(u,k+{\tau}(u))},\quad\tilde{\sigma}_{(u,k+1)}=(\sigma_0)_{(u,k+1+{\tau}(u))},
\end{align*}
and hence,
\begin{align*}
\sum_{k\in{\mathbb Z}}\eta_{\{(u,k),(u,k+1)\}}1_{(\sigma_0)_{(u,k)}\neq(\sigma_0)_{(u,k+1)}}=&\sum_{k\in{\mathbb Z}}\eta_{\{(u,k+{\tau}(u)),(u,k+1+{\tau}(u))\}}1_{(\sigma_0)_{(u,k+{\tau}(u))}\neq(\sigma_0)_{(u,k+1+{\tau}(u))}}\\
=&\sum_{k\in{\mathbb Z}}\eta^{\tau}_{\{(u,k),(u,k+1)\}}1_{\tilde{\sigma}_{(u,k)}\neq\tilde{\sigma}_{(u,k+1)}}.
\end{align*}
If $u\in\tilde{A}$, then
$$
\sum_{k\in{\mathbb Z}}\eta^{\tau}_{\{(u,k),(u,k+1)\}}1_{\tilde{\sigma}_{(u,k)}\neq\tilde{\sigma}_{(u,k+1)}}=\eta^{\tau}_{\{(u,0),(u,1)\}}=\eta_{\{(u,{\tau}(u)),(u,{\tau}(u)+1)\}}
$$
and hence, since by the definition of ${\tau}$, the configuration $\sigma_0$ has a sign change at height ${\tau}(u)$, i.e., $(\sigma_0)_{(u,{\tau}(u))}\neq (\sigma_0)_{(u,{\tau}(u)+1)}$,
\begin{multline*}
\sum_{k\in{\mathbb Z}}\eta_{\{(u,k),(u,k+1)\}}1_{(\sigma_0)_{(u,k)}\neq(\sigma_0)_{(u,k+1)}}-\sum_{k\in{\mathbb Z}}\eta^{\tau}_{\{(u,k),(u,k+1)\}}1_{\tilde{\sigma}_{(u,k)}\neq\tilde{\sigma}_{(u,k+1)}}\\
=\sum_{k\in{\mathbb Z}}\eta_{\{(u,k),(u,k+1)\}}1_{(\sigma_0)_{(u,k)}\neq(\sigma_0)_{(u,k+1)}}-\eta_{\{(u,{\tau}(u)),(u,{\tau}(u)+1)\}}\geq\alpha^{\parallel} \left(\mathcal{L}_{\{u\}}^{\parallel}(\sigma_0)-1\right).
\end{multline*}
The result follows.
\end{proof}

\begin{proof}[Proof of Lemma \ref{lem:vertical}]
For every $\{u,v\}\in E({\mathbb Z}^d)$ such that $\{u,v\}\cap\Lambda\neq\emptyset$ and $\{u,v\}\nsubseteq\tilde{A}$, it holds that $u,v\in B\cup\partial^{\rm out}B$ for some connected component $B$ of ${\mathbb Z}^d\setminus\tilde{A}$ and hence
${\tau}(u)={\tau}(v)$, by Corollary \ref{cor:s_constant}. 
Hence, by \eqref{eq:shifted_noise_perp}, $\eta^{\tau}_{\{(u,k),(v,k)\}}=\eta_{\{(u,k+\tau(u)),(v,k+\tau(u))\}}$ for every $k\in{\mathbb Z}$.
Additionally, with the aid of Observation \ref{obs:def_overlap} in case that $\lvert\{u,v\}\cap\tilde{A}\rvert=1$, it holds that $\tilde{\sigma}_{(u,k)}=(\sigma_0)_{(u,k+{\tau}(u))}$ and $\tilde{\sigma}_{(v,k)}=(\sigma_0)_{(v,k+{\tau}(v))}=(\sigma_0)_{(v,k+{\tau}(u))}$ for every $k\in{\mathbb Z}$. Therefore,
\begin{align*}
\sum_{k\in{\mathbb Z}}\eta_{\{(u,k),(v,k)\}}1_{(\sigma_0)_{(u,k)}\neq(\sigma_0)_{(v,k)}}
=&\sum_{k\in{\mathbb Z}}\eta_{\{(u,k+{\tau}(u)),(v,k+{\tau}(u))\}}1_{(\sigma_0)_{(u,k+{\tau}(u))}\neq(\sigma_0)_{(v,k+{\tau}(u))}}\\
=&\sum_{k\in{\mathbb Z}}\eta^{\tau}_{\{(u,k),(v,k)\}}1_{\tilde{\sigma}_{(u,k)}\neq\tilde{\sigma}_{(v,k)}}
\end{align*}

Hence, since for every $\{u,v\}\in E({\mathbb Z}^d)$ such that both $u$ and $v$ are in $\tilde{A}$ it holds that $\tilde{\sigma}_{(u,k)}=\tilde{\sigma}_{(v,k)}$ for every integer $k$, it follows that 
\begin{align*}
\mathcal{H}^{\eta,\Lambda,\perp}(\sigma_0)-
\mathcal{H}^{\eta^{\tau},\Lambda,\perp}(\tilde{\sigma})&=2\sum_{\substack{\{u,v\}\in E({\mathbb Z}^d)\\u,v\in\tilde{A}}}\sum_{k\in{\mathbb Z}}\eta_{\{(u,k),(v,k)\}}1_{(\sigma_0)_{(u,k)}\neq(\sigma_0)_{(v,k)}}\\
&\geq 2\alpha^{\perp}\sum_{\substack{\{u,v\}\in E({\mathbb Z}^d)\\u,v\in\tilde{A}}}\sum_{k\in{\mathbb Z}}1_{(\sigma_0)_{(u,k)}\neq(\sigma_0)_{(v,k)}}=2\alpha^{\perp}\mathcal{L}^{\perp}_{\tilde{A}}(\sigma_0).
\qedhere
\end{align*}
\end{proof}

\subsection{Bounding the trip entropy of ${\tau}$}

This section is devoted to proving the following proposition.

\begin{proposition}\label{prop:R_s}
The shift ${\tau}$ as defined above satisfies
$$
R({\tau})< R(E)+2(\lvert E\rvert-1)+\frac{16}{\min\{\alpha^{\parallel},\alpha^{\perp}\}d}\left(\mathcal{H}^{\eta,\Lambda}(\sigma_0) -\GE^{\Lambda}(\eta^{\tau})\right).
$$    
\end{proposition}

Denote 
${\mathcal L}:=\{u\in{\mathbb Z}^d\colon I_{\sigma_0}(u)={\rm ``layered"}\}$ and
${\mathcal D}:=\{\{u,v\}\in E({\mathbb Z}^d)\colon I_{\sigma_0}(u)\neq I_{\sigma_0}(v)\}$.

\begin{obs}\label{obs:LD}
It holds that ${\mathcal L}\subseteq V_{\sigma_0}$ and for every $\{u,v\}\in{\mathcal D}$, the vertices $u$ and $v$ are both in $V_{\sigma_0}$ and moreover, belong to the same connected component of $G_{\sigma_0}$.
\end{obs}

\begin{obs}\label{obs:CLD}
For every connected component $A$ of the graph $G_{\sigma_0}$ it holds that
$$\lvert A\rvert\leq 2\left(\mathcal{L}^{\parallel}_A(\sigma_0)-|A|+\mathcal{L}^{\perp}_A(\sigma_0)\right).$$
\end{obs}

\begin{proof}
Every $u\in A\setminus{\mathcal L}$ has at least one neighbour $v$ such that $\{u,v\}\in{\mathcal D}$. Hence, by using Observation \ref{obs:LD}, it holds that
\begin{equation*}
\lvert A\rvert\leq \lvert {\mathcal L}\cap A\rvert+2\left\lvert {\mathcal D}\cap\binom{A}{2} \right\rvert
\leq \frac{1}{2}\left(\mathcal{L}^{\parallel}_A(\sigma_0)-|A|\right)+2\mathcal{L}^{\perp}_A(\sigma_0). \qedhere\end{equation*}
\end{proof}

\begin{lemma}\label{lem:net0}
Let $A$ be a connected component of the graph $G_{\sigma_0}$.
There is a set $S\subseteq A$ such that 
$A\subseteq\bigcup_{a\in S}{\mathcal B}_4(a)$ and
$$\lvert S\rvert<\frac{1}{d}\left(\mathcal{L}^{\parallel}_A(\sigma_0)-|A|+\mathcal{L}^{\perp}_A(\sigma_0)\right).$$
\end{lemma}

\begin{proof}
We first show that for every $a\in A$,
\begin{equation}
\label{eq:disagree}\lvert{\mathcal L}\cap{\mathcal B}_2(a)\cap A\rvert+\left\lvert{\mathcal D}\cap\binom{{\mathcal B}_2(a)\cap A}{2}\right\rvert>d.
\end{equation}
If $I_{\sigma_0}(a)={\rm ``layered"}$, then $a\in {\mathcal L}$ and for every neighbour $b$ of $a$ it holds that $b\in A$ and either $b\in{\mathcal L}$ or $\{a,b\}\in{\mathcal D}$.
Hence, 
$$\lvert{\mathcal L}\cap{\mathcal B}_1(a)\cap A\rvert+\left\lvert{\mathcal D}\cap\binom{{\mathcal B}_1(a)\cap A}{2}\right\rvert\geq 2d+1.$$
If $I_{\sigma_0}(a)\neq{\rm ``layered"}$, then since $a\in V_{\sigma_0}$, it follows that there is a neighbour $b$ of $a$ such that $I_{\sigma_0}(b)\neq I_{\sigma_0}(a)$. 
Note that $b\in A$ and $\{a,b\}\in{\mathcal D}$.
With no loss of generality we may assume that $b=a+e_d$.
Then, for every $v\in\{\pm e_i\}_{i=1}^{d-1}$ it holds that $\{\{a,a+v\},\{b,b+v\},\{a+v,b+v\}\}\cap{\mathcal D}\neq\emptyset$. Hence, by using Observation \ref{obs:LD},
\begin{equation*}\left\lvert{\mathcal D}\cap\binom{{\mathcal B}_2(a)\cap A}{2}\right\rvert\geq 2d-1.
\end{equation*}

Now, let $S$ be a set of maximal cardinality in $A$ such that the sets $\{{\mathcal B}_2(a)\}_{a\in S}$ are mutually disjoint.
The maximality of $S$ implies that $A\subseteq \bigcup_{a\in S}{\mathcal B}_4(a)$, and by \eqref{eq:disagree}, 
\begin{align*}
\lvert S\rvert&<\frac{1}{d}\sum_{a\in S}\left(\lvert{\mathcal L}\cap{\mathcal B}_2(a)\cap A\rvert+\left\lvert{\mathcal D}\cap\binom{{\mathcal B}_2(a)\cap A}{2}\right\rvert\right)\\
&\leq\frac{1}{d}\left(\lvert{\mathcal L}\cap A\rvert+\left\lvert{\mathcal D}\cap\binom{A}{2}\right\rvert\right)\leq\frac{1}{d}\left(\frac{1}{2}\left(\mathcal{L}^{\parallel}_A(\sigma_0)-|A|\right)+\mathcal{L}^{\perp}_A(\sigma_0)
\right).
\qedhere
\end{align*}
\end{proof}

\begin{proof}[Proof of Proposition \ref{prop:R_s}]
By Corollary \ref{cor:s_constant},
every level component of $\tau$ intersects $\tilde{A}$.
Pick $\Omega\subseteq\tilde{A}$ that intersects once each level component of $\tau$. By Lemma \ref{lem:net0}, for every $A\in{\mathcal A}$ there is a set $S_A\subseteq A$ such that 
$A\subseteq\bigcup_{a\in S_A}{\mathcal B}_4(a)$ and
$\lvert S_A\rvert<\frac{1}{d}\left(\mathcal{L}^{\parallel}_A(\sigma_0)-|A|+\mathcal{L}^{\perp}_A(\sigma_0)\right)$.
Then, applying \eqref{eq:ham+} for the set $\Omega\cap A$, for each $A\in{\mathcal A}$, yields that
\begin{align}
R({\tau})<& R(E)+2(\lvert E\rvert-1)+\sum_{A\in{\mathcal A}}\left(2{\rm dist}(E,
A)+20\lvert S_A\rvert+8\lvert\Omega\cap A\rvert\right)\nonumber\\
=& R(E)+2(\lvert E\rvert-1)+8\lvert\mathcal{LC}({\tau})\rvert+\sum_{A\in{\mathcal A}}\left(2{\rm dist}(E,
A)+20\lvert S_A\rvert\right)\nonumber\\
<& R(E)+2(\lvert E\rvert-1)+8\lvert\mathcal{LC}({\tau})\rvert+\nonumber\\
&\sum_{A\in{\mathcal A}}\left(2{\rm dist}(E,
A)+\frac{20}{d}
\left(\mathcal{L}^{\parallel}_A(\sigma_0)-|A|+\mathcal{L}^{\perp}_A(\sigma_0)
\right)\right).\label{eq:Rsum1}
\end{align}
Let $A\in{\mathcal A}$ and let $v\in E\cap(A\cup{\rm in}(A))$.
If $v\in A$, then 
${\rm dist}(E,A)=0$. 
Otherwise, let $B:={\mathcal B}_{{\rm dist}(v,A)}(0)\cap\left({\mathbb Z}^{d-1}\times\{0\}\right)$. For every $u\in B$, since $v\in{\rm in}(A)$, it holds that $\{v+u+t\,e_d\colon t\in{\mathbb Z}\}\cap A\neq\emptyset$ and it follows that $\lvert B\rvert\leq\lvert A\rvert$.
Therefore, since the sequence $\left(\frac{1}{d}\binom{r+d-1}{d-1}\right)_{d=1}^{\infty}$ is non-decreasing for every positive integer $r$, 
\begin{equation*}
{\rm dist}(v,A)\leq\frac{1}{3}\binom{{\rm dist}(v,A)+2}{2}\leq\frac{1}{d}\binom{{\rm dist}(v,A)+d-1}{d-1}=\frac{1}{d}\lvert B\cap[0,\infty)^d\rvert<\frac{1}{d}\lvert B\rvert\leq\frac{1}{ d}\lvert A\rvert
\end{equation*}
and hence, by Observation \ref{obs:CLD},
\begin{equation*}
{\rm dist}(E,A)\leq {\rm dist}(v,A)<
\frac{2}{d}\left(\mathcal{L}^{\parallel}_A(\sigma_0)-|A|+\mathcal{L}^{\perp}_A(\sigma_0)
\right).
\end{equation*}
Plugging this in \eqref{eq:Rsum1} yields that
\begin{align*}
R({\tau})&< R(E)+2(\lvert E\rvert-1)
+8\lvert\mathcal{LC}({\tau})\rvert+\frac{24}{d}\sum_{A\in{\mathcal A}}\left(\mathcal{L}^{\parallel}_A(\sigma_0)-|A|+\mathcal{L}^{\perp}_A(\sigma_0)\right)\\
&= R(E)+2(\lvert E\rvert-1)
+8\lvert\mathcal{LC}({\tau})\rvert+\frac{24}{d}\left(\mathcal{L}^{\parallel}_{\tilde{A}}(\sigma_0)-|\tilde{A}|+\mathcal{L}^{\perp}_{\tilde{A}}(\sigma_0)\right).
\end{align*}
The result now follows by \eqref{eq:G_C} and since
\begin{equation*}
\lvert\mathcal{LC}({\tau})\rvert\leq\frac{1}{d}\TV({\tau})\leq \frac{1}{2\alpha^{\perp}d}\left(\mathcal{H}^{\eta,\Lambda}(\sigma_0) -\GE^{\Lambda}(\eta^{\tau})\right).
\end{equation*}
by \eqref{eq:no_LC} 
and \eqref{eq:G_E}.
\end{proof}

\subsection{Proof of Proposition \ref{prop:existance_Dob_shift_ground_config} and Lemma \ref{lem:spin different than sign shift existence}}
\label{sec:4_4}

\begin{proof}[Proof of Proposition \ref{prop:existance_Dob_shift_ground_config}]
Let $\sigma_0:=\sigma^{\eta,\Lambda,\Dob}$ and consider
the shift ${\tau}$ as defined above. 
Note that 
\begin{equation}\label{eq:G_H-GE}
G^{\eta,\Lambda}({\tau})=\GE^{\Lambda}(\eta)-\GE^{\Lambda}(\eta^{\tau})=\mathcal{H}^{\eta,\Lambda}\left(\sigma_0\right) -\GE^{\Lambda}(\eta^{\tau})
\end{equation}
and the result follows by \eqref{eq:G_E} and Proposition \ref{prop:R_s}.
\end{proof}

In the proof of Lemma \ref{lem:spin different than sign shift existence} we will use the following lemma, which is an immediate consequence of \cite{T}*{Lemma 1}.

\begin{lemma}\label{lem:Timar1}
For a positive integer $\ell$. let ${\mathcal G}_{\ell}$ be the graph on the vertex set $E({\mathbb Z}^{\ell})$, in which distinct $e,\tilde{e}\in E({\mathbb Z}^{\ell})$ are adjacent if 
$$e,\tilde{e}\in \{\{u,u+e_i\},\{u+e_i,u+e_i+e_j\},\{u+e_i+e_j,u+e_j\},\{u+e_j,u\}\}$$ for some $u\in{\mathbb Z}^{\ell}$ and $1\leq i<j\leq\ell$.
If $A\subseteq{\mathbb Z}^{\ell}$ is a connected set such that ${\mathbb Z}^{\ell}\setminus A$ is connected as well, then the set of edges $\{\{u,v\}\in E({\mathbb Z}^{\ell})\colon (u,v)\in\partial A\}$ is connected in ${\mathcal G}_{\ell}$.
\end{lemma}

\begin{proof}[Proof of Lemma \ref{lem:spin different than sign shift existence}]
Let $\sigma_0:=\sigma^{\eta,\Lambda,\Dob}$, let $E=\{0\}$, and consider
the shift ${\tau}$ as defined above (in Section~\ref{sec:s}). 
By using \eqref{eq:G_H-GE}, it follows from \eqref{eq:G_E} that
$G^{\eta,\Lambda}({\tau})\geq    
2\alpha^{\perp}\TV({\tau})$ and Proposition \ref{prop:R_s} yields that $G^{\eta,\Lambda}({\tau})\geq \frac{1}{16}\min\{\alpha^{\parallel},\alpha^{\perp}\}d\, R({\tau})$; hence, ${\tau}\in \AS^{\eta,\Lambda}(\alpha^{\parallel}, \alpha^{\perp})$; moreover, by \eqref{eq:G_C}, $G^{\eta,\Lambda}({\tau})\geq 2\alpha^{\perp}\mathcal{L}^{\perp}_{\tilde{A}}\left(\sigma_0\right)$.
Hence, it is left to show that $\mathcal{L}^{\perp}_{\tilde{A}}\left(\sigma_0\right)\geq \lvert k\rvert$.

With no loss of generality, assume that $k>0$ and that $k=\max\{h\in{\mathbb Z}\colon (\sigma_0)_{(0,h)}=-1\}$.

Define $\pi_{d+1}:{\mathbb Z}^d\times{\mathbb Z}\to{\mathbb Z}$ by $\pi_{d+1}(u,h)=h$, and for every $\{u,v\}\in E({\mathbb Z}^{d+1})$, let $\pi_{d+1}(\{u,v\})=\{\pi_{d+1}(u),\pi_{d+1}(v)\}$.
Note that if $e,\tilde{e}\in E({\mathbb Z}^{d+1})$ are adjacent in the graph ${\mathcal G}_{d+1}$ (defined above in Lemma \ref{lem:Timar1}), then $\pi_{d+1}(e)\subseteq\pi_{d+1}(\tilde{e})$ or $\pi_{d+1}(e)\supseteq\pi_{d+1}(\tilde{e})$.

Let $X:=\{x\in{\mathbb Z}^{d+1}\colon \sigma_0(x)=1\}$. The set $\{x\in 
X\colon \text{there is no path in } X \text{ from } x\text{ to } \infty\}$ is necessarily empty, otherwise we could flip all signs in this finite set to get a configuration in $\Omega^{\Lambda,{\rm Dob}}$ with smaller $H^{\eta}$.
Hence, $X$ is connected. 
Similarly, the set ${\mathbb Z}^{d+1}\setminus X=\{y\in{\mathbb Z}^{d+1}\colon \sigma_0(y)=-1\}$ is connected as well. 
Hence, by Lemma \ref{lem:Timar1},
the set 
$${\mathcal I}(\sigma_0)=\{\{x,y\}\in E({\mathbb Z}^{d+1})\colon (x,y)\in\partial X\}=\{\{x,y\}\in E({\mathbb Z}^{d+1})\colon \sigma_0(x)=1, \sigma_0(y)=-1\}$$ 
is connected in ${\mathcal G}_{d+1}$.

We now argue similarly to the proof of Lemma \ref{lem:tau_v_Av_Iv}.
Take some $w \in \Z^{d} \setminus(\Lambda\cup\partial^{\rm out}\Lambda)$.
Since $\sigma_{0} \in \Omega^{\Lambda\times\Z, \rho^{\Dob}} $, we have that $I_{\sigma_0}(w)=0$.
Since the set ${\mathcal I}(\sigma_0)$ is connected in ${\mathcal G}_{d+1}$, there is a sequence $(\tilde{e}_i)_{i=0}^N$ of edges in ${\mathcal I}(\sigma_0)$ such that 
$\tilde{e}_0=\{(0,k+1),(0,k)\}$, $\tilde{e}_N=\{(w,1),(w,0)\}$, and $\tilde{e}_{i-1}, \tilde{e}_i$ are adjacent in ${\mathcal G}_{d+1}$ for every $1\leq i\leq N$; hence, $\pi_{d+1}(\tilde{e}_0)=\{k+1,k\}$, $\pi_{d+1}(\tilde{e}_N)=\{1,0\}$, and $\pi_{d+1}(\tilde{e}_{i-1})\subseteq\pi_{d+1}(\tilde{e}_i)$ or $\pi_{d+1}(\tilde{e}_{i-1})\supseteq\pi_{d+1}(\tilde{e}_i)$ for every $1\leq i\leq N$.

Recall that $\mathcal C$ is the collection of all connected components of the graph $G_{\sigma_0}$ and ${\mathcal A}=\{A\in{\mathcal C}\colon 0\in A\cup{\rm in}(A)\}$, since $E=\{0\}$,
and let $S:=\bigcup_{A\in{\mathcal C}\setminus{\mathcal A}}( A\cup{\rm in}(A))$.
Note that $\tilde{e}_0\cap(S\times{\mathbb Z})=\emptyset$ and $\tilde{e}_N\cap (S\times{\mathbb Z})=\emptyset$.
We claim that for intervals where the sequence $(\tilde{e}_i)_{i=0}^N$ spends in $S\times{\mathbb Z}$, the value of $\pi_{d+1}$ at the edge just before entering $S\times{\mathbb Z}$ and at the subsequent edge leaving $S\times{\mathbb Z}$ is the same;
namely, that if $0<j\leq k<N$ are such that $\tilde{e}_{j-1}, \tilde{e}_{k+1}\nsubseteq S\times{\mathbb Z}$ and $\tilde{e}_i\subseteq S\times{\mathbb Z}$ for every $j\leq i\leq k$, then $\pi_{d+1}(\tilde{e}_{j-1})=\pi_{d+1}(\tilde{e}_{k+1})$.
To show this, first note that necessarily $\tilde{e}_{j-1}=\{(u,I_{\sigma_0}(u)),(u,I_{\sigma_0}(u)+1)\}$, $\tilde{e}_{k+1}=\{(v,I_{\sigma_0}(v)),(v,I_{\sigma_0}(v)+1)\}$ for some
$u,v\in {\mathbb Z}^d\setminus V_{\sigma_0}$. 
It is easy to verify that there is a path in $\Z^d$ from $u$ to $v$ such that all its internal vertices are in $S$.
If this path has internal vertices, then $I_{\sigma_0}(u)=I_{\sigma_0}(v)$ by Corollary \ref{cor:sameI}; otherwise, $u,v\in\Z^d\setminus V_{\sigma_0}$ are adjacent, hence $I_{\sigma_0}(u)=I_{\sigma_0}(v)$ in this case as well.  
Therefore,
$\pi_{d+1}(\tilde{e}_{j-1})=\{I_{\sigma_0}(u),I_{\sigma_0}(u)+1\}=\{I_{\sigma_0}(v),I_{\sigma_0}(v)+1\}=\pi_{d+1}(\tilde{e}_{k+1})$, as claimed.

It follows that for every $1\leq h\leq k$ there is 
$0<i_h<N$ such that $\tilde{e}_{i_h}\nsubseteq S\times{\mathbb Z}$ and $\pi_{d+1}(\tilde{e}_{i_h})=\{h\}$. 
Then, for every $1\leq h\leq k$ it holds that 
$\tilde{e}_{i_h}\in {\mathcal I}(\sigma_0)\cap E^{\perp}({\mathbb Z}^{d+1})$ and $\tilde{e}_{i_h}\cap(\tilde{A}\times{\mathbb Z})\neq\emptyset$ and hence, 
$${\mathcal L}^{\perp}_{\tilde{A}}(\sigma_0)=\lvert\{e\in {\mathcal I}(\sigma_0)\cap E^{\perp}({\mathbb Z}^{d+1})\colon e\cap(\tilde{A}\times{\mathbb Z})\neq\emptyset\}\rvert\geq k,$$
as desired.  
\end{proof}

\section{Convergence of finite-volume ground configurations}\label{sec:proof of convergence theorem}
In this section we prove Theorem~\ref{thm:convergence}, Corollary~\ref{cor:covariant Dobrushin configuration} and Theorem~\ref{thm:decay of correlations and rate of convergence}.

We assume throughout the section that we work with the anisotropic disordered ferromagnet in dimension $D\ge 4$, with disorder distributions $\nu^\parallel$ and $\nu^\perp$ satisfying~\eqref{eq:disorder distributions assumptions} and that condition~\eqref{eq:anisotropic condition} holds with a sufficiently small constant $c>0$ so that both the assumptions of Theorem~\ref{thm:localization} and Theorem~\ref{thm:Existence of admissible ensemble is unlikely} hold.

We introduce the notation, for integer $k\ge 0$,
\begin{equation}
\Lambda(k):= \{-k,\dots, k\}^{d}.
\end{equation}

\subsection{Proof of Theorem~\ref{thm:convergence}}
The proof is based on the following deterministic lemma, which is proved using similar methods as those in Section~\ref{sec:obtaining admissible shifts}.

\begin{lemma}\label{lem:weak dependence on boundary conditions}
Let $\eta\in\mathcal{D}(\alpha^{\parallel},\alpha^{\perp})$ for some $\alpha^{\parallel},\alpha^{\perp}>0$.
Let $L_1>L_0\ge0$ integer. Let $\Lambda^1,\Lambda^2\subset\Z^d$ be finite subsets containing $\Lambda(L_1)$. 
If
\begin{equation}
\sigma^{\eta,\Lambda^1,\Dob}|_{\Lambda(L_0)\times\Z}\not\equiv \sigma^{\eta,\Lambda^2,\Dob}|_{\Lambda(L_0)\times\Z}
\end{equation}
then for $i=1$ or for $i=2$ there exists ${\tau}\in \mathcal{AS}^{\eta,\Lambda^i}(\alpha^{\parallel},\alpha^{\perp})$ with
\begin{equation}
G^{\eta,\Lambda^i}({\tau})\geq \frac{\min \{\alpha^{\parallel},\alpha^{\perp}\}}{4} (L_1-L_0) ^{1-\frac{1}{d}}.
\end{equation} 
\end{lemma}
We postpone the proof of the lemma to Section~\ref{sec:proof of weak dependence on boundary lemma}. The following is an immediate consequence of the lemma and Theorem~\ref{thm:Existence of admissible ensemble is unlikely} (applied with $\Lambda=\Lambda^1$ and with $\Lambda=\Lambda^2$).

\begin{corollary}\label{cor:weak dependence on domain}
There exist constants $C,c>0$ such that the following holds under the assumptions of Theorem~\ref{thm:Existence of admissible ensemble is unlikely} (for a sufficiently small $c_0>0$). Let $L_1>L_0\ge0$ integer. Let $\Lambda^1,\Lambda^2\subset\Z^d$ be finite subsets containing $\Lambda(L_1)$. 
Then
\begin{multline}
\P\left(\sigma^{\eta,\Lambda^1,\Dob}|_{\Lambda(L_0)\times\Z}\not\equiv \sigma^{\eta,\Lambda^2,\Dob}|_{\Lambda(L_0)\times\Z}\right)\\\le C\exp \left(-\frac{c}{\kappa  d^2} \left( \min \left\{\frac{\alpha^{\parallel}}{\alpha^{\perp}},1\right\}\right)^{\frac{d-2}{d-1}}\left(L_1 - L_0\right)^{\frac{d-2}{d}} \right).
\end{multline} 
\end{corollary}

We proceed to prove Theorem~\ref{thm:convergence}. It suffices to prove that for every sequence $(\Lambda_n)$ of finite domains in $\Z^d$, satisfying that $\Lambda_n\supset\Lambda(n)$ for each $n$, and for every $v\in\Z^d$,
\begin{equation}
\label{eq:existence of limit at x}
\text{the restricted configuration $\sigma^{\eta,\Lambda_n,\Dob}|_{\{v\}\times\Z}$ is the same for all large $n$, almost surely.}
\end{equation}
Indeed, if we then define
\begin{equation}\label{eq:sigma eta dob def}
    \text{for each $v\in\Z^d$, $\sigma^{\eta,\Dob}|_{\{v\}\times\Z}$ is the eventual value of $\sigma^{\eta,\Lambda(n),\Dob}|_{\{v\}\times\Z}$ as $n\to\infty$}
\end{equation}
then we may conclude that the eventual value of $\sigma^{\eta,\Lambda_n,\Dob}|_{\{v\}\times\Z}$ also equals $\sigma^{\eta,\Dob}|_{\{v\}\times\Z}$ by applying~\eqref{eq:existence of limit at x} to the following two sequences of domains formed by interlacing subsequences of $(\Lambda_n)$ and $(\Lambda(n))$: either taking $(\Lambda_{2n})$ in the even positions and $(\Lambda(2n+1))$ in the odd positions or taking $(\Lambda_{2n+1})$ in the odd positions and $(\Lambda(2n))$ in the even positions.

We proceed to prove~\eqref{eq:existence of limit at x}, for some fixed sequence $(\Lambda_n)$ of finite domains in $\Z^d$, satisfying that $\Lambda_n\supset\Lambda(n)$ for each $n$. Let $L_0\ge 0$ be an integer.
Let $E_n:=\{\sigma^{\eta,\Lambda_n,\Dob}|_{\Lambda(L_0)\times\Z}\ne \sigma^{\eta,\Lambda_{n+1},\Dob}|_{\Lambda(L_0)\times\Z}\}$. By Corollary~\ref{cor:weak dependence on domain} (with $L_1=n$) we deduce that for every $n>L_0$,
\begin{equation}
    \P(E_n)\le C\exp \left(-c(\nu^\parallel,\nu^\perp,d) (n - L_0)^{\frac{d-2}{d}} \right)
\end{equation}
with $c(\nu^\parallel,\nu^\perp,d)>0$, depending only on $\nu^\parallel, \nu^\perp$ and the dimension $d$ and some absolute constant $C>0$. Thus $\sum_n \P(E_n)<\infty$. We conclude that only a finite number of the $E_n$ hold, almost surely, implying that~\eqref{eq:existence of limit at x} holds for all $v\in\Lambda(L_0)$. This finishes the proof of~\eqref{eq:existence of limit at x} as $L_0$ is arbitrary.

\subsection{Proof of Corollary~\ref{cor:covariant Dobrushin configuration}}

The probabilistic estimates~\eqref{eq:cor_localization at v k} and~\eqref{eq:cor_localization at v k moreover} hold as a consequence of Theorem~\ref{thm:localization} (with $\Lambda$ equal to some $\Lambda(n)$).

We proceed to define a $G^d$-invariant set $\mathcal{C}_0$ of coupling fields satisfying that $\P(\mathcal{C}_0)=1$ and, on $\mathcal{C}_0$, $\sigma^{\eta,\Dob}$ is well defined and is $G^d$ covariant.

For an automorpishm $h\in G^d$ and a set $\Lambda\subset\Z^d$ we define $h(\Lambda)$ as the set in $\Z^d$ satisfying that $h(\Lambda)\times\Z=h(\Lambda\times\Z)$ (such a set $h(\Lambda)$ exists by the definition~\eqref{eq:G D-1 automorphism group} of $G^d$).

Define
\begin{equation}
\mathcal{C}_\text{unique}:=\bigcap_{g,h\in G^d} \{\eta\colon \text{$\forall n$, there is a unique ground configuration in $\Omega^{h(\Lambda(n))\times\Z,\rho^{\Dob}}$ for $g(\eta)$}\}.
\end{equation}
As $G^d$ is countable and $g(\eta)$ has the same distribution as $\eta$ (since $g\in G^d$) we have that $\P(\mathcal{C}_\text{unique})=1$ by Lemma~\ref{lem:semi infinite volume ground configuration}. It is clear from the definition that $\mathcal{C}_\text{unique}$ is $G^d$ invariant. As before, the unique ground configuration (on $\mathcal{C}_\text{unique}$) in $\Omega^{h(\Lambda(n))\times\Z,\rho^{\Dob}}$ for $g(\eta)$ is denoted $\sigma^{g(\eta),h(\Lambda(n)),\Dob}$.

Now set
\begin{equation}
\mathcal{C}_0 := \left\{\eta\in\mathcal{C}_\text{unique}\colon \substack{\text{for each $g\in G^d$, there exists a configuration $\sigma^{g(\eta),Dob}:\Z^D\to\{-1,1\}$ such that}\\\text{$\lim_{n\to\infty}\sigma^{g(\eta),h(\Lambda(n)),\Dob}_x = \sigma^{g(\eta),\Dob}_x$ for all $h\in G^d$ and $x\in\Z^D$}}\right\}.
\end{equation}
Then $\P(\mathcal{C}_0)=1$ by Theorem~\ref{thm:convergence} (applied to the sequence $(h(\Lambda(n_0+n)))_n$ for $n_0=n_0(h)$ large enough so that $h(\Lambda(n_0+n))\supset\Lambda(n)$ for each $n$). It is again clear from the definition that $\mathcal{C}_0$ is $G^d$ invariant. 

We proceed to check that $\sigma^{\eta,\Dob}$ is $G^d$ covariant on $\mathcal{C}_0$. Note that, for $\Lambda\subset\Z^d$ and an automorphism $a$ of $\Z^D$, the set of ground configurations in $\Omega^{a(\Lambda)\times\Z,\rho^{\Dob}}$ for $a(\eta)$ equals $a$ applied to the set of ground configurations in $\Omega^{\Lambda\times\Z,\rho^{\Dob}}$ for $\eta$. In particular, if $\sigma^{\eta,\Lambda,\Dob}$ is well defined (i.e., uniqueness holds) then also $\sigma^{a(\eta),a(\Lambda),\Dob}$ is well defined and we have
\begin{equation}
\label{eq:ground configurations g relation}\sigma^{a(\eta),a(\Lambda),\Dob}_x=\sigma^{\eta,\Lambda,\Dob}_{a x} = a(\sigma^{\eta,\Lambda,\Dob})_x,\quad x\in\Z^D.
\end{equation} 
Now let $g\in G^d$ and $\eta\in\mathcal{C}_0$. For each $x\in\Z^D$,
\begin{equation}\label{eq:G d covariance identity}
g(\sigma^{\eta,\Dob})_x = \sigma^{\eta,\Dob}_{g x} = \lim_{n\to\infty}\sigma^{\eta,\Lambda(n),\Dob}_{g x}=\lim_{n\to\infty}\sigma^{g(\eta),g(\Lambda(n)),\Dob}_x = \sigma^{g(\eta),\Dob}_x,
\end{equation}
where the second and last equality use the definition of $\mathcal{C}_0$ and the third equality uses~\eqref{eq:ground configurations g relation}. Thus $\sigma^{\eta,\Dob}$ is a $G^d$-covariant ground configuration defined on $\mathcal{C}_0$.

It remains to define a $G^d$ invariant set of coupling fields $\mathcal{C}\subset\mathcal{C}_0$ with $\P(\mathcal{C})=1$ such that $\sigma^{\eta,\Dob}$ is a non-constant $G^d$-covariant ground configuration defined on $\mathcal{C}$. Define
\begin{alignat}{1}
&\mathcal{C}_{\text{non-const}}:=\{\eta\colon \sigma^{\eta,\Dob}\text{ is not a constant configuration}\},\\
&\mathcal{C}_{\text{ground}}:=\{\eta\colon\sigma^{\eta,\Dob}\text{ is a ground configuration for the coupling field $\eta$}\}.
\end{alignat}
Set
\begin{equation}
\mathcal{C} := \mathcal{C}_0\cap\mathcal{C}_{\text{non-const}}\cap \mathcal{C}_{\text{ground}}.
\end{equation}
Then $\P(\mathcal{C}_{\text{non-const}})=1$ by the estimate~\eqref{eq:cor_localization at v k}, and $\P(\mathcal{C}_{\text{ground}})=1$ since $\sigma^{\eta,\Dob}$ is the pointwise limit of $\sigma^{\eta,\Lambda(n),\Dob}$ and each of these configurations is, almost surely, a ground configuration in $\Omega^{\Lambda(n)\times\Z,\rho^{\Dob}}$ by Lemma~\ref{lem:semi infinite volume ground configuration}. The set $\mathcal{C}$ is $G^d$ invariant since the covariance identity~\eqref{eq:G d covariance identity} holds on $\mathcal{C}_0$ (and it maps ground configurations for $\eta$ to ground configurations for $g(\eta)$). This finishes the proof.

\subsection{Proof of Theorem~\ref{thm:decay of correlations and rate of convergence}}
The estimate~\eqref{eq:rate of convergence} is a direct consequence of Corollary~\ref{cor:weak dependence on domain}, applied with $\Lambda^1 = \Lambda(n)$ and $\Lambda^2=\Lambda$, by taking $n$ to infinity and applying the convergence result of Theorem~\ref{thm:convergence}.

We proceed to establish~\eqref{eq:correlation decay estimate}. Let $k := \lfloor\frac{\|u-v\|_\infty}{2}\rfloor$, so that $u+\Lambda(k)$ is disjoint from $v+\Lambda(k)$. By~\eqref{eq:rate of convergence} (after a translation by $u$ and by $v$),
\begin{alignat*}{1}
&\P\left(\sigma^{\eta,\Dob}|_{(u+\Lambda(L))\times\Z}\not\equiv \sigma^{\eta,u+\Lambda(k),\Dob}|_{(u+\Lambda(L))\times\Z}\right)\le C\exp \left(-c(\nu^{\parallel},\nu^{\perp},d)\left(k - L\right)^{\frac{d-2}{d}} \right),\\
&\P\left(\sigma^{\eta,\Dob}|_{(v+\Lambda(L))\times\Z}\not\equiv \sigma^{\eta,v+\Lambda(k),\Dob}|_{(v+\Lambda(L))\times\Z}\right)\le C\exp \left(-c(\nu^{\parallel},\nu^{\perp},d)\left(k - L\right)^{\frac{d-2}{d}} \right).
\end{alignat*}
The estimate~\eqref{eq:correlation decay estimate} follows (perhaps with a smaller $c>0$) as $\sigma^{\eta,u+\Lambda(k),\Dob}$ and $\sigma^{\eta,v+\Lambda(k),\Dob}$ are independent (as they are functions of disjoint subsets of the disorder).

Lastly, the fact that $(\eta,\sigma^{\eta,\Dob})$ is $G^d$-invariant is a rephrasing of the fact that $\sigma^{\eta,\Dob}$ is $G^d$-covariant (proved in Corollary~\ref{cor:covariant Dobrushin configuration}). The fact that it has a trivial $\Z^d$-tail sigma algebra is a consequence of~\eqref{eq:rate of convergence}, since for each finite $\Lambda\subset\Z^d$, $\sigma^{\eta,\Lambda,\Dob}$ is a function of $(\eta_e)$ for the edges $e$ above $\Lambda$, and hence $\sigma^{\eta,\Lambda,\Dob}$ is independent of the $\Z^d$-tail sigma algebra of $(\eta,\sigma^{\eta,\Dob})$. It is standard that an invariant and tail-trivial process is ergodic.

\subsection{Proof of Lemma~\ref{lem:weak dependence on boundary conditions}}
\label{sec:proof of weak dependence on boundary lemma}

For a configuration $\sigma:{\mathbb Z}^{d+1}\to\{-1,1\}$, define the graph $G_{\sigma}$ and the function $I_{\sigma}\colon {\mathbb Z}^d \to \Z \cup\{{\rm ``layered"}\}$ as in Section \ref{sec:s_0}.

\begin{lemma}
\label{lem:surrounds}
Let $\sigma:{\mathbb Z}^{d+1}\to\{-1,1\}$ be a configuration such that $\sigma=\rho^{\Dob}$ at all but finitely many points of ${\mathbb Z}^{d+1}$ and let $u_0\in{\mathbb Z}^d$ such that there exists $k_0\in \Z$ for which $\sigma_{(u_0,k_0)}\neq \rho ^{\Dob}_{(u_0,k_0)}$.
Then, there exists a connected component $A$ of $G_{\sigma}$ such that $u_0\in A\cup{\rm in}(A)$.
\end{lemma}

\begin{proof}
Let $U_0$ be the connected component of the set $\{u\in{\mathbb Z}^d\colon \sigma_{(u,k_0)}= \sigma_{(u_0,k_0)}\}$ containing $u_0$, let $U:=U_0\cup{\rm in}(U_0)$ and let $A_0:=\partial^{\rm in}U\cup\partial^{\rm out}U$. 

Recall the definition of the graph ${\mathcal G}_d$ from Lemma \ref{lem:Timar1}. Obviously, $U$ and ${\mathbb Z}^d\setminus U$ are both connected. Hence, by Lemma \ref{lem:Timar1}, the set $\{\{u,v\}\in E({\mathbb Z}^d)\colon (u,v)\in\partial U\}$ is connected in ${\mathcal G}_d$ and it readily follows that the set $A_0$ is connected. 

Clearly, $\partial^{\rm in}U\subseteq\partial^{\rm in}U_0\subseteq V_{\sigma}$ and $\partial^{\rm out}U\subseteq\partial^{\rm out}U_0\subseteq V_{\sigma}$, and hence $A_0\subseteq V_{\sigma}$. 
Let $A$ be the connected component of $G_{\sigma}$ such that $A_0\subseteq A$.

The set $\{u\in{\mathbb Z}^d\colon \sigma_{(u,k_0)}= \sigma_{(u_0,k_0)}\}\subseteq\{u\in{\mathbb Z}^d\colon \sigma_{(u,k_0)}\neq \rho^{\Dob}_{(u,k_0)}\}$ is finite, hence $U_0$ and $U$ are finite as well.
Therefore, since $u_0\in U_0\subseteq U$, it follows that every path from $u_0$ to $\infty$ intersects the set $\partial^{\rm in}U\subseteq A_0\subseteq A$ (as well as the set $\partial^{\rm out}U\subseteq A_0\subseteq A $) and hence $u_0\in A\cup{\rm in}(A)$. 
\end{proof}

\begin{lemma}\label{lem:path}
Under the assumptions of Lemma \ref{lem:weak dependence on boundary conditions}, 
there exists a path $u_1,\dots,u_n$ of points in $\Z^{d}$ starting in $u_1 \in \partial^{\rm out} \Lambda(L_0)$ and ending in $u_n\in \partial ^{\rm in}\Lambda(L_1)$
such that $u_{j-1}\sim u_j$ for every $1<j\leq n$, and for every $1\leq j\leq n$ there is an integer $k$ such that $\sigma^{\eta,\Lambda^1,\Dob}_{(u_j,k)} \neq\sigma^{\eta,\Lambda^2,\Dob}_{(u_j,k)}$.
\end{lemma}

\begin{proof}
Let $$U:=\{u\in\Lambda(L_1)\colon\forall k\in{\mathbb Z} \text{ it holds that } \sigma^{\eta,\Lambda^1,\Dob}_{(u_j,k)} =\sigma^{\eta,\Lambda^2,\Dob}_{(u_j,k)}\}.$$
By way of contradiction, assume that $\Lambda(L_0)\subseteq U\cup{\rm in}(U)$.
Consider the configurations $\tilde{\sigma}_1,\tilde{\sigma}_2\colon{\mathbb Z}^{d+1}\to\{-1,1\}$ defined as follows. For every $u\in{\mathbb Z}^d$ and $k\in{\mathbb Z}^d$,
$$
(\tilde{\sigma}_1)_{(u,k)}=\begin{cases}
\sigma^{\eta,\Lambda^2,\Dob}_{(u,k)} & u\in {\rm in}(U),\\
\sigma^{\eta,\Lambda^1,\Dob}_{(u,k)} & \text{otherwise},\end{cases}
\qquad\qquad (\tilde{\sigma}_2)_{(u,k)}=\begin{cases}
\sigma^{\eta,\Lambda^1,\Dob}_{(u,k)} & u\in {\rm in}(U),\\
\sigma^{\eta,\Lambda^2,\Dob}_{(u,k)} & \text{otherwise}.\end{cases}
$$
For any $i\in\{1,2\}$, clearly $\tilde{\sigma_i}\in\Omega^{\Lambda^i,\Dob}$ and hence $\mathcal{H}^{\eta,\Lambda^i}(\tilde{\sigma_i})\geq \mathcal{H}^{\eta,\Lambda^i}(\sigma^{\eta,\Lambda^i,\Dob})$.
Therefore, since it is easy to see that
$$
\mathcal{H}^{\eta,\Lambda^1}(\tilde{\sigma_1})+ \mathcal{H}^{\eta,\Lambda^2}(\tilde{\sigma_2})
=\mathcal{H}^{\eta,\Lambda^1}(\sigma^{\eta,\Lambda^1,\Dob})+ \mathcal{H}^{\eta,\Lambda^2}(\sigma^{\eta,\Lambda^2,\Dob}),
$$
it follows that $\mathcal{H}^{\eta,\Lambda^1}(\tilde{\sigma_1})= \mathcal{H}^{\eta,\Lambda^1}(\sigma^{\eta,\Lambda^1,\Dob})$ (as well as $\mathcal{H}^{\eta,\Lambda^2}(\tilde{\sigma_2})\geq \mathcal{H}^{\eta,\Lambda^2}(\sigma^{\eta,\Lambda^2,\Dob})$),
in contradiction to the uniqueness of $\sigma^{\eta,\Lambda^1,\Dob}$.

Hence, $\Lambda(L_0)\nsubseteq U\cup{\rm in}(U)$ and the claim follows.
\end{proof}

We proceed to prove Lemma \ref{lem:weak dependence on boundary conditions}.
For any $i\in\{1,2\}$ denote, for brevity, $\sigma_i:=\sigma^{\eta,\Lambda^i, \Dob}$ and let $\mathcal{C}_i$
be the collection of all connected components of $G_{\sigma_i}$. 
For the path $u_1,\dots,u_n$ guaranteed by Lemma \ref{lem:path}, note that for every $1\leq j\leq n$ it holds that $(\sigma_1)_{(u_j,k)} \neq \rho^{\Dob}_{(u_j,k)}$ or $(\sigma_2)_{(u_j,k)} \neq \rho^{\Dob}_{(u_j,k)}$ and hence, by Lemma \ref{lem:surrounds}, there is $A\in\mathcal{C}_1\cup\mathcal{C}_2$ 
such that $u_j\in A\cup{\rm in}(A)$.
Let $\mathcal{M}\subseteq\mathcal{C}_1\cup\mathcal{C}_2$ be a minimal collection (with respect to inclusion) such that for every $1\leq j\leq n$ there is $A\in\mathcal{M}$ such that $u_j\in A\cup{\rm in}(A)$.
For any $i\in\{1,2\}$, let 
$A^{(i)}:=\bigcup_{A\in\mathcal{M}\cap\mathcal{C}_i}A$
and
\begin{align*}
B^{(i)}_{\infty}:= \{u\in \Z^{d}: \text{there is a path from } u \text{ to } \infty \text{ that does not intersect } A^{(i)} \}.
\end{align*}

Consider $v\in{\mathbb Z}^d\setminus(A^{(i)}\cup B^{(i)}_{\infty})$ for $i\in\{1,2\}$.
Clearly, there is $A\in\mathcal{M}\cap\mathcal{C}_i$
such that $v\in{\rm in}(A)$, and this $A$ is unique by the second and third parts of Lemma \ref{lem:unique smallest contour} and the minimality of $\mathcal M$.
Let $\tilde{\tau}_i(v):=\tilde{I}_{\sigma_i}(v;A)$ (see Observation \ref{obs:svC}).
We complete the definition of a ``pre-shift" $\tilde{\tau}_i\colon {\mathbb Z}^d \to \Z \cup\{{\rm ``layered"}\}$ by setting $\tilde{\tau}_i(v)=I_{\sigma_i}(v)$ for $v\in A^{(i)}$ and $\tilde{\tau}_i(v)=0$ for $v\in B_{\infty}^{(i)}$.

We turn the "pre-shift" $\tilde{\tau}_i$ for $i\in\{1,2\}$ into a shift ${\tau}_i\colon {\mathbb Z}^d \to{\mathbb Z}$ exactly as done in Section \ref{sec:s}. 
The arguments presented in the proof of Proposition \ref{prop:total_variation_bound_s} imply, by using \eqref{eq:G_H-GE}, that for any $i\in\{1,2\}$,
\begin{equation}\label{eq:Gsi}
G^{\eta,\Lambda^i}({\tau}_i)\geq 2\alpha^{\parallel} \left( \mathcal{L}^{\parallel}_{A^{(i)}}(\sigma_i)-|A^{(i)}|\right)+2\alpha^{\perp}\mathcal{L}^{\perp}_{A^{(i)}}(\sigma_i)
\end{equation}
and consequently,
\begin{equation}\label{eq:Gsi_TVsi}
G^{\eta,\Lambda^i}({\tau}_i)\geq 2\alpha^{\perp}\TV({\tau}_i).
\end{equation}

\begin{lemma}\label{lem:R_GG}
For any $i\in\{1,2\}$,
\begin{equation*}
R({\tau}_i)< \frac{24}{\min \{\alpha^{\parallel}, \alpha^{\perp} \} d} \max\left\{G^{\eta,\Lambda^{1}}({\tau}_1),G^{\eta,\Lambda^{2}}({\tau}_2)\right\}. 
\end{equation*}
\end{lemma}

\begin{proof} 
We first show that the set $A^{(1)}\cup A^{(2)}=\bigcup_{A\in\mathcal{M}}A$ is connected.
By way of contradiction, assume that there is a set $\emptyset\neq\Gamma\subsetneq \bigcup_{A\in\mathcal{M}}A$ such that ${\rm dist}(\Gamma,\left(\bigcup_{A\in\mathcal{M}}A\right)\setminus\Gamma)>1$.
Since every $A\in\mathcal{M}$ is connected, there is necessarily $\emptyset\neq\mathcal{M}_0\subsetneq\mathcal{M}$ such that $\Gamma=\bigcup_{A\in\mathcal{M}_0}A$ and $\left(\bigcup_{A\in\mathcal{M}}A\right)\setminus\Gamma=\bigcup_{A\in\mathcal{M}\setminus\mathcal{M}_0}A$.
Note that it follows that ${\rm dist}(A,A')>1$ for every $A\in\mathcal{M}_0$ and $A'\in\mathcal{M}\setminus\mathcal{M}_0$.
The minimality of $\mathcal M$ implies that neither $\{u_j\}_{j=1}^n\subseteq\bigcup_{A\in\mathcal{M}_0}( A\cup{\rm in}(A))$ nor $\{u_j\}_{j=1}^n\subseteq\bigcup_{A\in\mathcal{M}\setminus\mathcal{M}_0}( A\cup{\rm in}(A))$.
Since $\{u_j\}_{j=1}^n\subseteq\bigcup_{A\in\mathcal{M}}( A\cup{\rm in}(A))$, it follows that
there are $1\leq j,j'\leq n$, $A\in\mathcal{M}_0$ and $A'\in\mathcal{M}\setminus\mathcal{M}_0$ such that $u_j\in( A\cup{\rm in}(A))$, $u_{j'}\in{\rm in}(A')\cup A'$ and $\lvert j-j'\rvert=1$, therefore $\lVert u_j-u_{j'}\rVert_1=1$ and hence ${\rm dist}(A\cup{\rm in}(A),A'\cup{\rm in}(A'))\leq 1$. 
Lemma \ref{lem:unique smallest contour} implies that either
$A\cup{\rm in}(A)\subsetneq A'\cup{\rm in}(A')$ or $A'\cup{\rm in}(A')\subsetneq A\cup{\rm in}(A)$, contradicting the minimality of $\mathcal M$.

Lemma \ref{lem:net0} implies that for any $i\in\{1,2\}$ there is a set $S_i\subseteq A^{(i)}$ such that 
$A^{(i)}\subseteq\bigcup_{a\in S_i}{\mathcal B}_4(a)$ and, by using \eqref{eq:Gsi},
\begin{align*}
\lvert S_i\rvert&<\sum_{A\in
\mathcal{M}\cap\mathcal{C}_i}\frac{1}{d}\left(\mathcal{L}^{\parallel}_A(\sigma_i)-|A|+\mathcal{L}^{\perp}_A(\sigma_i)\right)\\
&=\frac{1}{d}\left(\mathcal{L}^{\parallel}_{A^{(i)}}(\sigma_i)-|A^{(i)}|+\mathcal{L}^{\perp}_{A^{(i)}}(\sigma_i)\right)<\frac{1}{2\min \{\alpha^{\parallel}, \alpha^{\perp} \} d}G^{\eta,\Lambda^{i}}({\tau}_i).
\end{align*}

For any $i\in\{1,2\}$, the definition of ${\tau}_i$ implies, as in Sections \ref{sec:s_0} and \ref{sec:s}, that every level component of ${\tau}_i$ intersects $A^{(i)}$, and therefore intersects $A^{(1)}\cup A^{(2)}$.

Hence, by \eqref{eq:ham+}, 
since $A^{(1)}\cup A^{(2)}\subseteq\bigcup_{a\in S_1\cup S_2}{\mathcal B}_4(a)$,
$$
R({\tau}_i)<20\lvert S_1\cup S_2\rvert+8\lvert\mathcal{LC}({\tau}_i)\rvert<\frac{10}{\min \{\alpha^{\parallel}, \alpha^{\perp} \}d }\left(G^{\eta,\Lambda^{1}}({\tau}_1)+G^{\eta,\Lambda^{2}}({\tau}_2)\right)+8\lvert\mathcal{LC}({\tau}_i)\rvert
$$
and the result follows since
\begin{equation*}
\lvert\mathcal{LC}({\tau}_i)\rvert\leq\frac{1}{d}\TV({\tau}_i)\leq \frac{1}{2\alpha^{\perp}d}G^{\eta,\Lambda^i}({\tau}_i).
\end{equation*}
by \eqref{eq:no_LC} 
and \eqref{eq:Gsi_TVsi}.
\end{proof}

\begin{lemma}\label{lem:GGLL}
It holds that
\begin{equation*}
\max\left\{G^{\eta,\Lambda^1}({\tau}_1),G^{\eta,\Lambda^2}({\tau}_2)\right\}\geq\frac{\min\{\alpha^{\parallel},\,\alpha^{\perp}\}}{4}(L_1-L_0)^{1-\frac{1}{d}}.    
\end{equation*}
\end{lemma}

\begin{proof}
For every finite set $A\subset{\mathbb Z}^d$, by \eqref{eq:bdry_small},
$$
\lvert A\rvert\geq\frac{1}{2d}\lvert \partial({\rm in}(A))\rvert\geq\lvert{\rm in}(A)\rvert^{1-\frac{1}{d}}
$$
and therefore
$$
2\lvert A\rvert\geq\lvert{\rm in}(A)\rvert^{1-\frac{1}{d}}+\lvert A\rvert\geq \lvert{\rm in}(A)\rvert^{1-\frac{1}{d}}+\lvert A\rvert^{1-\frac{1}{d}}\geq \left(\lvert{\rm in}(A)\rvert+\lvert A\rvert\right)^{1-\frac{1}{d}}=\lvert{\rm in}(A)\cup  A\rvert^{1-\frac{1}{d}}.
$$
Hence, for any $i\in\{1,2\}$, by \eqref{eq:Gsi} and Observation \ref{obs:CLD},
\begin{align*}
\frac{2}{\min\{\alpha^{\parallel},\,\alpha^{\perp}\}}G^{\eta,\Lambda^i}({\tau}_i)&\geq 4\left(\mathcal{L}^{\parallel}_{A^{(i)}}(\sigma_i)-|A^{(i)}|+\mathcal{L}^{\perp}_{A^{(i)}}(\sigma_i)\right)\\
&=\sum_{A\in\mathcal{M}\cap\mathcal{C}_i}
4\left(\mathcal{L}^{\parallel}_A(\sigma_i)-|A|+\mathcal{L}^{\perp}_A(\sigma_i)\right)\\
&\geq \sum_{A\in
\mathcal{M}\cap\mathcal{C}_i}
2\lvert A\rvert\geq \sum_{A\in
\mathcal{M}\cap\mathcal{C}_i}
\lvert{\rm in}(A)\cup  A\rvert^{1-\frac{1}{d}}.
\end{align*}
Therefore,
\begin{align*}
\frac{4}{\min\{\alpha^{\parallel},\,\alpha^{\perp}\}}&\max\left\{G^{\eta,\Lambda^1}({\tau}_1),G^{\eta,\Lambda^2}({\tau}_2)\right\}\geq\frac{2}{\min\{\alpha^{\parallel},\,\alpha^{\perp}\}}\left(G^{\eta,\Lambda^1}({\tau}_1)+G^{\eta,\Lambda^2}({\tau}_2)\right)\\
\geq& \sum_{A\in\mathcal{M}}
\lvert{\rm in}(A)\cup  A\rvert^{1-\frac{1}{d}}\geq \left(\sum_{A\in\mathcal{M}}
\lvert{\rm in}(A)\cup  A\rvert\right)^{1-\frac{1}{d}}
\end{align*}
and the result follows since $\{u_1,\ldots,u_n\}\subseteq\bigcup_{A\in
\mathcal{M}}
\left( A\cup{\rm in}(A)\right)$, and hence
\begin{equation*}
\sum_{A\in\mathcal{M}}
\lvert{\rm in}(A)\cup  A\rvert\geq\Big\lvert\bigcup_{A\in\mathcal{M}}
\left( A\cup{\rm in}(A)\right)\Big\rvert\geq \lvert\{u_1,\ldots,u_n\}\rvert\geq L_1-L_0.
\qedhere\end{equation*}
\end{proof}

To conclude the proof of Lemma \ref{lem:weak dependence on boundary conditions}, take $i\in\{1,2\}$ for which $\max\{G^{\eta,\Lambda^1}({\tau}_1),G^{\eta,\Lambda^2}({\tau}_2)\}=G^{\eta,\Lambda^i}({\tau}_i)$.
Then ${\tau}_i$ is admissible by \eqref{eq:Gsi_TVsi} and Lemma \ref{lem:R_GG}, and
$$G^{\eta,\Lambda^i}({\tau}_i)\geq\frac{\min\{\alpha^{\parallel},\,\alpha^{\perp}\}}{4}(L_1-L_0)^{1-\frac{1}{d}}$$
by Lemma \ref{lem:GGLL}.

\section{Discussion and open problems}\label{sec:discussion and open problems}

\subsection{Localization of Dobrushin interface, roughening transitions and non-constant ground configurations} Question~\ref{q:localization of Dobrushin interface} asks whether the interface formed under Dobrushin boundary conditions remains localized uniformly in the volume. This question and its variants have received significant attention in the physics literature. 
As an approximation to the true interface (reminiscent of the disordered SOS model~\eqref{eq:disordered SOS Hamiltonian} but involving further approximations), it is suggested to study the ground configurations with zero boundary values of a \emph{real-valued} height function $\varphi:\Z^d\to\R$ whose energy is given by the formal ``disordered Gaussian free field (GFF)'' Hamiltonian
\begin{equation}\label{eq:disordered GFF Hamiltonian}
H^{\text{GFF}, \zeta}(\varphi):=\sum_{\{u,v\}\in E(\Z^d)} (\varphi_u - \varphi_v)^2 + \sum_{v\in\Z^d} \zeta_{v,\varphi_v}
\end{equation}
where $\zeta:\Z^d\times\R\to\R$ is an environment describing the quenched disorder, which is chosen with $(\zeta_{v,\cdot})_v$ independent and $t\mapsto\zeta_{v,t}$ having short-range correlations for each $v$ (and possibly also light tails). It is predicted that this height function delocalizes with power-law fluctuations in dimensions $d=1,2,3$, delocalizes with sub-power-law fluctuations in dimension $d=4$ and remains localized in dimensions $d\ge 5$. These predictions are put on a rigorous footing in the forthcoming work~\cite{DEHP23}. More precisely, it is predicted that on the cube $\{-L,\ldots,L\}^d$, the height function fluctuates to height $L^{2/3}$ in dimension $d=1$~\cites{HH85, HHF85, KN85}, to height $L^{0.41\ldots}$ in dimension $d=2$, to height $L^{0.22\ldots}$ in dimension $d=3$~\cites{F86, M95} and to height $(\log L)^{0.41\ldots}$ in dimension $d=4$~\cite{EN98}.

It is predicted, however, that the model~\eqref{eq:disordered GFF Hamiltonian} may display different behavior when restricted to \emph{integer-valued} height functions. Specifically, while it is believed that the ground configurations with zero boundary values are still delocalized in dimensions $d=1,2$, with the same power laws as the real-valued versions, and still localized in dimensions $d\ge 5$, a change of behavior occurs for $d=3,4$~\cites{HH85, F86, BG92, EN98}. In dimension $d=3$ a \emph{roughening transition} takes place in the disorder concentration: the height function is localized for sufficiently concentrated disorder and delocalized otherwise, having logarithmic fluctuations at the critical disorder concentration and power-law fluctuations, of the same order as the real-valued version, for less concentrated disorder~\cite{EN98}. In contrast, it is indicated that no such transition takes place in dimension $d=4$, where the height function is \emph{localized} at all disorder concentrations~\cite{EN98}. These predictions are also believed to govern the fluctuations of the disordered SOS model~\eqref{eq:disordered SOS Hamiltonian}, and the Dobrushin interface of the disordered ferromagnet on $\Z^D$, with our standard substitution $D=d+1$. Our work justifies the fact that the Dobrushin interface of the disordered ferromagnet is localized in dimensions $d\ge 3$ for sufficiently concentrated disorder (the analogous fact for the disordered SOS model is established in~\cite{BK}). It would be very interesting to extend it to establish the predicted roughening transition in dimension $d=3$ and the predicted localization for all disorder concentrations in dimensions $d\ge 4$. It would also be interesting to prove delocalization in dimension $d=2$ (and especially to prove power-law delocalization). We expect the methods of Aizenman--Wehr~\cite{AW90}, or their quantitative extension in~\cite{DHP21}, to be relevant in dimension $d=2$, as was the case for the disordered SOS model~\cite{BK96}. Power-law delocalization in dimension $d=1$ is proved by Licea--Newman--Piza~\cite{LNP96} (see also~\cites{DEP22, DEP23}).

Related to the above, we mention that a version of the model~\eqref{eq:disordered GFF Hamiltonian} in which the disorder is \emph{linear}, i.e., $\zeta_{v,\varphi_v} = \bar{\zeta}_v\varphi_v$ with the $(\bar{\zeta}_v)$ independently sampled from the Gaussian dsitribution $N(0,\lambda^2)$, is studied in~\cite{DHP23}. The (real-valued) model is exactly solvable and behaves similarly to~\eqref{eq:disordered GFF Hamiltonian} in the sense that it also exhibits power-law delocalization in dimensions $d=1,2,3$, square-root logarithmic delocalization when $d=4$ and is localized when $d\ge 5$. It is conjectured in~\cite{DHP23} that the \emph{integer-valued} version of this model should also exhibit a roughening transition: the model should transition from a localized to a delocalized regime as $\lambda$ increases in dimension $d=3$ (whether this also occurs for $d=4$ is unclear). The localization part of the transition is established in~\cite{DHP23}.

Lastly, Question~\ref{q:non-constant ground configurations} asks whether the disordered ferromagnet admits non-constant ground configurations. This is certainly the case whenever the Dobrushin interface is localized, as in this work. However, it may still be the case even when the Dobrushin interface is delocalized, as it may be that other boundary conditions on $\{-L,\ldots, L\}^D$ (possibly depending on the disorder~$\eta$) may lead to an interface passing near the origin. The fact that the predicted roughening exponent is relatively small already for $d=2$ (the prediction there is $\approx0.41$), together with the fact that there are more possibilities for the boundary conditions as $d$ grows leads us to believe that non-constant ground configurations will indeed exist for all $d\ge 2$ (see~\cite[Section 4.5.1]{ADH17} for a related heuristic of Newman for the $d=1$ case).

\subsection{Positive temperature and the random-field Ising model}\label{sec:positive temperature and RFIM} 
Our main result (Theorem~\ref{thm:non-constant ground configuration}) states that the disordered ferromagnet admits non-constant ground configurations in dimension $D\ge 4$ when the coupling constants are sampled independently from a sufficiently concentrated distribution. This is achieved by proving that the interface formed under Dobrushin boundary conditions remains localized uniformly in the volume (Theorem~\ref{thm:localization}). It is natural to ask for an extension of these results to the disordered ferromagnet at low, positive temperatures (instead of asking about non-constant ground configurations, one may ask whether there exist Gibbs states other than mixtures of the infinite-volume limits under constant boundary conditions). Such an extension is established for the disordered SOS model~\eqref{eq:disordered SOS Hamiltonian} by Bovier--K\"ulske~\cite{BK} and we believe that it holds also for the disordered ferromagnet. We also expect our methods to be relevant to the proof of such a result, though one obstacle which one will need to overcome is the existence of \emph{bubbles} in the Ising configuration: finite islands of one sign completely surrounded by spins of the other sign. Such bubbles occur at any non-zero temperature. Additional tools from the work of Dobrushin~\cite{Dob}, such as the use of cluster expansions and the notion of ``groups of walls'' may be of help here too.

Another model of interest is the \emph{random-field} Ising model~\eqref{eq:RFIM Hamiltonian}. It is natural to wonder whether our results (and their possible low-temperature extensions) hold also for the Dobrushin interface in the random-field Ising model. On the level of a disordered SOS approximation to the Dobrushin interface, this is stated to be true, in dimensions $D\ge 4$ and for sufficiently weak disorder, by Bovier--K\"ulske~\cite{BK}, following an earlier analysis of a hierarchical version~\cite{BK92}. We again believe that our methods will be relevant to the random-field Ising model case, but point out that an extra complication arising in this case compared to the disordered ferromagnet is that bubbles appear already at zero temperature.

\subsection{The set of non-constant covariant ground configurations}\label{sec:set of non-constant covariatn ground configurations} Theorem~\ref{thm:covariant ground configurations isotropic} shows the existence of a non-constant $G^{D-1}$-covariant ground configuration.  Additional configurations with these properties may be obtained from a given one by the following recipe: Suppose $\eta\mapsto\sigma(\eta)$ is a non-constant $G^{D-1}$-covariant ground configuration. For each integer $k$, define a configuration $\eta\mapsto\sigma^k(\eta)$ by the relation
\begin{equation}
T^k(\sigma^k(\eta)):=\sigma(T^k(\eta))
\end{equation}
where $T^k$ is the automorphism of $\Z^D$ given by $T^k(x):=x+ke_D$ (with $e_D$ being the last coordinate vector) and the action of automorphisms on coupling fields and configurations is given by~\eqref{eq:automorphism actions}. It is straightforward that $\eta\mapsto\sigma^k(\eta)$ is then also a non-constant $G^{D-1}$-covariant ground configuration. 

Suppose the coupling constants $(\eta_{\{x,y\}})$ are sampled independently from a disorder distribution which is non-atomic and has finite mean. We claim that the mappings $(\eta\mapsto\sigma^k(\eta))_k$ are all distinct (even modulo zero probability events). Indeed, to obtain a contradiction, suppose that $\sigma^{k+m}= \sigma^k$ almost surely, for some integers $k$ and $m\neq 0$. Then also $\sigma^m = \sigma$ almost surely. But this implies that $\eta\mapsto\sigma(\eta)$ is a $\Z^{D,m}$-covariant ground configuration, where $\Z^{D,m}$ is the group of translations by vectors of the form $x = (x_1,\ldots, x_D)\in\Z^D$ with $x_D$ divisible by~$m$. Recall that Wehr--Wasielak~\cite{WW16} prove that there are no non-constant $\Z^D$-covariant ground configurations (under the above assumptions on $\eta$). A minor extension of their proof also rules out non-constant $\Z^{D,m}$-covariant ground configurations, contradicting the fact that $\sigma$ is non-constant. 

It is natural to ask whether there is a \emph{unique} family $(\eta\mapsto\sigma^k(\eta))_{k\in\Z}$ of non-constant $G^{D-1}$-covariant ground configurations. We believe that the answer is positive under the assumptions of Theorem~\ref{thm:non-constant ground configuration}.

We also pose the following, somewhat related, problem. Theorem~\ref{thm:convergence} proves that for every sequence $(\Lambda_n)$ of finite subsets of $\Z^d$, with $\Lambda_n\supset\{-n,\ldots, n\}^d$ for each $n$, it holds almost surely that $\sigma^{\eta,\Lambda_n,\Dob}\to\sigma^{\eta,\Dob}$ pointwise. Are there exceptional sequences? That is, is there a \emph{random} sequence $(\Lambda_n)$ of subsets of $\Z^d$, converging to $\Z^d$, for which, with positive probability, the pointwise convergence $\sigma^{\eta,\Lambda_n,\Dob}\to\sigma^{\eta,\Dob}$ fails? We expect that the answer is negative under the assumptions of Theorem~\ref{thm:localization}.

\subsection{Tilted interfaces} Our work investigates the interface formed in the disordered ferromagnet's ground configuration when imposing the Dobrushin boundary conditions $\rho^{\Dob}_{(v,k)} = \sign(k - 1/2)$. It is also of interest to study the interfaces formed under other boundary conditions. For instance, one may consider ``tilted Dobrushin-type boundary conditions'' of the form $\rho^{\Dob, y}_x:=\sign(x\cdot y - 1/2)$, corresponding to a flat interface orthogonal to the vector $y\in\Z^D$ (so that $\rho^{\Dob}$ corresponds to $y=(0,\ldots,0,1)$). In analogy with predictions for the pure Ising model, we expect that whenever $y$ is not a multiple of one of $e_1,\ldots, e_D$ (the standard basis vectors) then the fluctuations of these tilted interfaces are of the same order as those of the real-valued disordered GFF model~\eqref{eq:disordered GFF Hamiltonian} discussed above, except perhaps in the critical dimension $d=4$. In particular, they are delocalized in dimensions $d\le 3$ and localized in dimensions $d\ge 5$ (a discussion of some simulation results is in~\cite[Section 7.2.3]{ADMR01}).

\subsection{Higher codimension surfaces}\label{sec:higher codimension surfaces}

The Dobrushin interface studied in this paper may be thought of as a $d$-dimensional surface embedded in a $D=d+1$ dimensional space. It is also of interest to consider surfaces of higher codimension, i.e., $d$-dimensional surfaces embedded in a $D=d+n$ dimensional space. Generalizing an approach of Borgs~\cite{B84}, let us describe how such surfaces arise in the context of (generalized) Ising lattice gauge theories.

Let $d,n\ge 1$ and set $D:=d+n$. An $m$-face in the lattice $\Z^D$ is a subset of the form $x+\{0,1\}^{I}\times\{0\}^{[D]\setminus I}$ for some $x\in\Z^D$ and $I\subset[D]$ with $|I|=m$ (a $0$-face is a vertex, a $1$-face is an edge, etc.). Denote the set of $m$-faces of $\Z^D$ by $F_m$. We consider Ising lattice gauge theories on $(n-1)$-faces, defined as follows. A configuration is a function $\sigma:F_{n-1}\to\{-1,1\}$. We also define
\begin{equation}
\sigma(f_n):=\prod_{\substack{f_{n-1}\in F_{n-1}\\f_{n-1}\subset f_n}}\sigma_{f_{n-1}}
\end{equation}
for an $n$-face $f_n$. The formal Hamiltonian is
\begin{equation}
H^{\text{gauge}}(\sigma) := -\sum_{f_n\in F_n}\sigma(f_n).
\end{equation}
The \emph{defects} of $\sigma$ are the $n$-faces $f_n$ satisfying $\sigma(f_n) = -1$. We think of the defect set as being dual to a $d$-dimensional surface (e.g., for $n=1$, the case of the standard Ising model, the defects are dual to the domain walls separating $-1$ and $1$). We wish to put the model under specific, Dobrushin-like, boundary conditions which will force such a surface through the volume. To this end, write vertices of $\Z^D$ as $x=(v,k)$ with $v=(v_1,\ldots, v_d)\in\Z^d$ and $k=(k_1,\ldots, k_n)\in\Z^n$. The Dobrushin-like boundary conditions are then
\begin{equation}
\rho^{\text{surface}}_{f_{n-1}}:=\begin{cases}-1&{\scriptstyle f_{n-1} = (v,k) + C\text{ with }v\in\Z^d,\, k=(k_1,0,\ldots, 0)\text{ having } k_1\le 0\text{ and }
C=\{0\}^{[d+1]}\times\{0,1\}^{[d+n]\setminus[d+1]}},\\1&\text{otherwise},
\end{cases}
\end{equation}
for each $f_{n-1}\in F_{n-1}$.
The important fact about this choice is that its defect set is exactly the set of $n$-faces $((v,0)+\{0\}^{[d]}\times\{0,1\}^{[d+n]\setminus[d]})_{v\in\Z^d}$ (other boundary conditions inducing the same defect set are also suitable). We note that $\rho^{\text{surface}}=\rho^{\Dob}$ when $n=1$. The problem of localization is then to decide whether the surface dual to the defects of $\sigma$ stays localized in the infinite-volume limit with $\rho^{\text{surface}}$ boundary conditions (i.e., to show that there are defects in the neighborhood of the origin with high probability, uniformly in the volume). Borgs~\cite{B84} considered the case $d=2, n=2$ and proved that localization occurs at low temperature (this is the so-called weak coupling regime). His results apply more generally when the (gauge) group $\{-1,1\}$ is replaced by a finite Abelian group.

The result of Borgs~\cite{B84} is analogous to the result of Dobrushin~\cite{Dob} in that he establishes the existence of a non-translation-invariant Gibbs measure for the non-disordered model. In our context, it is natural to consider a disordered version, with formal Hamiltonian
\begin{equation}
H^{\text{gauge}, \eta}(\sigma) := -\sum_{f_n\in F_n}\eta_{f_n}\sigma(f_n)
\end{equation}
with the $(\eta_{f_n})_{f_n\in F_n}$ sampled independently from a disorder distribution supported in $[0,\infty)$ (possibly satisfying additional assumptions). We highlight two special cases: (i) the case $n=1$ is the disordered ferromagnet studied in this work, (ii) for $d=1$, the defect surface of the finite-volume ground configuration with $\rho^{\text{surface}}$ boundary conditions is dual to a geodesic in first-passage percolation (in finite volume) in $\Z^{1+n}$.

In analogy with Question~\ref{q:non-constant ground configurations} and Question~\ref{q:localization of Dobrushin interface} we may ask whether the disordered model admits ground configurations with non-empty defect set and whether the ground configuration with $\rho^{\text{surface}}$ boundary conditions induces a localized surface. Regarding the second question, we believe that the localization/delocalization behavior is mostly determined by $d$ (though the quantitative delocalization magnitude depends also on $n$). In particular, we expect that when $d\ge 3$ an analogue of our results holds for each $n\ge 1$. 
Regarding the first question, it seems natural that the existence of ground configurations with non-empty defect set becomes easier as $n$ increases. We thus expect such configurations to exist (under mild assumptions on the disorder distribution) for all $d\ge 2$ and $n\ge 1$. For $d=1$, the question coincides with the open problem of whether infinite bigeodesics exist in first-passage percolation on $\Z^{1+n}$, where a negative answer is expected for $n=1$ but the situation for larger $n$ is unclear.

\subsection{More general disorder distributions} 

Our main results are proved for non-atomic disorder distributions (though in the anisotropic setup atoms are allowed for $\nu^\perp$)
with a strictly positive lower bound on their support, which are sufficiently concentrated in the sense that their ``width'', as defined in~\eqref{eq:width dist def}, is sufficiently small. 

Our notion of width allows either for compactly-supported distributions or distributions which are Lipschitz images of the standard Gaussian distribution. In fact, our only use of the concentration properties of the distribution is through the concentration inequality of Corollary~\ref{cor:concentration with bounded width}. Thus, our proof may be used for more general classes of distributions satisfying a similar concentration inequality; see~\cites{GRST17, GRSST18}.

In addition, our proof of the localization of the Dobrushin interface (Theorem~\ref{thm:localization}) applies also when the disorder variables $(\eta_e)$ are sampled independently from different disorder distributions, as long as the same disorder distribution is used for edges in the same ``column'' ($e_1,e_2$ are in the same column if $e_1 = e_2+(0,\ldots,0,k)$ for some $k\in\Z$), and our assumptions~\eqref{eq:disorder distributions assumptions} and~\eqref{eq:anisotropic condition} (for a sufficiently small $c_0>0$) are satisfied for each pair of parallel and perpendicular distributions.

The assumption that the disorder distribution is non-atomic is imposed only to ensure the uniqueness of \emph{finite-volume} ground configurations. We expect suitable versions of Theorem~\ref{thm:localization} and Theorem~\ref{thm:Existence of admissible ensemble is unlikely} to hold also in its absence, with minor adjustments to the proofs.

We also expect the results of this paper to continue to hold for some classes of disorder distributions $\nu$ with $\min(\supp(\nu))=0$. However, the assumption that $\min(\supp(\nu))>0$ is used more heavily in our proofs.

\subsection*{Acknowledgements}
We are grateful to Michael Aizenman, Sky Cao, Daniel S. Fisher, Reza Gheissari and David Huse for illuminating discussions. We thank Daniel Hadas and Sasha Sodin for helpful comments.

We are grateful to two anonymous referees for their thorough reading of our paper and for many astute suggestions which greatly improved the presentation.

The research of M.B. and R.P. is supported by the Israel Science Foundation grant
1971/19 and by the European Research Council Consolidator grant 101002733 (Transitions). Part of this work was completed while R.P. was a Cynthia and Robert Hillas Founders' Circle Member of the Institute for Advanced Study and a visiting fellow at the Mathematics Department of Princeton University. R.P. is grateful for their support.

\appendix

\section{Equivalence between notions of primitive contours in~$\Z^d$}
\label{app:BB}

For every $A\subseteq{\mathbb Z}^d$, 
let $\tilde{\partial}A:=\{\{u,v\}\colon (u,v)\in\partial A\}$.
Recall the definition of the graph ${\mathcal G}_d$ from Lemma \ref{lem:Timar1}.
Following \cites{LM,BB}, 
a set $E\subseteq E({\mathbb Z}^d)$ is called a \emph{contour} if it is connected in ${\mathcal G}_d$ and there is a finite set $A\subseteq{\mathbb Z}^d$ such that $E=\tilde{\partial}A$.
A contour is \emph{primitive} if it is not a disjoint union of two non-empty contours.
Let $\tilde{\mathbb B}_d:=\{A\subset{\mathbb Z}^d\colon A \text{ is finite and } \tilde{\partial}A \text{ is a primitive contour}\}$.
Recall that the family of finite $A\subset{\mathbb Z}^d$ such that both $A$ and ${\mathbb Z}^d\setminus A$ are connected in ${\mathbb Z}^d$ was denoted ${\mathbb B}_d$ in the proof of Proposition \ref{prop:enumerate_shift_functions}. 

The claim of \cite{BB}*{Theorem 6} is that $\lvert\{A\in \tilde{\mathbb B}_d\colon 0\in A,\,\lvert\partial A\rvert=b\rvert\leq (8d)^{2b/d}$ for every (even) integer $b\geq 2d$.
This is equivalent to \eqref{eq:BB} in light of the following proposition.

\begin{proposition}
$\tilde{\mathbb B}_d={\mathbb B}_d$.
\end{proposition}

\begin{proof}
First note that if $A\in{\mathbb B}_d$ then  $\tilde{\partial}A$ is connected in ${\mathcal G}_d$, 
by Lemma \ref{lem:Timar1}, 
and hence $\tilde{\partial}A$ is a contour.

Let $A\in\tilde{\mathbb B}_d$. 
Since $A$ is finite, 
the set ${\mathbb Z}^d\setminus A$ obviously has a unique infinite connected component. 
Let ${\mathcal A}$ be the collection of connected components of $A$ and finite connected components of ${\mathbb Z}^d\setminus A$.
Define a partial order $\preceq$ on the set $\mathcal A$ as follows: $C_1\preceq C_2$ if every path from $C_1$ to $\infty$ necessarily intersects $C_2$.
For every $C\in{\mathcal A}$, 
let $\bar{C}:=\cup_{S\preceq C}S$.
It is easy to see that  $\tilde{\partial}A$ is the disjoint union of the sets $\{\tilde{\partial}\bar{C}\}_{C\in{\mathcal A}}$ and that for every $C\in{\mathcal A}$ it holds that $\bar{C}\in{\mathbb B}_d$ and hence $\tilde{\partial}\bar{C}$ is a contour. 
Since $\tilde{\partial}A$ is a primitive contour, it follows that $\lvert{\mathcal A}\rvert=1$ and hence, $A\in{\mathbb B}_d$, as $A$ is finite.
Thus $\tilde{\mathbb B}_d\subseteq{\mathbb B}_d$.

By way of contradiction, 
assume there is $A\in{\mathbb B}_d\setminus\tilde{\mathbb B}_d$. 
Since $A\in{\mathbb B}_d$, 
the set $\tilde{\partial}A$ is a contour.
Therefore, since $A\notin\tilde{\mathbb B}_d$, 
it follows that $\tilde{\partial}A$ is not primitive, 
i.e., it is the disjoint union of two non-empty contours. 
In particular, there is a finite $A_0\subset{\mathbb Z}^d$ such that $\emptyset\neq\tilde{\partial}A_0\subsetneq \tilde{\partial}A$.
Since $A$ and $A_0$ are both finite, 
the set $({\mathbb Z}^d\setminus A)\cap({\mathbb Z}^d\setminus A_0)$ is infinite and in particular, non-empty.
Hence, since $({\mathbb Z}^d\setminus A)$ is connected and $\tilde{\partial}A_0\subset \tilde{\partial}A$, 
it follows that $({\mathbb Z}^d\setminus A)\subseteq({\mathbb Z}^d\setminus A_0)$, i.e., $A_0\subseteq A$.
The set $A\cap A_0=A_0$ is non-empty, since $\tilde{\partial}A_0\neq\emptyset$.
Hence, since $A$ is connected and $\tilde{\partial}A_0\subset \tilde{\partial}A$, it follows that $A_0\supseteq A$.
Therefore, $A_0=A$, in contradiction with $\tilde{\partial}A_0\neq \tilde{\partial}A$.
\end{proof}

\section{Passing to semi-infinite volume}\label{sec:small lemmas}

In this section we prove Lemma~\ref{lem:semi infinite volume ground configuration}, Observation ~\ref{obs:different Hamiltonians equivalent}, Lemma~\ref{lem:finite number of admissible shifts}, and Lemma~\ref{lem: bounded_layering_finite_to_semi_infinite_convergence}, which are concerned with properties of ground configurations in semi-infinite volume, $\Lambda\times\Z$ for $\Lambda\subset\Z^d$ finite.

The following observation will be used in the proof of Lemma \ref{lem: bounded_layering_finite_to_semi_infinite_convergence} as well as Lemma \ref{lem:finite number of admissible shifts}. 

\begin{obs}\label{obs: exponential moment exists to lip of gaussian}
Let $\eta=f(X)$ be Lipschitz function $f$ of a normalized Gaussian random variable $X\sim N(0,1)$. Then $\eta$ has exponential moments, i.e., $\E(e^{\eta} ) <\infty$.
\end{obs}

\begin{proof}
Denote by $C>0$ a constant such that $f$ is $C$-Lipschitz. Then the following holds:
\begin{align*}
\E(e^{\eta} )&= \frac{1}{\sqrt{2\pi}}\int_{x=-\infty}^{\infty} e^{f(x)} e^{-\frac{x^2}{2}} dx =
\frac{1}{\sqrt{2\pi}}\int_{x=0}^{\infty} 
\left( e^{f(x)}+e^{f(-x)} \right) e^{-\frac{x^2}{2}} dx \\
&\leq \frac{e^{|f(0)|}}{\sqrt{2\pi}}\int_{x=0}^{\infty} 
\left( e^{|f(x)-f(0)|}+e^{|f(-x)-f(0)|} \right) e^{-\frac{x^2}{2}} dx 
\leq \frac{e^{|f(0)|}}{\sqrt{2\pi}}\int_{x=0}^{\infty} 
e^{Cx} e^{-\frac{x^2}{2}} dx\\
&=\frac{e^{|f(0)|+\frac{C^2}{2}}}{\sqrt{2\pi}}\int_{x=0}^{\infty} 
e^{-\frac{(x-C)^2}{2}} dx<\frac{e^{|f(0)|+\frac{C^2}{2}}}{\sqrt{2\pi}}\int_{x=-\infty}^{\infty} 
e^{-\frac{(x-C)^2}{2}} dx=e^{|f(0)|+\frac{C^2}{2}},
\end{align*}
where the first equality is by definition of $\eta$, the first inequality by triangle inequality, the second by definition of $C$-Lipschitz function. 
\end{proof}

\begin{proposition} \label{prop:ground energy of ground configuration almost surely finite}
Let $\Lambda\subset \Z^{d}$ finite. Suppose the disorder distributions $\nu^{\parallel},\nu^{\perp}$ satisfies the condition in \eqref{eq:disorder distributions assumptions}.  For every integer $h$ and positive integer  $M$:
\begin{equation}\label{eq:exp_decay_h}
\P \left(\sum_{v\in \Lambda} \eta_{\{(v,h),(v,h+1)\}} \geq M\right) \leq \beta e^{-M},
\end{equation}
where $\beta= \beta(|\Lambda|,\nu^{\parallel})$ is a constant which depends on $|\Lambda|,\nu^{\parallel}$.
In particular, for $h=0$, it holds for every positive integer $M$ that
\begin{equation}\label{eq:exp_decay_0}
\P \left(\mathcal{H}^{\eta,\Lambda} (\rho^{\Dob}) \geq M\right) \leq \beta e^{-M/2}.
\end{equation}
Consequently, by the first Borel Cantelli lemma it holds that the configuration $\rho^{\Dob}$ has finite energy a.s and so the ground energy is a.s finite. 
\end{proposition}
\begin{proof}
Defining $\beta(|\Lambda|,\nu^{\parallel}):= \E(e^{X})^{|\Lambda|}<\infty$, for $X\sim \nu^{\parallel}$ and where the expectation is finite trivially for the compact case, and by Observation~\ref{obs: exponential moment exists to lip of gaussian} for the Lipschitz of Gaussian case. By the i.i.d nature of the disorders the following holds 
\begin{align*}
\P \left( \sum_{v\in \Lambda} \eta_{\{(v,h),(v,h+1)\}}\geq M \right)  &=
\P \left( e^{ \sum_{v\in \Lambda} \eta_{\{(v,h),(v,h+1)\}}}  
\geq e^{ M} \right) \leq \frac{\E \left( e^{ \sum_{v\in \Lambda} \eta_{\{(v,h),(v,h+1)\}}} \right)}{e^{M}}\\
&=\frac{\prod_{v\in\Lambda}\E (e^{\eta_{\{(v,h),(v,h+1)\}}})}{e^{M}} = \beta e^{- M}.
\qedhere\end{align*}
\end{proof}

The proposition below states that from the appearance of a sign change in a ground configuration in height $k$ one may deduce a lower bound on the energy of it. 

\begin{proposition}\label{prop: high face of the interface in ground configuration then high ground energy}
Let $\eta: E(\Z^{d+1}) \rightarrow [0,\infty)$ such that  $\eta_{e}\geq {\alpha}^{\perp}$ for every $e\in E^{\perp}(\Z^{d+1})$. Let $\Lambda\subset{\mathbb Z}^d$ be finite and let $\sigma\in\Omega^{\Lambda,\Dob}$ such that both the sets $\{x\in{\mathbb Z}^{d+1}\colon \sigma(x)=1\}$ and $\{y\in{\mathbb Z}^{d+1}\colon \sigma(y)=-1\}$ are connected. 
If $(u,h)\in\Lambda\times{\mathbb Z}$ such that $\sigma_{(u,h)}\neq \rho^{\Dob}_{(u,h)}$ 
then $\mathcal{H}^{\eta,\Lambda}(\sigma ) \geq 2\alpha^{\perp}\lvert h\rvert$.
\end{proposition}

\begin{proof}
The proof is similar to (though considerably simpler than) the argument used to prove Lemma~\ref{lem:spin different than sign shift existence}.

With no loss of generality, assume that $h>0$ and that $h=\max\{k\in{\mathbb Z}\colon \sigma_{(u,k)}=-1\}$.

Recall the definition of the graph ${\mathcal G}_{d+1}$ from Lemma \ref{lem:Timar1}.
Define $\pi_{d+1}:{\mathbb Z}^d\times{\mathbb Z}\to{\mathbb Z}$ by $\pi_{d+1}(u,h)=h$, and for every $\{x,y\}\in E({\mathbb Z}^{d+1})$, let $\pi_{d+1}(\{x,y\})=\{\pi_{d+1}(x),\pi_{d+1}(y)\}$.
Note that if $e,\tilde{e}\in E({\mathbb Z}^{d+1})$ are adjacent in ${\mathcal G}_{d+1}$, then $\pi_{d+1}(e)\subseteq\pi_{d+1}(\tilde{e})$ or $\pi_{d+1}(e)\supseteq\pi_{d+1}(\tilde{e})$.

By Lemma \ref{lem:Timar1}. the set 
$${\mathcal I}(\sigma):=\{\{x,y\}\in E({\mathbb Z}^{d+1})\colon \sigma(x)=1, \sigma(y)=-1\}$$ 
is connected in ${\mathcal G}_{d+1}$.
In particular, there is a sequence $(\tilde{e}_i)_{i=1}^N$ of edges in ${\mathcal I}(\sigma)$ such that $\tilde{e}_1=\{(u,h+1),(u,h)\}$, $\tilde{e}_N=\{(w,1),(w,0)\}$ for some $w\in {\mathbb Z}^d\setminus\Lambda$ and $\tilde{e}_i, \tilde{e}_{i+1}$ are adjacent in ${\mathcal G}_{d+1}$ for every $1\leq i<N$. 

Since $\pi_{d+1}(\tilde{e}_1)=\{h+1,h\}$, $\pi_{d+1}(\tilde{e}_N)=\{1,0\}$ and for every $1\leq i<N$ it holds that $\pi_{d+1}(\tilde{e}_{i-1})\subseteq\pi_{d+1}(\tilde{e}_i)$ or $\pi_{d+1}(\tilde{e}_{i-1})\supseteq\pi_{d+1}(\tilde{e}_i)$, it follows that for every $1\leq k\leq h$ there is necessarily $1< i_h< N$ such that $\pi_{d+1}(\tilde{e}_{i_k})=\{k\}$. 
Then, for every $1\leq k\leq h$ it holds that
$\tilde{e}_{i_k}\in {\mathcal I}(\sigma)\cap E^{\perp}({\mathbb Z}^{d+1})$ and $\tilde{e}_{i_k}\cap(\Lambda\times{\mathbb Z})\neq\emptyset$, and hence 
\begin{equation*}
\mathcal{H}^{\eta,\Lambda}(\sigma )=2\sum_{\substack{e\in {\mathcal I}(\sigma)\\e\cap(\Lambda\times{\mathbb Z})\neq\emptyset}}\eta_e \geq 2\sum_{k=1}^h\eta_{\tilde{e}_{i_k}}\geq 2h \alpha^{\perp}.
\qedhere\end{equation*}
\end{proof}

\begin{proof}[Proof of Lemma \ref{lem: bounded_layering_finite_to_semi_infinite_convergence}]
First notice that for any $M>0$
\begin{align*}
\P \Big(G^{\eta,\Delta_M, (b^{\parallel},b^{\perp})}&(\tau) \neq G^{\eta,\Lambda,(b^{\parallel},b^{\perp})}(\tau) \Big) \\
\leq &\P \left(\GE^{\Delta_M,\supp(\tau), (b^{\parallel},b^{\perp})}(\eta) \neq \GE^{\Lambda,\supp(\tau), (b^{\parallel},b^{\perp})}(\eta) \right)\\
&+\P \left(\GE^{\Delta_M,\supp(\tau), (b^{\parallel},b^{\perp})}(\eta^{\tau}) \neq \GE^{\Lambda,\supp(\tau), (b^{\parallel},b^{\perp})}(\eta^{\tau}) \right)\\
=& 2\P \left(\GE^{\Delta_M,\supp(\tau), (b^{\parallel},b^{\perp})}(\eta) \neq \GE^{\Lambda,\supp(\tau), (b^{\parallel},b^{\perp})}(\eta) \right),
\end{align*}
where the inequality is by union bound and equality afterwards by the fact the disorder is i.i.d and $\tau$ is fixed.

The set $\{x\in{\mathbb Z}^{d+1}\colon \sigma(x)=1\}$ obviously has a unique infinite connected component for every $\sigma\in\Omega^{\Lambda,\supp(\tau),(b^{\parallel},b^{\perp}),\Dob}$.
If the set $\{x\in{\mathbb Z}^{d+1}\colon \sigma(x)=1\}$ has  finite connected components then 
flipping all signs in such a component yields a configuration in $\Omega^{\Lambda,\supp(\tau),(b^{\parallel},b^{\perp}),\Dob}$
with smaller $H^{\eta}$.
Hence, if $\sigma_0\in\Omega^{\Lambda,\supp(\tau),(b^{\parallel},b^{\perp}),\Dob}$  is a configuration minimizing $\mathcal{H}^{\eta,\Lambda}$, then the set $\{x\in{\mathbb Z}^{d+1}\colon \sigma_0(x)=1\}$ is necessarily connected, and similarly, the set $\{y\in{\mathbb Z}^{d+1}\colon \sigma_0(y)=1\}$ is connected as well. 
Hence, if $\sigma_0\notin\Omega^{\Delta_M,\supp(\tau),(b^{\parallel},b^{\perp}),\Dob}$ then $\mathcal{H}^{\eta,\Lambda}(\sigma ) \geq 2\alpha^{\perp}(M+1)$, by Proposition~\ref{prop: high face of the interface in ground configuration then high ground energy}. It follows that
\begin{multline*}
\{\eta:\GE^{\Delta_M,\supp(\tau), (b^{\parallel},b^{\perp})}(\eta) \neq \GE^{\Lambda,\supp(\tau) (b^{\parallel},b^{\perp})}(\eta) \} \\
\subseteq\{\eta: \GE^{\Lambda,\supp(\tau), (b^{\parallel},b^{\perp})}(\eta) \geq 2\underline{\alpha}^{\perp}  (M+1)\} \subseteq \{\eta: \mathcal{H}^{\eta,\Lambda}(\rho^{\Dob}) \geq 2\underline{\alpha}^{\perp}(M+1) \}
\end{multline*}
where the second inclusion is by definition of ground energy, since $\rho^{\Dob}\in\Omega^{\Lambda,\supp(\tau),(b^{\parallel},b^{\perp})}$.
And so by the above inclusion, and \eqref{eq:exp_decay_0},
\begin{equation*}
\P \left(G^{\eta,\Delta_M, (b^{\parallel},b^{\perp})}(\tau) \neq G^{\eta,\Lambda,(b^{\parallel},b^{\perp})}(\tau) \right) \leq 2\beta e^{- \underline{\alpha}^{\perp}(M+1)} \xrightarrow[M\rightarrow\infty]{}0
\end{equation*}
as required.
\end{proof}

\begin{proof}[Proof of Observation \ref{obs:different Hamiltonians equivalent}]
By definition of $\Omega^{\Lambda,\Dob}$ the configurations $\sigma,\sigma'$ may only differ in finitely many places, and so the difference $H^{\eta}(\sigma)-H^{\eta}(\sigma')$ is well defined and denoting $D=\{x: \sigma_x \neq \sigma'_x \}$, both
\begin{equation*}
H^{\eta}(\sigma)- 
H^{\eta}(\sigma')=-2\sum_{\{ x,y\}\in \partial D } \eta_{\{x,y\} } \sigma_x \sigma_y  
\end{equation*}
and also
\begin{equation*}
\mathcal{H}^{\eta,\Lambda}(\sigma)- 
\mathcal{H}^{\eta, \Lambda}(\sigma')=\sum_{\{ x,y\}}  \eta_{\{x,y\} } \left(\sigma'_{x}\sigma'_y-\sigma_{x}\sigma_{y}) \right)=-2\sum_{\{ x,y\}\in \partial D}  \eta_{\{x,y\} } \sigma_{x}\sigma_{y}.
\qedhere
\end{equation*}
\end{proof}

\begin{proof}[Proof of Lemma \ref{lem:semi infinite volume ground configuration}] 
Let $M$ be an integer larger than $\frac{1}{\underline{\alpha}^{\perp}}\mathcal{H}^{\eta,\Lambda} (\rho^{\Dob})$ (which is finite by Proposition \ref{prop:ground energy of ground configuration almost surely finite}). Let $\Delta_M:=\Lambda \times \{-M,\dots, M\}$.
The function $\mathcal{H}^{\eta,\Lambda}$ is well defined on the finite $\Omega^{\Delta_M,\rho^{\Dob}}\subset\Omega^{\Lambda,\Dob}$.
Thus, there is a ground configuration $\sigma^{\eta,\Delta_M,\Dob}$ in $\Omega^{\Delta_M,\rho^{\Dob}}$ with respect to the Hamiltonian $\mathcal{H}^{\eta,\Lambda}$, which is unique by condition \eqref{eq:no finite binary zero combination}.
By Observation \ref{obs:different Hamiltonians equivalent}, $\sigma^{\eta,\Delta_M,\Dob}$ is also the unique ground configuration in $\Omega^{\Delta_M,\rho^{\Dob}}$ with respect to the Hamiltonian~$H^{\eta}$.

Consider a configuration $\sigma\in\Omega^{\Lambda,\Dob}\setminus\Omega^{\Delta_M,\rho^{\Dob}}$.
Each of the sets $\{x\in{\mathbb Z}^{d+1}\colon \sigma(x)=1\}$, $\{y\in{\mathbb Z}^{d+1}\colon \sigma(y)=-1\}$ obviously has a unique infinite connected component.
If either of the sets $\{x\in{\mathbb Z}^{d+1}\colon \sigma(x)=1\}$, $\{y\in{\mathbb Z}^{d+1}\colon \sigma(y)=-1\}$ has a finite connected component, then 
flipping all signs in such a component yields a configuration $\tilde{\sigma}\in\Omega^{\Lambda,\Dob}$ such that $H^{\eta}(\sigma)- H^{\eta}(\tilde{\sigma}) >0$.
If neither of the sets $\{x\in{\mathbb Z}^{d+1}\colon \sigma(x)=1\}$, $\{y\in{\mathbb Z}^{d+1}\colon \sigma(y)=-1\}$ has finite connected components, then both sets are connected and hence, by Observation \ref{obs:different Hamiltonians equivalent} and Proposition \ref{prop: high face of the interface in ground configuration then high ground energy},
$$H^{\eta}(\sigma)-H^{\eta}(\rho^{\Dob}) =
\mathcal{H}^{\eta,\Lambda}(\sigma)-\mathcal{H}^{\eta, \Lambda}(\rho^{\Dob})\geq \underline{\alpha}^{\perp}(M+1)-\mathcal{H}^{\eta, \Lambda}(\rho^{\Dob})>0.$$
Therefore, no $\sigma\in\Omega^{\Lambda,\Dob}\setminus\Omega^{\Delta_M,\rho^{\Dob}}$ is a ground configuration in $\Omega^{\Lambda,\Dob}$ and hence $\sigma^{\eta,\Delta_M,\Dob}$ is also the unique ground configuration in $\Omega^{\Lambda,\Dob}$.

Finally, since every configuration that differs from 
$\sigma^{\eta,\Delta_M,\Dob}$ in finitely many places is necessarily in $\Omega^{\Lambda,\Dob}$, it follows that $\sigma^{\eta,\Delta_M,\Dob}$ is also a ground configuration in $\Omega^{\Lambda\times{\mathbb Z},\rho^{\Dob}}$; 
it is also almost surely unique, since almost surely no configuration in $\Omega^{\Lambda\times{\mathbb Z},\rho^{\Dob}}\setminus\Omega^{\Lambda,\Dob}$ is a ground configuration in $\Omega^{\Lambda\times{\mathbb Z},\rho^{\Dob}}$ with respect to the Hamiltonian $H^{\eta}$, as we will prove next.

Take $\gamma>\ln\beta(|\Lambda|,\nu^{\parallel})$ (see Proposition \ref{prop:ground energy of ground configuration almost surely finite}), and let $$J:=\{j\in{\mathbb Z} :\sum_{v\in \Lambda} \eta_{\{(v,j),(v,j+1)\}} <\gamma\}.$$
We will show that if the set $J$ is infinite, then no configuration in $\Omega^{\Lambda\times{\mathbb Z},\rho^{\Dob}}\setminus\Omega^{\Lambda,\Dob}$ is a ground configuration in $\Omega^{\Lambda\times{\mathbb Z},\rho^{\Dob}}$ with respect to the Hamiltonian $H^{\eta}$.
Since the disorder $\eta$ is an independent field, it readily follows from \eqref{eq:exp_decay_h} that $J$ is almost surely infinite, and we are done.

Therefore, assume that the set $J$ is infinite and let $\sigma\in\Omega^{\Lambda\times{\mathbb Z},\rho^{\Dob}}\setminus\Omega^{\Lambda,\Dob}$. 
Then, the set $I:=\{i\in{\mathbb Z}\colon\exists u\in{\mathbb Z}^d \text{ such that } \sigma_{(u,i)}\neq\rho^{\Dob}_{(u,i)} \}$ is infinite,
and hence there are $0\leq j_1<j_2$ or $j_1<j_2\leq 0$ in $J$ such that $\lvert I\cap[j_1+1,j_2]\rvert>\frac{\gamma}{ \underline{\alpha}^{\perp}d}$.
Consider the configuration $\tilde{\sigma}$, defined as follows: for every $u\in{\mathbb Z}^d$, $\tilde{\sigma}_{(u,i)}=\rho^{\Dob}_{(u,i)}$ if $j_1+1\leq i\leq j_2$ and $\tilde{\sigma}_{(u,i)}=\sigma_{(u,i)}$ otherwise.
Clearly, $\tilde{\sigma}\in\Omega^{\Lambda\times{\mathbb Z},\rho^{\Dob}}$, $\{x\in{\mathbb Z}^{d+1}\colon\tilde{\sigma}_x\neq\sigma_x\}$ is finite and
\begin{align*}
H^{\eta}(\sigma)-H^{\eta}(\tilde{\sigma})\geq & 4\underline{\alpha}^{\perp}d\lvert I\cap[j_1+1,j_2]\rvert-2\sum_{v\in \Lambda} \eta_{\{(v,j_1),(v,j_1+1)\}} -2\sum_{v\in \Lambda} \eta_{\{(v,j_2),(v,j_2+1)\}} \\
>&4\gamma-4\gamma=0.
\end{align*}
Hence, $\sigma$ is not a ground configuration, as desired. 
\end{proof}

\begin{proof}[Proof of Lemma ~\ref{lem:finite number of admissible shifts}] 
First notice that for any $\tau\in \mathcal{S} \cap \{\tau: \max_{u\in \Lambda} |\tau(u)| = r \} $ it holds that $\TV(\tau) \geq 2dr$, and so
\begin{equation}\label{eq: big admissible shift than big ground energy}
\{\eta: \tau\in \AS ^{\eta,\Lambda}(\alpha^{\parallel},\alpha^{\perp}) \}\subset 
\{\eta: |G^{\eta,\Lambda}(\tau)|\geq d {\alpha}^{\perp}\}
\end{equation}
by the definition of $({\alpha}^{\parallel},{\alpha}^{\perp})$-admissibility.

Now use union bound and \eqref{eq: big admissible shift than big ground energy} to get
\begin{align*}
\P &\left(\exists \tau \in \AS^{\eta,\Lambda}(\alpha^{\parallel},\alpha^{\perp}) \colon \max_{u\in \Lambda}|\tau(u)|=r \right) 
\leq (2r+1)^{|\Lambda|}
\max_{\mysubstack{\tau\in \mathcal{S}}{\max_{u\in \Lambda}|\tau(u)|=r} }\P (\tau \in \AS)\\
&\leq (2r+1)^{|\Lambda|} \P\left(|G^{\eta,\Lambda}(\tau)|\geq d {\alpha}^{\perp}  \right) \leq 2(2r+1)^{|\Lambda|} \P\left(\GE^{\Lambda,\Dob}(\eta)\geq d {\alpha}^{\perp}  \right)\\
&\leq 2(2r+1)^{|\Lambda|} 
\P\left(\mathcal{H}^{\eta,\Lambda}(\rho^{\Dob}) \geq d {\alpha}^{\perp} \right)\leq \beta 2(2r+1)^{|\Lambda|} e^{-d {\alpha}^{\perp}/2} 
\end{align*}
where the first inequality is by union bound, the second is by \eqref{eq: big admissible shift than big ground energy}, the third is by the fact the parallel disorder is i.i.d, the forth is by definition of ground energy, and the fifth by \eqref{eq:exp_decay_0}. 

Noticing that by the above $\P \left(\exists \tau \in \AS^{\eta,\Lambda}(\alpha^{\parallel},\alpha^{\perp}) \colon \max_{u\in \Lambda}|\tau(u)|=r \right)$ is summable with respect to $r$ and applying Borel--Cantelli one gets that with probability one finitely many of the events 
$$
\{\exists \tau \in \AS^{\eta,\Lambda}(\alpha^{\parallel},\alpha^{\perp}) \colon \max_{u\in \Lambda}|\tau(u)|=r \} 
$$
and in particular with probability one $|\mathcal{AS}^{\eta,\Lambda}(\alpha^{\parallel},\alpha^{\perp})|<\infty$ as required.
\end{proof}

\end{document}